%% file: main.tex
\documentclass[a4paper,noarxiv,onecolumn,accepted=2024-07-11]{quantumarticle}

\usepackage[top=1 in,bottom=1in, left=1 in, right=1 in]{geometry}

\usepackage[english]{babel}
\usepackage{color}
\usepackage{graphicx}
\usepackage{framed}
\usepackage[normalem]{ulem}
\usepackage{mathtools}
\usepackage{amsmath}
\usepackage{amsthm}
\usepackage{amssymb}
\usepackage{amsfonts}
\usepackage{bbm}
\usepackage{enumerate}
\usepackage{algorithm}
\usepackage[noend]{algpseudocode}
\usepackage[utf8]{inputenc}
\usepackage{etoolbox}
\usepackage{multirow}
\usepackage{appendix}
\usepackage{stmaryrd}
\usepackage{hyperref}
\usepackage{url}

\newcounter{ex}

\theoremstyle{plain}
\newtheorem{theorem}{Theorem}[section]
\newtheorem{lemma}[theorem]{Lemma}

\newtheorem{corollary}[theorem]{Corollary}

\newtheorem{example}[ex]{Example}

\theoremstyle{definition}
\newtheorem{definition}{Definition}[section]

\makeatletter
\renewenvironment{proof}[1][\proofname]{\par
  \vspace{-\topsep}
  \pushQED{\qed}%
  \normalfont
  \topsep3pt \partopsep3pt 
  \trivlist
  \item[\hskip\labelsep
        \itshape
    #1\@addpunct{.}]\ignorespaces
}{%
  \popQED\endtrivlist\@endpefalse
  \addvspace{6pt plus 6pt} 
}
\makeatother

\newcommand{\setmargins}[0]{\setlength{\itemsep}{-16pt}\setlength{\parsep}{0pt}\setlength{\parskip}{-4pt}}

\makeatletter
\newtheorem*{rep@theorem}{\rep@title}
\newcommand{\newreptheorem}[2]{%
\newenvironment{rep#1}[1]{%
 \def\rep@title{#2 \ref{##1}}%
 \begin{rep@theorem}}%
 {\end{rep@theorem}}}
\makeatother

\newreptheorem{theorem}{Theorem}
\newreptheorem{lemma}{Lemma}
\newreptheorem{corollary}{Corollary}

\newcommand{\ket}[1]{\ensuremath{\vert{#1}\rangle}}

\newcommand{\Tr}[1]{\ensuremath{\mathrm{Tr}\left({#1}\right)}}
\newcommand{\supp}[1]{\ensuremath{\mathrm{supp}\left({#1}\right)}}

\newcommand{\sys}[1]{\ensuremath{\mathrm{sys}\left({#1}\right)}}
\newcommand{\hsys}[1]{\ensuremath{\mathrm{hsys}\left({#1}\right)}}
\newcommand{\rank}[1]{\ensuremath{\mathrm{rank}\left({#1}\right)}}

\newcommand{\pauligrp}[1]{\mathcal{\hat{P}}_{#1}}
\newcommand{\ham}[1]{\text{Ham}(#1)}
\newcommand{\com}[2]{\text{Com}(#1,#2)}
\newcommand{\commap}[2]{\text{Com}_{#1}(#2)}

\newcommand{\edit}[1]{{\color{black}{#1}}}

\DeclarePairedDelimiter\floor{\lfloor}{\rfloor}
\DeclarePairedDelimiter\ceil{\lceil}{\rceil}

\DeclareMathOperator{\lcm}{lcm}

\title{A graph-based formalism for surface codes and twists}

\author{Rahul Sarkar}
\address{Institute for Computational and Mathematical Engineering, Stanford University, Stanford, CA 94305}
\email{rsarkar@stanford.edu}
\thanks{}

\author{Theodore J.~Yoder}
\address{IBM T.J. Watson Research Center, Yorktown Heights, NY}
\email{ted.yoder@ibm.com}
\thanks{}

\date{}

\begin{document}

\maketitle

\begin{abstract}
Twist defects in surface codes can be used to encode more logical qubits, improve the code rate, and implement logical gates. In this work we provide a rigorous formalism for constructing surface codes with twists generalizing the well-defined homological formalism introduced by Kitaev for describing CSS surface codes. In particular, we associate a surface code to \emph{any} graph $G$ embedded on \emph{any} 2D-manifold, in such a way that (1) qubits are associated to the vertices of the graph, (2) stabilizers are associated to faces, (3) twist defects are associated to odd-degree vertices. In this way, we are able to reproduce the variety of surface codes, with and without twists, in the literature and produce new examples. We also calculate and bound various code properties such as the rate and distance in terms of topological graph properties such as genus, systole, and face-width.
\end{abstract}

\tableofcontents

\newpage
\input{Introduction.tex}
\input{Rotation-systems-v2.tex}
\input{Majorana-codes-v2.tex}
\input{Logical-operators-v4.tex}

\input{Code-examples.tex}
\input{Open-problems.tex}

\section*{Acknowledgements}
We are grateful to a number of people for helpful discussions, including Sergey Bravyi, Andrew Cross, Carollan Helinski, Tomas Jochym-O'Connor, Alex Kubica, Andrew Landahl, and Guanyu Zhu. We would especially like to thank Leonid Pryadko and Alexey Kovalev for sharing ideas that led to the distance proof for cyclic toric codes. Rahul Sarkar was partially supported by the Schlumberger Innovation Fellowship. Ted Yoder was partially supported by the IBM Research Frontiers Institute.
\bibliographystyle{quantum}
\bibliography{bibliography}

\newpage
\input{Appendix.tex}

\end{document}

%% file: Introduction.tex
\section{Introduction}

Since their introduction \cite{kitaev2003fault}, surface codes have been an essential tool in the quantum engineer's fight against noise, combining a high threshold with relatively simple 2-dimensional qubit connectivity. These codes also exemplify a surprising connection between quantum error-correction and topology. In effect, one associates a CSS quantum code to a graph embedded on a 2-dimensional manifold by assigning qubits to edges and stabilizers to vertices and faces of the graph. The homology of the surface guarantees that the stabilizers commute and that non-trivial logical operators correspond to homologically non-trivial cycles in the graph or its dual. 

However, the simple 2-dimensional connectivity comes with a drawback. Like all 2-dimensional stabilizer codes, surface codes are limited by the Bravyi-Poulin-Terhal bound \cite{bravyi2010tradeoffs}, which says that a code family that uses $N$ qubits to encode $K$ qubits with distance $D$ must satisfy \edit{$cKD^2\le N$} for some constant $c$. For instance, the rotated surface code \cite{bombin2007optimal} achieves $c=1$ in the plane, while the square-lattice toric code \cite{kitaev2003fault} achieves $c=1$ on the torus.

Motivated in part by improving the constant $c$, twist defects have been introduced to surface codes \cite{bombin2010topological,kitaev2012models}. Fundamentally, twist defects are local regions where Pauli errors can create (or destroy) unpaired excitations of both types, $X$ and $Z$, say, whereas elsewhere in the code the parities of these excitations are conserved. However, twist defects appear in several quantitatively different forms. For example, sometimes weight-five stabilizers are involved \cite{bombin2010topological,kitaev2012models}, sometimes only weight-four as in the usual surface code \cite{yoder2017surface,landahl2020surface}, and sometimes twists are embodied by a combination of weight-two and weight-six stabilizers \cite{krishna2020topological}. Qualitatively, these defects behave the same, with $M$ defects adding $(M-2)/2$ logical qubits to the code and with non-trivial logical operators forming paths between pairs of defects. 

Clever placement and implementation of twist defects can lead to savings in qubit count by producing codes with $N=cKD^2$ for relatively small constant $c$. For instance, in the plane, $c=3/4$ \cite{yoder2017surface} and $c=1/2$ \cite{kesselring2018boundaries} have been achieved for surface codes, and $c=3/8$ for color codes \cite{kesselring2018boundaries}. However, there is not a rigorous formalism for constructing surface codes with twists comparable to the well-defined homological formalism introduced by Kitaev for describing CSS surface codes \cite{kitaev2003fault}. Such a formalism would be helpful for looking for better surface codes with perhaps even smaller constants $c$.

Our main goal in this work is to introduce such a formalism, capable of describing surface codes with and without twist defects in a unified manner. Again we associate a code to an embedded graph, but find it most natural to do this in a different way than the homological formalism described above. Instead of placing qubits on edges, we place them on vertices. Only faces of the graph support stabilizers in our description, and odd-degree vertices act as the twist defects. Placing qubits on vertices is not a substantially new idea (for instance, appearing for the degree-4 case in \cite{bombin2007optimal}), but we go beyond prior works in applying the construction to arbitrary embedded graphs, yielding new codes with improved parameters. \edit{We highlight several code families we discuss in this work in Table~\ref{tab:code_summary}.}

Let us briefly summarize the main results of each section. Section~\ref{sec:rotation-systems} rigorously defines embedded graphs and related notation used throughout the paper. Our description of choice for an embedded graph is via a rotation system. Effectively, this is a combinatoric description of the embedded graph, which abstracts away unnecessary details about how exactly the graph is drawn on the surface. What remains is just the adjacency of the graph components (e.g.~vertices, edges, and faces). The rotation system description is very similar to the formalisms used for hyperbolic codes in, for example, \cite{delfosse2013tradeoffs,breuckmann2016constructions} but slightly more general as it applies to non-regular graphs. Although the rotation system description may be overkill for individual codes, for which a picture may suffice, we believe that creating a more comprehensive theory of surface codes deserves such a precise framework. Having the rotation system description is also convenient for computer-assisted searches for codes, as we do for hyperbolic manifolds in Section~\ref{sec:improved_hyperbolic_codes}. We identify a key property of embedded graphs called checkerboardability, which means one can two-color faces of the graph such that adjacent faces are differently colored.

In Section~\ref{sec:Majorana_and_qubit_codes}, we give our surface code construction. We actually provide two equivalent descriptions, one using Majoranas to encode qubits and one using qubits to encode qubits. The Majorana description is similar to that used to describe surface codes in \cite{wen2003quantum,kitaev2006anyons,bravyi2018correcting}. We place two Majoranas on each edge and a single Majorana at each odd-degree vertex. Stabilizers, products of Majoranas, are associated to faces and vertices of the graph. Identifying qubit subspaces at the vertices turns the Majorana code into a qubit surface code with qubits on the vertices, as we described above. A significant result in this section is the computation of the number of encoded qubits in Theorem~\ref{thm:number_encoded_qubits}. This depends on the genus of the surface, its orientability, the number of odd-degree vertices, and the checkerboardability of the surface. Finally, we note that, when the graph is checkerboardable, Kitaev's homological construction, applied to a different but related graph, can be used to describe the same code.

Our focus in Section~\ref{sec:locating_logical_operators} is on finding logical operators and bounding the code distance. A key idea here is to derive a ``decoding" graph from the original embedded graph. Vertices in the decoding graph represent stabilizers, edges represent Pauli errors, and cycles represent logical operators. In principle, one can perform perfect matching in the decoding graph to recover from errors, though for non-checkerboardable codes this may not succeed in achieving the code distance. Like Kitaev's homological construction, we wish to associate homologically non-trivial cycles to non-trivial logical operators. This is easily done in the checkerboardable case of our construction. However, in the non-checkerboardable case, the decoding graph might not embed in the same manifold as the original graph. In this case, we instead find a way to embed the decoding graph in a manifold of potentially higher genus in such a way that this correspondence of homological and logical non-triviality is maintained. Bounds on the code distance follow in terms of lengths of cycles in the decoding graph\edit{, and we provide efficient algorithms to compute those bounds}.

Finally, Section~\ref{sec:code_examples}, we provide several code families that serve as examples of our construction. The first family, \edit{rotated} square-lattice codes on the torus \cite{wen2003quantum}, serves as a familiar introduction. The second family \edit{introduces general rotations of the square-lattice} on the torus to get codes like those in \cite{kovalev2012improved}. Included in this case is a family of cyclic codes that contains the famous 5-qubit quantum code as a member. The third family consists of hyperbolic codes defined using regular tilings of high genus surfaces. Because we are using our qubit-on-vertices definition of surface codes, this leads to new codes, different from those in, for instance \cite{breuckmann2017hyperbolic,breuckmann2017homological}. In these first examples, we take pains to calculate the code distances formally. A final section of examples consists of a planar code family generalizing the triangle code \cite{yoder2017surface} and an embedding of stellated codes \cite{kesselring2018boundaries} on higher genus surfaces. In these cases, we do not formally prove the code distances but provide conjectures. We believe that stellated codes on higher genus surfaces provide the smallest known constant $c$, approaching $1/4$, for codes with stabilizers that are mostly weight four and vanishingly many stabilizers of weight five.

\begin{table}[ht]
    \centering
    \begin{tabular}{|c|c|c|c|}
        \hline
        Code & Section & Notes & References \\\hline\hline
        Majorana surface codes & Sec.~\ref{sec:Majorana_and_qubit_codes}, Fig.~\ref{fig:Maj_code_example} & Encodes qubits into Majoranas & \cite{wen2003quantum,kitaev2006anyons,vijay2015majorana} \\\hline
        Rotated toric codes & Sec.~\ref{sec:square_lattice_toric_codes}, Fig.~\ref{fig:square_lattice_examples} & Encode 1 or 2 qubits on a torus & \cite{wen2003quantum} \\\hline
        General rotated toric codes & Sec.~\ref{sec:rotated_toric_codes}, Fig.~\ref{fig:5-qubit_family} & Contains the $\llbracket 5,1,3\rrbracket$ code & \cite{kovalev2012improved,kovalev2011low} \\\hline
        New hyperbolic codes & Sec.~\ref{sec:improved_hyperbolic_codes}, Fig.~\ref{fig:hyperbolic_codes}, Tab.~\ref{tab:noncheckerboardable_hyperbolic_codes} & Regular tilings, qubits on vertices & -- \\\hline
        Circle-packing codes & Sec.~\ref{sec:generalized_triangle}, Fig.~\ref{fig:circle_packing_code} & Planar codes with $N\approx KD^2/2$ & -- \\\hline
        Stellated codes & Fig.~\ref{fig:stellated_codes} & Planar codes with $N\approx KD^2/2$ & \cite{kesselring2018boundaries} \\\hline
        High-genus stellated codes & Sec.~\ref{sec:generalized_triangle}, Fig.~\ref{fig:high_genus_stellated} & $N\approx KD^2/4$ on 2D manifolds & -- \\
        \hline
    \end{tabular}
    \caption{\edit{A summary of code families that we discuss in this work, some appearing for the first time here. Except for the Majorana codes, these are codes that use qubits to protect logical qubits. 
    }}
    \label{tab:code_summary}
\end{table}

%% file: Rotation-systems-v2.tex
\section{Preliminaries}
\label{sec:rotation-systems}

In this section, we first provide an introduction to the combinatorial description of graph embeddings (also called \textit{maps}) on closed surfaces. We refer to this description as a \textit{rotation system}, following terminology in \cite{stahl1978embeddings,stahl1978generalized, schaefer2018crossing}. We recall that a \textit{closed surface} is a 2-dimensional, compact, and connected topological manifold without boundary (see \cite{lee2010topological, lee2012smooth} for an introduction to topological manifolds), and henceforth we simply say a \textit{manifold} to mean a closed surface (unless specified otherwise). All manifolds admit a unique smooth structure up to diffeomorphisms \cite{rado1925begriff}, and by the \textit{closed surface classification theorem} \cite{seifert1980seifert, lee2010topological, francis1999conway}, each manifold is homeomorphic to a space of one of three types --- (a) Type I: $\mathbb{S}^2$, (b) Type II: a connected sum of one or more copies of $\mathbb{T}^2$, (c) Type III: a connected sum of one or more copies of $\mathbb{RP}^2$. This greatly simplifies the picture when we physically think of ``drawing'' a graph on a manifold, more formally captured using the concept of graph embeddings. A rotation system is simply a combinatorial description of any such graph emdedding. \edit{To garner the most intuition, we believe that rotation systems are most clearly presented by first defining them for the general case of embeddings in orientable and non-orientable manifolds (Types I, II, III), and then discussing the case of embeddings in orientable manifolds (Types I, II) as a special case in Appendix~\ref{sec:app-gluing-orientable}, which is a lot easier to visualize.}
To a large extent, the content of this section follows \cite{stahl1978embeddings,nedela2001regular,siran2001triangle}. Graph theoretic terminology appearing in this section is borrowed from \cite{diestel2000graduate,gross2001topological}. 

\edit{This section is structured as follows. We begin by discussing graph embeddings in Section~\ref{subsec:graph_embeddings}. General rotation systems are then introduced in Section~\ref{subsec:embedding-non-orientable}. In the remainder of this section, we introduce a range of definitions and terminologies from topological graph theory, essential for the subsequent discussions in this paper. In Section~\ref{subsec:checkerboardability}, we introduce the concept of graph checkerboardability. Next in Section~\ref{subsec:homology}, we discuss the notion of $\mathbb{F}_2$-homology of a graph embedding, and in particular introduce the notion of the homological systole. Finally, in Section~\ref{subsec:covers} we discuss covering maps and contractible loops on a manifold.}


\subsection{Graph embeddings}
\label{subsec:graph_embeddings}
Let us first cover graphs and their embeddings. A graph $G(V,E)$ is a collection of vertices $V$ and edges $E$. We will mostly deal with finite graphs, i.e. $1 \leq |V| < \infty$, and $1 \leq |E| < \infty$, and explicitly state when this assumption is violated. 
An edge adjacent to a single vertex is a \textit{loop}. If an edge is not a loop, then it is adjacent to exactly two vertices. The \textit{degree} of a vertex $v$ is denoted $\text{deg}(v)$, and it is the number of edges adjacent to it, with loops counted with multiplicity two. We assume in this section that $G$ is a connected graph.

A graph embedding requires that we actually draw the graph on a manifold $\mathcal{M}$. There may be multiple ways to do this even for a single graph and a fixed manifold. Generally, this drawing procedure is described by a \textit{graph embedding map} $\Gamma:V\cup E\rightarrow\mathcal{M}$ that assigns unique points in $\mathcal{M}$ to each $v\in V$, and \textit{arcs} in $\mathcal{M}$ to each $e\in E$. If $e \in E$ is not a loop adjacent to distinct vertices $v_1$ and $v_2$, then the arc assigned to it must be the image of a homeomorphism $\gamma:[0,1]\rightarrow\mathcal{M}$, with \textit{endpoints} $\gamma(0) = \Gamma(v_1)$ and $\gamma(1) = \Gamma(v_2)$. If $e$ is a loop adjacent to a vertex $v$, then the arc assigned to it is the image of a homeomorphism $\gamma: \mathbb{S}^1 \rightarrow\mathcal{M}$, with both endpoints given by $\gamma(0) = \Gamma(v)$ (here $\mathbb{S}^1$ is parameterized in polar coordinates). Moreover, two arcs may not intersect except possibly at their endpoints.

The set $\mathcal{F}=\mathcal{M}\setminus\Gamma(E)$ is a collection of regions called the \textit{faces} of the embedding, and two points $p,q \in \mathcal{F}$ belong to the same face if and only if there exists a continuous map $\gamma' : [0,1] \rightarrow \mathcal{F}$ such that $\gamma'(0) = p$, and $\gamma'(1) = q$. We say that an embedding of $G$ is a \textit{2-cell embedding} if and only if every face is homeomorphic to the open disc $\mathcal{B}(0,1)= \{(x,y) \in \mathbb{R}^2 : x^2 + y^2 < 1\}$. If the closure of every face is homeomorphic to the closed disc $\overline{\mathcal{B}(0,1)}$, and the set difference of the closure of the face and the face itself is homeomorphic to the circle $\mathbb{S}^{1}$, we say that it is a \textit{closed 2-cell embedding} or a \textit{strong embedding}. Every connected graph admits a 2-cell embedding in some orientable manifold (see Theorem A.16. in \cite{schaefer2018crossing}), and also in some non-orientable manifold (see Theorem 3.4.3 in \cite{gross2001topological}). One can see a graph (a square lattice) embedded in the torus in \edit{Fig.~\ref{fig:rot_sys_ex_A} in Appendix~\ref{sec:app-gluing-orientable}, while a graph embedded in the projective plane is shown in Fig.~\ref{fig:rot_sys_ex} below}.

\subsection{General rotation systems}
\label{subsec:embedding-non-orientable}

\edit{The goal of this section is to establish an equivalence between graphs embedded in manifolds (both orientable and non-orientable) and a combinatorial object called a general rotation system. The description of a graph embedding in the previous section requires a variety of homeomorphisms whose exact details, topologically speaking, do not matter. The rotation system description dispenses with those details by breaking the embedded graph down into a set of \textit{flags} $H$, four for every edge, and describing vertices, edges, and faces of the graph embedding as sets of these flags. Each edge should be thought of as two \textit{half-edges}\footnote{Half-edges are formally defined for the special case of oriented rotation systems in Appendix~\ref{sec:app-gluing-orientable}, and Fig.~\ref{fig:rot_sys_ex_A} shows an example.}, and of the four flags per edge, each half-edge corresponds to exactly two of the flags. The two flags per half-edge effectively introduce an orientation to the half-edge, with the flags corresponding to the two ``sides'' of the manifold, and it is this notion that allows us to handle both orientable and non-orientable manifolds simultaneously. Fig.~\ref{fig:rot_sys_ex} makes it clear how flags can be visualized in the embedding.}
The formal definition of a general rotation system is as follows.
\begin{definition}
\label{def:rotation_system}
A \textit{general rotation system} is a quadruple $R=(H, \lambda, \rho, \tau)$, where $H$ is a \edit{nonempty,} finite set of flags, and $\lambda, \rho, \tau$ are permutations on $H$ satisfying the properties
\begin{enumerate}[(i)]
\item $\lambda, \rho$, and $\tau$ are fixed-point-free involutions, \edit{such that $h$, $\lambda h$, $\rho h$ and $\tau h$ are distinct for every $h \in H$.}
\item $\lambda\tau = \tau\lambda$, or equivalently, $\lambda\tau$ is an involution.
\item the $\textit{monodromy}$ group $M(R)=\langle\lambda,\rho,\tau\rangle$ acts transitively on $H$.
\end{enumerate}
\end{definition}
The involutions $\lambda, \rho, \tau$ permute flags as shown in \edit{Fig.~\ref{fig:rot_sys_ex}}. We define  vertices, edges, and faces as orbits of $\langle \rho,\tau\rangle$, $\langle\lambda,\tau\rangle$, and $\langle\rho,\lambda\rangle$ respectively. From properties (i) and (ii) of Definition~\ref{def:rotation_system}, orbits of $\langle\lambda,\tau\rangle$ have the form $\{h,\lambda h,\tau h, \lambda\tau h : h \in H\}$, i.e. are exactly of size four, and so $|H| \equiv 0 \pmod{4}$. Moreover, one deduces from property (i) that the orbits of $\langle\rho,\lambda\rangle$ and $\langle\rho,\tau\rangle$ have even sizes, with the canonical form of an orbit of size \edit{$2n+4$} given by $\{h,\lambda h,\rho\lambda h, \lambda\rho\lambda h,\dots, (\lambda \rho)^{n\edit{+1}} \lambda h : h \in H\}$, and $\{h,\tau h,\rho\tau h, \tau\rho\tau h,\dots (\tau \rho)^{n\edit{+1}} \tau h : h \in H\}$ respectively. 

The sets of vertices, edges, and faces, which we denote $V$, $E$, and $F$ respectively, define a graph embedding and satisfy Euler's formula
\begin{equation}
\label{eq:Eulers_formula}
\chi = |V|-|E|+|F|,
\end{equation}
where $\chi$ is the Euler characteristic of the manifold, i.e.~if $g$ is the manifold's genus (orientable or non-orientable), then $\chi=2-2g$ if the manifold is orientable, and $\chi=2-g$ if it is non-orientable. Note that property (iii) ensures that $G$ is a connected graph. We will write $G(V,E,F)$ to denote such an embedded graph, as compared to $G(V,E)$ which only denotes the graph without the embedding.

The orbits of the free group $\langle \tau \rangle$ are exactly the half-edges, and additionally we define \emph{sectors} as the orbits of the free group $\langle \rho \rangle$. As $\tau$ and $\rho$ are fixed-point-free involutions, these orbits are just two flags each. We denote a half-edge as $[h]_\tau=\{h,\tau h\}$, and a sector as $[h]_\rho=\{h,\rho h\}$.

\edit{Finally, we note that, in Definition~\ref{def:rotation_system}, the requirement of distinctness in property (i) is unlike prior literature (see e.g.~\cite{nedela2001regular}),
but conveniently rules out the possibilities of having 
degree-1 vertices (because $\rho h\neq\tau h$), unpaired half-edges ($\lambda h\neq\tau h$), and faces bounded by just one edge ($\rho h\neq\lambda h$), all of which are unnecessary and distracting for the later definitions of quantum codes.}

\begin{figure}[h]
    \centering
    \includegraphics[width=0.4\textwidth]{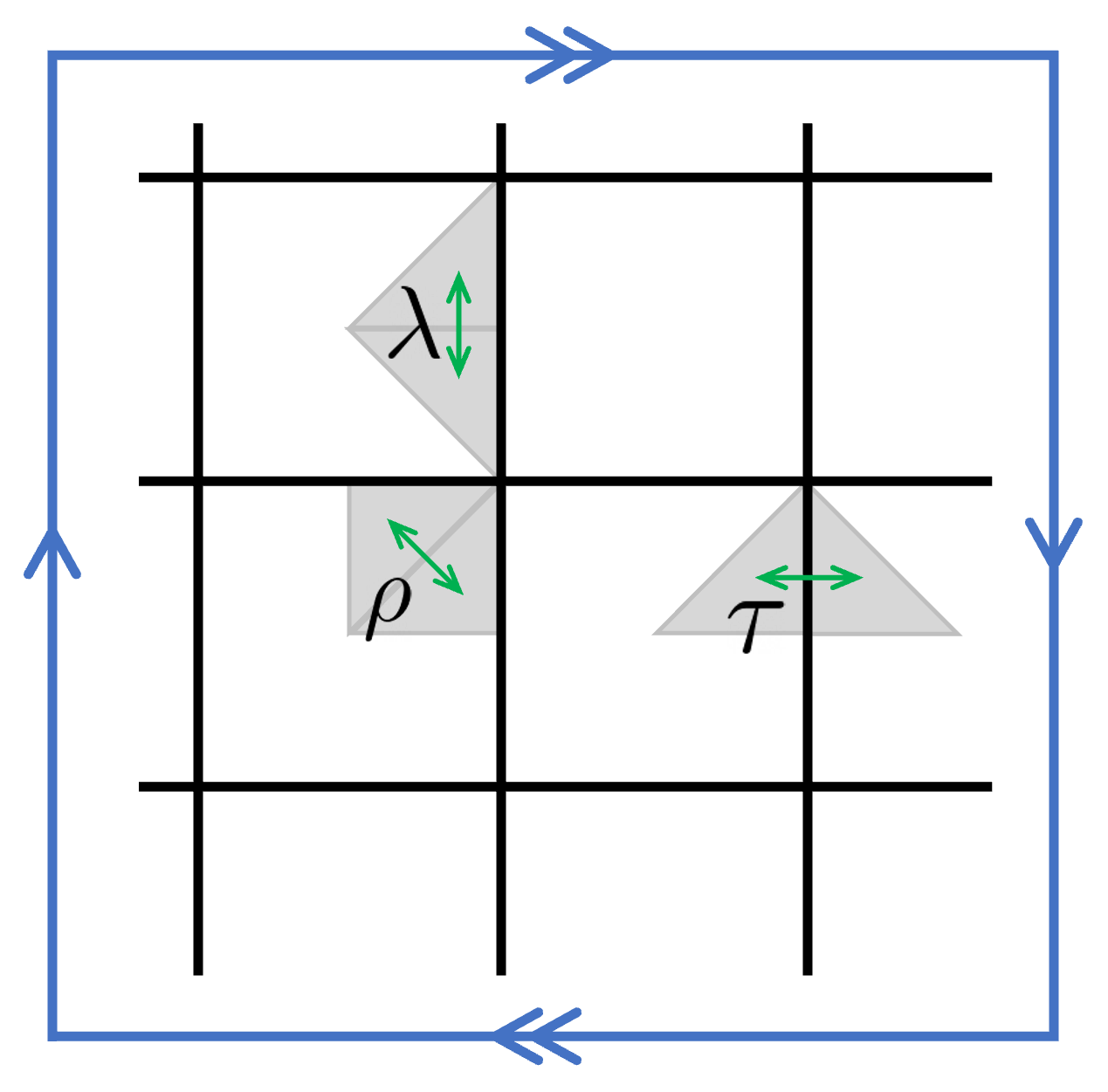}
    \caption{We show a graph embedded on the projective plane. A general rotation system is defined by a set of flags (gray triangles) and three permutations: $\lambda$ swaps flags along the same side of an edge, $\rho$ swaps them within a face adjacent to a vertex, and $\tau$ swaps them across an edge.}
    \label{fig:rot_sys_ex}
\end{figure}

\edit{While Definition~\ref{def:rotation_system} captures graph embedding on both orientable and non-orientable manifolds, the particular case of orientable manifolds is much simpler. This is carried out in Appendix~\ref{sec:app-gluing-orientable}, and the relevant combinatorial object there is the oriented rotation system, defined in Definition~\ref{def:oriented_rotation_system}. However, as mentioned before, this is a special case of the general rotation system, and thus} Definition~\ref{def:rotation_system} generalizes Definition~\ref{def:oriented_rotation_system}. Suppose we have an oriented rotation system $R_O=(H_O,\nu,\epsilon)$. There is a corresponding general rotation system $R=(H,\lambda,\rho,\tau)$ where $H=H_O\times\{1,-1\}$ and for $(h,i)\in H$,
\begin{equation}
\lambda(h,i)=(\epsilon h,-i),\quad \rho(h,i)=(\nu^{i} x,-i),\quad \tau(h,i)=(h,-i).
\end{equation}
This definition ensures that the general rotation system represents the same embedded graph as the given oriented rotation system. 

A general rotation system $(H,\lambda,\rho,\tau)$ describes a graph embedded in an orientable manifold if and only if $H$ can be partitioned into two sets $H_{\pm1}$ such that $\lambda$, $\rho$, and $\tau$ each map elements of either set to elements of the other set. Otherwise, the graph is embedded in a non-orientable manifold. If the manifold is orientable, there are at least two oriented rotation systems describing the same graph embedding, on the same set of half-edges. Take $H_O=H/\tau$ and let $[h]_\tau=\{h,\tau h\}\in H_O$ be a generic half-edge. Define $\nu [h]_\tau = [\rho\tau h]_\tau$, $\nu' [h]_\tau=[\tau\rho h]_\tau$, and $\epsilon [h]_\tau = [\lambda\tau h]_\tau$. Then, both $(H_O,\nu,\epsilon)$ and $(H_O,\nu',\epsilon)$ are oriented rotation systems for the same embedded graph.


Adjacency is an important concept for graphs and their components. Using the general rotation system description, flags are the primitives, and we say a flag is adjacent to a vertex, edge, or face $x\in V\cup E\cup F$, if it is contained in $x$. We define adjacency matrices in terms of flags. The vertex-flag adjacency matrix $A\in\mathbb{F}_2^{|V|\times|H|}$ has $A_{vh}=1$ if and only if $h\in v$. Similarly define an edge-flag adjacency matrix $B \in \mathbb{F}_2^{|E|\times|H|}$, and a face-flag adjacency matrix $C \in \mathbb{F}_2^{|F|\times|H|}$. Although flags are primitives, we can also define adjacency of the more typical graph components. \edit{For example, we say that a vertex $v\in V$ and an edge $e\in E$ are adjacent if $v\cap e\neq\emptyset$.}
Likewise, adjacency of edges and faces, and of vertices and faces can be defined by non-trivial intersection. Note that each edge in an embedded graph is adjacent to at least one and at most two faces. The degree of a vertex $v$ is half the number of flags it is adjacent to, $\text{deg}(v)=|v|/2$, and this definition coincides with the one in Section~\ref{subsec:graph_embeddings}. We also say that two vertices (resp. two faces) are adjacent, if there exists an edge adjacent to both vertices (resp. both faces).

\subsection{Dual graph and checkerboardability}
\label{subsec:checkerboardability}

If an embedded graph $G(V,E,F)$ is described by the general rotation system $(H,\lambda,\rho,\tau)$, then the \textit{embedded dual graph} $\overline{G}(\overline{V}, \overline{E}, \overline{F})$ of $G$ is defined by the general rotation system $(H,\tau,\rho,\lambda)$, i.e. by simply exchanging the permutations $\lambda$ and $\tau$. This implies that $\overline{G}$ has a vertex (resp. face) for every face (resp. vertex) of $G$, while $|E| = |\overline{E}|$, and we conclude that the Euler characteristic of $G$ and $\overline{G}$ are the same. Note also that the graph $\overline{G}$ is connected (by properties of a general rotation system). Moreover $\overline{G}$ defines an embedding in an orientable manifold if and only if $G$ does too --- thus the graphs $G$ and $\overline{G}$ are embedded in the same manifold. The edges of $\overline{G}$ have a natural interpretation: for every edge $e \in E$ adjacent to faces $f,f' \in F$, not necessarily distinct, there is an edge $\overline{e} \in \overline{E}$ between the vertices of $\overline{G}$ corresponding to $f$ and $f'$ (in particular if $f=f'$ then $\overline{e}$ is a loop). The dual of the graph $\overline{G}$ is isomorphic to the graph $G$. There are several more modifications of a graph $G$ that we use in the paper. In Table~\ref{tab:graph_notation}, we list these and the sections where they are defined.

\begin{table}[]
    \centering
    {\renewcommand{\arraystretch}{1.5}
    \begin{tabular}{|c||c|c|c|c|c|c|}
        \hline
         Name & Original & Dual & Medial & Decoding & Doubled & Face-vertex \\\hline
         Notation & $G$ & $\overline{G}$ & $\widetilde{G}$ & $G_{\text{dec}}$ & $G^2$ & $G_{\text{fv}}$ \\\hline
         Section & \ref{subsec:graph_embeddings} & \ref{subsec:checkerboardability} & \ref{subsec:relate_to_homological} & \ref{sec:dec_graph} & \ref{sec:doubled_graphs} & \ref{sec:face-width} \\\hline
    \end{tabular}}
    \caption{Embedded graph $G$ and the derived graphs used in this paper. \edit{Without loss of generality for constructing codes (see the end of Section~\ref{subsec:maj_codes_on_graphs}), we assume $G$ has no vertices of degree less than three.}}
    \label{tab:graph_notation}
\end{table}

Before proceeding any further, given a graph embedding $G(V,E,F)$, we define the face-edge adjacency matrix $\Phi=\frac{1}{2} CB^\top \in \mathbb{F}_2^{|F|\times|E|}$ \edit{(with this matrix multiplication followed by division performed over the integers, and then reduced modulo two)}, which captures the adjacency of the faces and edges of $G$. It is worth pointing out the structure of the matrix $\Phi$. First note that each column of $B^\top$ contains exactly four non-zeros, as each edge $e \in E$ is a set of four flags. Now there are two cases: (i) the two sides of $e$ belong to distinct faces $f,f' \in F$, and (ii) $e$ is surrounded by a single face on both sides, i.e. $e$ is a subset of that face. In the former case, $\Phi_{fe}=\Phi_{f'e}=1$ are the only non-zero entries of the column of $\Phi$ corresponding to $e$, while in the latter case $\Phi_{fe} = 0$ for all $f \in F$. \edit{The face-edge adjacency matrix $\Phi$ can also be interpreted in homological language directly, and this is mentioned in Appendix~\ref{app:algebraic_topology}.}

We now discuss an important property of embedded graphs from the perspective of quantum codes, called checkerboardability.

\begin{definition}
\label{def:checkerboardability}
A graph embedding $G$ is \textit{checkerboardable} if the dual graph $\overline{G}$ is bipartite.
\end{definition}

An equivalent, and a more intuitive way of stating this definition is the following lemma.

\begin{lemma}
\label{lem:checkerboardable_0}
A graph embedding $G(V,E,F)$ described by a general rotation system $(H,\lambda,\rho,\tau)$, is checkerboardable if and only if the faces of $G$ can be two-colored, such that for every flag $h \in H$, $h$ and $\tau h$ are differently colored.
\end{lemma}

\begin{proof}
One direction of this equivalence is clear. Suppose $\overline{G}$ is bipartite, and we color the two disjoint vertex sets of $\overline{G}$, corresponding to the bipartition differently. Using the correspondence between the faces of $G$ and vertices of $\overline{G}$, we get a two-coloring of $F$, and hence of the flags $H$. Now pick any flag $h \in H$, and let $h \in e \in E$. Then $h,\tau h$ must belong to distinct faces $f,f' \in F$ respectively, otherwise if $e$ is adjacent to a single face then $\overline{G}$ will contain a loop (as mentioned above), and hence cannot be bipartite. Since $f$ and $f'$ are both adjacent to $e$, they must be differently colored. For the converse, suppose we are given a two-coloring of $F$ (which now induces a coloring of the vertices of $\overline{G}$) such that $h$, and $\tau h$ are differently colored for all $h \in H$. Consider $\overline{G}(\overline{V},\overline{E})$, and partition the vertices as $\overline{V} = \overline{V}_w \sqcup \overline{V}_b$, where each disjoint partition contains all vertices of $\overline{G}$ of a single color. Assume that $\overline{G}$ is not bipartite with this bipartition. If $\overline{G}$ has a loop then there exists $e \in E$ adjacent to a single face. If $\overline{G}$ has no loop, then either $\overline{V}_w$ or $\overline{V}_b$ is not an independent set; so there exist distinct faces $f,f' \in F$ of the same color, adjacent to an edge $e \in E$. In both cases, all four flags of $e$ have the same color which is a contradiction. Thus $\overline{G}$ is bipartite.
\end{proof}

There are numerous other equivalent statements of checkerboardability. We give two more below that we find useful. Note that the only permutation out of $\lambda$, $\rho$, $\tau$ that moves flags across edges of the graph is $\tau$. Therefore using Lemma~\ref{lem:checkerboardable_0}, we have the following:

\begin{lemma}
\label{lem:checkerboardability_rotation_system}
A graph embedding $G(V,E,F)$ described by general rotation system $R=(H,\lambda,\rho,\tau)$, is checkerboardable if and only if we can partition $H$ into two disjoint sets $H_w$ and $H_b$ such that (i) $\lambda$ and $\rho$ map both sets to themselves, and (ii) $\tau$ maps elements of either set to the other set.
\end{lemma}

\begin{proof}
Suppose first that $G$ is checkerboardable. Then there is a two-coloring of the faces of $F$ such that for every $h \in H$, $h$ and $\tau h$ are differently colored. Defining $H_w$ and $H_b$ to be the sets of flags of each color then shows that these sets have the required properties (i) and (ii), proving the only if direction. Now assume that $H$ can be partitioned into disjoint sets $H_w$ and $H_b$ satisfying the two properties. If $f \in F$, by property (i) all flags of $f$ must belong to either $H_w$ or $H_b$. Let us color all flags of $H_w$ and $H_b$ black and white respectively, which induces a two-coloring of all the faces of $G$. Property (ii) now implies that for every $h \in H$, $h$ and $\tau h$ are differently colored, proving that $G$ is checkerboardable. 
\end{proof}

In this paper we adopt the notation that $\vec1$ and $\vec0$ are the row vectors of all 1s and all 0s respectively. The arguments above show that in a checkerboard coloring, each edge is adjacent to exactly one face of each color (white and black, say). Thus, checkerboardability implies the existence of a vector $x\in\mathbb{F}_2^{|F|}$ (e.g.~$x_f=1$ if and only if face $f$ is colored white) such that $x\Phi=\vec1$. In fact we also have that $(\vec1-x) \Phi=\vec1$, and these are the only two vectors with this property. This is summarized in the following lemma.

\begin{lemma}
\label{lem:checkerboardability}
Let $G(V,E,F)$ be a graph embedding described by the general rotation system $(H,\lambda,\rho,\tau)$, and let $\Phi$ be its face-edge adjacency matrix. Then
\setmargins
\begin{enumerate}[(a)]
    \item $G$ is checkerboardable if there exists $x\in\mathbb{F}_2^{|F|}$ such that $x\Phi=\vec1$.
    \item If $G$ is checkerboardable then the set $\{x\in\mathbb{F}_2^{|F|} : x\Phi=\vec1 \}$ has exactly two distinct elements $x$ and $x'$ (different from $\vec 0$ and $\vec 1$), which satisfy $x + x' = \vec1$. Thus the partitioning of the flags into the two disjoint sets as described in Lemma~\ref{lem:checkerboardability_rotation_system} is unique.
\end{enumerate}
\end{lemma}

\vspace{0.5cm}
\begin{proof}
\setmargins
\begin{enumerate}[(a)]
    \item Suppose there exists $x$ such that $x\Phi=\vec1$, and define $F_b = \{f \in F : x_f = 1\}$, and $F_w = F \setminus F_b$. Color the faces in $F_b$ and $F_w$ differently. Assume for contradiction that $h \in H$ is a flag such that $h$ and $\tau h$ have the same color, and let $h \in e \in E$. Since $h$ and $\lambda h$ belong to the same face, this implies that all four flags of $e$ have the same color. But this implies $\Phi_{fe} = 0$ for all $f \in F$ (by properties of $\Phi$ discussed above), and so $(x\Phi)_e = \vec0$, which is a contradiction.
    \item Assume $G$ is checkerboardable, which also implies that every column of $\Phi$ has exactly two non-zeros (or equivalently each edge $e \in E$ is adjacent to two distinct faces), and so $\vec1 \Phi = \vec0$. We already know from the discussion in the paragraph before this lemma that there exists $y\in\mathbb{F}_2^{|F|}$ such that $y\Phi=\vec1$. Thus we also conclude that $(\vec1 - y) \Phi = \vec1$. Now suppose there is a vector $x \in \mathbb{F}_2^{|F|}$ such that $x \Phi = \vec1$, and define a disjoint partition $F = F_w \sqcup F_b$ of the faces of $G$, as in the proof of (a). This induces a disjoint vertex partition $\overline{V} = \overline{V}_w \sqcup \overline{V}_b$ of the dual graph $\overline{G}(\overline{V},\overline{E})$. Then combining the arguments of the proof of (a), and the proof of the converse part of Lemma~\ref{lem:checkerboardable_0} shows that $\overline{G}$ is bipartite. Since, $\overline{G}$ is connected, the bipartition is unique, and so $x$ is equal to either $y$ or $\vec1 - y$. This proves the first part, and the second part follows from it easily. 
\end{enumerate}
\end{proof}

A consequence of Lemma~\ref{lem:checkerboardability}(b) is that the checkerboard coloring of a checkerboardable graph is essentially unique up to a permutation of the colors. One easy way a graph can fail to be checkerboardable is if it contains an odd-degree vertex $v$. Then, $(\rho\tau)^{\text{deg}(v)}$ would map $h\in v$ to itself, but also because $\tau$ is applied an odd number of times, it would map $h$ from one set, say $H_w$, to the other $H_b$ (where $H_w$ and $H_b$ are as in Lemma~\ref{lem:checkerboardability_rotation_system}), a contradiction. If $g=0$ (the graph is planar, i.e.~embedded on the 2-sphere $\mathbb{S}^2$), we also have have the well known converse: a planar graph embedding is checkerboardable if all vertices have even degree.



\subsection{Homology of graph embeddings}
\label{subsec:homology}

One topological invariant that is important for this paper (and for other quantum codes such as homological codes) is that of $\mathbb{F}_2$-homology. We introduce this concept briefly here, and refer the reader to \cite{hatcher2002algebraic} for a more comprehensive treatment of homology. To do this in the context of embedded graphs, we will first introduce some graph theoretic terminology. \edit{We prefer this approach of presenting homological concepts in graph theoretic language as it places a lighter burden on the reader than having to know algebraic topology; however, the corresponding connections with known objects in algebraic topology (specifically chain complexes) are described in Appendix~\ref{app:algebraic_topology}.}

Recall first that a \textit{trail} $t$ in a graph is a sequence of vertices and edges,
\begin{equation}
t=(v_0,e_1,v_1,e_2,\dots,v_{\ell-1},e_\ell,v_\ell),
\end{equation}
where edge $e_i$ is adjacent to vertices $v_{i-1}$ and $v_i$, and each edge is distinct. Trails can be open (if $v_0\neq v_\ell$) or closed (if $v_0=v_\ell$). Open trails have endpoints $v_0$ and $v_\ell$. Closed trails are also called cycles. Including the vertices in the specification of a trail is standard but superfluous --- a trail $t$ can be represented by a vector $t\in\mathbb{F}_2^{|E|}$ such that $t_i=1$ if and only if the $i^{\text{th}}$ edge of the graph is in the trail. A \emph{path} is a trail without repeated vertices, except the first and last if it is a closed trail. The set of trails generates a group $\mathcal{T}(G)$ with the group operation the symmetric difference of the trails or, equivalently, the addition modulo two of the vectors $y\in\mathbb{F}_2^{|E|}$ representing the trails. Because each edge is a trail, $\mathcal{T}(G)$ is isomorphic to the vector space $\mathbb{F}_2^{|E|}$, and a generating set of $\mathcal{T}(G)$ (or equivalently a basis for $\mathbb{F}_2^{|E|}$) can be constructed to consist only of paths. 

Now let $G(V,E,F)$ be an embedded graph on a manifold, either orientable or non-orientable. Two subgroups of $\mathcal{T}(G)$ are particularly important in the definition of $\mathbb{F}_2$-homology: the subgroup $\mathcal{Z}(G)$ of all the cycles of $G$, and the subgroup $\mathcal{B}(G)$ generated by the cycles of $G$ that are \textit{boundaries} of faces of the embedded graph. Formally
\begin{equation}
\label{eq:BG-ZG}
\begin{split}
    \mathcal{Z}(G) &= \left \langle c : c \text{ is a cycle of } G \right \rangle, \\
    \mathcal{B}(G) &= \left \langle c : c \text{ is a cycle of } G, \; \exists f \in F \text{ such that } \forall \; \text{edges } e \in c, \; \Phi_{fe} = 1 \right \rangle.
\end{split}
\end{equation}

We call cycles in $\mathcal{B}(G)$ homologically trivial and cycles in $\mathcal{Z}(G)$ but not in $\mathcal{B}(G)$ homologically non-trival. The homological systole of the graph, $\hsys{G}$, is the length of the shortest homologically non-trivial cycle. The first homology group over $\mathbb{F}_2$ is the quotient group (or equivalently as the quotient vector space)
\begin{equation}
\label{eq:F2-homology}
    H_1(G, \mathbb{F}_2) = \mathcal{Z}(G)/\mathcal{B}(G).
\end{equation}

It turns out that for two different graphs $G, G'$ embedded in the same manifold $\mathcal{M}$, $H_1(G, \mathbb{F}_2) \cong H_1(G', \mathbb{F}_2)$. Thus in the rest of the paper, we will simply write $H_1(\mathcal{M})$ for any given graph embedding $G(V,E,F)$ in $\mathcal{M}$.

\subsection{Covering maps and contractible loops}
\label{subsec:covers}

A cycle drawn on a manifold $\mathcal{M}$ is contractible if it can be continuously deformed to a point and non-contractible otherwise. Homological non-triviality implies non-contractibility \cite{cabello2007finding}. However, the converse is not true -- consider for instance a cycle winding twice around a torus, which is non-contractible but homologically trivial. Similar to the homological case, we let $\sys{G}$ denote the length of the shortest non-contractible cycle in a graph $G$ embedded on $\mathcal{M}$ and note $\hsys{G}\ge\sys{G}$. A manifold is simply-connected if all cycles are contractible.

We say a connected manifold $\mathcal{M}'$ is a cover of $\mathcal{M}$ if there is a map $\pi:\mathcal{M}'\rightarrow\mathcal{M}$ such that each open disk $U\subseteq\mathcal{M}$ has a pre-image $\pi^{-1}(U)$ that is a union of disjoint open disks, each mapped homeomorphically onto $U$ by $\pi$ \cite{bredon2013topology}. We say $\pi$ is a covering map. An $l$-fold cover is one in which $\pi^{-1}(U)$ is a union of $l$ open disks for all $U$. It is also acceptable to call a 2-fold cover a double cover, a 3-fold cover a triple cover, etc. 

A universal cover $\mathcal{U}$ of a manifold $\mathcal{M}$ is a cover that covers all other covers of $\mathcal{M}$. 
Equivalently, the universal cover of $\mathcal{M}$ is the unique simply-connected cover of $\mathcal{M}$ (see for instance Chapter III.4 of \cite{bredon2013topology}). Compact manifolds are universally covered by either the sphere (if $\mathcal{M}$ is the sphere or projective plane) or the plane (otherwise).

A curve drawn in $\mathcal{M}'$ can be projected down to $\mathcal{M}$ by applying $\pi$. Conversely, a curve $\gamma:[0,1]\rightarrow\mathcal{M}$ can be \emph{lifted} to $\mathcal{M}'$. Roughly one constructs a lift by applying $\pi^{-1}$ to $\gamma$. However, since $\pi^{-1}$ is one-to-many, more precisely a lift of $\gamma$ is defined to be any compact, continuous curve in $\mathcal{M}'$ that maps to $\gamma$ when applying $\pi$.

An interesting connection presents itself between contractibility and the universal cover. A cycle in $\mathcal{M}$ is contractible if and only if its lift into the universal cover $\mathcal{U}$ is also a cycle. Non-contractible curves, on the other hand, lift to curves with two distinct endpoints \cite{cabello2007finding}.


%% file: Majorana-codes-v2.tex
\section{Majorana and qubit surface codes from rotation systems}\label{sec:Majorana_and_qubit_codes}

In this section, we define our framework for constructing qubit surface codes from rotation systems. Majorana surface codes appear as a useful intermediary in the construction. So, building on the last section's introduction to rotation systems, the logical progression of this section is roughly
\begin{equation*}
\begin{array}{ccccc}
\text{Rotation system \textbackslash\space graph embedding}& \longrightarrow& \text{Majorana surface code}& \longrightarrow& \text{Qubit surface code}\\\text{(Section~\ref{sec:rotation-systems})}&&\text{(Section~\ref{subsec:maj_codes_on_graphs})}&&\text{(Section~\ref{subsec:qub_surface_codes})}
\end{array}.
\end{equation*}
Before diving into surface codes, Section~\ref{subsec:maj_ops_intro} gives a quick introduction to Majorana fermions and Majorana fermion codes in general. Afterward, in Section~\ref{subsec:relate_to_homological}, we show how our framework for qubit surface codes generalizes the standard homological definition \cite{kitaev2003fault}, in which qubits are placed on edges, while vertices and faces support stabilizers. We assume some familiarity of the reader with the Pauli group, and Appendix~\ref{sec:cal-prop} contains a brief recap.

\subsection{Majorana operators and codes}
\label{subsec:maj_ops_intro}
\edit{To introduce Majorana operators in this subsection, we largely follow and summarize \cite{bravyi2010majorana}. For even integer $m$,} the Majorana operators $\{\gamma_0,\gamma_1,\dots,\gamma_{m-1}\}$ are linear Hermitian operators acting on the fermionic Fock space $\mathcal{H}_{m/2}=\{\ket{\vec b}:\vec b\in\mathbb{F}_2^{m/2}\}$, or equivalently the $m/2$-qubit complex Hilbert space, satisfying
\begin{align}
\label{eq:Maj_rules}
\gamma_i^2=I,\quad \gamma_i\gamma_j=-\gamma_j\gamma_i,\quad \forall \; 0 \leq i<j \leq m-1.
\end{align} 
Eq.~\eqref{eq:Maj_rules} ensures that each $\gamma_i$ is distinct and different from $I$. The total number of Majorana operators is even because two Majoranas correspond to each fermion in the system. We define a group $\mathcal{J}_m$ consisting of finite products of the Majorana operators $\gamma_i$, and phase factor $i = \sqrt{-1}$. By Eq.~\eqref{eq:Maj_rules}, this group is finite with size $|\mathcal{J}_m|=2^{m+2}$. Elements of $\mathcal{J}_m$ either commute or anticommute. We indicate an element of $\mathcal{J}_m$ uniquely by $\eta \gamma_{a}$ where $\eta\in\{\pm1,\pm i\}$, $a\in\mathbb{F}_2^{m}$, and $\gamma_{a}=\prod_{i:a_i=1}\gamma_i$ (this product is ordered so that the $\gamma_i$ with smaller indices are on the left, and the empty product is defined to be equal to $I$), and the uniqueness of this representation follows from Eq.~\eqref{eq:Maj_rules}. We define the \textit{support} of $\eta \gamma_{a}$ to be the set $\supp{\eta \gamma_{a}}:=\{i : a_i=1\}$, and its \textit{weight} $|\eta \gamma_{a}|:=|\supp{\eta \gamma_{a}}|=|a|$, where $|a|$ is the Hamming weight of $a$. The commutation of the elements of $\mathcal{J}_m$ is now easily expressed as
\begin{equation}
\label{eq:majorana-commutation}
\gamma_{a}\gamma_{b}=(-1)^{|a||b|+a\cdot b}\gamma_{b}\gamma_{a} = (-1)^{\xi(a,b)} \gamma_{a+b}, \;\; \xi(a,b) = \sum_{i: b_i = 1} |\{j > i : a_j = 1\}|,
\end{equation}
where it is understood the operations ``$\cdot$'' and ``$+$'' are performed over $\mathbb{F}_{2}$ (the quantity $|a||b|$ is computed over integers and reduced modulo 2). These vectors are assumed to have context-appropriate length.

Particular choices of Majorana operators can be made by associating them with appropriately chosen Pauli operators on $m/2$ qubits. This can be done in a variety of ways. A particularly famous example is the Jordan-Wigner transformation, which makes the association
\begin{align}
\label{eq:JW}
\text{JW}(\gamma_{2k})=X_k\prod_{i=0}^{k-1}Z_i,\quad \text{JW}(\gamma_{2k+1})=Y_k\prod_{i=0}^{k-1}Z_i,
\end{align}
for all $0 \leq k \leq m/2 - 1$, where $X_i,Y_i,Z_i$ are the Paulis acting on qubit $i$. One can check that the Paulis associated with the Majoranas obey Eq.~\eqref{eq:Maj_rules}.

Majorana fermion codes (see e.g.~\cite{bravyi2010majorana,vijay2017quantum}) are created by specifying a subgroup $\mathcal{S}\le\mathcal{J}_m$ to be a stabilizer \edit{group}. The corresponding codespace is defined as usual to be the $+1$-eigenspace of all stabilizers, a subspace of $\mathcal{H}_{m/2}$. We require the stabilizer \edit{group} be chosen so that
\begin{enumerate}[i,wide]
\renewcommand\labelenumi{(\theenumi)}
\item $\mathcal{S}$ does not contain $-I$.
\item Each $\gamma\in\mathcal{S}$ has even weight.
\end{enumerate}
The first condition arises because we wish the codespace to be non-empty, and the second condition is imposed because we want the stabilizer operators to be physical, preserving fermion parity in the system. In fact, the first condition also implies that $\mathcal{S}$ is Abelian and that all elements of $\mathcal{S}$ are Hermitian. We note this fact and a few others in the following lemma\edit{, whose proof is in Appendix~\ref{app:majorana-proofs}.}

\begin{lemma}
\label{lem:majorana-group-props}
Let $\mathcal{S}$ be a subgroup of $\mathcal{J}_m$, and let $\mathcal{I} \subseteq \mathcal{J}_m$ be non-empty, such that elements of $\mathcal{I}$ commute and are Hermitian. Then $\langle \mathcal{I} \rangle$ is Abelian and Hermitian, and moreover the following holds:

\setmargins
\begin{enumerate}[(a)]
    \item There exists a set $\mathcal{I}'$ formed by multiplying each element of $\mathcal{I}$ by either $1$ or $-1$, such that $-I \not \in \langle \mathcal{I}' \rangle$.
    \item If $-I \not \in \mathcal{S}$, then $\mathcal{S}$ is Abelian and Hermitian. Conversely, if $\mathcal{S}$ is Abelian and Hermitian, then either $-I \not \in \mathcal{S}$, or $\mathcal{S} = \langle \mathcal{S}', -I \rangle$ for some subgroup $\mathcal{S}'$ of $\mathcal{S}$ with $-I \not\in \mathcal{S}'$.
\end{enumerate}
\end{lemma}

Because of Pauli representations of Majoranas like Eq.~\eqref{eq:JW}, it is clear that a Majorana fermion code on $m$ Majoranas corresponds to a stabilizer code on $m/2$ qubits. Applying the Jordan-Wigner transformation, Eq.~\eqref{eq:JW}, to $\mathcal{S}$ converts it to an Abelian group of Pauli operators on $m/2$ qubits, which has size upper bounded by $2^{m/2}$. Therefore, $|\mathcal{S}|\le 2^{m/2}$. Thus, if $\mathcal{S}$ is generated by $m/2-k$ independent operators, then the code encodes $k$ qubits (or $2k$ Majoranas). The centralizer $\mathcal{C}(\mathcal{S})$ of $\mathcal{S}$ is the subset of $\mathcal{J}_m$ that commutes with all elements of $\mathcal{S}$. The \textit{distance} $d$ of the Majorana code is the minimum weight of an element of $\mathcal{C}(\mathcal{S})$ that is not (up to factors of $i$) in $\mathcal{S}$, i.e.~$d=\min\{|\gamma|:\gamma\in\mathcal{C}(\mathcal{S})\setminus\langle \{iI\}\cup\mathcal{S}\rangle\}$.

Consider, for example, a simple Majorana code, generated by just one stabilizer:
\begin{equation}
\label{eq:S_even}
\mathcal{S}_{\text{even}}=\langle i^{m/2}\gamma_0\gamma_1\dots\gamma_{m-1}\rangle.
\end{equation}
The phase $i^{m/2}$ guarantees this stabilizer squares to $I$ and not $-I$. This code encodes $k=(m/2 - 1)$ qubits into $m$ Majoranas with distance $2$, because the centralizer consists of all $\gamma_{a}$ with even weight. The code $\mathcal{S}_{\text{even}}$ will play an important role in our converting Majorana codes defined on graphs to qubit stabilizer codes defined on graphs.


\subsection{Majorana surface codes defined on embedded graphs}
\label{subsec:maj_codes_on_graphs}


\edit{This subsection defines Majorana surface codes that generalize those in prior works \cite{wen2003quantum,kitaev2006anyons,bravyi2010majorana,vijay2015majorana} by starting with an arbitrary} embedded graph $G$ given by a rotation system $R$. We define a Majorana stabilizer code $\mathcal{S}(R)$ (or equivalently $\mathcal{S}(G)$) by placing Majorana operators on each half-edge and each odd-degree vertex, and associating stabilizers to each vertex and face of the embedded graph. The formal definition is as follows.
\begin{definition}
\label{def:majorana_surface_code}
Given a rotation system $R=(H,\lambda,\rho,\tau)$ and the associated embedded graph $G=(V,E,F)$ with $M$ odd-degree vertices, associate a single Majorana $\gamma_{[h]_\tau}$ to each half-edge $[h]_\tau=\{h,\tau h\}\in H/\tau$ and a single Majorana $\bar{\gamma}_v$ to each odd-degree vertex in $v\in V$. Now to every vertex $v\in V$, associate a \textit{vertex stabilizer} $S_v$ by
\begin{equation}
\label{eq:s_v}
S_v=\bigg\{\begin{array}{ll}i^{\text{deg}(v)/2}\left(\prod_{[h]_\tau\subseteq v}\gamma_{[h]_\tau}\right),&\text{deg}(v)\text{ is even}\\i^{(\text{deg}(v)+1)/2}\left(\prod_{[h]_\tau\subseteq v}\gamma_{[h]_\tau}\right)\bar{\gamma}_v,&\text{deg}(v)\text{ is odd}\end{array}
\end{equation}
and to every face $f\in F$, a \textit{face stabilizer} $S_f$ by
\begin{equation}
\label{eq:s_f}
S_f=i^{|\{h\in f : \tau h \not \in f\}|/2}\prod_{\{h\in f : \tau h \not \in f\}}\gamma_{[h]_\tau}.
\end{equation}
The order of Majoranas in these products is important. To clarify this, we give each Majorana a unique label from $\{0,1,\dots,m-1\}$ satisfying certain rules, where $m = 2|E| + M$, and then the products are sorted by ascending order in these labels. When $G$ is not checkerboardable, the labeling can be arbitrary as long as it ensures that any two Majoranas on the same edge, e.g.~$\gamma_{[h]_\tau}$ and $\gamma_{[\lambda h]_{\tau}}$, have successive labels. When $G$ is checkerboardable, the labeling should additionally ensure that for all $0 \leq j \leq m/2-1$, the Majoranas corresponding to labels $2j+1$ and $2j+2$ (evaluating the labels modulo $m$) belong to half-edges which are subsets of the same vertex. We denote the group generated by all $S_v$ and $S_f$ as $\mathcal{S}(R)$ (or equivalently $\mathcal{S}(G)$). The Majorana code associated with $R$ (or $G$) has the stabilizer group $\mathcal{S}(R)$.
\end{definition}

We note a few easy consequences of this definition. First, notice that the total number of Majoranas is even since $M$ is, which is true because $2|E|=\sum_{v\in V}\text{deg}(v)$. Second, notice that the vertex and the face stabilizers have even weight, are Hermitian (since they square to $I$), and commute. The commutation follows from Eq.~\eqref{eq:majorana-commutation} since supports of $S_v$ and $S_{v'}$ do not overlap for distinct $v,v' \in V$, supports of $S_f$ and $S_{f'}$ for distinct $f,f' \in F$ overlap if and only if they share adjacent edges and therefore share an even number of Majoranas, and supports of $S_v$ and $S_f$ for $v \in V$ and $f \in F$ overlap only if $v$ and $f$ are both adjacent to one or more sectors in which case they share an even number of Majoranas (note that this holds even if $v$ and $f$ are both adjacent to any edge $e$ such that $e$ is not adjacent to any other face).

It turns out that the Majorana labeling scheme used in Definition~\ref{def:majorana_surface_code} ensures that the stabilizer group $\mathcal{S}(R)$ has the property $-I \not \in \mathcal{S}(R)$, which is formally proved below in Lemma~\ref{lem:stabilizer_dependence}(g), and since each stabilizer has even weight, so does every element of $\mathcal{S}(R)$; thus $\mathcal{S}(R)$ satisfies the two properties of a Majorana fermion code. If a different Majorana labeling was used, or if the products in Eqs.~\eqref{eq:s_v} and \eqref{eq:s_f} were ordered differently then it is not necessarily true that the first property holds, but in that case by Lemma~\ref{lem:majorana-group-props}(a) one can multiply the stabilizers by either $1$ or $-1$ and get a new set of stabilizers for which the property would still be true. By Lemma~\ref{lem:majorana-group-props}(b), $\mathcal{S}(R)$ is Abelian and Hermitian. Two examples of Majorana surface code stabilizers are shown in Fig.~\ref{fig:Maj_code_example}. 

We note briefly that a labeling scheme satisfying the demands of Definition~\ref{def:majorana_surface_code} exists in the checkerboardable case (the non-checkerboardable case is trivial). If $G$ is checkerboardable it has no odd-degree vertices, so there exists an Euler cycle $(e_0,e_1,\dots,e_{|E|-1})$ that uses all the edges in $E$. Majoranas can be labeled in the order they are encountered by following this cycle.

\begin{figure}[t]
    \centering
    \includegraphics[width=\textwidth]{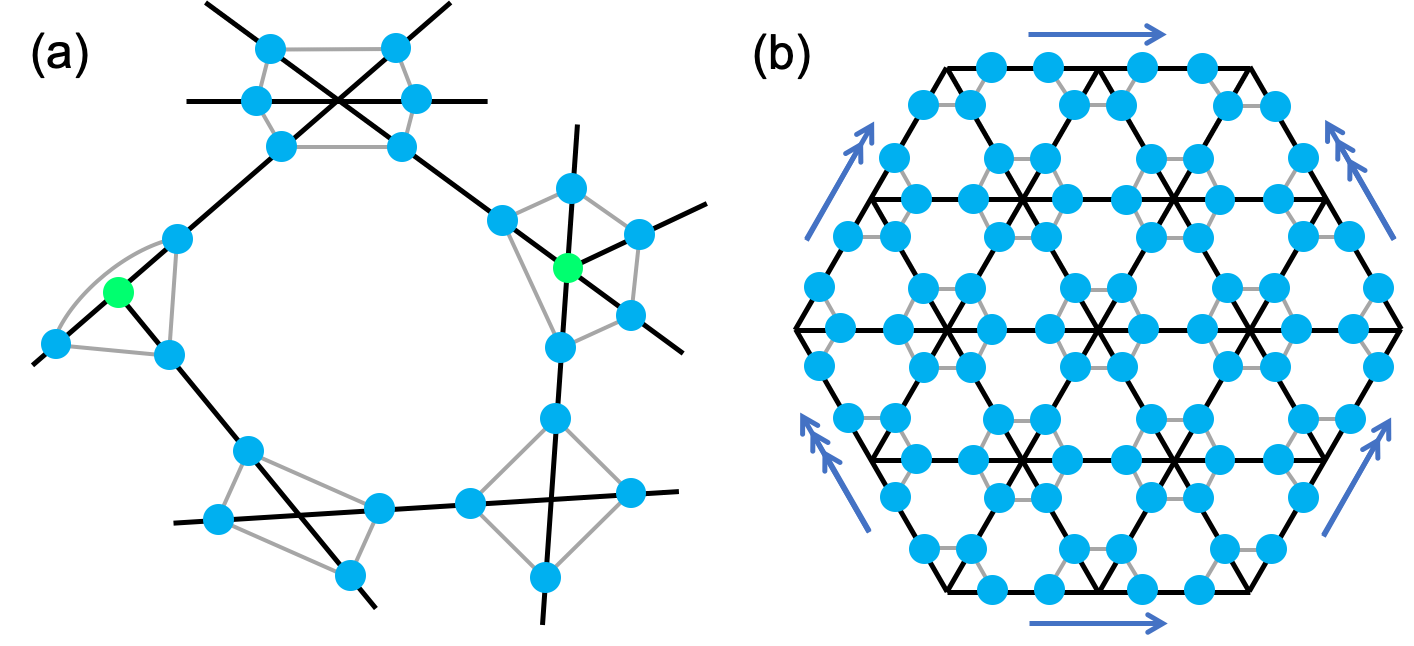}
    \caption{(a) An example section of a Majorana surface code. Majorana operators are placed on each half-edge (blue circles) and at each odd-degree vertex (green circles). The stabilizer associated to an even-degree vertex is the product of Majoranas around that vertex. The stabilizer associated to an odd-degree vertex is the product of Majoranas around that vertex and the Majorana located at that vertex. Gray lines cordon off the Majoronas involved in these stabilizers. The stabilizer associated to the pentagonal face is the product of the ten blue Majoranas on edges around that face. (b) A Majorana surface code of reference \cite{vijay2015majorana}. In our framework, it arises from a graph triangulating the torus. Here opposite sides of the hexagon are identified (directionally, as indicated by arrows) along with Majoranas on those edges. A total of 72 Majoranas survive this identification. All stabilizers associated to both vertices and faces are the product of six Majoranas.}
    \label{fig:Maj_code_example}
\end{figure}

The next lemma is useful to understand the dependence between the stabilizers of the Majorana surface code. \edit{Its proof is in Appendix~\ref{app:majorana-proofs}.}
\begin{lemma}
\label{lem:stabilizer_dependence}
The stabilizers of the Majorana surface code for the embedded graph $G(V,E,F)$ satisfy
\setmargins
\begin{enumerate}[(a)]
\item There is no non-empty subset $V'\subseteq V$ such that $\prod_{v\in V'}S_v= \pm I$.
\item There is no non-empty, proper subset $V'\subset V$ and no subset $F'\subseteq F$ such that $\prod_{v\in V'}S_v= \pm \prod_{f\in F'}S_f$. If there is any odd-degree vertex, then the statement holds for all non-empty subsets $V'\subseteq V$.
\item If $\prod_{v\in V}S_v=\prod_{f\in F'}S_f$ for some subset $F' \subseteq F$, then $G$ is checkerboardable. If $G$ is checkerboardable then $\prod_{v\in V}S_v=\prod_{f\in F'}S_f = \prod_{f\in F \setminus F'}S_f$ for a non-empty proper subset $F' \subset F$, which is determined uniquely up to taking complements.
\item There is no non-empty, proper subset $F'\subset F$ such that $\prod_{f\in F'}S_f= \pm I$.
\item $\prod_{f\in F}S_f=I$.
\item $|F|=1$ if and only if there is $f\in F$ such that $S_f=I$.
\item $-I \not \in \mathcal{S}(G)$.
\end{enumerate}
\end{lemma}

A consequence of Lemma~\ref{lem:stabilizer_dependence}(g) is that since $\mathcal{S}(R)$ is Abelian and Hermitian, it implies that a non-empty subset of the stabilizers is dependent if and only if the product of its elements is $I$. Let us determine the number of qubits encoded by this Majorana code. There are $|V|+|F|$ stabilizers as defined. However, not all these stabilizers are independent. If the total number of independent stabilizers is $|V|+|F|-\alpha$ for a non-negative integer $\alpha$, then the number of encoded qubits is
\begin{equation}
\label{eq:K_prelim}
K = (2|E|+M)/2 - \left(|V|+|F|-\alpha\right)=M/2-\chi+\alpha,
\end{equation}
using Eq.~\eqref{eq:Eulers_formula}. The remaining work is in determining $\alpha$ by counting dependencies in the stabilizers, which can be done easily using Lemma~\ref{lem:stabilizer_dependence}.

\begin{theorem}
\label{thm:number_encoded_qubits}
A Majorana surface code encodes
\begin{equation}
K=\left\{\begin{array}{ll}2g&,\text{ orientable}\\g&,\text{ non-orientable}\end{array}\right\}+\left\{\begin{array}{ll}0&,\text{ checkerboardable}\\(M-2)/2&,\text{ not checkerboardable}\end{array}\right\}
\end{equation}
qubits, where conditions in brackets are properties of the rotation system $(H,\lambda,\rho,\tau)$, or equivalently the graph embedding, defining the code.
\end{theorem}

\begin{proof}
We will argue that (i) $\alpha \in \{1,2\}$, and (ii) $\alpha = 2$ if and only if the graph is checkerboardable. Once we have done so and because a checkerboardable graph does not contain any odd degree vertices, we can use Eq.~\eqref{eq:K_prelim} to complete the proof. 

To prove (i), let $\mathcal{I} = \{S_v : v \in V\} \cup \{S_f : f \in F\}$ be the set of stabilizers, so $\langle \mathcal{I} \rangle = \mathcal{S}(R)$. By Lemma~\ref{lem:stabilizer_dependence}(e), one can always remove a single face stabilizer $S_f$ so that $\langle \mathcal{I} \setminus S_f \rangle = \mathcal{S}(R)$, which implies $\alpha \geq 1$. If $\mathcal{I} \setminus S_f$ is not independent, by Lemma~\ref{lem:stabilizer_dependence}(a,b,d) it must be that $\prod_{v\in V}S_v = \prod_{f\in F'}S_f$ for some non-empty subset $F' \subseteq F \setminus f$. If this happens then we remove a vertex stabilizer $S_v$ arbitrarily after which Lemma~\ref{lem:stabilizer_dependence}(b) guarantees that the set $\mathcal{I} \setminus \{S_f, S_v\}$ is independent. Thus $\alpha \leq 2$. To prove (ii), by the preceding arguments $\alpha = 2$ if and only if $\prod_{v\in V}S_v = \prod_{f\in F'}S_f$ for some non-empty proper subset $F' \subset F$, and the latter is true if and only if the graph is checkerboardable by Lemma~\ref{lem:stabilizer_dependence}(c).
\end{proof}

We also point out that if the graph defining the Majorana surface code has vertices with degree one or degree two, then there are vertex stabilizers $S_v$ that are the product of just two modes. These two modes and the stabilizer that is their product can be removed from the code without affecting the number of encoded qubits or code distance. What may be affected is the degree of protection arising from superselection rules \cite{bravyi2010majorana}, e.g.~Kitaev's 1D chain \cite{kitaev2001unpaired} has many such weight-two stabilizers. Nevertheless, our main goal is to convert these Majorana surface codes to qubit codes where superselection is not a relevant protection. Thus, we assume from now on that our graphs have no vertices of degree less than three.

\subsection{From Majorana surface codes to qubit surface codes}
\label{subsec:qub_surface_codes}
\edit{In this subsection we replace Majorana operators with Paulis to make qubit surface codes on arbitrary graphs, thus slightly generalizing an idea from \cite{wen2003quantum,kitaev2006anyons,bravyi2018correcting}.} \edit{First,} recognize that at each even (resp.~odd) degree vertex $v$, because of the vertex stabilizer $S_v$, there is a copy of the code $\mathcal{S}_{\text{even}}$, Eq.~\eqref{eq:S_even}, defined on $\text{deg}(v)$ (resp.~$\text{deg}(v)+1$) distinct Majorana modes and encoding $(\text{deg}(v)-2)/2$ (resp.~$(\text{deg}(v)-1)/2$) qubits. Therefore, to define the qubit code given an embedded graph $G=(V,E,F)$, the first step is to place these many qubits at each vertex. We already assumed the graph has vertices of degree at least three, so there is at least one qubit at each vertex.

Now, consider a single vertex $v$. By definition, Paulis that act only on the qubits at $v$ can be associated to logical operators of the code $\mathcal{S}_{\text{even}}$ located at $v$. We are most interested in the Paulis associated to the sector operators $q_{[h]_\rho}=i\gamma_{[h]_\tau}\gamma_{[\rho h]_\tau}$ for each flag $h\in v$. That is, $q_{[h]_\rho}$ is the product of the two Majoranas on half-edges adjacent to the sector. We are interested particularly in $q_{[h]_\rho}$ because the face stabilizers of the Majorana code can be alternatively written as the product of these sector operators, 
\begin{equation}
\label{eq:face_stabilizer_as_sectors}
S_f=\pm\prod_{[h]_\rho\subseteq f}q_{[h]_\rho}.
\end{equation}
Therefore, choosing Paulis $q_{[h]_\rho}$ will also give $S_f$ in terms of Paulis.

Let us list the sector operators $\{q_{[h]_\rho},q_{[\tau\rho h]_\rho},\dots,q_{[(\tau\rho)^{\text{deg}(v)-1}h]_\rho}\}$ starting from some arbitrary $h\in v$. Notice that these operators have very specific commutation rules, with adjacent (and the first and last elements also considered adjacent) elements anticommuting with one another, and all other pairs commuting. This leads us to study the following lists of Paulis.

\begin{definition}
\label{def:CAL}
A \textit{cyclically anticommuting list} (CAL)  acting on $n$ qubits is a list of $n$-qubit Paulis $\{p_0,p_1,\dots,p_{\ell-1}\}$, in which for distinct $i$ and $j$, $p_i$ and $p_j$ anticommute if and only if $i=j \pm1 \mod \ell$. A CAL of length $\ell\ge 1$ is called \textit{extremal} if there does not exist a CAL of the same length acting on fewer qubits.
\end{definition}

Clearly, the list of sector operators $\{q_{[h]_\rho},q_{[\tau\rho h]_\rho},\dots,q_{[(\tau\rho)^{\text{deg}(v)-1}h]_\rho}\}$ is a CAL of length $\text{deg}(v)$ acting on $(\text{deg}(v)-2)/2$ qubits if $v$ has even degree or acting on $(\text{deg}(v)-1)/2$ qubits if $v$ has odd degree. Let us establish existence of CALs with a construction.

\begin{theorem}
\label{thm:CAL_construct}
A CAL of length $\ell \ge 3$ acting on $n$ qubits in which the Paulis have weight at most two exists if and only if
\begin{equation}
\label{eq:minimal_qubits_CAL}
n\ge\bigg\{\begin{array}{ll}(\ell-2)/2,& \ell \text{ even},\\
(\ell-1)/2,& \ell \text{ odd}.\end{array}
\end{equation}
\end{theorem}

\begin{proof}
We construct extremal CALs with length $\ell$ acting on $n$ qubits saturating the lower bound in Eq.~\eqref{eq:minimal_qubits_CAL}. To get non-extremal CALs, more qubits can always be added by acting on them with identity, or just Pauli $Z$, for instance, so long as the commutation relations are unaffected. We note that extremal CALs of length three and four, acting on a single qubit, are $\{X,Y,Z\}$ and $\{X,Z,X,Z\}$. To construct extremal CALs of longer length, we define a composition operation $\circ$ that takes a CAL of length $\ell$ (even) and a CAL of length $\ell'$ (either even or odd) to a CAL of length $\ell+\ell'-2$:
\begin{equation}
\{p_0,\dots,p_{\ell-1}\}\circ\{p'_0,\dots,p'_{\ell'-1}\} := \{p_0,\dots,p_{\ell/2-2},p_{\ell/2-1}p'_0,p'_1,p'_2,\dots,p'_{\ell'-2},p'_{\ell'-1}p_{\ell/2},p_{\ell/2+1}\dots p_{\ell-1}\},
\end{equation}
where the $p_j$s and $p'_j$s represent Paulis that act on $n$ and $n'$ qubits, respectively. Thus the CAL on the right acts on $n+n'$ qubits. Repeated application of this composition suffices to make CALs of any length out of the length three and four CALs given above. Composition preserves extremality --- if $n$ and $n'$ are minimal qubit counts for CALs of lengths $\ell$ and $\ell'$, then
\begin{equation}
n+n'=\bigg\{\begin{array}{ll}(\ell+\ell'-4)/2,& \ell'\text{ even},\\
(\ell+\ell'-3)/2,& \ell'\text{ odd}
\end{array}
\end{equation}
is the minimum number of qubits needed for a length $\ell+\ell'-2$ CAL. A graphical depiction of this composition operation is shown in Fig.~\ref{fig:spinal_expansion}. It should be clear that the construction results in CALs consisting of Paulis of weight just one or two.

\begin{figure}[t!]
    \centering
    \includegraphics[width=0.8\textwidth]{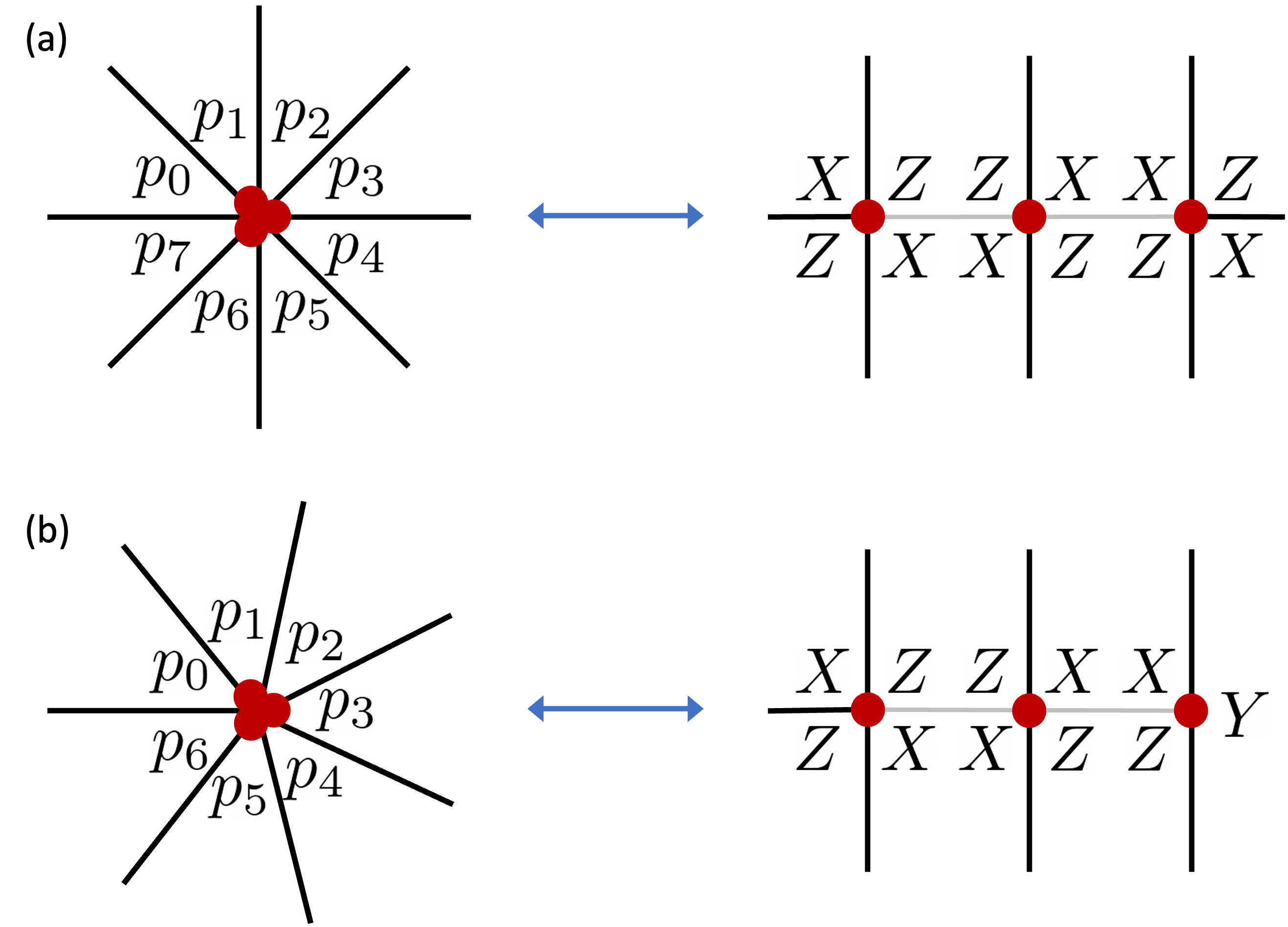}
    \caption{One way to create a CAL of length $\ell$ is to compose CALs of lengths three or four together. In (a), we compose three CALs with length four (right) to make a CAL with length eight (left) acting on three qubits (red circles). In (b), we use a similar composition of two length-four CALs and one length-three CAL to make a length-seven CAL. \edit{In both parts, to go from right to left, shrink the interior, light edges until all three qubits occupy the central vertex. Outgoing dark edges are meant to be identified in cyclical fashion between left and right, which in turn determines the identification of Paulis on the left and right.} In terms of surface codes defined on graphs, this composition operation means that we can always decompose higher degree vertices into vertices with degrees just three or four. Still, we generally will include higher degree vertices in our drawings to highlight symmetries when possible.}
    \label{fig:spinal_expansion}
\end{figure}

There are no CALs of length $\ell$ acting on fewer qubits than in Eq.~\eqref{eq:minimal_qubits_CAL}. Suppose, by way of contradiction, that a CAL $\{p_0,p_1,\dots,p_{\ell-1}\}$ acting on $n$ qubits has length $\ell>2n+2$. Then, there is a set of Paulis 
\begin{equation}
\left\{a_j=\prod_{k=0}^{j}p_k:j\in\{0,1,\dots,\ell-2\}\right\}
\end{equation}
with size $>2n+1$ in which all Paulis anticommute with one another. However, this contradicts known results on the maximum size of such anticommuting sets \cite{sarkar-vandenberg2019}.
\end{proof}

CALs can be thoroughly studied in their own right. The interested reader can find much more discussion of CALs in Appendix~\ref{sec:cal-prop}. For instance, we show some structural properties of extremal CALs, such as the following lemmas which we find useful (see Corollary~\ref{cor:cal-dim-extremal} and Corollary~\ref{cor:extremal-cal-mult-identity}).

\begin{lemma}
\label{lem:CAL_generating}
If $\{p_0,p_1,\dots,p_{\ell-1}\}$ is an extremal CAL of length $\ell \geq 3$ acting on $n$ qubits, then taking products of the $p_i$ generates the entire $n$-qubit Pauli group (up to global phases).
\end{lemma}

\begin{lemma}
\label{lem:CAL_identities}
Suppose $\{p_0,p_1,\dots,p_{\ell-1}\}$ is an extremal CAL of length $\ell \geq 3$. Then $\prod_{j=0}^{\ell-1}p_j\propto I$ and, if $\ell$ is even, $\prod_{j=0}^{\ell/2-1}p_{2j}\propto I$ and $\prod_{j=0}^{\ell/2-1}p_{2j+1}\propto I$. Moreover, these are the only products that are proportional to identity.
\end{lemma}

These lemmas lead to the following simple corollary.
\begin{corollary}
\label{cor:CAL_commutation}
Let $\mathcal{C} = \{p_0,p_1,\dots,p_{\ell-1}\}$ be an extremal CAL of length $\ell \geq 3$ acting on $n$ qubits. For any Pauli $q$ acting on $n$ qubits, define the (row) vector $C_q\in\mathbb{F}_2^{\ell}$ whose $j^{\text{th}}$ element is $1$ if and only if $q$ and $p_j$ anticommute. If $\ell$ is odd define $M^\top=\left(\begin{smallmatrix}1&1&\dots&1&1\end{smallmatrix}\right) \in \mathbb{F}_2^{1 \times \ell}$ and if $\ell$ is even $M^\top=\left(\begin{smallmatrix}1&1&\dots&1&1\\1&0&\dots&1&0\end{smallmatrix}\right) \in \mathbb{F}_2^{2 \times \ell}$. Then $C_q M=\vec0$. Moreover, if $x \in \mathbb{F}_2^{\ell}$ satisfies $x M=\vec0$, then there exists a (unique, up to phase) Pauli $q$ such that $x=C_q$. Finally, $C_q=\vec0$ if and only if $q\propto I$. 
\end{corollary}

\begin{proof}
Lemma~\ref{lem:CAL_identities} implies that the last one (two) Paulis of a CAL are not independent from the others if $\ell$ is odd (even) and that these are the only dependencies. These dependencies result in $C_q M=0$. Lemma~\ref{lem:CAL_generating} implies that the first $\ell-1$ ($\ell-2$) independent Paulis of the CAL generate the entire $n$-qubit Pauli group. Therefore, we are guaranteed the existence of a (unique, up to phase) Pauli $q$ satisfying the commutation relations expressed by the first $\ell-1$ ($\ell-2$) bits of $x$ and, since $x M=\vec 0$, the final one (two) bits of $x$ must be consistent with those commutations. Finally, $C_q=\vec 0$ implies $q$ commutes with the entire Pauli group, which is possible if and only if $q\propto I$. 
\end{proof}

With CALs forming the basis for the logical space of the Majorana code $\mathcal{S}_{\text{even}}$ we can complete our construction of qubit surface codes.
\begin{definition}\label{def:qubit_surface_code}
Suppose we have an embedded graph $G$ described by some rotation system $R=(H,\lambda,\rho,\tau)$ and the vertices of $G$ are at least degree three. At vertex $v$, place $N_v=\lceil(\text{deg}(v)-2)/2\rceil$ qubits, and associate a Hermitian Pauli $q_{[h]_\rho}$ to each sector $[h]_\rho$ around that vertex. These Paulis $\{q_{[h]_\rho},q_{[\tau\rho h]_\rho},\dots,q_{[(\tau\rho)^{\text{deg}(v)-1}h]_\rho}\}$ should form an extremal CAL of length $\text{deg}(v)$ acting on the $N_v$ qubits at $v$. To every face $f$ of the graph we associate a stabilizer, the product of all Paulis associated to sectors $[h]_\rho\subseteq f$.
\end{definition}

By the discussion starting this subsection, this definition ensures the codespaces of the Majorana code associated to $G$ (Definition~\ref{def:majorana_surface_code}) and the qubit surface code associated to $G$ (Definition~\ref{def:qubit_surface_code}) are the same, both encoding $K$ qubits from Theorem~\ref{thm:number_encoded_qubits}, leading to the following corollary, which is formally proved in Appendix~\ref{app:majorana-qubit-code-equivalence}.

\begin{corollary}
\label{cor:number_encoded_qubits}
A qubit surface code encodes
\begin{equation}
K=\left\{\begin{array}{ll}2g&,\text{ orientable}\\g&,\text{ non-orientable}\end{array}\right\}+\left\{\begin{array}{ll}0&,\text{ checkerboardable}\\(M-2)/2&,\text{ not checkerboardable}\end{array}\right\}
\end{equation}
qubits, where conditions in brackets are properties of the rotation system $(H,\lambda,\rho,\tau)$, or equivalently the graph embedding, defining the code.
\end{corollary}

Fig.~\ref{fig:Qub_code_example}, the qubit versions of the Majorana codes from Fig.~\ref{fig:Maj_code_example}, should help in understanding the correspondence of Majorana and qubit surface codes. Further examples of qubit surface codes fitting our definition and appearing in prior literature are shown in Fig.~\ref{fig:stellated_codes}.

Notice that while a Majorana surface code has stabilizers associated to both vertices and faces, a qubit surface code only has stabilizers associated to faces. This is because we introduced just the right number of qubits at each vertex to automatically enforce the vertex stabilizers of the Majorana surface code. Alternatively, we introduced just enough physical qubits to span the codespaces of the Majorana codes $\mathcal{S}_{\text{even}}$ that exist at each vertex. This fact can also be illustrated with a simple qubit count -- with one qubit for every two Majoranas (e.g.~using the Jordan-Wigner transformation, Eq.~\eqref{eq:JW}) we would have $(2|E|+M)/2=|E|+M/2$ qubits (recall $M$ is the number of odd-degree vertices), but instead we have used
\begin{equation}
N=\sum_{v\in V}N_v=M/2+\sum_{v\in V}\left(\text{deg}(v)-2\right)/2=|E|-|V|+M/2
\end{equation}
qubits. Each of the $|V|$ ``missing" qubits is a degree of freedom eliminated by enforcing a stabilizer $S_v$.

For stabilizer code with stabilizer group $\mathcal{S}$, the \emph{logical operators} are the elements of the centralizer $\mathcal{C}(\mathcal{S})$, i.e.~the group of all Pauli operators that commute with all elements of $\mathcal{S}$. The centralizer includes the stabilizer group itself, and so we refer to elements of $\mathcal{S}$ as trivial logical operators as they apply the logical identity to encoded qubits. The code distance $D$ of a stabilizer code is the minimum weight of an element of $\mathcal{C}(\mathcal{S})\setminus\langle\{iI\}\cup\mathcal{S}\rangle$. We use the notation $\llbracket N,K,D\rrbracket$ to concisely present the code parameters of a stabilizer code.

\begin{figure}[t]
    \centering
    \includegraphics[width=\textwidth]{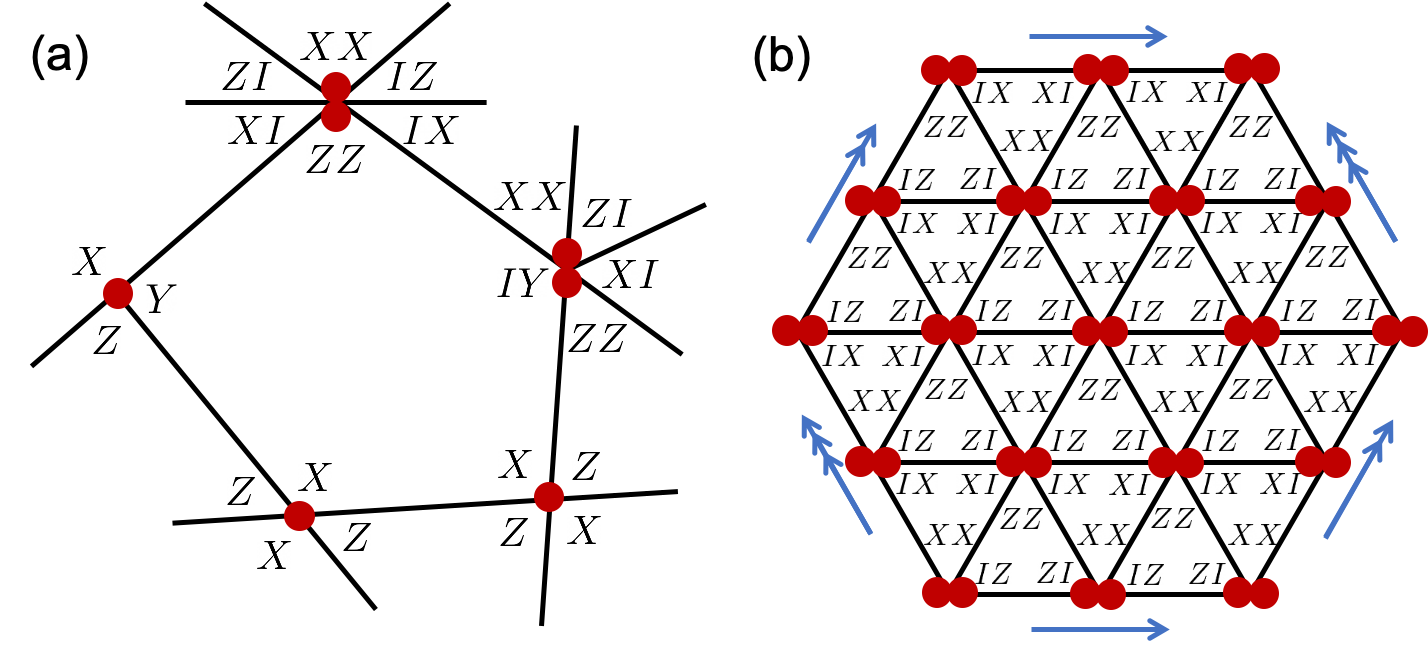}
    \caption{Qubit versions of the Majorana codes in Fig.~\ref{fig:Maj_code_example}. Place a single qubit on vertices with degrees three or four and a pair of qubits on vertices with degrees five or six. Around each vertex write a cyclically anticommuting list of Paulis (acting on qubits at that vertex). Each face represents a stabilizer defined as the product of all Paulis written within that face. Note that (b) depicts a toric code with $X^{\otimes4}$ and $Z^{\otimes4}$ stabilizers but on a lattice different from Kitaev's square lattice \cite{kitaev2003fault}. In this case, the code is $\llbracket24,2,4\rrbracket$. The number of encoded qubits is clear by Corollary~\ref{cor:number_encoded_qubits} as this graph embedding is checkerboardable. 
    }
    \label{fig:Qub_code_example}
\end{figure}

\begin{figure}[t!]
    \centering
    \includegraphics[width=\textwidth]{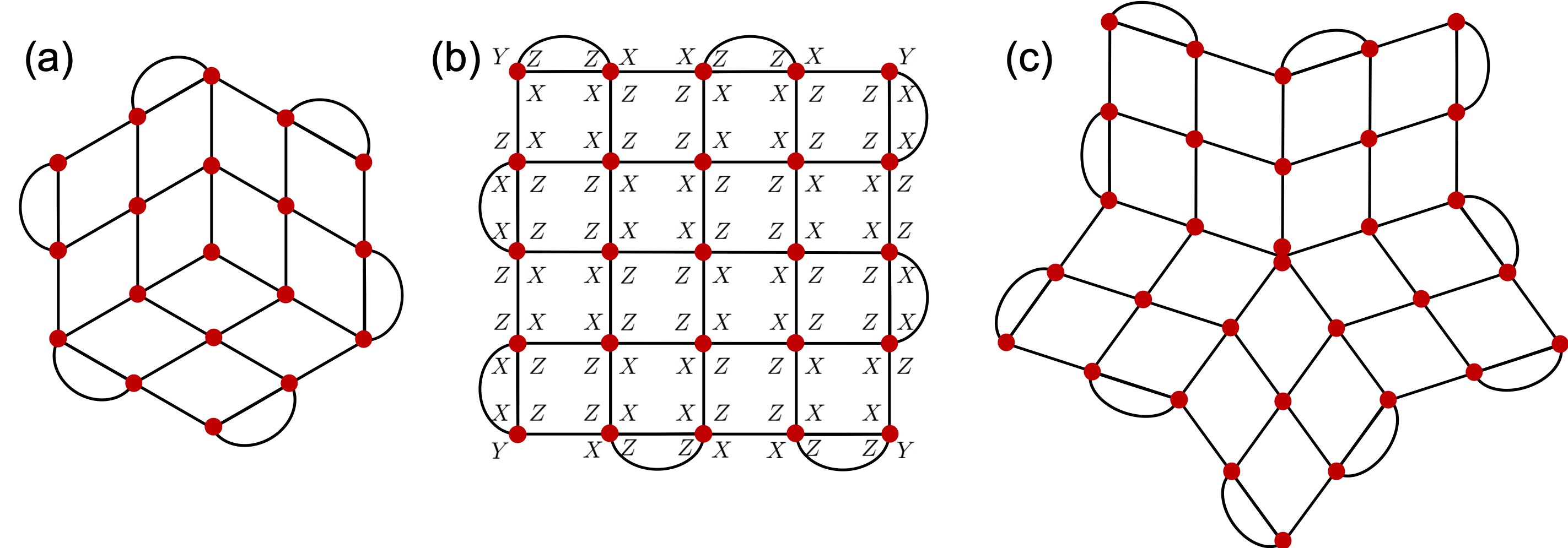}
    \caption{(a) The triangular surface code \cite{yoder2017surface}, (b) the rotated surface code \cite{bombin2007optimal} and (c) a stellated surface code \cite{kesselring2018boundaries} with even greater symmetry. Each code fits in our framework -- in this case, by Definition~\ref{def:qubit_surface_code} applied to these planar graphs. In (b), we show explicitly the assignment of CALs that gives the familiar rotated surface code. Notice that the outer face also gets assigned a stabilizer. Because of Lemma~\ref{lem:stabilizer_dependence}(e), the product of stabilizers on all non-outer faces is proportional to the stabilizer on the outer face.}
    \label{fig:stellated_codes}
\end{figure}


To further demonstrate Definition~\ref{def:qubit_surface_code}, in Fig.~\ref{fig:5-qubit_family} we show a family of topological codes that includes the well-known $\llbracket5,1,3\rrbracket$ code as its smallest member. Also included as a subset are the cyclic codes given in Example 11 and Figure 3 of \cite{kovalev2011low}. The general construction gives cyclic codes defined on the torus. Consider the typical fundamental square for the torus, i.e.~a unit square with opposite sides identified, and draw the lines $y=bx/a$ and $y=-ax/b$, where we assume $b > a\ge 1$ are integers and $\text{gcd}(a,b)=1$. These two lines intersect at $N=a^2+b^2$ points within the fundamental square, including the single point at $(0,0)=(0,1)=(1,0)=(1,1)$. Interpret these intersections as vertices of a 4-regular graph with edges the line segments between them. With correct choice of CALs, so that each face stabilizer has two $X$s and two $Z$s (see Fig.~\ref{fig:5-qubit_family} for the convention we use), and with qubits labeled in a sequence along the line $y=bx/a$ (moving outwards from the origin in the direction of increasing $x$), the surface code defined by this graph and Definition~\ref{def:qubit_surface_code} is cyclic --- an overcomplete generating set of the stabilizer group consists of $Z\otimes X\otimes I^{\otimes s}\otimes X\otimes Z\otimes I^{N-s-4}$ and its cyclic permutations, where $s=\lceil Nt/b\rceil-2$ 
and $t$ is the minimal positive 
integer such that $(ta+1)/b$ 
is an integer.\footnote{The extended Euclidean algorithm on $(a,b)$ can be used to determine $t$. A B\'{e}zout pair $(x,y)$ of $(a,b)$ is a pair of integers such that $xa+yb=\text{gcd}(a,b)$. In particular, $(x,y)=(-t,(ta+1)/b)$ is the unique B\'{e}zout pair of $(a,b)$ in which $-b\le x\le0$.} We have the following parameters for these codes.

\begin{figure}[t!]
    \centering
    \includegraphics[width=\textwidth]{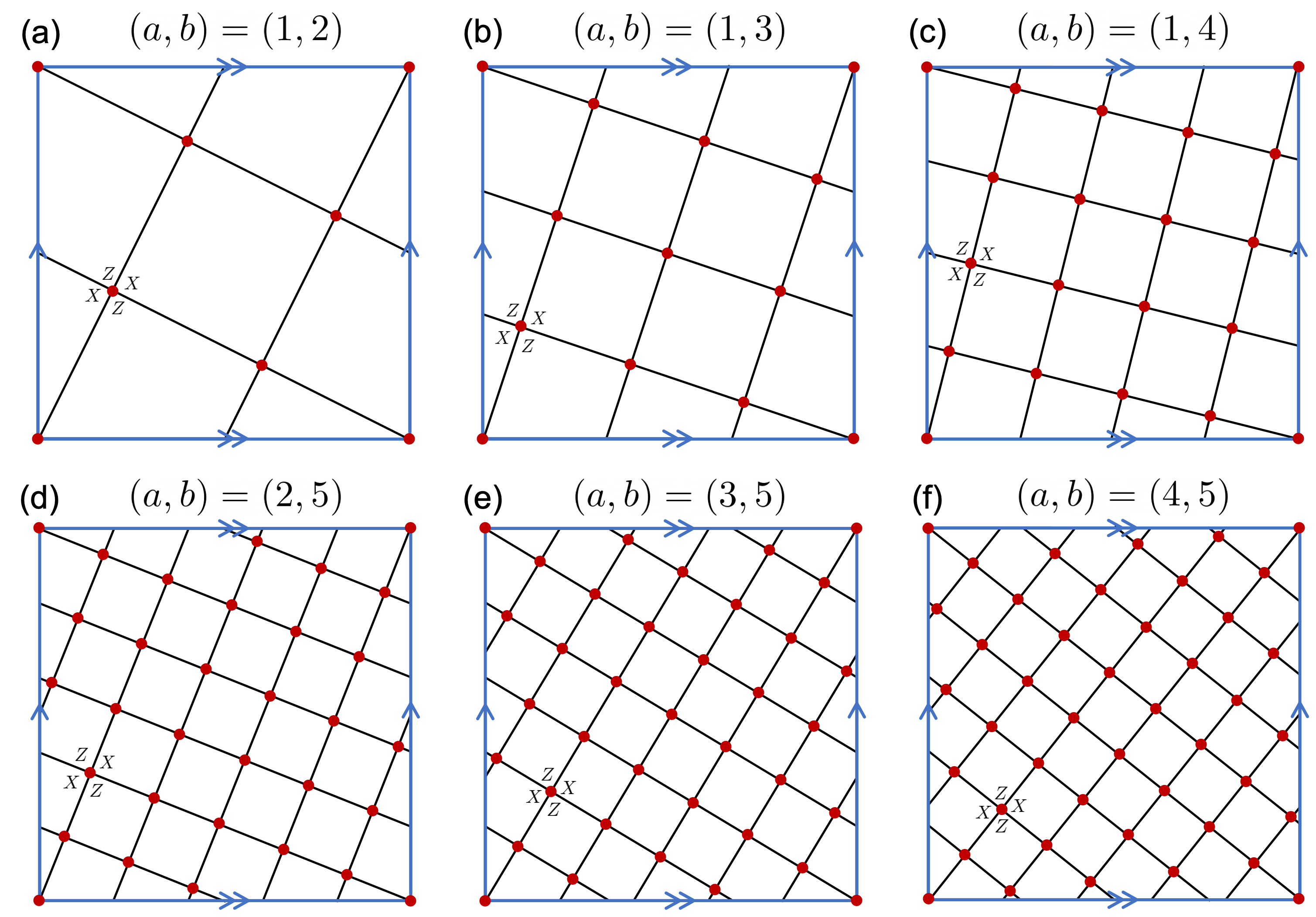}
    \caption{(a) The smallest error-correcting quantum code, the $\llbracket5,1,3\rrbracket$ code, is a surface code defined by an embedding of the 5-vertex complete graph $K_5$ in the torus. There are different embeddings of $K_5$ in the torus, but no other gives a 5-qubit code with distance more than two. Parts (b-f) show members of the cyclic toric code family with more qubits. See Theorem~\ref{thm:cyclic_toric_codes} and the paragraph above it for the code definition and parameters. A demonstrative CAL is written around one vertex in each graph. The codes are cyclic when this same CAL is orientated similarly around each vertex.}
    \label{fig:5-qubit_family}
\end{figure}

\begin{minipage}{\textwidth}
\begin{theorem}
\label{thm:cyclic_toric_codes}
The cyclic toric code with integer parameters $a$ and $b$, where without loss of generality $\text{gcd}(a,b)=1$ and $b > a\ge1$, is a $\llbracket N=a^2+b^2,K,D\rrbracket$ code, with
\begin{enumerate}[(a)]
\item If $N$ is odd, $K=1$ and $D=a+b$,
\item If $N$ is even, $K=2$ and $D=b$.
\end{enumerate}
\end{theorem}
\end{minipage}

The odd and even cases differ in the number of encoded qubits because of checkerboardability (the latter case being checkerboardable and the former not). Note that a consequence of the theorem is that the code distance $D$ is always odd. The codes in Example 11 and Figure 3 of \cite{kovalev2011low} are cyclic toric codes with $a=b-1$.

We delay the proof of this theorem, in particular the code distances, until Section~\ref{sec:rotated_toric_codes}, after we have developed sufficient tools in Section~\ref{sec:locating_logical_operators}. We point out that in two situations, the cyclic toric codes achieve $N=\frac12KD^2+O(D)$ -- when $a=b-1=(D-1)/2$ and when $b=D$ is odd with $a=1$. Moreover, $1/2$ is the best constant achievable for the cyclic toric codes.

To conclude this section, we point out a corollary on the structure of qubit surface codes that follows from the proof of Theorem~\ref{thm:CAL_construct}, and in particular the vertex decomposition demonstrated in Fig.~\ref{fig:spinal_expansion}. We use this corollary later to justify considering only graphs of low degree.

\begin{corollary}\label{cor:simplify_graph}
For any graph $G$, there is a choice of CALs such that the surface code associated with $G$ is the same as the code associated with another graph $G'$ that only has vertices with degrees three and four.
\end{corollary}

\subsection{Relation to homological surface codes}
\label{subsec:relate_to_homological}
Finally, we want to point out the relation of Definition~\ref{def:qubit_surface_code} with the well-known homological construction of surface codes \cite{kitaev2003fault,bravyi1998quantum,freedman2001projective}. In the homological construction, a single qubit is placed on each edge of a graph embedding $G=(V,E,F)$. Stabilizers correspond to each vertex $v\in V$ and face $f\in F$:
\begin{equation}
S'_v=\prod_{e\in v}X_e,\quad S'_f=\prod_{e\in f}Z_e,
\end{equation}
where $e\in v$ and $e\in f$ are shorthand for saying edge $e$ is adjacent to the vertex $v$ or face $f$, and of course $X_e$, $Z_e$ are Paulis on the qubit on that edge.

The \emph{medial graph} $\widetilde G$ of an embedded graph $G=(V,E,F)$ has a vertex for each edge of $G$ and connects them with edges such that vertices and faces of $G$ form faces of $\widetilde G$. See some examples in Fig.~\ref{fig:medial_examples}. It should be clear that the homological surface code defined on graph $G$ is exactly the same code as our surface code, Definition~\ref{def:qubit_surface_code}, defined on the medial graph $\widetilde G$. Thus, Definition~\ref{def:qubit_surface_code} includes all homological codes.

We can formally define the medial graph using rotation systems.
\begin{definition}\label{def:medial_graph}
The \emph{medial rotation system} $\widetilde R=(\widetilde H,\widetilde \lambda,\widetilde \rho, \widetilde \tau)$ of a rotation system $R=(H,\lambda,\rho,\tau)$ is defined by setting $\widetilde H=H\times\{1,-1\}$ and for $(h,j)\in\widetilde H$ letting $\widetilde\tau(h,j)=(h,-j)$, $\widetilde\lambda(h,j)=(\rho(h),j)$, and
\begin{equation}
\widetilde\rho(h,j)=\bigg\{\begin{array}{ll}(\lambda(h),j),&\text{if }j=-1\\
(\tau(h),j),&\text{if }j=1
\end{array}.
\end{equation}
\end{definition}
Some facts about medial graphs become apparent (we prove these in Appendix~\ref{app:medial-graphs-and-trisection}). First, the medial graph embeds into the same manifold as the original graph. Second, medial graphs are also always 4-valent (every vertex is degree four). Third, medial graphs are always checkerboardable -- e.g.~a natural way to color faces in $\widetilde G$ is to color black faces that correspond to vertices of $G$ and color white faces corresponding to faces of $G$. In fact, an embedded graph is the medial graph of another embedded graph if and only if it is 4-valent and checkerboardable. As a result, there are codes resulting from Definition~\ref{def:qubit_surface_code} that cannot be described as homological codes.

The reader may also refer to \cite{bombin2007optimal} where checkerboardable, 4-valent graphs (i.e.~medial graphs) are used to define homological surface codes directly and to \cite{anderson2013homological} where this connection between homological codes and medial graphs is also pointed out.

\begin{figure}
    \centering
    \includegraphics[width=\textwidth]{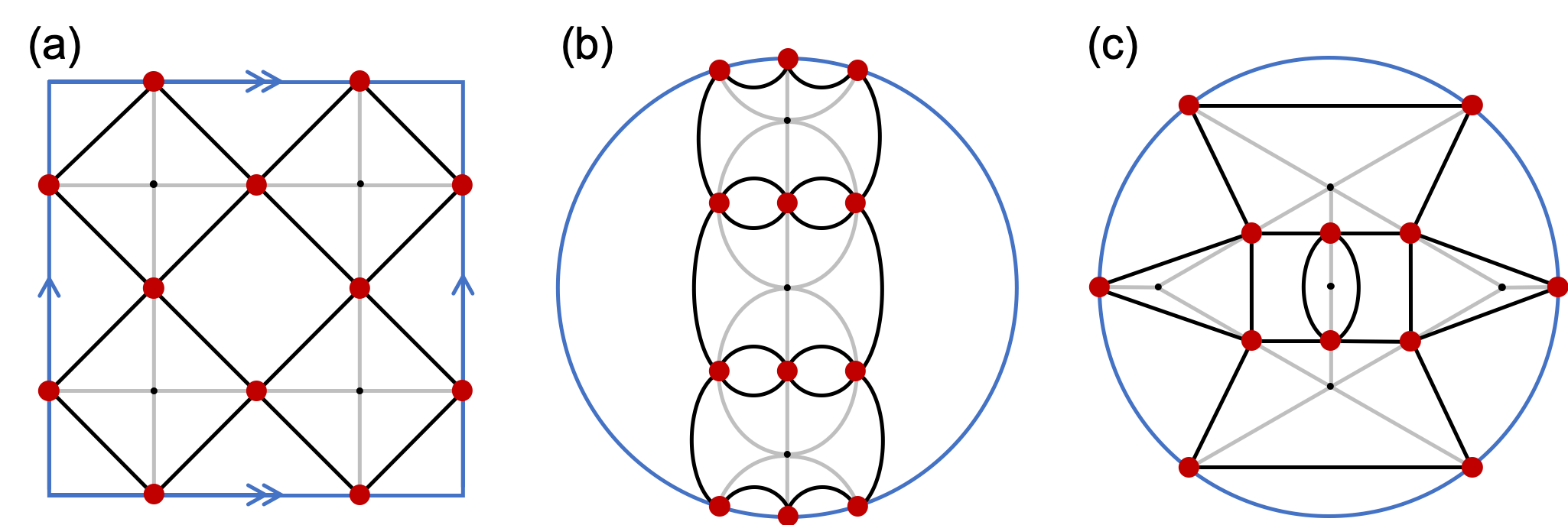}
    \caption{Graphs (gray edges with small, black vertices) on the torus (a) and projective plane, (b) and (c), are drawn with their medial graphs (black edges and large, red vertices). In (a), we get a $\llbracket8,2,2\rrbracket$ toric code. In (b), we get Shor's $\llbracket9,1,3\rrbracket$ code \cite{shor1995scheme} with the gray graph coming from \cite{freedman2001projective}. In (c), we show another $\llbracket9,1,3\rrbracket$ code from \cite{freedman2001projective}. }
    \label{fig:medial_examples}
\end{figure}

Of course, forays have already been made beyond the strict homological surface code construction. We have already seen some examples in the forms of the triangular, rotated, and stellated codes in Fig.~\ref{fig:stellated_codes}. Our goal however is a more rigorous framework that generalizes these specific instances.

One way in which the homological surface code construction can be rigorously generalized to produce more general planar codes is described by Freedman and Meyer \cite{freedman2001projective}. Let us describe this construction using the language of medial graphs. Start with a homological surface code defined on an improper embedding of a graph in orientable surface of genus $g$. By an improper embedding, we mean one face is not an open disk but instead a $g$-punctured disk, e.g.~if $g=1$ the $g$-punctured face is a cylinder. Now, form the medial graph of this improperly embedded graph\footnote{We defined the medial graph of an embedded graph. In this case, the graph is not properly embedded. When we construct the medial graph here, we do so by drawing proper, disk-like faces around each vertex.}. The medial graph will also be improperly embedded, but it is still 4-valent and checkerboardable. Because of the redundancy in stabilizers (see proof of Theorem~\ref{thm:number_encoded_qubits}), we can delete, i.e.~remove from the stabilizer group, any single black as well any single white face without changing the code. Thus, delete the face that is $g$-punctured (say it is colored white), a black face sharing an adjacent edge, and any edges shared between them. The result is a planar graph with some degree-3 vertices (because of the deletion of edges). Examples producing the rotated surface code and the Bravyi-Kitaev planar code \cite{bravyi1998quantum} are shown in Fig.~\ref{fig:improper_embeddings}. 

\begin{figure}
    \centering
    \includegraphics[width=0.9\textwidth]{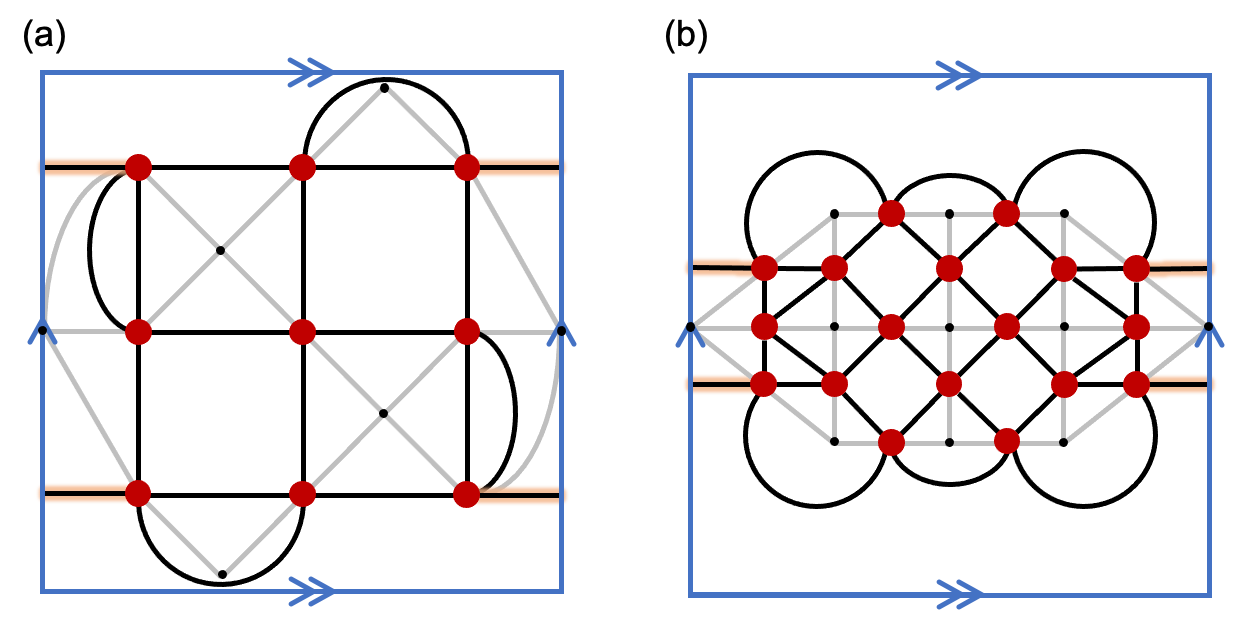}
    \caption{Improper embeddings of graphs (gray), and their medial graphs (black), on the torus. In both pictures, the face (of the medial graph embedding) intersecting the top and bottom frames is not an open disk but instead a cylinder. Deleting this face, the face intersecting the left and right frames, and their shared adjacent edges (highlighted) results in familiar planar codes. In (a) we find the $\llbracket9,1,3\rrbracket$ rotated surface code and in (b) Bravyi and Kitaev's planar code \cite{bravyi1998quantum}, in this instance $\llbracket18,1,3 \rrbracket$.}
    \label{fig:improper_embeddings}
\end{figure}

Notably, this improper embedding procedure does not produce non-checkerboardable graph embeddings in surface of larger genus $g>0$, so it does not produce the cyclic toric codes in Fig.~\ref{fig:5-qubit_family}(a,c,d,f) for instance. We have also not found a way to produce stellated codes with odd symmetry, e.g.~parts (a) and (c) of Fig.~\ref{fig:stellated_codes}, using the improper embedding technique, the apparent obstacle being the presence of an odd degree vertex in the center of the planar graph, i.e.~not adjacent to the outer face. This appears an obstacle because the outer face and the odd degree vertices are effectively created by deleting faces and their shared adjacent edges from the improper embedding, with the result being that any odd degree vertices are adjacent to the outer face. Hence, we find this construction involving improper embeddings and face deletions to be more complicated than our construction while also producing fewer codes.

%% file: Logical-operators-v4.tex
\section{Locating logical operators}
\label{sec:locating_logical_operators}

In this section, we discuss the structure of logical operators of the qubit surface codes of Definition~\ref{def:qubit_surface_code}. One of our main findings is that logical operators are exactly the homologically non-trivial cycles in a related graph called the decoding graph. This is probably not a huge surprise to someone familiar with homological codes, but there are some technical hurdles to overcome in proving it in our generalized setting. This includes the fact that the decoding graph for non-checkerboardable codes is not guaranteed to embed in the same manifold $\mathcal{M}$ as the original graph. Without a graph embedding, there is no notion of homology, so finding an appropriate embedding of the decoding graph is essential to the proof.

In Section~\ref{sec:checker_w_defects} we generalize checkerboardability to include \textit{defects}, useful in the subsequent sections. In Section~\ref{sec:dec_graph} we define the decoding graph, show that cycles in the decoding graph correspond to logical operators of the code (both trivial and non-trivial), and prove that it embeds in $\mathcal{M}$ in the checkerboardable case. In Section~\ref{sec:doubled_graphs}, we define a manifold which is generally of higher genus than $\mathcal{M}$, in which we can embed the decoding graph, and prove that homological non-triviality in this new manifold is equivalent to logical non-triviality. Finally, in Section~\ref{sec:canon_paths}, we discuss how some logical operators, including all trivial ones, can be represented by paths in the original graph.

\subsection{Checkerboardability with defects}\label{sec:checker_w_defects}
In Definition~\ref{def:checkerboardability} we defined what it means for a graph embedding to be checkerboardable. In this section, we generalize this notion to include defects. Informally put, a defect is a set of edges that can be deleted from the graph embedding to make it checkerboardable. Defects have a few uses in finding or identifying logical operators in surface codes which are covered in the subsequent subsections.

To define checkerboardability with defects more rigorously, we use the matrix $\Phi\in\mathbb{F}_2^{|F|\times|E|}$ from Section~\ref{subsec:checkerboardability}, which encodes the edges belonging to each face of the embedding.
\begin{definition}\label{def:checkerboardable_w_defect}
We say a graph embedding $G=(V,E,F)$ is checkerboardable with defect $\delta\in\mathbb{F}_2^{|E|}$ if there is a vector $x\in\mathbb{F}_2^{|F|}$ such that $x\Phi=\vec 1+\delta$.
\end{definition}
Clearly, being checkerboardable according to Definition~\ref{def:checkerboardability} is equivalent to being checkerboardable with the trivial defect $\delta=\vec0$. The analogous lemma to Lemma~\ref{lem:checkerboardability_rotation_system} is as follows.
\begin{lemma}\label{lem:checkerboardable_w_defect_rs}
$G$, described by rotation system $R=(H,\lambda,\rho,\tau)$, is checkerboardable with a defect if and only if we can partition $H$ into two sets $H_w$ and $H_b$ such that $\lambda$ and $\rho$ map both sets to themselves, while $\tau$ maps elements of either set to the other set, \emph{except} when those flags are part of an edge in the defect.
\end{lemma}

A given embedded graph does not have a unique defect. However, different choices of defect are related by the addition of rows of $\Phi$. That is, by simple application of the definition, the following is true.
\begin{lemma}
Suppose a graph embedding $G=(V,E,F)$ is checkerboardable with defect $\delta_1$. Then, $G$ is checkerboardable with defect $\delta_2$ as well if and only if there exists $x$ such that $x\Phi=\delta_1+\delta_2$.
\end{lemma}

Moreover, it is an interesting fact that a defect cannot be just any subset of edges. In fact, a defect must be a sum of certain \emph{trails} in the graph. Recall that $\mathcal{T}(G)\cong\mathbb{F}_2^{|E|}$ is the group of all trails in the graph. A subgroup of $\mathcal{T}(G)$ is generated by just the closed trails and trails whose endpoints are odd-degree vertices. We call this subgroup $\mathcal{T}_0(G)$. 

\begin{lemma}\label{lem:defect_decomp}
Suppose the graph embedding $G=(V,E,F)$ is checkerboardable with defect $\delta\in\mathbb{F}_2^{|E|}$. Then, $\delta\in\mathcal{T}_0(G)$. Let $a_v=\sum_{e\ni v}\delta_e$ evaluated over $\mathbb{F}_2$ indicate the parity of the number of edges that are adjacent to $v\in V$ and in the defect. Then, $a_v=1$ if and only if $\deg(v)$ is odd.
\end{lemma}
\begin{proof}
The way that $\delta$ can fail to be in $\mathcal{T}_0(G)$ is if $a_v=1$ for some even-degree vertex $v$. However, if this were the case, it would imply that $G$ is not checkerboardable with defect $\delta$ because it fails to be so locally around vertex $v$. We would need to two-color faces around $v$ such that faces on opposite sides of an edge $e\ni v$ are different colors if $\delta_e=0$ and the same color if $\delta_e=1$. This is impossible for an even-degree vertex with $a_v=1$. This shows $\delta\in\mathcal{T}_0(G)$. The remaining claim, that $a_v=1$ for odd-degree vertices $v$ follows from a similar argument of colorability in the local neighborhood of $v$.
\end{proof}

Finally, we point out that for any graph embedding $G$ it is efficient (polynomial time in graph size) to perform two related tasks: (1) determine some $\delta\in\mathbb{F}_2^{|E|}$ such that $G$ is checkerboardable with defect $\delta$ and (2) given a candidate $\delta\in\mathbb{F}_2^{|E|}$ determine whether $G$ is checkerboardable with defect $\delta$. We present an algorithm in Algorithm~\ref{alg:checkerboard} that can perform both tasks via a greedy strategy. The inputs are, first, a representation of the faces in the graph embedding (to be concrete, we use the binary vectors $\phi_i\in\mathbb{F}_2^{|E|}$, $i=1,2\dots,|F|$, that are the rows of $\Phi$) and, second, a candidate defect $\delta$. The outputs are a list of faces colored black, a list of those colored white, and a defect. The algorithm starts by arbitrarily coloring one face and continues by coloring faces that neighbor colored faces appropriately (the same color if they neighbor across a defect edge and differently otherwise). If ever a face cannot be colored either color without contradiction, edges are added or removed from the defect to make it work.

It is clear this algorithm solves task (1). To solve (2), we note the algorithm guarantees that the output defect $\gamma$ is the same as the input $\delta$ if and only if $\delta$ is a defect. This solves (2). In the course of the algorithm, each edge is considered at most once in the while loop, the body of which takes $O(\max_i|\phi_i|)$ time using appropriate data structures,\footnote{In particular, one needs adjacency oracles that take a part of the graph embedding (e.g.~a face) and return other adjacent parts (e.g.~the adjacent edges). These oracles or similar are standard in the field of computational graph theory -- for instance, the ``doubly-connected edge list" data structure used in \cite{berg1997computational}.} leading to a total time complexity of just $O(|E|\max_i|\phi_i|)$ for both tasks (1) and (2).

\begin{algorithm}
\caption{Returns checkerboard colors for each face and a list of edges in the defect}\label{alg:checkerboard}
\begin{algorithmic}[1]
    \Procedure{Checkerboard}{$\{\phi_1,\phi_2,\dots,\phi_{|F|}\}\in\left(\mathbb{F}_2^{|E|}\right)^{|F|}$, $\delta\in\mathbb{F}_2^{|E|}$}
        \State Set $B\leftarrow\{1\}$, $W\leftarrow\emptyset$ \quad \# Sets of faces colored black and white
        \State Set $\gamma\leftarrow \delta$ \quad \# This is a record of the defect edges (vector in $\mathbb{F}_2^{|E|}$)
        \State Set $b\leftarrow\phi_1$ \quad \# This is a list of ``boundary" edges of the faces already colored (vector in $\mathbb{F}_2^{|E|}$)
        \State For any edge $e$ that is internal to a face (same face on both sides), set $\gamma_e\leftarrow1$.
        \State While $b$ is not $\vec 0$:
        \State \quad Let $e$ be the first edge such that $b_e=1$
        \State \quad If $e$ is adjacent to a face $f\not\in B\cup W$ and a face $g\in B\cup W$:
        \State \quad \quad If $g\in W$ and $\gamma_e=0$, append $f$ to $B$.
        \State \quad \quad If $g\in B$ and $\gamma_e=1$, append $f$ to $B$.
        \State \quad \quad If $f\not\in B$, append $f$ to $W$.
        \State \quad \quad Set $b_{e'}=1$ for any $e'$ that borders both $f$ and a face $f'\not\in B\cup W$.
        \State \quad else:
        \State \quad \quad If the two faces adjacent to $e$ are both in $B$ or both in $W$, set $\gamma_e\leftarrow 1$. Else, set $\gamma_e\leftarrow 0$.
        \State \quad Set $b_e=0$
    \State Return $B,W,\gamma$
    \EndProcedure
\end{algorithmic}
\end{algorithm}



\subsection{Characterizing logical operators as cycles in the decoding graph}\label{sec:dec_graph}
Recall that $\mathcal{C}(\mathcal{S})$ is the group of logical operators of the stabilizer code with stabilizer $\mathcal{S}$. We characterize $\mathcal{C}(\mathcal{S})$ for our surface codes as the cycles in a related graph, referred to here as the decoding graph.
\begin{definition}\label{def:decoding_graph}
Given an embedded graph $G=(V,E,F)$, an associated decoding graph $G_{\text{dec}}$ is a (unembedded) graph with vertices $V_{\text{dec}}$ of three types: (1) a vertex $u_f$ associated to each face of $G$, (2) a vertex $w^{(0)}_v$ associated to each odd-degree vertex of $G$, and (3) two vertices $w^{(1)}_v$ and $w^{(-1)}_v$ associated to each even-degree vertex of $G$. For each sector $[h]_{\rho}$ of $G$, there is an edge $e_{[h]_{\rho}}$ in $G_{\text{dec}}$. Choose $v\in V$ and $f\in F$ such that $h\in v$ and $h\in f$. If $v$ is odd-degree, $e_{[h]_{\rho}}$ connects vertices $w^{(0)}_v$ and $u_f$. If $v$ is even-degree, $e_{[h]_{\rho}}$ connects vertices $u_f$ to $w^{(j)}_v$ where $j$ is such that the edges associated to the adjacent sectors, i.e.~$e_{[\tau h]_{\rho}}$ and $e_{[\tau\rho h]_{\rho}}$, are incident to $w^{(-j)}_v$. All decoding graphs of $G$ are equivalent up to relabeling vertices $w^{(1)}_v$ and $w^{(-1)}_v$, so we just speak of \emph{the} decoding graph of $G$.
\end{definition}

In Fig.~\ref{fig:decoding_graph_locally}, we draw the decoding graph locally around vertices of small degree. Our definition is slightly different from the decoding graphs commonly used decode surface codes (e.g.~\cite{dennis2002topological}) in that we introduce vertices at the vertices of $G$ as well as the faces. One reason we do this is so we can make precise statements about embedding the decoding graph or its connected components into various manifolds. A consequence of this decision however is that the decoding graph's systole is doubled compared to a definition not including vertices at the vertices of $G$. The reader knowledgeable in topological quantum error correction should not be concerned by the appearance of a factor of $1/2$ when we calculate the code distance in terms of this systole.

Consider now the problem of relating cycles in $G_{\text{dec}}$ to elements of $\mathcal{C}(\mathcal{S})$. Notice that the decoding graph is bipartite, with vertices associated to faces of $G$ (type 1 in the Definition) on one side and vertices associated to vertices of $G$ (types 2 and 3) on the other. Therefore, any cycle in $G_{\text{dec}}$ is even-length and can be broken into paths of consecutive edges $(e_{[h]_{\rho}},e_{[k]_{\rho}})$, where $h,k$ are flags in the same vertex $v$ but not in the same sector. We let $\mathcal{T}_2(G_{\text{dec}})\le\mathcal{T}(G_{\text{dec}})$ denote the group generated by two-edge paths of $G_{\text{dec}}$. 

Define a map $\sigma:\mathcal{T}_2(G_{\text{dec}})\rightarrow\mathcal{C}(\mathcal{S})$ that translates each two-edge path $(e_{[h]_{\rho}},e_{[k]_{\rho}})$ to a Pauli acting on qubits at the vertex $v$. This Pauli is chosen to anti-commute with only two elements, $q_{[h]_{\rho}}$ and $q_{[k]_{\rho}}$, of the CAL associated with $v$ that defines the surface code on $G$ (see Definition~\ref{def:qubit_surface_code}). This Pauli exists and is unique (up to phase) by Corollary~\ref{cor:CAL_commutation}.

Now we state some consequences of Corollary~\ref{cor:CAL_commutation} for $\sigma$. First, if we ignore phases on the Paulis, $\sigma$ is actually surjective -- every Pauli up to phase is represented by some collection of two-edge paths in $G_{\text{dec}}$ -- and, second, $\sigma$ is a homomorphism -- for all $t_1,t_2\in \mathcal{T}_2(G_{\text{dec}})$, we have $\sigma(t_1+t_2)=\sigma(t_1)\sigma(t_2)$. Finally, $\sigma(t)\propto I$ if and only if $t$ is the empty set.

\begin{figure}
    \centering
    \includegraphics[width=0.85\textwidth]{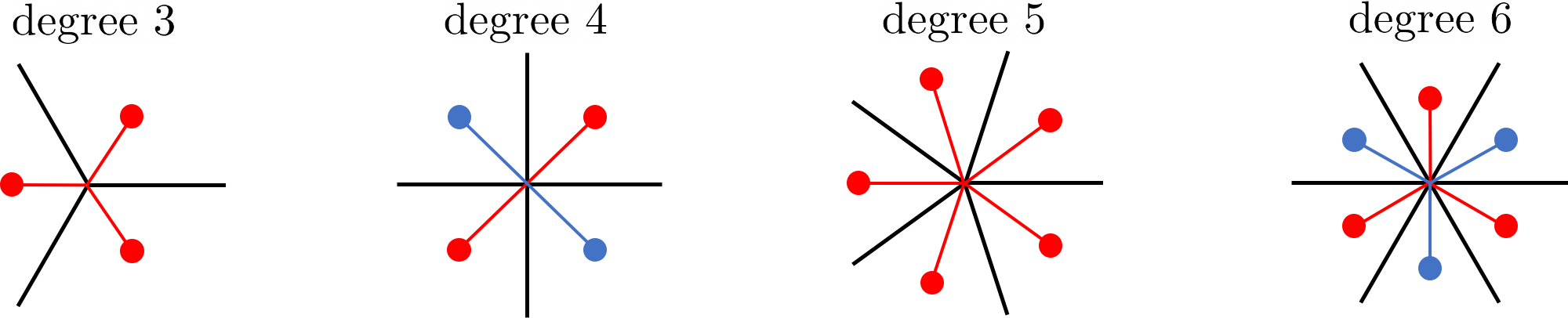}
    \caption{The decoding graph (red/blue) pictured locally around vertices of degrees 3 through 6. In odd-degree cases, all vertices in the surrounding faces connect to the central vertex. In the even-degree cases, there are two distinct central vertices, red and blue (not shown), and faces connect to these in alternating fashion.}
    \label{fig:decoding_graph_locally}
\end{figure}

By surjectivity of $\sigma$, if we have a logical operator $l\in\mathcal{C}(\mathcal{S})$, then it can be represented by a set of edges $E_l\subseteq E_{\text{dec}}$ in the decoding graph. Because the logical operator must commute with all stabilizers $S_f$, the number of edges from $E_l$ incident to any vertex $u_f$ of the decoding graph (i.e.~a face in the original graph) must be even. Therefore, $E_l$ is a sum of cycles in the decoding graph. The converse is also clearly true -- any cycle $c$ in the decoding graph maps to a logical operator $\sigma(c)\in\mathcal{C}(\mathcal{S})$. Thus, we have a characterization of the logical operators of our codes.
\begin{theorem}
Let $G=(V,E,F)$ be an embedded graph. Cycles in the decoding graph of $G$ (Defintion~\ref{def:decoding_graph}) represent all logical operators of the qubit surface code corresponding to $G$ (Definition~\ref{def:qubit_surface_code}) via the map $\sigma$.
\end{theorem}

\begin{figure}
    \centering
    \includegraphics[width=0.9\textwidth]{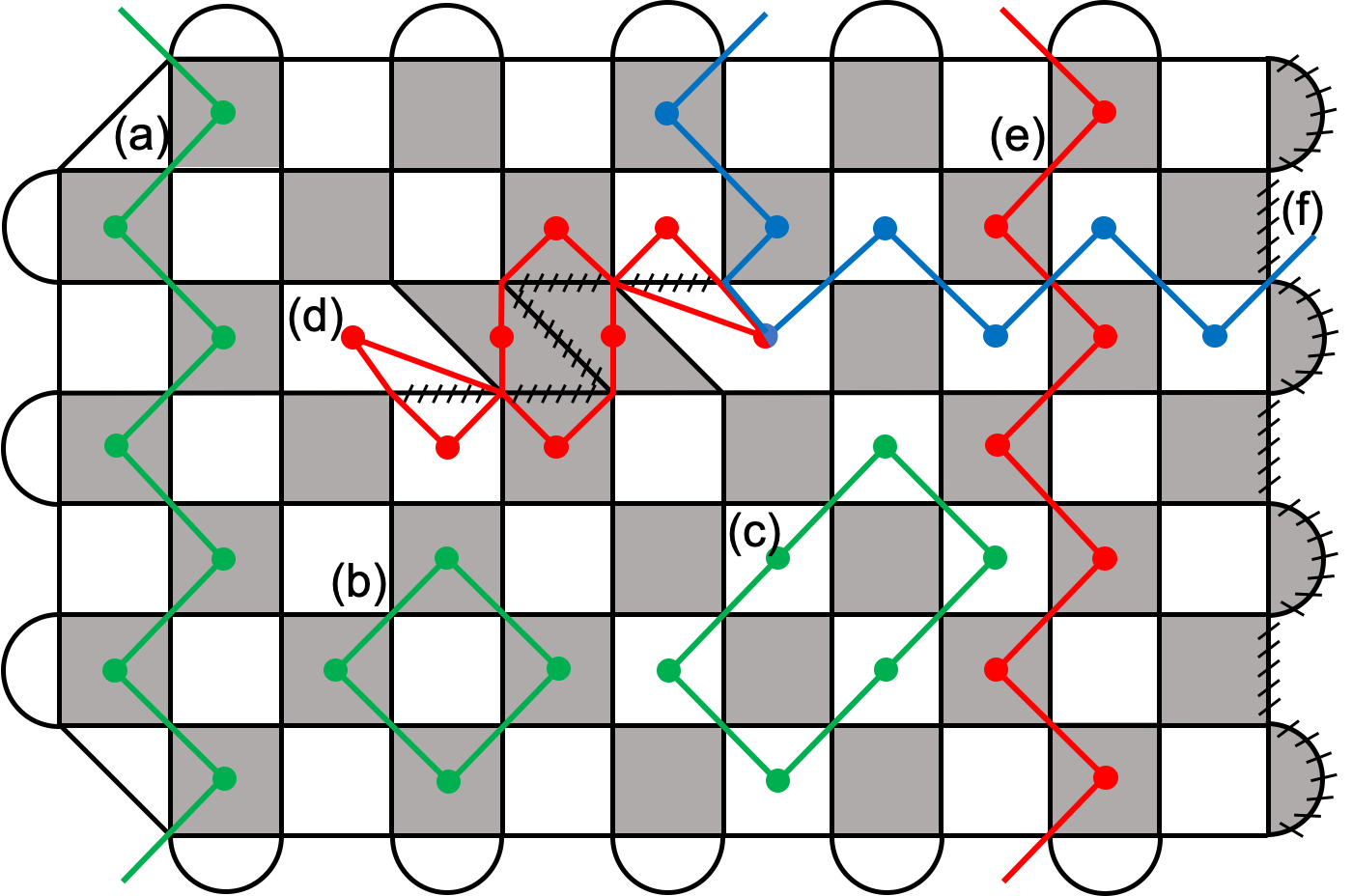}
    \caption{Example cycles in a decoding graph. The original graph has black edges, and stitched edges are part of a (particular choice of) defect. Faces are checkerboardable with this defect (the outer face should be colored black but is not shown). The entire decoding graph is not shown, but cycles labeled (a), (b), (c), (d), (e), and (f) illustrate some of it. Cycles (a), (b), (c) in green represent trivial logical operators (i.e.~stabilizers), while cycles (d) and (e) in red and cycle (f) in blue represent non-trivial logical operators. We know they are non-trivial because (d) and (e) anticommute with (f). Cycles (a), (e), and (f) connect to the vertex in the outer face (not shown) via their dangling edges.}
    \label{fig:cycles_in_dec_graph}
\end{figure}

See Fig.~\ref{fig:cycles_in_dec_graph} for some examples of cycles in a decoding graph and their correspondence with logical operators.

In checkerboardable codes, we can say even more about locating nontrivial logical operators by relating them to the homology of the surface. For this to work, we must embed the decoding graph of the checkerboardable code into the same surface.
\begin{lemma}\label{lem:checkerboardable_decoding_graph_embedding}
Suppose $G$ is a checkerboardable graph embedded in a 2-manifold $\mathcal{M}$. Then the decoding graph $G_{\text{dec}}$ consists of exactly two connected components, each of which can be embedded in the manifold $\mathcal{M}$.
\end{lemma}
\begin{proof}
Suppose $G$ is described by the rotation system $(H,\lambda,\rho,\tau)$ and, because it is checkerboardable, there is a partition $H=H_w\sqcup H_b$ as in Lemma~\ref{lem:checkerboardability_rotation_system}. Of course, each vertex of $G$ is necessarily even-degree. It should be clear from Fig.~\ref{fig:decoding_graph_locally} that the decoding graph $G_{\text{dec}}$ has two connected components, one corresponding to faces that are subsets of $H_w$ and another corresponding to faces that are subsets of $H_b$.

Let $G_{\text{dec},w}$ and $G_{\text{dec},b}$ be the connected components of $G_{\text{dec}}$. These are described by rotation systems $R_{\text{dec},w}=(H_w\times\{\pm1\},\lambda_{\text{dec}},\rho_{\text{dec}},\tau_{\text{dec}})$ and $R_{\text{dec},b}=(H_b\times\{\pm1\},\lambda_{\text{dec}},\rho_{\text{dec}},\tau_{\text{dec}})$, where
\begin{align}\label{eq:dec_lambda}
\lambda_{\text{dec}}(h,j) &= (h,-j),\\\label{eq:dec_rho}
\rho_{\text{dec}}(h,j) &= \bigg\{\begin{array}{ll}(\tau\rho\tau(h),j),&j=-1\\(\lambda(h),j),&j=1\end{array},\\\label{eq:dec_tau}
\tau_{\text{dec}}(h,j)&=(\rho(h),j).
\end{align}
This provides an explicit embedding of the components of $G_{\text{dec}}$. 

We do still need to show that they are indeed embedded in the manifold $\mathcal{M}$ and not some other one. Suppose $G$ has sets of vertices, edges, faces $V$, $E$, $F=F_w\sqcup F_b$ (with this partition due to checkerboardability) and $G_{\text{dec},w}$ has sets $V_{\text{dec}}$, $E_{\text{dec}}$, and $F_{\text{dec}}$. One can check from the rotation system, Eqs.~(\ref{eq:dec_lambda}-\ref{eq:dec_tau}), that $|V_{\text{dec}}|=|F_w|+|V|$, $|E_{\text{dec}}|=\frac12\sum_{v\in V}\text{deg}(v)=|E|$, and $|F_{\text{dec}}|=|F_b|$. This implies the Euler characteristics of $R$ and $R_{\text{dec},w}$ are the same.

We also need to show that $R_{\text{dec},w}$ is orientable if and only if $R$ is. Note $R$ being orientable means one can partition $H$ into two sets $H_{\pm1}$ such that $\lambda$, $\rho$, $\tau$ applied to an element of one set maps it to an element of the other set. Define $H'_{\pm1}\subseteq H':=H_w\times\{\pm1\}$ such that $(h,j)\in H_{\text{dec},k}$ if and only if $h\in H_{jk}$. These sets clearly partition $H'$ and $\lambda_{\text{dec}}$, $\rho_{\text{dec}}$, and $\tau_{\text{dec}}$ all map either set to the other. So $R_{\text{dec}}$ is orientable. The other direction -- if $R_{\text{dec}}$ is orientable, then $R$ is -- follows a very similar argument that we leave for the reader.
\end{proof}

Here we comment on the case where a checkerboardable graph $G$ has only degree-4 vertices. In this case, the decoding graph has degree-2 vertices associated to the vertices of $G$. Suppose we remove these vertices, joining the edges incident to them into one edge. It is not hard to see that this vertex removal makes the connected components of $G_{\text{dec}}$ into graphs $Q$ and $\overline{Q}$ (duals of one another) where the medial graph $\widetilde{Q}$ is equal to $G$. Therefore, in this checkerboardable and 4-valent case, the operation of building the decoding graph essentially inverts the operation of building the medial graph. Then $Q$ is a graph that defines a homological code (see Section~\ref{subsec:relate_to_homological}) equivalent to the code we have defined on $G$. \edit{For an illustration, refer back to Fig.~\ref{fig:medial_examples}: in each part of that figure, the gray edged graph is one connected component of the decoding graph of the black edged graph.}

With the decoding graph embedded, we can relate homological and logical nontriviality in the checkerboardable case. 

\begin{theorem}\label{thm:checkerboardable_nontriviality}
For checkerboardable graph $G$, a cycle $c$ in $G_{\text{dec}}$ is homologically nontrivial if and only if $\sigma(c)$ is a nontrivial logical operator of the surface code associated with $G$.
\end{theorem}
\begin{proof}
The key fact to use is that a Pauli $p$ is a stabilizer generator associated to a face in $G$ if and only if $p\propto\sigma(c)$ for some facial cycle $c$ in $G_{\text{dec}}$. In Fig.~\ref{fig:face_correspondence}, we show this correspondence explicitly. We use this fact to complete the proof.

\begin{figure}
    \centering
    \includegraphics[width=0.8\textwidth]{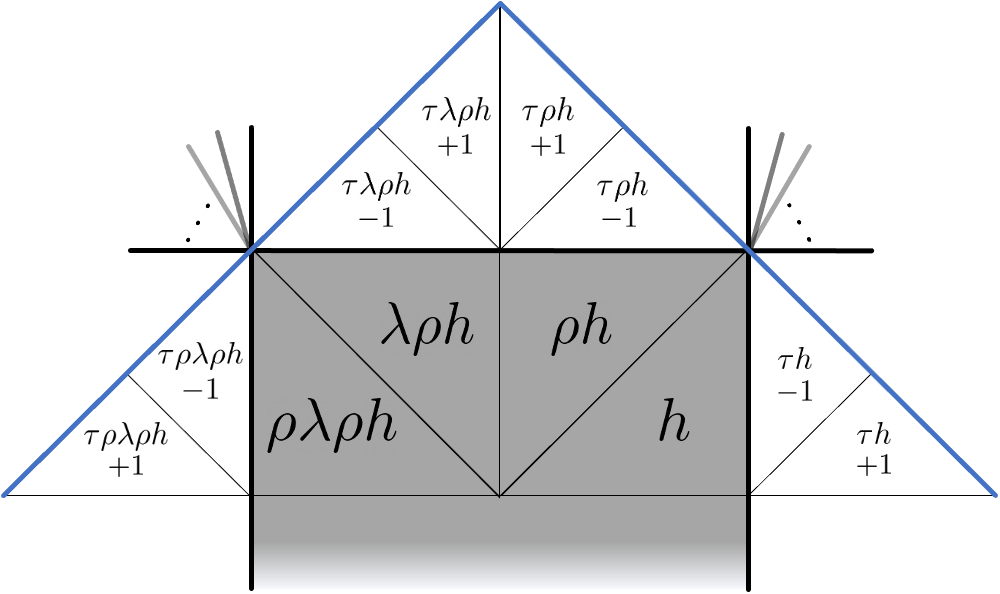}
\caption{Picturing the correspondence between a face $f$ of checkerboardable graph $G$ (thick, black edges) and a face $f'$ of a connected component of its decoding graph $G_{\text{dec}}$ (blue edges). Notation comes from the proof of Lemma~\ref{lem:checkerboardable_decoding_graph_embedding}. Flags of $G$ and $G_{\text{dec}}$ are outlined by thin lines. Flags of $G$ making up $f$ include $h$, $\rho h$, $\lambda\rho h$, and $\rho\lambda\rho h$. These each relate to two flags in the corresponding face $f'$. For instance, $h\in f$ becomes $(\tau h,\pm1)\in f'$. If $f$ is colored black in $G$, then $f'$ is a face of $G_{\text{dec},w}$.}
\label{fig:face_correspondence}
\end{figure}

Let $c$ be a homologically trivial cycle in $G_{\text{dec}}$. We need to show $\sigma(c)$ is a stabilizer. Since $G_{\text{dec}}$ has two connected components, $c=c_w+c_b$ where $c_w$ is a cycle in $G_{\text{dec},w}$ and $c_b$ a cycle in $G_{\text{dec},b}$. Since $\sigma(c)=\sigma(c_w)\sigma(c_b)$, we just show $\sigma(c_w)$ is a stabilizer with an analogous argument for $\sigma(c_b)$. Since $c_w$ is homologically trivial it is a sum of facial cycles $c_w=c_1+c_2+\dots+c_l$ and $\sigma(c_w)=\sigma(c_1)\sigma(c_2)\dots\sigma(c_l)$. By the fact in the previous paragraph, $\sigma(c)$ is indeed a stabilizer.

Likewise, if Pauli $p$ is a stabilizer, then it is the product of stabilizer generators associated to faces of $G$, or $p=p_1p_2\dots p_l$. By fact (1), there are facial cycles $c_1,c_2,\dots,c_l$ in $G_{\text{dec}}$ such that $\sigma(c_j)=p_j$ and therefore a homologically trivial cycle $c=c_1+c_2+\dots+c_l$ so that $\sigma(c)=p$.
\end{proof}

In a slight abuse of notation, we let $\hsys{G_{\text{dec}}}$ be the smaller of the homological systoles of the two connected components of $G_{\text{dec}}$ embedded as in Lemma~\ref{lem:checkerboardable_decoding_graph_embedding}. 

\begin{corollary}\label{cor:checkerboardable_distance}
Suppose $G$ is a checkerboardable, embedded graph and $D$ is the code distance of the surface code defined by $G$. Then, $\frac12\hsys{G_{\text{dec}}}\le D$ and if $G$ is 4-valent, $\frac12\hsys{G_{\text{dec}}}=D$.
\end{corollary}
\begin{proof}
For the proof of the bound, note that any non-trivial logical operator with weight $D$ is represented by a homologically non-trivial cycle $c$ confined to a single connected component of the decoding graph $G_{\text{dec}}$. Suppose $c$ visits $n\le D$ vertices of $G_{\text{dec}}$ that are  associated to vertices of $G$. There is a homologically non-trivial sub-cycle of $c$, call it $c'$, that visits each of the $n$ vertices at most once and so has at most $2n$ edges. Since $\hsys{G_{\text{dec}}}$ lower bounds the length of $c'$, $\hsys{G_{\text{dec}}}\le2n\le2D$.

Now suppose $G$ is 4-valent. The shortest homologically non-trivial cycle $c$ in $G_{\text{dec}}$ is confined to just one connected component. By 4-valency, there is only one qubit at each vertex and each two-edge sub-path in $c$ is mapped to single-qubit Pauli by $\sigma$. Since $c$ has length $\hsys{G_{\text{dec}}}$, the weight of the non-trivial logical Pauli it represents is at most $\frac12\hsys{G_{\text{dec}}}$, which therefore is an upper bound on $D$. 
\end{proof}

The equality statement in this corollary does not hold for graphs with higher degree vertices because, in that case, two-edge paths in the decoding graph may represent Paulis of weight greater than one. However, if one uses the CAL construction in Theorem~\ref{thm:CAL_construct} to assign Paulis to sectors, one can effectively decompose these higher-degree vertices into degree-4 vertices (see Corollary~\ref{cor:simplify_graph}) and apply Corollary \ref{cor:checkerboardable_distance} to the resulting 4-valent graph.

\subsection{Doubled graphs}\label{sec:doubled_graphs}

Our goal in this section is to find some manifold to embed $G_{\text{dec}}$ into when the original graph $G$ is not checkerboardable. In general, this is a manifold with genus larger than the original surface. In particular, for any non-checkerboardable graph $G$ embedded on a manifold $\mathcal{M}$, we show that there is a checkerboardable graph $G^2$, referred to as a \emph{doubled} graph, embedded on manifold $\mathcal{M}^2$ such that $G_{\text{dec}}$ is isomorphic to any one of the two connected components of $G^2_{\text{dec}}$. A similar doubled construction is given by Barkeshli and Freedman \cite{barkeshli2016modular} in the context of topological phases.

\begin{figure}
    \centering
    \includegraphics[width=0.9\textwidth]{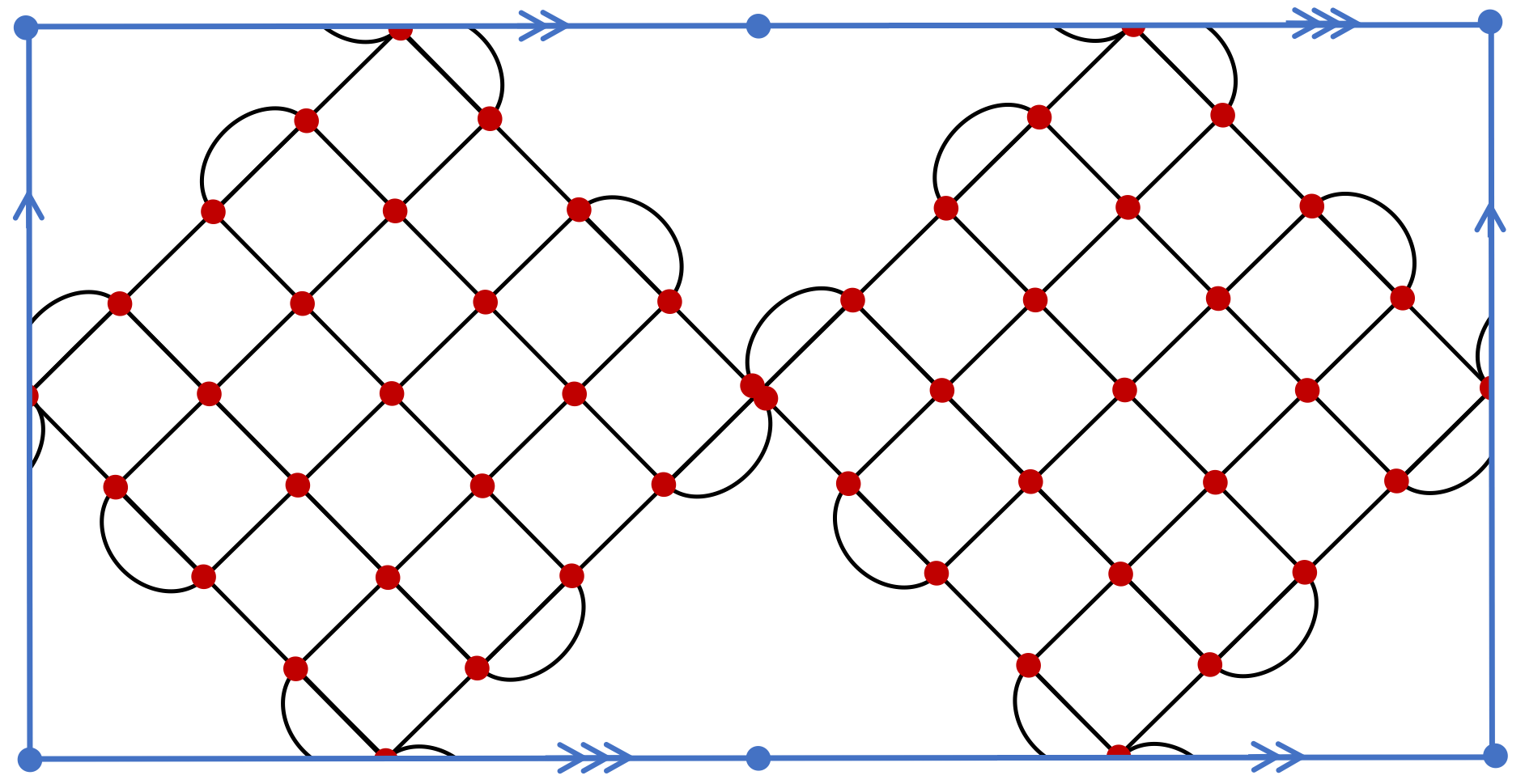}
    \caption{A doubled graph for the rotated surface code of Fig.~\ref{fig:stellated_codes}(b) is embedded in the torus (six-sided blue border).}
    \label{fig:doubled_surface_code}
\end{figure}

\begin{figure}
    \centering
    \includegraphics[width=0.75\textwidth]{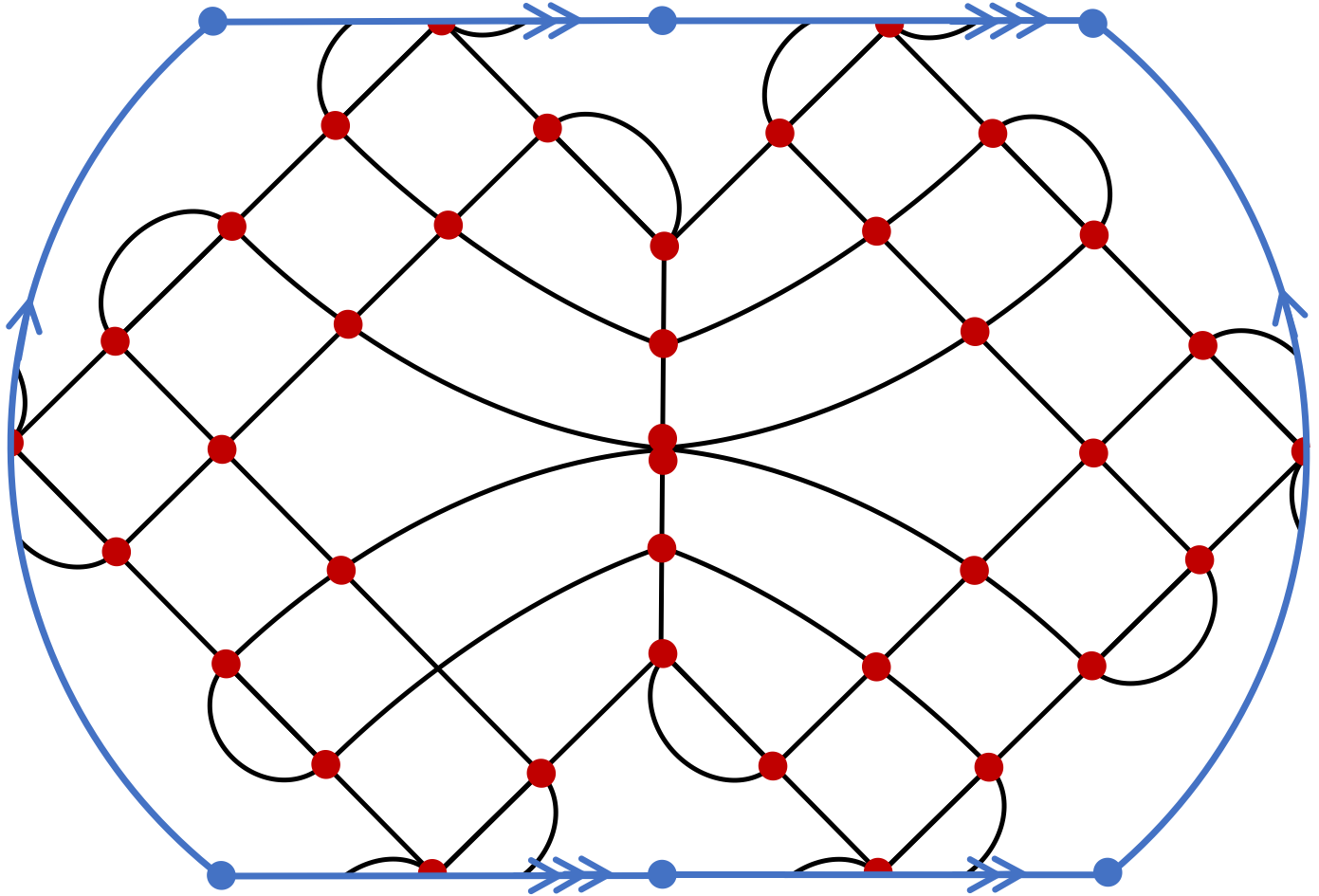}
    \caption{A doubled graph for the triangle code of Fig.~\ref{fig:stellated_codes}(a) is embedded in the torus (six-sided blue border).}
    \label{fig:doubled_triangle_code}
\end{figure}

Informally, the doubled manifold $\mathcal{M}^2$ is created by taking two copies of $\mathcal{M}$, cutting along the edges of a defect of the embedded graph $G$, and gluing the two copies together. This gluing is done in such a way that in crossing from one copy of $\mathcal{M}$ to the other, a path ends up at the same place in the destination manifold as it would have in the original manifold if it were still intact. This means, for instance, that a defect trail between two odd degree vertices is cut open into a disk and glued with its corresponding disk on the other copy. In contrast, cycles in the defect that are homologically non-trivial result in the handles or cross-caps of the surface being cut apart before they are glued with their copies. We provide a formal construction of the doubled graph in terms of rotation systems and a construction of the doubled manifold subsequently. We comment that the doubled graph $G^2$ depends on the choice of defect, but the topology of $\mathcal{M}^2$ does not.

\begin{definition}\label{def:doubled_rotation_system}
The doubled rotation system $R^2=(H',\lambda',\rho',\tau')$ of a rotation system $R=(H,\lambda,\rho,\tau)$ with defect $\delta$ is defined by setting $H'=H\times\{1,-1\}$ and for $(h,j)\in H'$ letting $\lambda'(h,j)=(\lambda(h),j)$, $\rho'(h,j)=(\rho(h),j)$, and
\begin{equation}
\tau'(h,j)=\bigg\{\begin{array}{ll}(\tau(h),-j),&\exists e\in\delta\text{ s.t.~}h\in e\\
(\tau(h),j),&\text{otherwise}\end{array}.
\end{equation}
\end{definition}

One can study orbits of $\langle\lambda',\rho'\rangle$ to relate faces of $G^2$ with faces of $G$. Since neither $\lambda'$ nor $\rho'$ acting on $(h,j)\in H'$ flips the sign of $j$, there are two faces in $G^2$ for every face of $G$, one made up of flags with $j=1$ and one with flags $j=-1$. It turns out we can also checkerboard $G^2$ such that these two faces are oppositely colored.

\begin{lemma}\label{lem:Gsq_is_checkerboardable}
For any embedded graph $G$, the doubled graph $G^2$ is checkerboardable.
\end{lemma}
\begin{proof}
Let $\delta$ be the defect used to define $G^2$ from $G$, and let $R=(H,\lambda,\rho,\tau)$ be the rotation system of $G$. As $G$ is checkerboardable with defect $\delta$, there is a partition of $H=H_w\sqcup H_b$, where $\lambda$ and $\rho$ map either set to itself and $\tau$ maps an element $h$ from either set to the other \emph{except} when $h$ is in an edge belonging to the defect (Lemma~\ref{lem:checkerboardable_w_defect_rs}). 

Now we can partition $H'=H\times\{1,-1\}$, the set of flags for the doubled rotation system $R^2=(H',\lambda',\rho',\tau')$, into two sets
\begin{align}
H'_w&=\{(h,j):h\in H_w,j=1\text{ or }h\in H_b,j=-1\},\\
H'_b&=\{(h,j):h\in H_b,j=1\text{ or }h\in H_w,j=-1\}.
\end{align}
One can check that $\lambda'$ and $\rho'$ map either set to itself and $\tau'$ maps between the sets. Thus, $G^2$ is checkerboardable.
\end{proof}

Just as there are two faces in $G^2$ for every face in $G$, it follows from Definition~\ref{def:doubled_rotation_system} that $G^2$ has two edges for every edge of $G$, two degree-$\text{deg}(v)$ vertices for every even-degree vertex $v$ of $G$, and one degree-$2\text{deg}(v)$ vertex for every odd-degree vertex $v$ of $G$. This counting allows us to establish the genus of $G^2$.

\begin{theorem}
Let $G$ be a non-checkerboardable graph containing $M$ odd-degree vertices. 
\setmargins
\begin{enumerate}[(a)]
\item If $G$ embeds into a genus $g$, orientable manifold, then $G^2$ embeds into a genus $2g+(M-2)/2$, orientable manifold.
\item Suppose $G$ embeds into a genus $g$, non-orientable manifold.
\begin{enumerate}[(i)]
\item If $G^2$ embeds into an orientable manifold, its genus is $g+(M-2)/2$.
\item If $G^2$ embeds into a non-orientable manifold, its genus is $2\left(g+(M-2)/2\right)$.
\end{enumerate}
\end{enumerate}
\end{theorem}

\begin{proof}
Part of this theorem that is implicit is that if $G$ is orientable, then $G^2$ is orientable. Represent $G$ and $G^2$ with rotation systems as in Definition~\ref{def:doubled_rotation_system}. Then, $G$ being orientable means there is a partition of $H=H_{+1}\sqcup H_{-1}$ such that $\lambda$, $\rho$, and $\tau$ all map elements of either set $H_{\pm1}$ to elements of the other. Clearly, $H'=H'_{+1}\sqcup H'_{-1}$ can also be partitioned so that $(h,j)\in H'_{k}$ if and only if $h\in H_k$. As required for orientability, $\lambda'$, $\rho'$, and $\tau'$ all map elements of either set $H'_{\pm1}$ to elements of the other.

The genus of $G^2$ can be obtained by counting. As discussed previously, the doubled graph has $|F'|=2|F|$ faces, $|E'|=2|E|$ edges, and $|V'|=2|V|-M$ vertices where $|F|$, $|E|$, and $|V|$ are the face, edge, and vertex counts of $G$. The conclusions follow.
\end{proof}

Because $G^2$ is checkerboardable, its decoding graph $G^2_{\text{dec}}$ has two connected components, each embedding in the same manifold as $G^2$. 

\begin{theorem}\label{thm:distance_bounds_from_doubled_graph}
Either connected component of $G^2_{\text{dec}}$ is isomorphic to $G_{\text{dec}}$. Moreover, homologically non-trivial cycles in $G^2_{\text{dec}}$ correspond to non-trivial logical operators of the surface code associated to $G$. If that code has distance $D$, then $\frac14\hsys{G^2_{\text{dec}}}\le D\le\frac12\hsys{G^2_{\text{dec}}}$.
\end{theorem}
\begin{proof}
We describe the isomorphism by describing how vertices of $G_{\text{dec}}$ and $G^2_{\text{dec},w}$ (one of the two connected components of $G^2_{\text{dec}}$) are mapped to each other. There are three maps that should be composed in the right way, the map from $G$ to $G_{\text{dec}}$, the map from $G$ to $G^2$, and the map from $G^2$ to $G^2_{\text{dec},w}$.

First, note $G^2_{\text{dec},w}$ has a single vertex for each white face of $G^2$ (which is checkerboardable). Also, $G^2_{\text{dec},w}$ has a single vertex for each vertex of $G^2$. Next, recall $G_{\text{dec}}$ has a vertex for each face of $G$, two vertices for each even-degree vertex of $G$, and one vertex for each odd-degree vertex of $G$. Finally, $G^2$ has two (differently colored) faces for each face of $G$, two vertices for each even-degree vertex of $G$, and one vertex for each odd-degree vertex of $G$. With the vertices of $G_{\text{dec}}$ and $G^2_{\text{dec},w}$ associated to one another by the implied map, vertex adjacency is also preserved, showing the isomorphism.



We argued in Lemma~\ref{lem:Gsq_is_checkerboardable} that every face of $G$ corresponds to two oppositely colored faces of $G^2$. Since $G^2_{\text{dec},w}$ has faces corresponding to black faces of $G^2$, these actually represent all the faces of $G$. Therefore, a cycle in $G^2_{\text{dec},w}$ is homologically non-trivial if and only if it is a non-trivial logical operator of the code. The upper bound on the code distance follows immediately.

The lower bound is a result of two vertices in $G^2$ representing each even-degree vertex of $G$. If a homologically non-trivial cycle in $G^2_{\text{dec},w}$ crosses both of these vertices, it may represent only one qubit of support in the corresponding logical operator. Thus, we have an additional factor of $1/2$ compared to the upper bound.
\end{proof}

With the theory above, we have seen how the doubled manifold effectively describes the homology of nontrivial logical operators in a higher genus surface. One consequence for surface codes is that for every non-checkerboardable code defined by graph $G$ with parameters $\llbracket N,K,D\rrbracket$, there is a checkerboardable code defined by graph $G^2$ with parameters $\llbracket2N,2K,D'\rrbracket$ where $D\le D'\le 2D$. The checkerboardable code has the same rate and at least as good a code distance, but with potentially higher weight stabilizers and larger genus. Thus, non-checkerboardable codes can offer improvements over checkerboardable codes given that stabilizer weight and genus are relevant parameters for practical implementations.

\subsection{Face-width as a lower-bound on code distance}
\label{sec:face-width}
The \emph{face-width} $\text{fw}(G)$ of an embedded graph $G$ (with genus $g>0$) is the minimum number of times any non-contractible cycle drawn on the manifold intersects the graph \cite{cabello2007finding}. This can also be defined as the length of the shortest non-contractible cycle in a related graph, the face-vertex graph $G_{\text{fv}}$.

\begin{definition}
The face-vertex graph $G_{\text{fv}}$ of an embedded graph $G$ is an embedded graph possessing a vertex $w_v$ for each vertex $v$ of $G$ and a vertex $u_f$ for each face $f$ of $G$. An edge is drawn between $w_v$ and $u_f$ if and only if $v$ and $f$ are adjacent in $G$. To specify the embedding, if $G$ has rotation system $(H,\lambda,\rho,\tau)$, the face-vertex graph has rotation system $(H_{\text{fv}},\lambda_{\text{fv}},\rho_{\text{fv}},\tau_{\text{fv}})$, where $H_{\text{fv}}=H\times\{1,-1\}$ and
\begin{align}
\lambda_{\text{fv}}(h,j)=(h,-j),\quad\quad
\rho_{\text{fv}}(h,j)=\bigg\{\begin{array}{ll}(\tau(h),j),&j=-1\\(\lambda(h),j),&j=1\end{array},\quad\quad
\tau_{\text{fv}}(h,j)=(\rho(h),j).
\end{align}
\end{definition}

We point out some easily-proved facts about the face-vertex graph. First, it is bipartite since vertices of type $w_v$ are only ever connected to vertices of type $u_f$. Second, the face-vertex graph is the dual of the medial graph (Definition~\ref{def:medial_graph}), $\overline{\widetilde{G}}=G_{\text{fv}}$. Third, the face-vertex graph is embedded in the same manifold as $G$. Fourth, if we recall that a graph $G$ and its dual $\overline{G}$ have isomorphic medial graphs, we also find that $G$ and $\overline{G}$ have isomorphic face-vertex graphs. Finally, the face-vertex graph is the decoding graph (see Def.~\ref{def:decoding_graph}) with the vertices $w^{(\pm1)}_v$ merged at each even-degree vertex $v$.

Because $G_{\text{fv}}$ is bipartite, $\sys{G_{\text{fv}}}$ is even. Any cycle drawn on the manifold is homeomorphic to a cycle in the face-vertex graph. Hence, it is easy to see that $\text{fw}(G)=\frac12\sys{G_{\text{fv}}}$. This is the most convenient definition of face-width for our purposes.

The rest of this subsection is devoted to proving the following.
\begin{theorem}\label{thm:face_width}
Suppose $G$ is an embedded graph with only even-degree vertices and genus $g>0$. Let $D$ be the code distance of the surface code associated to $G$. Then $\text{fw}(G)\le D$.
\end{theorem}

To prove this theorem, we need a simple lemma, which we prove separately.
\begin{lemma}\label{lem:double_cover}
Suppose $G$ is a non-checkerboardable graph with only even-degree vertices embedded in manifold $\mathcal{M}$ with genus $g>0$. Then the doubled graph $G^2$ is embedded in a manifold $\mathcal{M}^2$ that is a double cover of $\mathcal{M}$.
\end{lemma}
\begin{proof}
To create the manifold $\mathcal{M}^2$ from the rotation system description in Def.~\ref{def:doubled_rotation_system} (whose notation we import here), we glue the flags together according to $\lambda'$, $\rho'$, $\tau'$. Suppose we map flags in $H'=H\times\{1,-1\}$ to flags in $H$, i.e.~$\Pi(h,j)=h$ for $h\in H$ and $j=\pm1$. This map extends to a map $\pi:\mathcal{M}^2\rightarrow\mathcal{M}$ by simply mapping points in the flag $(h,j)$ to points in flag $h$ in homeomorphic fashion, i.e.~if both $(h,j)$ and $h$ are viewed as identically-sized triangular patches of surface with boundaries associated to the $\lambda$, $\rho$, and $\tau$ involutions, then map the former to latter point-by-point. Note that $\Pi\circ\lambda'=\lambda\circ\Pi$, $\Pi\circ\rho'=\rho\circ\Pi$, and $\Pi\circ\tau'=\tau\circ\Pi$. This establishes that if two flags are glued together in $\mathcal{M}^2$, their images are glued together in $\mathcal{M}$.

We should check that for all open disks $U$, $\pi^{-1}(U)$ is the union of two open disks, thus proving it is a double cover. This is clearly true for open disks contained entirely within a single flag $h$. We also look at open disks that straddle two flags (across a flag's edge) and open disks that straddle more than two flags (because they contain a flag's corner).

Suppose an open disk straddles two flags $h,\lambda h\in H$. Applying $\Pi^{-1}$ to this pair of flags, we get either $(h,1)$ and $\lambda'(h,1)=(\lambda h,1)$ or $(h,-1)$ and $\lambda'(h,-1)=(\lambda h,-1)$. Thus, $\pi^{-1}$ gives two open disks, each straddling one of these pairs of flags in $\mathcal{M}^2$.

The same argument holds for flags $h,\rho h$. Slightly more interesting is the case of flags $h,\tau h$. In this case, if the edge of the graph $G$ containing $h$ is in the defect, $\Pi^{-1}$ gives either $(h,1)$ and $(\tau h,-1)$ or $(h,-1)$ and $(\tau h,1)$. If the edge is not in the defect, then the result is $(h,1)$ and $(\tau h,1)$ or $(h,-1)$ and $(\tau h,-1)$.

Finally, corners of the flags correspond to vertices, edges, and faces of the graphs. An open disk at the corner $\{h,\rho h,\lambda h, \rho\lambda h\}$, a face of the graph, is the most straightforward case since no defect is involved. Here $\Pi^{-1}$ maps the face of $G$ to two faces of $G^2$, namely $\{(h,1),(\rho h,1),(\lambda h,1), (\rho\lambda h,1)\}$ and $\{(h,-1),(\rho h,-1),(\lambda h,-1), (\rho\lambda h,-1)\}$. Likewise, it is not hard to see that an edge (containing $h$) in $G$ maps to two edges (the one containing $(h,1)$ and the one containing $(h,-1)$) in $G^2$, whether or not that edge is part of the defect. For the case of a vertex $v$ in $G$, one must recall the fact that an even number of edges incident to $v$ are part of the defect, Lemma~\ref{lem:defect_decomp}. This is the only place where we are using that $G$ has only even-degree vertices. Suppose $h\in v$. The vertices in $G^2$ that map to $v$ under $\Pi$ are the one containing $(h,1)$ and the one containing $(h,-1)$, and these are distinct because of the aforementioned fact.
\end{proof}

\begin{proof}[Proof of Theorem~\ref{thm:face_width}]
If $G$ is checkerboardable, then notice that $\text{fw}(G)=\frac12\text{sys}(G_{\text{fv}})\le\frac12\text{sys}(G_{\text{dec}})\le\frac12\text{hsys}(G_{\text{dec}})=D$, with the final equality following from Corollary~\ref{cor:checkerboardable_distance}.

Thus, we are left with the non-checkerboardable case. Lemma~\ref{lem:double_cover} says the doubled manifold $\mathcal{M}^2$ is a double cover of the manifold $\mathcal{M}$ on which $G$ is embedded. This means that the universal cover of $\mathcal{M}$ (which is the plane $\mathcal{U}=\mathbb{R}^2$) is also the universal cover of $\mathcal{M}^2$. This implies the existence of three covering maps, the projections $\pi_{21}:\mathcal{M}^2\rightarrow\mathcal{M}$, $\pi_{U2}:\mathcal{U}\rightarrow\mathcal{M}^2$, and $\pi_{U1}:\mathcal{U}\rightarrow\mathcal{M}$ such that $\pi_{U1}=\pi_{21}\pi_{U2}$.

Now, choose a logical operator represented by some homologically non-trivial cycle $c_2$ (in the decoding graph $G^2_{\text{dec}}$) drawn on $\mathcal{M}^2$. Any lift $\pi_{U2}^{-1}(c_2)$ is necessarily a non-closed curve in $\mathcal{U}$. We also see that $c_1=\pi_{21}(c_2)$, which is easily seen to be homeomorphic to a cycle in $G_{\text{fv}}$, also lifts to $\pi_{U2}^{-1}(c_2)=\pi_{U1}^{-1}\pi_{21}(c_2)$, and therefore $c_1$ is non-contractible. 

Finally, there must be some simple sub-cycle of $c_1$, call it $c_1'$, that is also non-contractible. It visits exactly $\frac12\text{len}(c_1')$ vertices and so this is at least the Pauli weight $w$ of the logical operator it represents, $w\ge\frac12\text{len}(c_1')\ge\text{fw}(G)$. Since our arguments apply to all choices of $c_2$ (representing all non-trivial logical operators), the distance of the code is also lower bounded by $\text{fw}(G)$.
\end{proof}

\subsection{Logical operators from trails in the original graph}\label{sec:canon_paths}
In this section, we provide another upper bound on the code distance of surface codes, complementing the bounds in Theorem~\ref{thm:distance_bounds_from_doubled_graph}. This bound comes from identifying logical operators with certain trails in the graph $G$ that defines the surface code. Conveniently, we do not have to work with any derived graph, such as the decoding graph, to obtain these logical operators.

Recall that $\mathcal{T}_0(G)$ is the subgroup of $\mathcal{T}(G)$ that is generated by closed trails and open trails whose endpoints are odd-degree vertices. To each trail in $\mathcal{T}_0(G)$ we assign a logical operator, or equivalently an element of $\mathcal{C}(\mathcal{S})$. We define a homomorphism $\omega:\mathcal{T}_0(G)\rightarrow\mathcal{C}(\mathcal{S})$. It is easy to define $\omega$ using the Majorana code picture from Section~\ref{subsec:maj_codes_on_graphs}. Let $t\in\mathcal{T}_0(G)$. If $t$ is a closed trail, then we define $\omega(t)$ to be the product of all Majoranas on the edges in $t$. If $t$ is an open trail, its endpoints are at odd degree vertices by definition, and we define $\omega(t)$ to be the product of all Majoranas on the edges in $t$ and the two Majoranas at the odd-degree endpoints. By the arguments in Section~\ref{subsec:qub_surface_codes}, these products of Majoranas are equivalent to Paulis in the qubit surface code.  We do not concern ourselves with the sign of the Pauli, and so the products can be taken in arbitrary order. Moreover, the Pauli $\omega(t)$ commutes with all stabilizers by construction and so is in $\mathcal{C}(\mathcal{S})$.

What is the weight of the Pauli $\omega(t)$? If $t$ does not visit a vertex, clearly $\omega(t)$ is not supported on qubits at that vertex. Alternatively, if $t$ traverses all edges adjacent to a vertex, then $\omega(t)$ is also not supported on qubits at that vertex. These are the only vertices lacking support -- $\omega(t)$ is supported on qubits at all other vertices. If the underlying graph has only degree 3 or 4 vertices (higher degree vertices may be decomposed with Corollary~\ref{cor:simplify_graph}), then the weight of $\omega(t)$ is equal to the number of vertices visited exactly once by $t$. Paths by definition visit vertices at most once, so that for a path $p\in\mathcal{T}_0(G)$, the weight of $\omega(p)$ is exactly the number of vertices visited by $p$. From this discussion, we can find the kernel of $\omega$ (for our usual assumption of a connected graph $G$).
\begin{equation}\label{eq:ker_tau}
\ker\omega=\{b\in\mathcal{T}_0(G):\omega(b)=I\}=\{\vec0,\vec1\}.
\end{equation}
This means the map $\omega$ is 2-to-1, ignoring phases on Paulis in $\mathcal{C}(\mathcal{S})$.

When is $\omega(t)$ a non-trivial logical operator? It turns out we can efficiently check this using the notion of checkerboardability and Algorithm~\ref{alg:checkerboard}. With some abuse of notation, we call $t$ non-trivial if $\omega(t)$ is non-trivial.
\begin{theorem}\label{thm:tau(t)_trivial}
Let $t\in\mathcal{T}_0(G)$. Then $\omega(t)\in\mathcal{C}(\mathcal{S})$ is a stabilizer if and only if $G$ is checkerboardable with defect $t$ or with defect $\vec 1+t$.
\end{theorem}
\begin{proof}
We start with the reverse direction. We are assured the existence of $x\in\mathbb{F}_2^{|F|}$ such that $x\Phi=\vec1+t$ or such that $x\Phi=t$. Let $\Phi_i$ denote the $i^{\text{th}}$ row of $\Phi$. Because $\omega$ is a homomorphism and $\omega(\vec1)=I$, we have
\begin{align}\label{eq:tau_and_checkerboard}
\omega(t)=\omega(x\Phi)=\omega\left(\sum_{i:x_i=1}\Phi_i\right)=\prod_{i:x_i=1}\omega\left(\Phi_i\right),
\end{align}
which is a stabilizer because each $\omega(\Phi_i)$ is the stabilizer $S_i$ in Eq.~\eqref{eq:s_f} associated to the $i^{\text{th}}$ face.

For the forward direction of the theorem, we start with $\omega(t)=\prod_{i:x_i=1}\omega(\Phi_i)$ for some $x\in\mathbb{F}_2^{|F|}$, since the faces generate the stabilizer group by definition. By Eq.~\eqref{eq:tau_and_checkerboard} once again, we get $\omega(t)=\omega(x\Phi)$. This implies $t=x\Phi+b$ where $b\in\ker\omega$. By Eq.~\eqref{eq:ker_tau}, we have $\vec1+t=x\Phi$ or $t=x\Phi$ or, equivalently, $G$ is checkerboardable with defect $t$ or with defect $\vec1+t$.
\end{proof}

To summarize the above, for any graph $G$ embedded in a 2-manifold defining surface code with stabilizer $\mathcal{S}$, we have shown
\begin{equation}
\mathcal{S}\le\omega(\mathcal{T}_0(G))\le\mathcal{C}(\mathcal{S}).
\end{equation}
We now point out a strengthening of this result when $G$ is planar (i.e.~$g=0$).
\begin{theorem}
Let $G$ be a planar graph and $\mathcal{S}$ be the stabilizer of the surface code defined on $G$ according to Definition~\ref{def:qubit_surface_code}. Then $\mathcal{S}\le\omega(\mathcal{T}_0(G))=\mathcal{C}(\mathcal{S})$ (ignoring Pauli's phases).
\end{theorem}
\begin{proof}
If $G$ contains no odd-degree vertices, then $K=0$ by Theorem~\ref{thm:number_encoded_qubits} and so $\mathcal{S}=\mathcal{C}(\mathcal{S})=\tau(\mathcal{T}_0(\mathcal{S}))$.

Otherwise, $G$ contains an even number of odd-degree vertices $M\ge2$. Construct a spanning tree of $G$ rooted at one of these odd-degree vertices $v$. There is a path $p_w$ from $v$ to any other odd-degree vertex $w$ contained within the spanning tree, and $\omega(p_w)\in\mathcal{C}(\mathcal{S})$. Notice that for $w\neq w'$, $\omega(p_w)$ anticommutes with $\omega(p_{w'})$, because, written as products of Majoranas, $p_w$ and $p_{w'}$ have odd overlap -- they share edges, two Majoranas each, as well as the Majorana located at odd-degree vertex $v$. Thus, we have a set of $M-1$ logical Paulis $\mathcal{B}=\{\omega(p_w):\text{odd-degree vertex }w\neq v \}$ that mutually anticommute. This implies, because the maximum size of an anticommuting set of $K$-qubit Paulis is $2K+1$ \cite{sarkar-vandenberg2019}, that $K\ge(M-2)/2$. However, $K=(M-2)/2$ is exactly the number of encoded qubits by Theorem~\ref{thm:number_encoded_qubits}, meaning $\mathcal{B}$ is a basis of all logical operators. In turn, this implies $\mathcal{C}(\mathcal{S})=\mathcal{S}\mathcal{B}\le\omega(\mathcal{T}_0(G))$, which completes the proof.
\end{proof}

The rest of this section concerns the actual calculation of an element of $\mathcal{T}_0(G)$ that visits the fewest number of vertices. First, note that there are bases of $\mathcal{T}_0(G)$ that consist only of trails that do not visit any vertex more than once, i.e.~bases of paths. Let the length of a path be the number of vertices it visits. The total length of a path basis is the sum of lengths of all trails in the basis. This implies the existence of a path basis that minimizes the total length of the basis, called a minimum path basis. Such a minimum path basis must also contain the shortest non-trivial path (with non-triviality determined by Theorem~\ref{thm:tau(t)_trivial}). If it did not, we could replace some element of the basis with this shortest non-trivial path, obtaining another path basis with lower total weight.

This discussion of minimum path bases is mirrored for bases of $\mathcal{Z}(G)$, the cycle space of the graph. Likewise, there is a minimum simple cycle basis that consists of simple cycles (cycles without repeated vertices) and minimizes their combined length. Polynomial time algorithms for finding a minimum simple cycle basis have been known since Horton's algorithm \cite{horton1987polynomial}, which has time complexity $O(|E|^3|V|)$ for a general graph $G=(V,E)$. Better algorithms now exist, including a $O(|E|^\Omega)$ time Monte-Carlo algorithm \cite{amaldi2009breaking} (where $\Omega$ is the exponent of matrix multiplication) and $O\left(|E|^2|V|\right)$ \cite{kavitha2004faster} and $O\left(|E|^2|V|/\log|V|+|E||V|^2\right)$ \cite{mehlhorn2009minimum} deterministic algorithms. We can use these algorithms to also find a minimum path basis and efficiently place an upper bound on the code distance.

\begin{theorem}\label{thm:dist_ub}
For embedded graph $G$, let $J$ be the number of vertices in a non-trivial path $p\in\mathcal{T}_0(G)$ visiting the fewest vertices. Then the surface code defined by graph $G$ has distance $D\le J$ and $J$ can be calculated in polynomial time in the size of graph $G$.
\end{theorem}
\begin{proof}
Since $\omega(p)$ is a non-trivial logical operator with weight $J$, clearly $J$ upper bounds the code distance.

It remains to show $J$ is efficient to calculate. Construct a graph $G'$ from $G$ by adding edges between all pairs of odd-degree vertices, a total of $M(M-1)/2$ new edges if there are $M$ odd-degree vertices. Call the set of new edges $E_o$. Find the minimum cycle basis $B$ of $G'$ using one of the efficient algorithms referenced above, e.g.~\cite{horton1987polynomial}. Each cycle $b\in B$ that traverses an edge outside $E_o$ corresponds to a path $q\in\mathcal{T}_0(G)$. In particular, $q$ is formed from $b$ by removing all edges from $E_o$. The length of $b$ is the number of vertices visited by $q$. Thus, by replacing $b$ with $q$ we can turn the minimum simple cycle basis $B$ into a minimum path basis $B_0$. To calculate $J$, we just need to determine whether each path in $B_0$ is non-trivial using Theorem~\ref{thm:tau(t)_trivial} and take the one visiting the fewest vertices.
\end{proof}

Finally, we point out that the upper bound from Theorem~\ref{thm:dist_ub} and the bounds from Theorem~\ref{thm:distance_bounds_from_doubled_graph} are not saturated in general. An example is shown in Fig.~\ref{fig:obstruction}.

\begin{figure}
    \centering
    \includegraphics[width=0.72\textwidth]{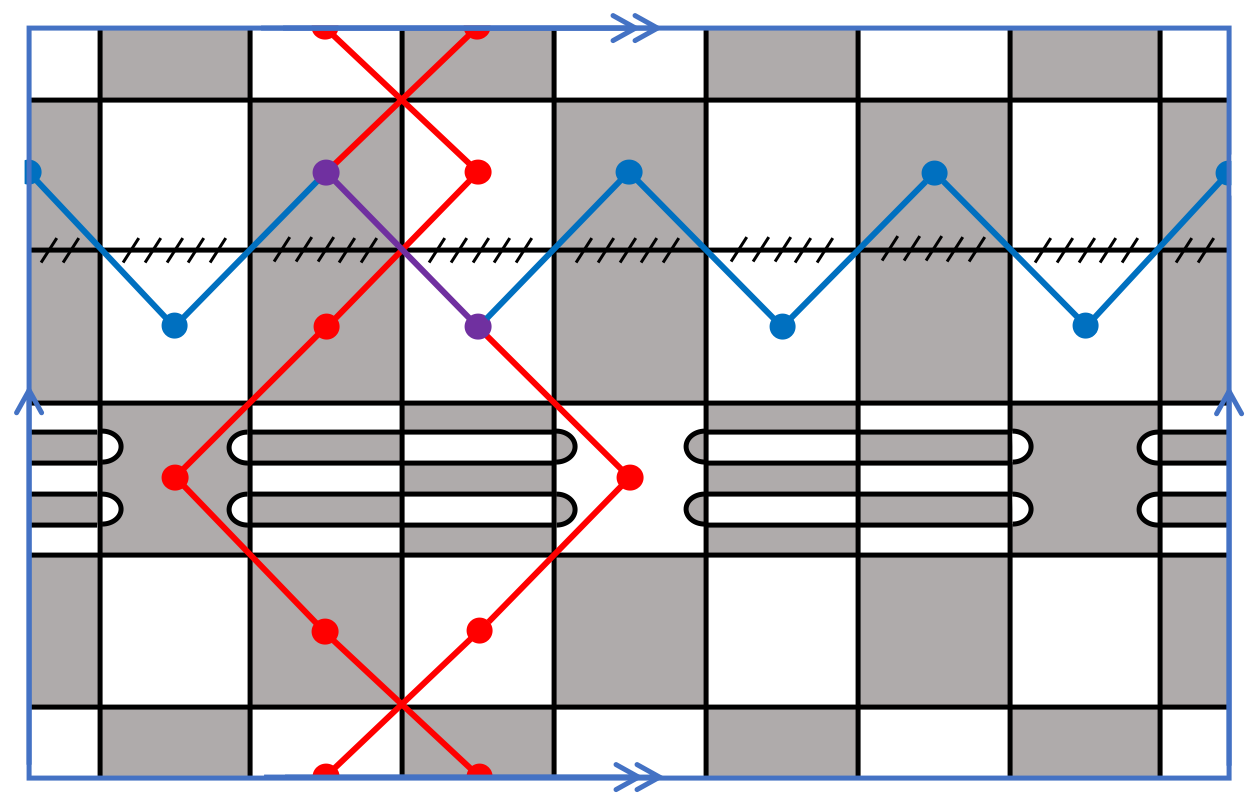}
    \caption{An example graph $G$ in which the bounds in Theorems~\ref{thm:distance_bounds_from_doubled_graph} and \ref{thm:dist_ub} are not saturated. The graph is largely a square lattice on the torus, but with obstructions along one row that force minimum weight logical operators to go around. We checkerboard the graph with stitched edges indicating the defect. The code encodes one logical qubit. The shortest non-trivial cycles in the decoding graph $G_{\text{dec}}$ have length $L=8$; an example is shown in blue (and purple where it overlaps the red cycle). The shortest non-trivial cycles in the original graph $G$ also visit $J=8$ vertices. Thus, Theorem~\ref{thm:dist_ub} proves $4\le D\le 8$. However, the distance of the code actually satisfies $D\le7$, as demonstrated by the red cycle in the decoding graph (which has length $10$). Also, $D\ge5$ follows from Lemma~\ref{lem:disjoint_set_lb} and arguments similar to those in Section~\ref{sec:square_lattice_toric_codes}.
    }
    \label{fig:obstruction}
\end{figure}

%% file: Code-examples.tex
\section{Code examples}\label{sec:code_examples}

In this section, we present five example code families. In the first three subsections we consider square lattice toric codes, rotated toric codes (with a subfamily of cylic toric codes first mentioned in Section~\ref{subsec:qub_surface_codes}), and hyperbolic codes. For these codes we make some rigorous arguments about code parameters (especially the code distance) using the theorems from Section~\ref{sec:locating_logical_operators}. In the final subsection, we present two more code families that generalize stellated codes \cite{kesselring2018boundaries}.

\subsection{Rotated toric codes}\label{sec:square_lattice_toric_codes}
As the name suggests, this family of codes is created from simple square lattices on the torus. Up to local Cliffords, these are the codes defined by Wen \cite{wen2003quantum}. Because of checkerboardability properties, the dimensions of the lattice are relevant to $K$ as well as $D$. We let the lattice be $m\times n$, i.e.~it takes $m$ steps to traverse the torus in one direction and $n$ to traverse the other direction. See Fig.~\ref{fig:square_lattice_examples}, for two examples.

\begin{figure}[t]
    \centering
    \includegraphics[width=0.8\textwidth]{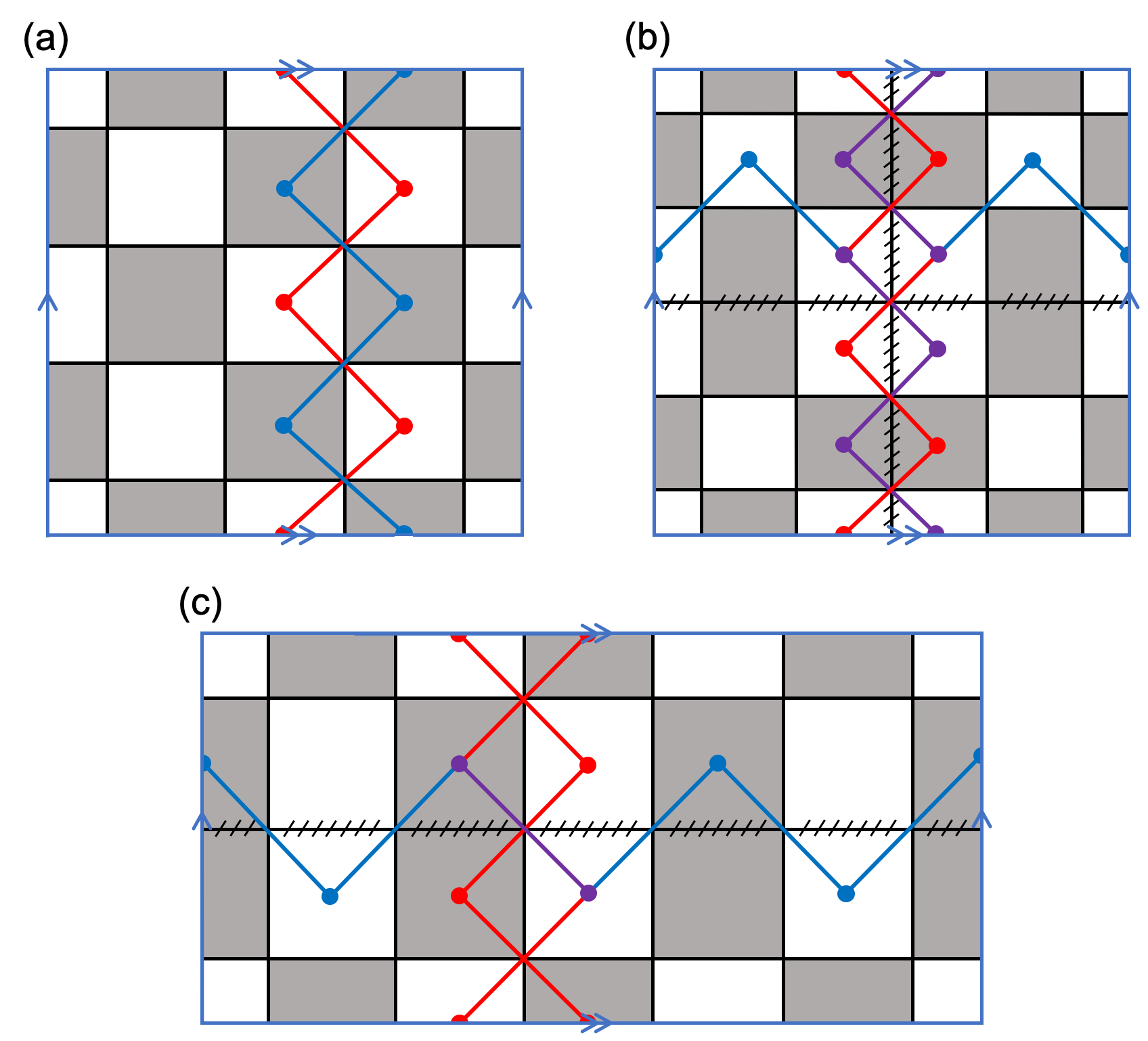}
    \caption{Three examples of \edit{rotated} toric codes. (a) An even$\times$even code can be checkerboarded and encodes two logical qubits. We show two cycles, red and blue, in the decoding graph representing two commuting, non-trivial logical operators, e.g.~$\bar X_1$ and $\bar Z_2$. (b) An odd$\times$odd code can be checkerboarded with a defect (stitched edges) and encodes one logical qubit. We show two cycles, red and blue (with purple edges and vertices where they overlap), representing two anti-commuting logical operators, e.g.~$\bar X$ and $\bar Z$. (c) An odd$\times$even code is also checkerboardable with a defect and encodes one logical qubit.}
    \label{fig:square_lattice_examples}
\end{figure}

Note that a \edit{rotated} toric code is only checkerboardable if both $m$ and $n$ are even. Therefore, by Corollary~\ref{cor:number_encoded_qubits}
\begin{equation}
K=\bigg\{\begin{array}{ll}
2,&m,n\text{ are even},\\
1,&\text{otherwise}.
\end{array}
\end{equation}
Since they are checkerboardable and four-valent, the $K=2$ codes can also be described by the homological code model (see Section~\ref{subsec:relate_to_homological}), while the $K=1$ codes cannot.

What is the code distance of a $m\times n$ \edit{rotated} toric code? It will surprise few to learn that it is
\begin{equation}
D=\min(m,n).
\end{equation}
However, we would like to prove this. We divide the proof into cases -- even$\times$even, odd$\times$even, odd$\times$odd -- though the proof is similar for each case. The general idea is to assume that there is a logical operator $l$ that has weight less than $\min(m,n)$, then show that $l$ commutes with all logical operators of the code, so it must be trivial (i.e.~in the center of the group, i.e.~a stabilizer). Thus, $D$ is at least $\min(m,n)$. Separately, we can show there is some non-trivial logical operator with that weight.

The odd$\times$even case is slightly simpler so we start there. Say $m$ (number of rows) is odd and $n$ (number of columns) is even. There is one logical qubit. Any column $c$ is a cycle, so $\omega(c)$ is a logical operator. Theorem~\ref{thm:tau(t)_trivial} tells us $\omega(c)$ is a non-trivial logical operator, but we can also see this a different way. Let $r$ be any row and $r'$ be any cycle in the decoding graph such that the support of the logical operator $\sigma(r')$ is only within row $r$. Then $\sigma(r')$ and $\omega(c)$ anticommute and so are both non-trivial logical operators. Also, $\omega(c)$, $\sigma(r')$, and the stabilizers generate the entire group of logical operators, since there is only one encoded qubit. Fig.~\ref{fig:square_lattice_examples}(c) shows an example of these two logical operators for $m=3$, $n=6$. Note, we can find such anticommuting pairs for any column $c$ and row $r$. Therefore, if logical operator $l$ has weight less than $\min(m,n)$, there is both a column $c$ and a row $r$ where it is not supported. Thus, it commutes with the entire logical group and must be trivial.

In the even$\times$even case, the only complication is that there are two logical qubits. Non-trivial logical operators can still be confined to a single column $c$ and a single row $r$. There are two different cycles $c'$ and $c''$ in the decoding graph such that support of $\sigma(c')$ and $\sigma(c'')$ lie within column $c$. Likewise, there are two different cycles $r'$ and $r''$. See Fig.~\ref{fig:square_lattice_examples}(a). By their commutation relations, we note that $\sigma(c'), \sigma(c''), \sigma(r'), \sigma(r''),$ and stabilizers generate the entire group of logical operators. Thus, as before, a logical operator $l$ with weight less than $\min(m,n)$ will not have support in one column and row, thus commutes with all logical operators, and so must be trivial.

Finally, in the odd$\times$odd case, there is one logical qubit. Any column $c_1$ is a closed trail and $\omega(c_1)$ a non-trivial logical operator. However, the complementary logical operator is not confined to a single row. Instead, given any row $r$ \emph{and} column $c_2$, there is a cycle $z$ in the decoding graph such that $\sigma(z)$ has support only in $r$ and $c_2$. Moreover, $\omega(c_1)$ and $\sigma(z)$ anticommute no matter the row and columns chosen. Set $c_1=c_2=c$. See Fig.~\ref{fig:square_lattice_examples}(b). Now the same argument from before can be applied: if logical operator $l$ has weight less than $\min(m,n)$, then it is not supported in some column $c$ and some row $r$, but then it commutes with all logical operators and is trivial.

The preceding proof technique can be summarized in a lemma that is potentially useful for lower bounding the code distance of any stabilizer code. If $l$ is a logical Pauli operator of stabilizer code with stabilizer group $\mathcal{S}$, we let $D(l)$ be the minimum weight of an element of the coset $l\mathcal{S}$.
\begin{lemma}\label{lem:disjoint_set_lb}
Let $l$ be a logical Pauli operator for stabilizer code $\mathcal{S}$ on $N$ qubits and suppose there exists positive integer $\mu$ and collections $C_1,C_2,\dots,C_{\mu}\subseteq P(\{0,1,\dots,N-1\})$ (where $P(S)$ denotes the power set of $S$) such that
\begin{enumerate}
    \item For all $C_i$ and all distinct $c,c'\in C_i$, $c\cap c'=\emptyset$,
    \item For all choices of $c_i\in C_i$, there exists $l_0\in l\mathcal{S}$, with $\supp{l_0}\subseteq\bigcup_{i=1}^{\mu}c_i$.
\end{enumerate}
Then, for all logical Paulis $l'$ that anticommute with $l$, $D(l')\ge \min_i|C_i|$.
\end{lemma}
\begin{proof}
Since logical operators $l$ and $l'$ anticommute, every element of $l\mathcal{S}$ anticommutes with every element of $l'\mathcal{S}$. Suppose, by way of contradiction, there is $l'_0\in l'S$ with weight $|l'_0|<\min_i|C_i|$. For all $i$, because the sets in each $C_i$ are disjoint, there must be a set $c_i\in C_i$ such that $\supp{l'_0}\cap c_i=\emptyset$. However, some $l_0\in l\mathcal{S}$ exists such that $\supp{l_0}\subseteq\bigcup_{i=1}^{\mu}c_i$. Therefore, $\supp{l_0}\cap\supp{l'_0}=\emptyset$. So $l_0$ and $l'_0$ do not anticommute, a contradiction.
\end{proof}

The total distance of the code is $D=\min_{l\in\mathcal{C}(\mathcal{S})\setminus\mathcal{S}}D(l)$. To lower bound the code distance it is sufficient to fix a basis of logical operators for a code, e.g.~$\{\bar X_i,\bar Z_i\}$, and find appropriate collections satisfying the conditions of the lemma for each element of the basis. 

For any \edit{rotated toric} code, we need only two collections, $C_1$ the set of all columns and $C_2$ the set of all rows. With these collections, \emph{every} non-trivial logical operator of the code satisfies the second condition of the lemma (see the constructions of these operators in Fig.~\ref{fig:square_lattice_examples}), and we immediately get the lower bound $D\ge\min(m,n)$.

One thing to note when applying the lemma, is that because of the conditions on the collections $C_i$, $N$ must be at least $D^2/\mu$. Thus, the lemma is unlikely to be useful for codes with $N$ scaling better than a constant times $D^2$. Luckily, our codes must scale this way due to the Bravyi-Poulin-Terhal bound \cite{bravyi2010tradeoffs}. For instance, symmetric \edit{rotated toric codes} with $m=n=D$ satisfy $N=KD^2$ if $D$ is odd and $N=\frac12KD^2$ if $D$ is even. Because of the factor of $1/2$, the latter family has twice the code rate of the rotated surface code, Fig.~\ref{fig:stellated_codes}(b).

\subsection{General rotated toric codes}\label{sec:rotated_toric_codes}
In this subsection we focus on a generalization of toric codes introduced by Kovalev and Pryadko \cite{kovalev2012improved} and prove Theorem~\ref{thm:cyclic_toric_codes}. These codes are defined by two vectors, $L_1=(a_1,b_1)$ and $L_2=(a_2,b_2)$, where $a_i,b_i$ are integers. Consider the infinite square lattice, and equate points $x,y\in\mathbb{R}^2$ if
\begin{equation}\label{eq:cyclic_toric_cover}
x-y\in S(L_1,L_2):=\{m_1L_1+m_2L_2:m_1,m_2\in\mathbb{Z}\}.
\end{equation}
The vertices and edges of the infinite square lattice now map, using this equivalence, to vertices and edges of a finite graph $G(L_1,L_2)$ embedded on the torus. With single qubits at each vertex and stabilizers associated to faces, $G(L_1,L_2)$ defines a surface code via Def.~\ref{def:qubit_surface_code}. For instance, if $b_1=0$ and $a_2=0$, these codes are just the square lattice toric codes discussed in the previous section. Referring back to Section~\ref{subsec:covers}, the infinite square lattice is acting as the universal cover of the graph $G(L_1,L_2)$ on the torus, with the covering map described by equating points as in \edit{Eq.~\eqref{eq:cyclic_toric_cover}}.

Some of these rotated toric codes are equivalent despite being defined by different vectors $L_1$ and $L_2$. The sets $S(L_1,L_2)$ and $S(L_1',L_2')$ are equal if and only if $L_i'=g_{i1}L_1+g_{i2}L_2$ for some integer-valued matrix $g$ with $\det(g)=\pm1$ \cite{kovalev2012improved}. In these cases, the codes associated to graphs $G(L_1,L_2)$ and $G(L_1',L_2')$ are also equivalent.

Code parameters $N$ and $K$ are relatively obvious given our general framework. First, $N$ can be calculated as the area of the parallelogram with sides $L_1$ and $L_2$, or \cite{kovalev2012improved}
\begin{equation}
N=|L_1\times L_2|=|a_1b_2-b_1a_2|.
\end{equation}
As also pointed out in \cite{kovalev2012improved}, checkerboardability of the graph $G(L_1,L_2)$ depends on the parities of $\|L_1\|_1=a_1+b_1$ and $\|L_2\|_1=a_2+b_2$. A checkerboard coloring of the infinite square lattice can be mapped to a checkerboard coloring of the finite graph $G(L_1,L_2)$ if and only if, for all points $P\in\mathbb{R}^2$, the faces at points $P$, $P+L_1$, and $P+L_2$ are colored the same. Moreover, those points are colored the same if and only if $\|L_1\|_1$ and $\|L_2\|_1$ are both even integers. Thus, we have
\begin{equation}
K=\bigg\{\begin{array}{ll}
1,&\|L_1\|_1\text{ or }\|L_2\|_1\text{ is odd,}\\
2,&\|L_1\|_1\text{ and }\|L_2\|_1\text{ are even.}
\end{array}
\end{equation}

With regards to distance $D$, we first relate the (non)-triviality of logical operators to the topological (non)-triviality of cycles in the decoding graph. To do this for non-checkerboardable codes, we need to use the doubled graph of Section~\ref{sec:doubled_graphs}. Even the doubled graph does not resolve the issue that cycle length in the non-checkerboardable code's decoding graph is not necessarily the Pauli weight (reflected in the gap in the upper and lower bounds of Theorem~\ref{thm:distance_bounds_from_doubled_graph}), and the next step is to resolve this.

Consider first the checkerboardable rotated toric codes as also considered in \cite{kovalev2012improved}. In this case, the faces of $G(L_1,L_2)$ are two-colorable, black and white. The decoding graph resolves into two connected components which are duals of one another: the black component has vertices associated to only black faces and the white component has vertices associated to only white faces. As both these connected components embed naturally onto the torus, with qubits associated to edges and stabilizers to vertices and faces, the homological description of toric codes applies \cite{kitaev2003fault}. We conclude that a cycle in the decoding graph is non-trivial if and only if it is homologically non-trivial, or equivalently, when mapped back to the infinite square lattice, it connects a point $P$ with the point $P+m_1L_1+m_2L_2$ where at least one of $m_1$ or $m_2$ is odd. Edges in the decoding graph follow the vectors $(1,1)$, $(1,-1)$, $(-1,1)$, or $(-1,-1)$. Therefore, to get from $P$ to $P+m_1L_1+m_2L_2$ takes $\|m_1L_1+m_2L_2\|_{\infty}$ steps in the decoding graph. We conclude the code distance is \cite{kovalev2012improved}
\begin{equation}\label{eq:cyc_ckr_D}
D=\min_{\substack{m_1,m_2\in\mathbb{Z}\\(m_1,m_2)\neq(0,0)}}\|m_1L_1+m_2L_2\|_{\infty}.
\end{equation}

We divide the non-checkerboardable case into two sub-cases. For now, assume $\|L_1\|_1$ is odd but $\|L_2\|_2$ is even. We claim that the decoding graph of $G(L_1,L_2)$, call it $G_{\text{dec}}$ is isomorphic to either connected component of the decoding graph of the doubled graph $G(2L_1,L_2)$, a checkerboardable code. If one believes that $G(2L_1,L_2)$ is legitimately a doubled graph of $G(L_1,L_2)$, then this isomorphism follows from Theorem~\ref{thm:distance_bounds_from_doubled_graph}. However, we can be more explicit in the argument for this example.

Therefore, let us show that $G_{\text{dec}}$ and the black component $G'_{\text{dec}}$ of the decoding graph of $G(2L_1,L_2)$ are isomorphic. Showing it for the white component is analogous. To establish some coordinates, we imagine that graph $G(L_1,L_2)$ lies within the parallelogram $P$ with corners $(0,0)$, $L_1$, $L_1+L_2$, and $L_2$. Similarly, the graph $G(2L_1,L_2)$ lies within the parallelogram $P'$ with corners $(0,0)$, $2L_1$, $2L_1+L_2$, and $L_2$. See Fig.~\ref{fig:odd_rotated_codes}(a).

\begin{figure}[t]
    \centering
    \includegraphics[width=\textwidth]{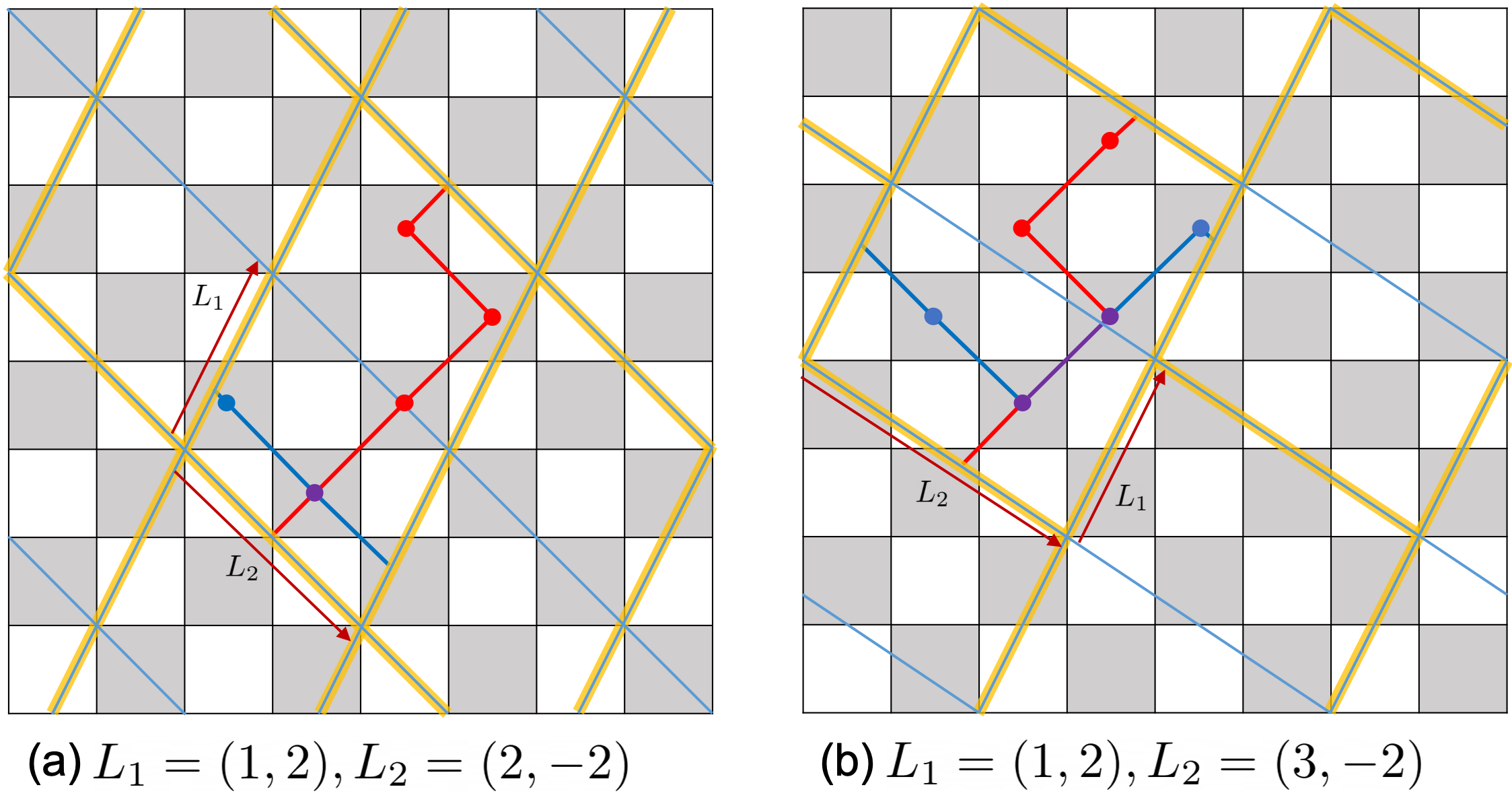}
    \caption{Regular tessellations of the infinite square lattice. Inside a single cell of a tessellation, the subgraph of the infinite square lattice defines a toric code. In fact, both images show two tessellations, one outlined in blue and one outlined in yellow. The blue tessellations have periodicity vectors $L_1$ and $L_2$ and define non-checkerboardable rotated toric codes. As argued in the text, the decoding graphs of these non-checkerboardable codes are isomorphic to a connected component of the decoding graphs of the corresponding checkerboardable codes defined by the yellow tessellations. Some paths (blue and red, purple where they overlap) in the yellow codes' decoding graphs are shown, which, when mapped to a single cell of the blue tessellation, are anticommuting logical operators.}
    \label{fig:odd_rotated_codes}
\end{figure}

Consider the graphs $G_{\text{dec}}$ and  $G'_{\text{dec}}$ as a sets of points making up their edges and vertices. The idea is that any point $p$ in $G'_{\text{dec}}$ that is also in the parallelogram $P'\setminus P$ should be mapped to $p-L_1$. This places it within $P$. It is easily seen that this map takes the points of $G'_{\text{dec}}$ to the points of $G_{\text{dec}}$ in one-to-one fashion and is a graph isomorphism.

The other crucial property of the map is that any cycle in $G_{\text{dec}}$ that corresponds to single face stabilizer maps to a cycle in $G'_{\text{dec}}$ corresponding to a single face stabilizer, and vice versa. In this way, we know that all trivial (both topologically and logically) cycles in $G'_{\text{dec}}$ map to (logically) trivial cycles in $G_{\text{dec}}$. The cycle space of $G'_{\text{dec}}$ is generated by these trivial cycles as well as two non-trivial cycles (because it is embedded on a torus). When mapped to $G_{\text{dec}}$, these two non-trivial cycles must be anticommuting logical operators of the code.

Now, we start the other sub-case of non-checkerboardable codes. Say both $\|L_1\|_1$ and $\|L_2\|_1$ are odd. To come up with a related checkerboardable code, we have to alter the rotated toric code construction slightly. Let $G_E(L_1,L_2)$ be the graph defined by equating points $x,y\in\mathbb{R}^2$ in the infinite square lattice such that
\begin{equation}
x-y\in S_E(L_1,L_2):=\{m_1L_1+m_2L_2:m_1,m_2\in\mathbb{Z},m_1+m_2\text{ is even}\}.
\end{equation}
This corresponds to tiling the plane with offset parallelograms, or hexagons with opposite sides identified, see Fig.~\ref{fig:odd_rotated_codes}(b). Clearly, $G_E(L'_1,L'_2)$ is always a checkerboardable code for any $L'_1$ and $L'_2$. We claim that either connected component of the decoding graph of $G_E(L_1,L_2)$ is isomorphic to the decoding graph of $G(L_1,L_2)$.

Again, let $G_{\text{dec}}$ be the decoding graph of $G(L_1,L_2)$ and $G'_{\text{dec}}$ be the black connected component of the decoding graph of $G_E(L_1,L_2)$. As before, we can imagine $G(L_1,L_2)$ lying within the parallelogram $P$ and $G_E(L_1,L_2)$ lying within the parallelogram $P'$. The isomorphism is also the same map. Each point $p$ in $P'\setminus P$ and in $G'_{\text{dec}}$ gets mapped to a point $p-L_1$ in $P$. Once again this maps trivial cycles to trivial cycles. The two remaining generators of the cycle space of $G'_{\text{dec}}$, when mapped to $G_{\text{dec}}$, are the anticommuting logical operators of the non-checkerboardable code.

With this setup, we finish the proof of Theorem~\ref{thm:cyclic_toric_codes}.
\begin{proof}[Proof of Theorem~\ref{thm:cyclic_toric_codes}]
The nontrivial part of Theorem~\ref{thm:cyclic_toric_codes} that remains is the proof of the distance. Here, in the special case of cyclic codes $L_1=(a,b)$ and $L_2=(-b,a)$ with $b>a>0$ and $\text{gcd}(a,b)=1$. For checkerboardable codes, those in part (b) of the theorem, Eq.~\eqref{eq:cyc_ckr_D} applies and simplifies to $D=\max(a,b)=b$, as claimed.

The non-checkerboardable case, in which both $\|L_1\|_1$ and $\|L_2\|_1$ are odd, makes up the bulk of this proof. Above, we characterized the nontrivial logical operators as topologically nontrivial cycles in the decoding graph of a related checkerboardable code. We will take this a step further and change coordinates to align with the decoding graph itself. This amounts to rotating and re-scaling $L_1$ and $L_2$.
\begin{align}
L_1^{\text{dec}}&=\frac12\left(\begin{array}{cc}1&1\\-1&1\end{array}\right)L_1=\left(\frac12(b+a),\frac12(b-a)\right),\\
L_2^{\text{dec}}&=\frac12\left(\begin{array}{cc}1&1\\-1&1\end{array}\right)L_2=\left(-\frac12(b-a),\frac12(b+a)\right).
\end{align}
Again, points $x,y\in\mathbb{R}^2$ are equated if $x-y\in S_E(L_1^{\text{dec}},L_2^{\text{dec}})$. It is now the decoding graph, rather than the doubled graph, that is represented by applying this equivalence to the square lattice.

We look at topologically nontrivial paths in the decoding graph. By symmetry, we just consider two cases: (1) paths from the origin $O=(0,0)$ to $2L_1^{\text{dec}}$ and (2) paths from $O$ to $L_1^{\text{dec}}+L_2^{\text{dec}}$. The shortest path from point $A\in\mathbb{R}^2$ to $B\in\mathbb{R}^2$ clearly has length equal to $\|B-A\|_1$. The complication is that, due to the graph isomorphism between non-checkerboardable and checkerboardable decoding graphs, edges located at points $e_1,e_2\in\mathbb{R}^2$ act on the same qubit (one acts with Pauli $X$ say, and the other with Pauli $Z$) if
\begin{equation}\label{eq:XZ_edge_correspondence}
e_2-e_1\in T(L_1^\text{dec},L_2^{\text{dec}}):=\{m_1L_1^{\text{dec}}+m_2L_2^{\text{dec}}:m_1,m_2\in\mathbb{Z},m_1+m_2\text{ is odd}\}.
\end{equation}
Thus, the Pauli weight of a path is at most $\|B-A\|_1$ but may be less. Eq.~\eqref{eq:XZ_edge_correspondence} implies that if $e_1$ is horizontal then $e_2$ is vertical and vice-versa. So if a path has, say, more horizontal edges than vertical edges, at most each vertical edge may be acting on the same qubit as a horizontal edge, and thus the Pauli weight is at least $\|B-A\|_{\infty}$. In summary, if we let $|P(A,B)|$ denote the Pauli weight of a minimum weight path from $A$ to $B$, we have shown
\begin{equation}\label{eq:path_weight_bounds}
\|B-A\|_{\infty}\le|P(A,B)|\le\|B-A\|_1.
\end{equation}
Moreover, because adding an edge to a path never decreases the Pauli weight, there is a path with minimum Pauli weight that also has minimum length, namely, length $\|B-A\|_1$.

For case (1), we argue that the path with lowest Pauli weight has weight $|P(O,2L_1^{\text{dec}})|=\|2L_1^{\text{dec}}\|_{\infty}=a+b$, matching the lower bound. We need only demonstrate an explicit path from $O$ to $2L_1^{\text{dec}}=(a,b)$ with this weight. We can describe this with a string of symbols $N$ and $E$ representing moving north and east, respectively. So, for instance, $NENE=(NE)^2$ would mean moving north, then east, then north, then east. The path with minimum Pauli weight is then
\begin{equation}\label{eq:NE_sequence}
({\color{red}NE})^{\frac12(b-a-1)}{\color{red}N}E^a({\color{blue}EN})^{\frac12(b-a-1)}{\color{blue}E}E^{a}.
\end{equation}
Edges of opposite orientation act on the same qubit if they are separated in this sequence by $\frac12(b+a-1)$ $E$ edges and $\frac12(b-a-1)$ $N$ edges. For instance, the red edges in the sequence \eqref{eq:NE_sequence} act on the same $b-a$ qubits as the blue edges. Thus, while the sequence's length is $2b$, its Pauli weight is $2b-(b-a)=a+b$, as claimed. See Fig.~\ref{fig:cyc_toric_dec_graph} for an example.

\begin{figure}[t]
    \centering
    \includegraphics[width=0.65\textwidth]{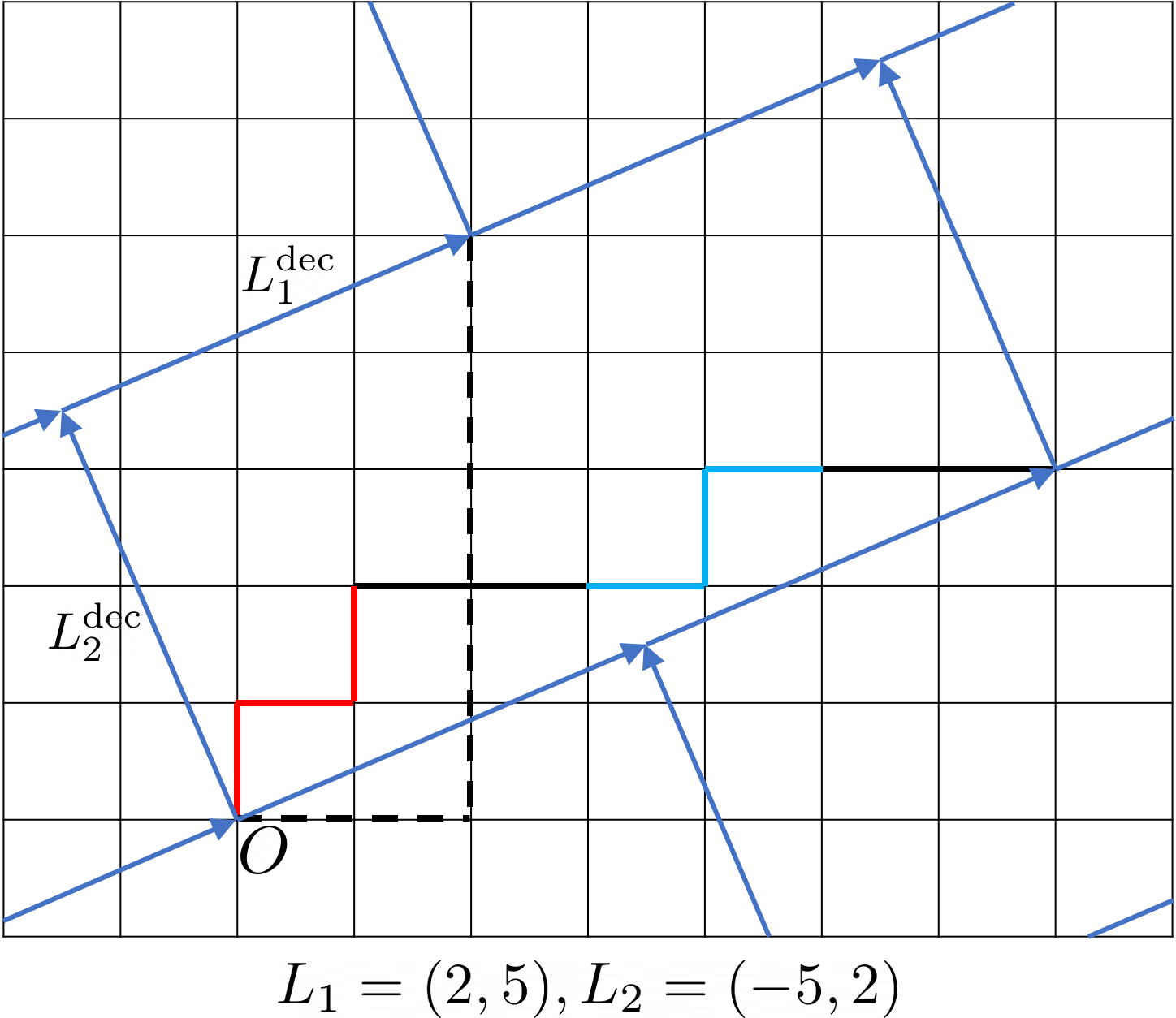}
    \caption{Illustrating the proof of Theorem~\ref{thm:cyclic_toric_codes} with a specific example. We show a path from $O$ to $O+2L_1^{\text{dec}}$ consisting of red, black, and blue edges. Each red edge can be paired with a blue edge that is located a vector $L_1^{\text{dec}}$ away and acts on the same qubit. The dashed lines show a path from $O$ to $O+L_1^{\text{dec}}+L_2^{\text{dec}}$ with minimal Pauli weight in which no two edges act on the same qubit. The logical operators represented by the solid and dashed paths anticommute.}
    \label{fig:cyc_toric_dec_graph}
\end{figure}

For case (2), we argue that the path with lowest Pauli weight has weight $|P(O,L_1^{\text{dec}}+L_2^{\text{dec}})|=\|L_1^{\text{dec}}+L_2^{\text{dec}}\|_1=a+b$, matching the upper bound on path weight. Therefore, we must reason that, in going from $O$ to $L_1^{\text{dec}}+L_2^{\text{dec}}=(a,b)$ one can take no advantage of edges acting on the same qubit. According to the comment just below Eq.~\eqref{eq:path_weight_bounds}, we can restrict our investigation to minimum length paths. We argue that, for any $t\in T(L_1^\text{dec},L_2^{\text{dec}})$ and any edge $e$, no minimum length path contains both $e$ and $e+t$. If we show this claim for $t$, it implies the claim for $-t$. So assume $t=m_1L_1^{\text{dec}}+m_2L_2^{\text{dec}}=(t_x,t_y)$ and that, without loss of generality, $m_1>m_2$. Note, because $t$ associates edges of opposite orientation, both $t_x$ and $t_y$ are half-integers and so necessarily nonzero.

We introduce three easily-verified conditions that preclude any shortest path from containing both edge $e$ and $e+t$. If (A) $|t_x|>a$ or (B) $t_xt_y<0$ or (C) $|t_y|>b$, then the shortest paths are restricted in this way and the shortest path's length equals its Pauli weight. We show all $t\in T(L_1^\text{dec},L_2^{\text{dec}})$ fall into one of these two cases.

First, assume $m_1>0$. Then, because $m_1>m_2$ and $b>a>0$, $t_x>m_1a\ge a$. This falls into case (A). Next, assume $m_2<m_1\le0$. We find $t_y\le -a/2<0$. If $t_x>0$, case (B) applies and we are done. So, suppose $t_x<0$. Then $(m_1+m_2)<-(m_1-m_2)b/a$, which implies $|t_y|/b>(m_1-m_2)(a^2+b^2)/2ab>1$, and we have case (C).
\end{proof}

\subsection{New hyperbolic codes}
\label{sec:improved_hyperbolic_codes}

In this section, we discuss how one can construct hyperbolic quantum codes using Definition~\ref{def:qubit_surface_code}. We restrict to regular tilings of hyperbolic space, and as a result the set of codes we create here is incomparable with the set of regular homological codes created in, for instance, \cite{breuckmann2016constructions,breuckmann2017hyperbolic,breuckmann2017homological}. Codes in our set that are not in the regular homological set include non-checkerboardable codes, as well as some checkerboardable codes defined on graphs with even vertex degrees larger than four. Codes in the homological set not included in our set include those with face (e.g.~$Z$-type) and vertex (e.g.~$X$-type) stabilizers of different weight. In principle, using our code construction on irregular tilings would reproduce all homological hyperbolic codes and more, but irregular graphs are more difficult to study systematically than regular graphs, and so we do not consider them here.

The automorphism group $\text{Aut}(R)$ of a rotation system $R=(H,\lambda,\rho,\tau)$ is the group of all permutations of $H$ that commute with $\lambda$, $\rho$, and $\tau$ \cite{nedela2001regular}. Because of the transitivity of the monodromy group $M(R)$ on $H$, knowing where $\kappa\in\text{Aut}(R)$ sends a single flag $h\in H$ is enough to fix the action of $\kappa$ on all of $H$. Therefore, the largest the automorphism group can be is $|H|$. If $|\text{Aut}(R)|=|H|$, then the rotation system is \textit{regular}, and it represents a regular graph, i.e.~all vertices have the same degree, say $n$, and all faces have the same number of adjacent edges, say $m$. These graphs are called $(m,n)$-regular. A simple calculation and Eq.~\eqref{eq:Eulers_formula} shows
\begin{equation}\label{eq:regular_Eulers_formula}
\frac12\chi=\left(\frac1m+\frac1n-\frac12\right)|E|.
\end{equation}

All flags in regular rotation systems are equivalent up to automorphism. Therefore, it is sensible that there is a description without an explicit set of flags $H$. This description is realized by first specifying a free group along with identity relations on its generators:
\begin{equation}
\mathcal{F}=\langle\lambda,\rho,\tau|\lambda^2=\rho^2=\tau^2=(\lambda\tau)^2=(\lambda\rho)^m=(\rho\tau)^n=1\rangle.
\end{equation}
However, only if $\chi>0$, or equivalently, $1/m+1/n>1/2$, is the group $\mathcal{F}$ finite. Additional identities are necessary to compactify infinite graphs. Compactified regular graphs are in one-to-one correspondence with normal subgroups of $\mathcal{F}$ with finite index, and moreover these take the form of free groups with additional identities $r_i$ imposed on the generators, e.g.~
\begin{equation}\label{eq:free_group_subgroup}
\mathcal{F}_c=\langle\lambda,\rho,\tau|\lambda^2=\rho^2=\tau^2=(\lambda\tau)^2=(\lambda\rho)^m=(\rho\tau)^n=1\text{ and }r_i(\lambda,\rho,\tau)=1,\forall i\rangle.
\end{equation}

To find normal subgroups of $\mathcal{F}_c$, one can use a computer algebra package such as GAP \cite{GAP4,GAP_LINS}. Appendix~\ref{app:GAP_example} gives an example of using GAP for this purpose. We tabulate some small examples of non-checkerboardable hyperbolic codes in Table~\ref{tab:noncheckerboardable_hyperbolic_codes} and draw some explicit examples in Fig.~\ref{fig:hyperbolic_codes}.

\begin{table}[]
    \centering
    \begin{tabular}{|c|c|c|c|c|c|}
         \hline
         Tiling & $N$ & $K$ & $D$ & Generators & Orientable? \\ \hline
         $(5,4)$ & 20 & 5 & 4 & $(\rho\tau\rho\lambda)^4, (\rho\lambda\tau)^5$ & No \\ \hline
         \multirow{4}{*}{\edit{$(4,6)$}} & 6 & 3 & 2 & $(\rho\lambda\tau)^3$ & No \\ \cline{2-6}
         & 15 & 6 & 2 & $(\rho\tau\rho\lambda)^3, (\rho\lambda\tau)^5$ & No \\ \cline{2-6}
         & 24 & 9 & 3 & $(\rho\lambda\tau\rho\lambda)^3, (\rho\tau\rho\lambda)^4,(\rho\lambda\tau)^6$ & No \\ \cline{2-6}
         & 30 & 11 & 3 & $(\rho\tau\rho\lambda)^3$ & Yes \\ \hline
    \end{tabular}
    \caption{Some very small hyperbolic codes on non-checkerboardable graphs. The $(5,4)$ example in particular uses fewer qubits than any distance four hyperbolic code in \cite{breuckmann2017homological}.}
    \label{tab:noncheckerboardable_hyperbolic_codes}
\end{table}

\begin{figure}
    \centering
    \includegraphics[width=0.6\textwidth]{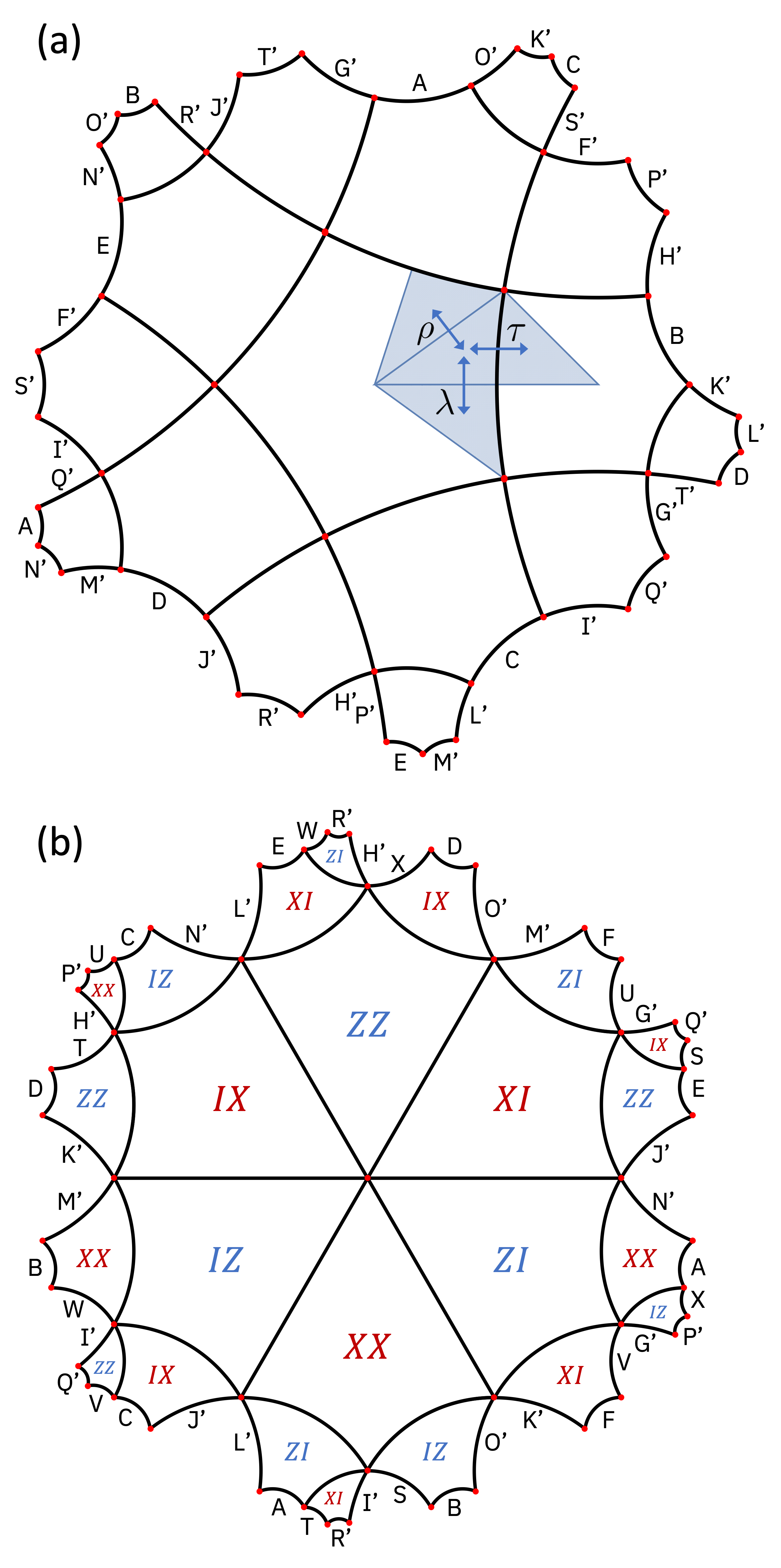}
    \caption{(a) A $\llbracket20,5,4\rrbracket$ hyperbolic code from a $(5,4)$-regular tiling of a genus six non-orientable manifold. (b) A $\llbracket32,10,4\rrbracket$ hyperbolic code from a $(4,6)$-regular tiling of a genus ten non-orientable manifold using the relations $(\lambda\tau\rho\tau\rho)^3=(\tau\rho\lambda\rho)^4=(\rho\lambda\tau)^6=1$. As indicated by letters, edges are identified to compactify the tilings. Non-primed lettered edges are matched orientably, i.e.~one edge is directed clockwise, the other counterclockwise. Primed lettered edges are matched non-orientably, i.e.~both edges directed clockwise. In (a), example actions of the generator permutations are shown as well. In (b), because the code distance depends on the choice and orientation of the CALs, we label each face with the Pauli that should act on each vertex of the face.}
    \label{fig:hyperbolic_codes}
\end{figure}

Of course, non-checkerboardable codes like those in Table~\ref{tab:noncheckerboardable_hyperbolic_codes} and Fig.~\ref{fig:hyperbolic_codes}(a) are not equivalent to homological hyperbolic codes. What is less obvious is that even some checkerboardable codes, like the $\llbracket32,10,4\rrbracket$ code shown in Fig.~\ref{fig:hyperbolic_codes}(b), are not equivalent to homological hyperbolic codes on \emph{regular} tilings (though they are equivalent to a homological code on some irregular tiling). To see this, suppose we define a surface code in our framework on a $(m,n)$-regular tiling with $m>4$ being even. Since $m>4$ this is not a medial graph of any other embedded graph. However, we should also show that even after we decompose vertices into degree-four vertices using Fig.~\ref{fig:spinal_expansion} the graph is still not a medial graph. Suppose it is, so that after the vertex decomposition we have a $(m',4)$-regular tiling. If before the decomposition we have $V$ vertices and $F$ faces, after the decomposition there are $V'=V+(n-4)V/2=(n-2)V/2$ vertices and $F$ faces. By regularity of the tilings, $F=nV/m=4V'/m'$. This implies $m'=(2n-4)m/n$, which must be an integer. When $(2n-4)m/n$ is not an integer, it is not possible to end up with a medial graph after the vertex decomposition. An example is when $n=4$ and $m=6$, as for the code in Fig.~\ref{fig:hyperbolic_codes}(b).

Finally, we point out that the face-width lower bound on code distance, Theorem~\ref{thm:face_width}, is applicable to the hyperbolic codes in this section. Using it and explicitly finding logical operators with that weight, we can verify the distances of the non-checkerboardable codes in Table~\ref{tab:noncheckerboardable_hyperbolic_codes}. The $\llbracket32,10,4\rrbracket$ code in Fig.~\ref{fig:hyperbolic_codes}(b), being checkerboardable, has distance lower bounded by the shortest homologically non-trivial cycle in its decoding graph, Corollary~\ref{cor:checkerboardable_distance}. This gives $3\le D$, or, in other words, any non-trivial logical operator acts on qubits at three different vertices. However, we argue it must act on both qubits at at least one of these vertices, and thus be at least weight four. Focus on the $X$-type logical operators; the $Z$-type are analogous. Note each vertex is surrounded by three $Z$-type faces, $IZ$, $ZI$, and $ZZ$. Thus, the homologically non-trivial cycle with length three must visit a face of each type. At the vertex where it passes from the $IZ$ to the $ZI$ type face, it is supported on both qubits, namely as $XX$. Therefore, the minimum weight logical operator is at least weight four, and it is easy enough to find one with this weight.

\subsection{Two more ways to generalize the triangle code}\label{sec:generalized_triangle}
In \cite{yoder2017surface}, the authors presented the family of triangle codes; the distance 5 version is pictured here in Fig.~\ref{fig:stellated_codes}(a). This family has been generalized by Kesselring et al.~to the stellated codes \cite{kesselring2018boundaries}. Stellated codes feature an improved constant in the relation $N=cKD^2$, namely $c=s/(2s-2)$ for odd integers $s$ specifying the degree of symmetry ($s=3$ for the triangle code, $s=5$ for the code in Fig.~\ref{fig:stellated_codes}(c), etc.), and so $c$ approaches $1/2$ as $s\rightarrow\infty$.

One curiosity in the stellated codes is that the qubit density in the center becomes large as $s$ grows. One might wonder if there is a family of planar codes for which $c$ approaches $1/2$ that does not have this singularity. There is indeed such a family of codes, provided one accepts a conjecture on its code distance.

We begin the construction of this family by imagining ``gluing" square surface code patches, like that pictured in Fig.~\ref{fig:stellated_codes}(b), together. For instance, one can glue three distance $D=3$ square surface code patches together to obtain the $D=5$ triangle code, Fig.~\ref{fig:stellated_codes}(a), and five distance $D=3$ patches to obtain the code in Fig.~\ref{fig:stellated_codes}(c). To ensure that the qubit density is bounded by a constant throughout the resulting graph, we should not glue too many patches together around a single vertex. We glue at most four patches around a vertex, and the maximum vertex degree in the graph is six, leading to at most two qubits at a vertex.

Twist defects, the odd-degree vertices in the resulting graph, should be spaced apart by graph distance $D-1$ to get a distance $D$ code. Therefore, we imagine each square patch is $(D+1)/2$-by-$(D+1)/2$ qubits in size, and that the only odd-degree vertices are located at two corners of the patch, diagonally opposite one another. If we glue the patches together accordingly, we get codes as in Fig.~\ref{fig:circle_packing_code}. We refer to this family as circle-packing codes, because of the way in which odd-degree vertices are closely packed. Because the addition of a square patch of roughly $D^2/4$ qubits adds one additional twist, and two additional twists are required to add a logical qubit, we have $N\rightarrow KD^2/2$ for these codes as $K$ gets large.

\begin{figure}
    \centering
    \includegraphics[width=0.8\textwidth]{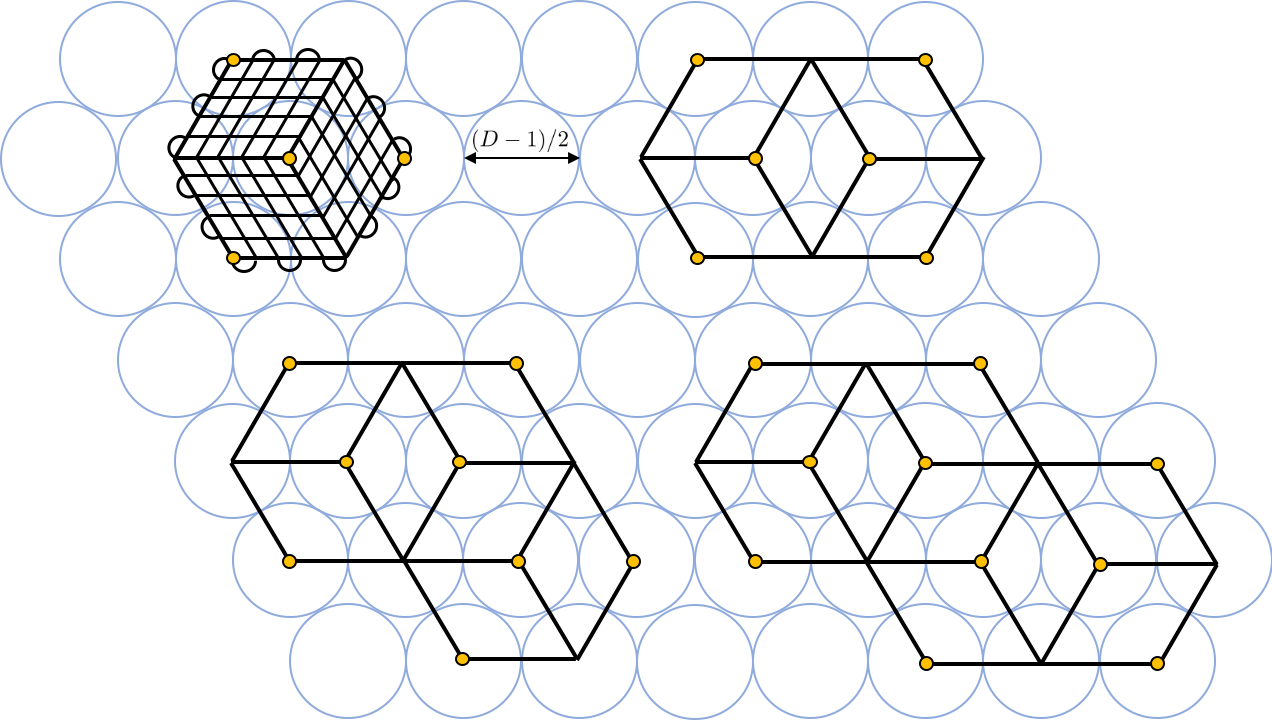}
    \caption{Circle-packing codes encoding $K=1,2,3,$ and $4$ qubits in order of size. The packed circles in the background act as a coordinate system, and are not part of the graphs or codes. In the top left, we show the triangle code and the entire square lattices of qubits making it up. In other drawings, the lattice interiors are abstracted away. Odd-degree vertices, or twist defects, are shown as yellow circles. One can continue attaching square patches, following the pattern established in these first four examples, to grow the code down and to the right and obtain a family where $N=\frac14(2K+1)D^2$. Therefore, $c=(2K+1)/4K$ approaches $1/2$ as $K$ increases.}
    \label{fig:circle_packing_code}
\end{figure}

A seemingly strange feature of the circle-packing codes is how they do not fill a 2-D area but are rather grown outward in one direction. However, one can easily see that filling a 2-D area with the same pattern of square patches results in a worse constant $c$. Suppose we have glued $F$ square patches together, filling a large 2-D area, so that the perimeter is much smaller than the volume. Then, because we can neglect the perimeter, there are $2F/3$ odd-degree vertices, implying $K\approx 2F/6$ and $N\approx F(D/2)^2=3KD^2/4$.

We note explicitly that we did not prove that the distance of the circle-packing codes is actually $D$, unlike the more exaggerated proofs in Secs.~\ref{sec:square_lattice_toric_codes} and \ref{sec:rotated_toric_codes}. We expect this is the code distance due to comparison with the similar codes in \cite{yoder2017surface} and \cite{kesselring2018boundaries} and the confidence with which those authors state the code distances (though one might say they are also lacking rigorous proofs).

A second way we can generalize the stellated codes is by embedding them in higher genus ($g>0$) surfaces. By Corollary~\ref{cor:number_encoded_qubits}, a larger genus can lead to more encoded qubits, assuming any reduction in the number of odd-degree vertices does not outweigh the effect. Our higher genus embeddings double the number of encoded qubits with small reductions in both $N$ and $D$.

\begin{figure}
    \centering
    \includegraphics[width=0.85\textwidth]{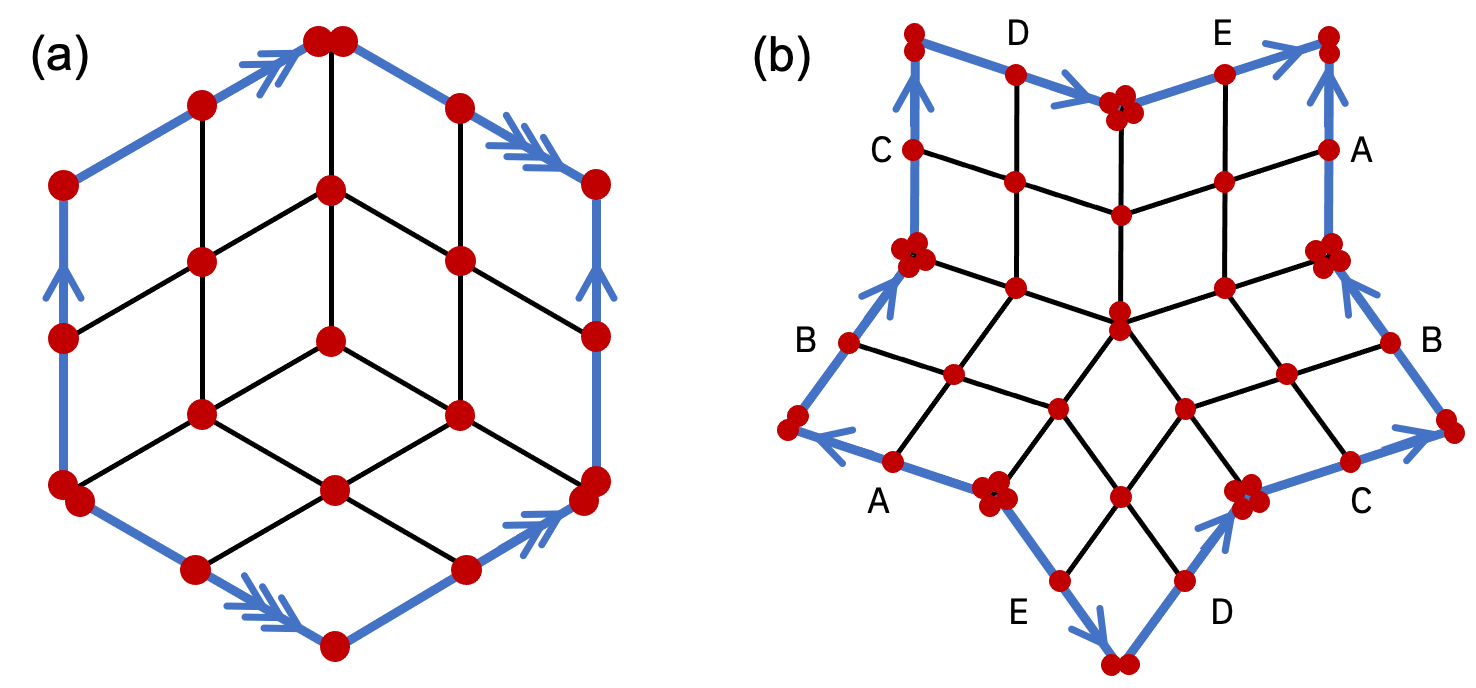}
    \caption{Two examples of stellated codes embedded in higher genus surfaces, (a) a $\llbracket13,2,4\rrbracket$ code with $s=3$ and $t=2$ and (b) a $\llbracket23,4,4\rrbracket$ code with $s=5$ and $t=2$. They are drawn within a $2s$-gon with opposite sides identified. The $\llbracket13,2,4\rrbracket$ code actually has optimal distance for a quantum code with $N=13$ and $K=2$ \cite{Grassl:codetables}.}
    \label{fig:high_genus_stellated}
\end{figure}

The idea of the higher genus embedding is to draw the stellated code with odd symmetry parameter $s$ into a polygon with $2s$ sides and identify opposite sides of the polygon to make an orientable manifold. In fact, this manifold is a torus with $g=(s-1)/2$ holes. This embedding results in two odd-degree vertices with degree $s$, a single vertex with degree $2s$, and all other vertices having degree four. Suppose the graph distance from the code's center to the boundary is $t$. Then, $N=st^2+s-2$, $K=s-1$, and we conjecture $D=2t$. This leads to $N=\frac14(K+1)D^2+K-1=\frac14KD^2+O(D^2+K)$, so that if both $K$ and $D$ are large, we have $c=1/4$.

We believe the code distance is $2t$ because the code distance of a stellated code in the plane is $2t+1$, the length of some logical operators that cross the code from one boundary to another. In the higher genus case, the qubits on the boundary are identified, reducing the weight of such operators by one. It is easy to see that $2t$ is indeed an upper bound on $D$ by finding a logical operator (one crossing the code, as described) with that weight.

%% file: Open-problems.tex
\section{Open problems}\label{sec:open_problems}

In this work, we took a rigorous look at qubit surface codes defined from embedded graphs. We related them to Majorana surface codes, calculated the number of encoded qubits, and bounded the code distance in general with more specific cases worked out exactly.

There remain some interesting questions, however. First, it would be very nice to be able to calculate (efficiently) the code distance of an arbitrary surface code defined by Definition~\ref{def:qubit_surface_code}. We illustrated a prickly example in Fig.~\ref{fig:obstruction} and the codes in Section~\ref{sec:generalized_triangle} lack rigorous distance proofs, for instance. If the distance cannot be efficiently calculated exactly, can a multiplicative approximation be found better than the factor $1/2$ approximation of Theorem~\ref{thm:distance_bounds_from_doubled_graph}? Inapproximability (see e.g.~\cite{trevisan2004inapproximability}) in complexity theory might be applicable if indeed the exact distance calculation is NP-hard. \edit{We note that existing algorithms for exact distance calculation in the literature are limited to CSS surface codes \cite{breuckmann2017hyperbolic} or require exponential time for some of the codes in our formalism \cite{cao2023quantum}.}

A second interesting problem is to create a graph-based formalism for color code twists as described by Kesselring et al. \cite{kesselring2018boundaries}. By this we mean that one should be able to define twisted color codes on arbitrary embedded graphs, where the variety of twists correspond to (local) features of the graph, just as we found that surface code twists arise at odd-degree vertices. For instance, while untwisted color codes are defined from graphs with 3-colorable, even-degree faces, a twist may arise in a more general graph where a face has odd-degree.

A third direction would be to literally add a third direction: can a lattice-based formalism be used to describe the variety of surface codes and twist defects in 3-dimensions and beyond? Here things may in general get messy due to the lack of classification theorems like those for 1D and 2D codes \cite{bombin2014structure,haah2016algebraic}. Still, we do not need to demand that the framework describe every 3D code by a 3D lattice, just that every 3D lattice gives rise to some 3D code. We probably would want some known family, like the toric code, to arise from the construction on the cubic lattice, and that something interesting (better code parameters, fractal logical operators, etc.) happen for other lattices.

Finally, probably the most practical open problem is to find lower bounds on the constant $c$ in the relation $N=cKD^2$ \cite{bravyi2010tradeoffs} that holds for 2-dimensional codes. Note that $c(w,\mathcal{M})$ may depend on both the maximum stabilizer weight $w$ and the manifold $\mathcal{M}$. Therefore, a concrete question is, for instance, limited to stabilizers of at most weight five in the plane $\mathbb{R}^2$, what is the smallest possible value of $c(5,\mathbb{R}^2)$? We believe the best known is $c(5,\mathbb{R}^2)$ approaching $1/2$, see Figure~\ref{fig:circle_packing_code}. If one allows weight-six stabilizers, as in the color codes of \cite{kesselring2018boundaries}, then $c(6,\mathbb{R}^2)$ can approach $1/4$. One can generally concatenate two 2-dimensional codes (e.g.~surface codes) with the $\llbracket4,2,2\rrbracket$ code to get a color code with $N'=4N$, $K'=2K$, $D'=2D$, and $w'=2w$ \cite{criger2016noise}. Therefore $c(2w,\mathcal{M})\le c(w,\mathcal{M})/2$ for any $w$ and $\mathcal{M}$.

%% file: Appendix.tex
\appendix
\appendixpage
\addappheadtotoc
\input{appendix/gluing-orientable-v3.tex}
\input{appendix/algebraic-topology.tex}
\input{appendix/cal-v3.tex}
\input{appendix/majorana_qubit_code_proofs.tex}
\input{appendix/majorana_qubit_code_equivalence_v1.tex}
\input{appendix/cyclic_codes_params_v2.tex}
\input{appendix/medial-graphs-and-trisection.tex}

\input{appendix/GAP_example.tex}

%% file: appendix/gluing-orientable-v3.tex
\section{Graph embeddings from oriented rotation systems}
\label{sec:app-gluing-orientable}

The goal of this appendix is to establish an equivalence between graphs embedded in orientable manifolds and a combinatorial object called an oriented rotation system. It is easier to see this equivalence in one direction --- start with an embedded graph in an orientable manifold and develop the oriented rotation system description. We do this first below, and sketch a more formal argument for the converse direction after that. 

\edit{As was already mentioned before in Section~\ref{subsec:embedding-non-orientable}, the oriented rotation system is a special case of the general rotation system. Since the manifold is orientable, one can unambiguously define a ``side'' of the manifold, and thus one does not need to maintain a flag for each half-edge (defined below in Defintion~\ref{def:oriented_rotation_system}), one for each side of the manifold. Thus one can simply work with the half-edges, two for every edge of the graph. As in the case of general rotations systems, the vertices, edges, and faces of the graph embedding will be sets of these half-edges.}

\begin{definition}
\label{def:oriented_rotation_system}
An \textit{oriented rotation \edit{system}} is a triple $R_O=(H_O,\nu,\epsilon)$, where $H_O$ is a finite set of \textit{half-edges} (sometimes called \textit{darts}), and two permutations $\nu:H_O\rightarrow H_O$ and $\epsilon:H_O\rightarrow H_O$ satisfying the additional properties
\setmargins
\begin{enumerate}[(i)]
    \item $\epsilon$ is a fixed-point-free involution, meaning $\epsilon^2 = \text{Id}$, and $\epsilon h \neq h$ for all $h \in H_O$.
    \item The free group $\langle \nu, \epsilon \rangle$ generated by $\nu$ and $\epsilon$, acts transitively on $H_O$. This means that for all $h, h' \in H_O$, one can find $\kappa \in \langle \nu, \epsilon \rangle$ such that $\kappa h = h'$.
\end{enumerate}
\end{definition}
We now define the three sets
\begin{equation}
\label{eq:graph-from-rotation-system}
\begin{split}
    V &= \{v \subseteq H_O: v \text{ is an orbit of the free group } \langle \nu \rangle \}, \\
    E &= \{e \subseteq H_O: e \text{ is an orbit of the free group } \langle \epsilon \rangle \}, \\
    F &= \{f \subseteq H_O: f \text{ is an orbit of the free group } \langle \nu \epsilon \rangle \}.
\end{split}
\end{equation}
We call elements of these sets the \textit{vertices}, \textit{edges}, and \textit{faces} of the oriented rotation system, respectively. So, condition (i) ensures edges are sets of exactly two half-edges, and thus $|H_O|$ is even, while condition (ii) ensures the graph is connected.

Thus given an embedded graph $G(V,E)$ in an oriented manifold $\mathcal{M}$, to get the oriented rotation system $(H_O,\nu,\epsilon)$ corresponding to it, we first create a set of half-edges $H_O$ with each edge in $E$ contributing two half-edges. Let $e \in E$ with corresponding half-edges $h_e, h'_e \in H_O$. The permutation $\epsilon$ is defined so that $\epsilon h_e = h'_e$, and $\epsilon h'_e = h_e$ for every $e \in E$, and note that property (i) of Definition~\ref{def:oriented_rotation_system} is satisfied. If $e$ is adjacent to vertices $v,v' \in V$ ($v=v'$ if $e$ is a loop), we assign $h_e$ and $h'_e$ to $v$ and $v'$ respectively. Thus each vertex $v$ gets assigned exactly $\text{deg}(v)$ half-edges, and conversely each half-edge in $H_O$ is assigned to exactly one vertex in $V$. This partitions $H_O$ into disjoint subsets indexed by the vertices. Now choose a \edit{continuously} varying outward normal on $\mathcal{M}$, which is possible since it is orientable, fix a vertex $v \in V$, and let $(H_O)_v$ be the half-edges assigned to it. We can enumerate the edges adjacent to $v$ in counterclockwise order with respect to the outward normal, where an edge appears twice successively in this enumeration if it is a loop. Then replacing each edge in this ordering by the corresponding half-edge that it contributes to $v$, gives a counterclockwise ordering of $(H_O)_v$. The permutation $\nu$ for $(H_O)_v$ is defined to be the cyclic permutation given by this counterclockwise ordering, and repeating this at every vertex completely specifies $\nu$. It follows easily that property (ii) of Definition~\ref{def:oriented_rotation_system} is satisfied if $G$ is connected. This completes the construction of the oriented rotation system $(H_O,\nu,\epsilon)$. See Fig.~\ref{fig:rot_sys_ex_A} for an example. However, the description we just provided fails to capture graph embeddings in non-orientable manifolds, because in that case we cannot choose a \edit{continuously} varying outward normal.

\begin{figure}[h]
    \centering
    \includegraphics[width=0.4\textwidth]{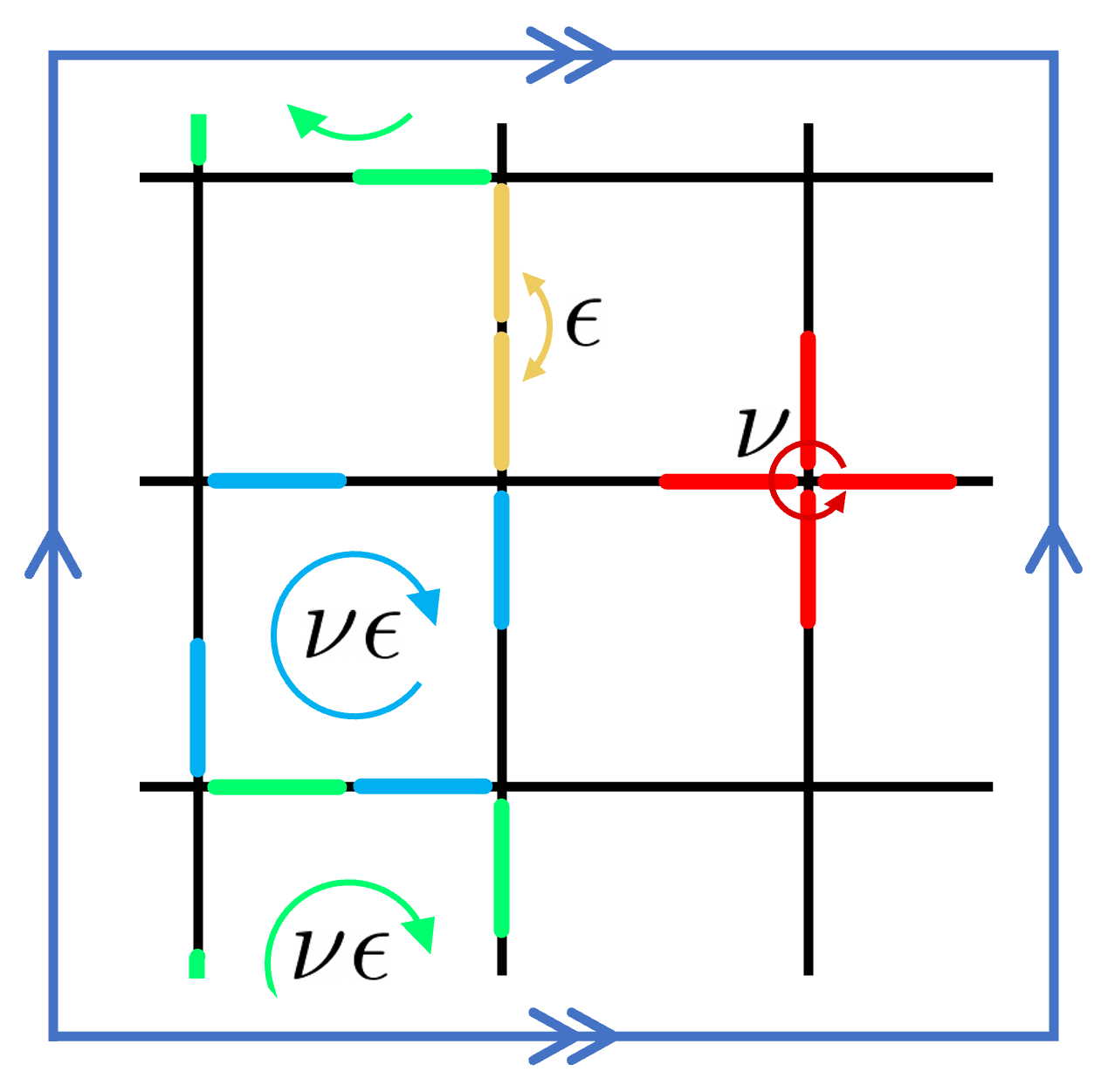}
    \caption{A graph embedded onto the torus and its oriented rotation system. An orbit of $\langle \nu \rangle$, an orbit of $\langle \epsilon \rangle$, and two orbits of $\langle \nu \epsilon\rangle$ are shown. These correspond to an edge, a vertex, and two faces of the graph embedding.}
    \label{fig:rot_sys_ex_A}
\end{figure}

Next, we show how to get a graph embedding on an orientable manifold starting from an oriented rotation system $R_O=(H_O,\nu,\epsilon)$. 
Recall that the rotation system gives rise to a set of vertices $V$, edges $E$, and faces $F$, which are orbits of $\langle \nu \rangle$, $\langle \epsilon \rangle$ and $\langle \nu \epsilon \rangle$ respectively. The graph specified by $R_O$, consists of vertices and edges labeled by the elements of $V$ and $E$ respectively (in a bijective fashion), so we can use the notation $G(V,E)$ to denote the graph without any ambiguity. \edit{A vertex $v \in V$ and edge $e \in E$ are adjacent if $v \cap e \neq \emptyset$.} \edit{One notes that} property (ii) of Definition~\ref{def:oriented_rotation_system} guarantees that $G$ is connected.

In addition to specifying the graph connectivity, $R_O$ also specifies an embedding of $G$, where the faces of the graph embedding are in bijection with the elements of $F$ --- this justifies the use of the notation $G(V,E,F)$ to denote the graph embedding. We will now briefly sketch how to obtain this embedding starting from $R_O$. For this we establish some additional notation: for every half-edge $h\in H_O$, let $[h]_\nu \in V$ be the unique vertex that contains $h$ and, likewise, let $[h]_\epsilon \in E$ be the unique edge that contains $h$. Adjacency of vertices, edges and faces for oriented rotation systems are defined by non-trivial intersection (just as for general rotation systems), so for e.g. a vertex $v$ and a face $f$ are adjacent if and only if $v \cap f \neq \emptyset$.

The construction starts by assigning an open disc $\mathcal{B}_{f}(0,1) \subseteq \mathbb{R}^2$ to each face $f \in F$ of $R_O$. For a face $f$, let us order the half-edges in it as $\{h_0, \dots, h_{|f|-1}\}$, where $h_i = (\nu \epsilon)^{i}h$ for any $h \in f$ chosen arbitrarily. Now for the closed disc $\overline{\mathcal{B}_{f}(0,1)}$, first let $S^1_{f}$ denote the boundary circle parameterized in polar coordinates, and place $|f|$ distinct points $0 = \theta_{0} < \dots < \theta_{|f|-1} < 2\pi$ on $S^1_f$. Let $\omega_{0}, \dots, \omega_{|f|-1}$ denote the open segments that result defined as $\omega_i = (\theta_i, \theta_{i+1})$ for $0 \leq i \leq |f| - 2$, and $\omega_{|f|-1}= (\theta_{|f|-1},2\pi)$, and then define a marking function $M_{f} : \bigcup_{i=0}^{|f|-1} \{\theta_{i}, \omega_{i}\}  \rightarrow V \cup H_O$ as
\begin{equation}
\label{eq:marking-func}
    M_f (x) =
    \begin{cases}
        [h_i]_\nu & \;\; \text{if } x = \theta_i \\
        h_i & \;\; \text{if } x = \omega_i.
    \end{cases}
\end{equation}
The combined effect of the marking functions for all the faces is that each half-edge $h$ gets associated with exactly one open segment of a boundary circle $S^{1}_{f}$ for some face $f$, while each vertex $v$ gets associated with exactly $|v|$ points in different boundary circles. Finally, we take the disjoint union space $\bigsqcup_{f \in F} \overline{\mathcal{B}_{f}(0,1)}$ and define an identification rule $\sim$ for points on the boundary circles as follows:
\begin{enumerate}[(i)]
    \item For every half-edge $h \in H_O$, if $(\theta_{a}, \theta_{b})$ and $(\theta'_{a}, \theta'_{b})$ are the open segments associated with $h$ and  $\epsilon h$ respectively, identify these segments using the map $(\theta_{a}, \theta_{b}) \ni x \mapsto x (\theta'_{b} - \theta'_{a}) / (\theta_{a} - \theta_{b}) + (\theta_{a} \theta'_{a} - \theta_{b} \theta'_{b}) / (\theta_{a} - \theta_{b}) \in (\theta'_{a}, \theta'_{b})$.
    \item If points $\theta, \theta'$ are associated with vertices $v,v'$ respectively, then they are identified if and only if $v = v'$.
\end{enumerate}
The quotient space $\mathcal{M} = \bigsqcup_{f \in F} \overline{\mathcal{B}_{f}(0,1)} / \sim$ that results under this identification is an orientable manifold (we omit the topological details), and $G(V,E)$ is 2-cell embedded in it. The graph embedding map $\Gamma$ is obtained as follows: (a) if a vertex $v \in V$ is associated with points $\{\theta_{1},\dots,\theta_{|v|}\}$, then $\Gamma(v)$ equals the equivalence class of these points under $\sim$, (b) if a half-edge $h$ is associated with an open segment $(\theta_{a}, \theta_{b})$ then for the corresponding edge $e = \{h,\epsilon h\}$, $\Gamma(e)$ equals the closure in $\mathcal{M}$ of the equivalence class of $(\theta_{a}, \theta_{b})$ under $\sim$.

%% file: appendix/algebraic-topology.tex
\section{Chain complexes}\label{app:algebraic_topology}

\edit{In this appendix, we translate the discussion of Section~\ref{subsec:homology} into the language of chain complexes. Any graph embedding $G(V,E,F)$ can be associated with a \textit{chain complex}, consisting of $0$-chains, $1$-chains, and $2$-chains. A $2$-chain is a formal linear combination of the faces of the graph embedding with coefficients in $\mathbb{F}_2$, that is an expression of the form
\begin{equation}
    \sum_{f \in F} c_f f, \;\; c_f \in \mathbb{F}_2,
\end{equation}
and the set of all $2$-chains form an $\mathbb{F}_2$-vector space of dimension $|F|$, which we will denote as $C_2(G)$. Similarly, $1$-chains and $0$-chains are defined to be formal linear combinations of edges and vertices of the graph respectively, with coefficients in $\mathbb{F}_2$. The set of all $1$-chains, which we will denote $C_1(G)$, and the set of all $0$-chains, which we will denote $C_0(G)$, form $\mathbb{F}_2$-vector spaces of dimensions $|E|$ and $|V|$ respectively. It is clear that the set of faces, edges, and vertices form a basis for $C_2(G)$, $C_1(G)$, and $C_0(G)$ respectively.

Associated to these vector spaces are two linear maps, called \textit{boundary maps} $\partial_2$ and $\partial_1$. To define these maps, we need the face-edge adjacency matrix $\Phi$ defined in Section~\ref{subsec:checkerboardability}, and also need the edge-vertex adjacency matrix $\widetilde{\Phi} \in \mathbb{F}_2^{|E| \times |V|}$, which we define to be $\widetilde{\Phi} = \frac{1}{2} B A^{\top}$ (recall that $A,B$ are defined in Section~\ref{subsec:embedding-non-orientable}), with the result reduced modulo $2$. Note that $|e \cap v| = 0, 2$, or $4$, with the case $|e \cap v|=4$ corresponding to the case when both the endpoints of $e$ are incident on $v$. Thus given an edge $e$ and a vertex $v$, we have that $\widetilde{\Phi}_{ev} = 0$ if $|e \cap v| = 0$ or $4$, and $\widetilde{\Phi}_{ev} = 1$ if $|e \cap v| = 2$. The maps $\partial_1$ and $\partial_2$ are now defined by their action on the basis elements of their domains as follows, with the convention that an empty sum is the zero vector:
\begin{equation}
\label{eq:def-boundary-maps-homology}
\begin{split}
    &\partial_2: C_2(G) \rightarrow C_1(G), \;\; \partial_2(f) := \sum_{ \{e:  \Phi_{fe}=1 \}} e, \; \forall f \in F, \\
    &\partial_1: C_1(G) \rightarrow C_0(G), \;\; \partial_1(e) := \sum_{ \{v:   \widetilde{\Phi}_{ev}=1\}} v, \;\; \forall e \in E,
\end{split}
\end{equation}
and these maps are extended uniquely by linearity to the whole domain. This also shows that the matrices $\widetilde{\Phi}^{\top}$ and $\Phi^{\top}$ are the matrix representations of the linear maps $\partial_1$ and $\partial_2$ respectively. With these boundary maps defined, we have the chain complex given by the diagram
\begin{equation}
    C_2(G) \xrightarrow[]{\partial_2} C_1(G) \xrightarrow[]{\partial_1} C_0(G).
\end{equation}
The first homology group over $\mathbb{F}_2$ can now be equivalently defined to be the quotient vector space
\begin{equation}
\label{eq:F2-homology-1}
    H_1(G, \mathbb{F}_2) = \frac{\text{kernel}( \partial_1)}{ \text{range}( \partial_2)}.
\end{equation}
This definition makes sense because $\partial_1 \circ \partial_2 = 0$, which can be verified directly by showing $\sum_{e}\Phi_{fe}\widetilde{\Phi}_{ev}=0$ modulo two for all $f\in F$, $v\in V$. Calculate
\begin{align}\label{eq:phiphi_step1}
\sum_{e\in E}\Phi_{fe}\widetilde{\Phi}_{ev}=\sum_{e:\Phi_{fe}=1}\widetilde{\Phi}_{ev}&=\sum_{e:\Phi_{fe}=1}\frac12\sum_{h\in H}B_{eh}A_{vh}\mod{2}\\
\label{eq:phiphi_step2}
&=\sum_{e:\Phi_{fe}=1}\left|\{[h]_\tau\in H/\tau:[h]_\tau\subseteq e,v\}\right|\mod{2}\\
\label{eq:phiphi_step3}
&=\left|\{[h]_\tau\in H/\tau:[h]_\tau\subseteq v,\left|[h]_\tau\cap f\right|=1\}\right|\mod{2}\\\label{eq:phiphi_step4}
&=|\{h\in H:h\in v,f\}|\mod{2}=0.
\end{align}
Going from Eq.~\eqref{eq:phiphi_step1} to Eq.~\eqref{eq:phiphi_step2} involves noticing that if $B_{eh}A_{vh}=1$ then $B_{e,\tau h}A_{v,\tau h}=1$, $B_{eh}=1$, and $A_{vh}=1$. If $A_{vh}=1$, then $[h]_\tau=\{h,\tau h\}\subseteq v$ by definition of vertex-flag adjacency matrix $A$. Likewise, if $B_{eh}=1$, then $[h]_\tau\subseteq e$ by definition of edge-flag adjacency matrix $B$. To show Eq.~\eqref{eq:phiphi_step3} use the definition of $\Phi$ to conclude that any edge $e=\{h,\tau h,\lambda h,\tau\lambda h\}$ for which $\Phi_{fe}=1$ shares exactly one of $h$ and $\tau h$ with $f$. Because we evaluate the set size modulo 2, we can count flags shared between $v$ and $f$ directly as in Eq.~\eqref{eq:phiphi_step4}. We conclude the number of shared flags is even by noting that for every flag $h\in v,f$, there is another flag $\rho h\in v,f$.

One should also note the following. Let us decompose the set of edges $E$ into disjoint sets $E'$ and $E''$, where $E' = \{e \in E: \Phi_{fe}=0 \; \;\; \forall f \in F\}$, and $E'' = E \setminus E'$. Then defining the subspaces $C'_1(G)$ (resp. $C''_1(G)$) to be formal linear combinations of edges in $E'$ (resp. $E''$) with coefficients in $\mathbb{F}_2$, we see that $C_1(G)$ can be written as a direct sum $C_1(G) = C'_1(G) \oplus C''_1(G)$. From the definition of $\partial_2$ in Eq.~\ref{eq:def-boundary-maps-homology}, it is then clear that $\text{range}( \partial_2) \subseteq C''_1(G)$.

Since trails are elements of $\mathbb{F}_2^{|E|}$, we now see that they are $1$-chains. However, not every $1$-chain is a trail. On the other hand, every edge is a trail and an $1$-chain, and thus the group $\mathcal{T}(G)$ can be canonically identified with the $\mathbb{F}_2$-vector space $C_1(G)$, such that trails are identified with their corresponding $1$-chains, and the identification is a group isomorphism. Once this identification is made, we can trace through the definitions of $\mathcal{Z}(G)$ and $\mathcal{B}(G)$ in Eq.~\eqref{eq:BG-ZG}, and that of $\partial_1$ and $\partial_2$ in Eq.~\eqref{eq:def-boundary-maps-homology}, to conclude that $\mathcal{Z}(G)$ is identified with $\text{kernel}(\partial_1)$, and $\mathcal{B}(G)$ is identified with $\text{range}(\partial_2)$. Thus cycles of $G$ are elements of $\text{kernel}(\partial_1)$ and boundaries of faces of $G$ are elements of $\text{range}(\partial_2)$. This proves the equivalence of the two definitions of $H_1(G,\mathbb{F}_2)$ in Eq.~\eqref{eq:F2-homology} and Eq.~\eqref{eq:F2-homology-1}.
}

%% file: appendix/cal-v3.tex
\section{Paulis with arbitrary commutation patterns and CALs}
\label{sec:cal-prop}

In this appendix we first do a very brief review of the Pauli group and recall some of its basic properties. We then analyze the problem of finding a list of Paulis given a desired commutation pattern. As a special case, we study the case of cyclically anticommuting list of Paulis in some detail, which is a key tool in the construction of the qubit surface codes of Section~\ref{subsec:qub_surface_codes}.

\subsection{Basic definitions}
\label{subsec:pauli-definitions}

A \textit{multiset} $M$ is a set with repeated elements. The \textit{multiplicity} $\nu (m)$ of an element $m \in M$ is defined as the number of occurrences of $m$ in $M$. For a set $A$, if $m \in M$ implies $m \in A$, then we say $M \subseteq A$. If $M'$ is another multiset, then we say $M' \subseteq M$ if and only if the following condition holds: for each $m \in M'$ with multiplicity $\nu'(m)$, it must hold that $m \in M$ and $\nu'(m) \leq \nu(m)$. If $M' \subseteq M$ and $M' \neq M$ then we write $M' \subset M$. An ordered multiset is a \textit{list}.

The Pauli group $\mathcal{P}_1$ on one qubit is the 16 element group consisting of the $2 \times 2$ identity matrix $I_2$, and the Pauli matrices
\begin{equation}
    X = 
    \begin{bmatrix}
    0 & 1 \\
    1 & 0
    \end{bmatrix}, \;\;
    Y = 
    \begin{bmatrix}
    0 & -i \\
    i & 0
    \end{bmatrix}, \;\;
    \text{and } 
    Z = 
    \begin{bmatrix}
    1 & 0 \\
    0 & -1
    \end{bmatrix}, \;\;
\end{equation}
together with all possible phase factors of $\pm 1$ and $\pm i$. The $n$-Pauli group $\mathcal{P}_n$ on $n$ qubits ($n \geq 1$), with qubits numbered $0,\dots,n-1$, is now defined to be the group 
\begin{equation}
\label{eq:pauligrp-def}
    \mathcal{P}_n = \{\eta \; (\sigma_0 \otimes \dots \otimes \sigma_{n-1}) : \eta \in \{\pm 1, \pm i\}, \sigma_j \in \{I_2, X, Y, Z\} \; \forall \; j\},
\end{equation}
where $\otimes$ denotes the Kronecker product of matrices. The \textit{support} of a Pauli $p = \eta \; (\sigma_0 \otimes \dots \otimes \sigma_{n-1}) \in \mathcal{P}_n$, which we define as $\supp{p} := \{0 \leq j \leq n-1 : \sigma_j \neq I_2\}$, indicates all those qubits where it acts by either $X$, $Y$ or $Z$. For brevity we will often simply list the Pauli matrices acting on the support of a Pauli, with a subscript per matrix indicating the qubit where it acts. For example, $p = i X_0 Y_2 Z_3$ indicates that $p$ acts on qubits $0$, $2$ and $3$ with $X$, $Y$ and $Z$ respectively, and with $I_2$ on all remaining qubits. Another equivalent notation for representing $p$ (frequently used in Appendix~\ref{sec:cyclic-code-families}) is to drop the $\otimes$ symbol, which gives a string of length $n$ with a phase factor in front. This is called the string representation of Paulis. In a string representation, since there is never any chance for confusion, $I_2$ is replaced by $I$ to simplify the notation. For example if $n=5$, the Pauli $i X_0 Y_2 Z_3$ has the string representation $i XIYZI$. Moreover one can optionally combine alphabets appearing consecutively in a string representation, for which we give an example: for $n=5$, if $XXIIZ$ is the string representation of a Pauli, then it is also equivalently written in compressed form as $X^{\otimes 2}I^{\otimes 2}Z$. We define the abelian group $\pauligrp{n}$ to be the quotient group $\mathcal{P}_n / \langle iI \rangle$, i.e. if $\hat{p} \in \pauligrp{n}$, then it is an equivalence class of the form $\hat{p} = \{p, -p, ip, -ip\} = \langle iI, p \rangle$ for some $p \in \mathcal{P}_n$, where $I$ is the identity element of $\mathcal{P}_n$, the $2^n \times 2^n$ identity matrix. In this appendix and also in the next, if $p \in \mathcal{P}_n$, then $[p]$ will refer to the equivalence class of $p$ in $\pauligrp{n}$, and the support of $[p]$ is defined to be the support of $p$. The \textit{weight} of any element of $\mathcal{P}_n$ or $\pauligrp{n}$ is defined to be the number of qubits in its support. For any $[p] \in \pauligrp{n}$, where $p = \eta \; (\sigma_0 \otimes \dots \otimes \sigma_{n-1}) \in \mathcal{P}_n$, we define the \textit{restriction} of $p$ (resp. $[p]$) to qubit $j$ for any $0 \leq j \leq n-1$, to be $\sigma_j \in \mathcal{P}_1$ (resp. $[\sigma_j] \in \pauligrp{1}$). An element $p \in \mathcal{P}_n$ (resp. $\hat{p} \in \pauligrp{n}$) is called $X$-type if and only if the restriction of $p$ (resp. $\hat{p}$) to qubit $j$ is either $I_2$ or $X$ (resp. $[I_2]$ or $[X]$), for all $0 \leq j \leq n-1$. We similarly define $Y$-type and $Z$-type elements of $\mathcal{P}_n$ and $\pauligrp{n}$.

$\mathcal{P}_n$ is a non-Abelian group of size $4^{n+1}$, and each element of $\mathcal{P}_n$ is traceless, while we have $|\pauligrp{n}| = 4^n$. Any two elements $p,q \in \mathcal{P}_n$ either commute or anticommute, and we will use the notations $[p,q] = 0$ and $\{p,q\} = 0$, to denote these two cases respectively. Note that $\{p,q\}$ can also refer to the set containing the elements $p$ and $q$, but this difference will either be clear from context, or we will explain the usage whenever there is a chance for confusion. For convenience, we define a function $\text{Com}: \mathcal{P}_n \times \mathcal{P}_n \rightarrow \mathbb{F}_2$, that will be useful in some places, as
\begin{equation}
\com{p}{q} = \frac{1}{2} \left(1-\Tr{p q p^\dag q^\dag}/2^n\right),
\end{equation}
which satisfies $\com{p}{q} = 0$ if and only if $[p,q]=0$. Thus $\text{Com}$ is a symmetric function of its arguments. Two elements $\hat{p},\hat{q} \in \pauligrp{n}$ will be said to commute (resp. anticommute)\footnote{This notion of commutativity of elements of $\pauligrp{n}$ is different from that due to the group operation on $\pauligrp{n}$. Elements of $\pauligrp{n}$ always commute under the group operation on $\pauligrp{n}$.} if and only if for any chosen representatives $p \in \hat{p}$ and $q \in \hat{q}$, $p$ and $q$ commute (resp. anticommute), and we will use the same notation as above to express commutativity. For any non-empty, ordered multiset $\mathcal{H} = \{\hat{p}_1, \dots, \hat{p}_\ell\} \subseteq \pauligrp{n}$ such that $|\mathcal{H}| = \ell$, we define following \cite{sarkar-vandenberg2019}, the \textit{commutativity map} with respect to $\mathcal{H}$, to be the function $\text{Com}_{\mathcal{H}} : \pauligrp{n} \rightarrow \mathbb{F}_2^{\ell}$,
\begin{equation}
\label{eq:comm-map}
    (\text{Com}_{\mathcal{H}} (\hat{q}))_i = \com{q}{p_i}, \text{ where } \hat{q} = [q],\; \hat{p_i} = [p_i], \;\;\; \forall\; 1 \leq i \leq \ell,
\end{equation}
and we say that $\hat{q} \in \pauligrp{n}$ generates the \textit{commutativity pattern} $\text{Com}_{\mathcal{H}} (\hat{q})$ with respect to $\mathcal{H}$.

Commutativity of Paulis can also be represented in yet another equivalent form. Since the phase factors of two Paulis $p,q \in \mathcal{P}_n$ do not influence their commutativity, it is often convenient to represent a Pauli $p = \eta \; (\sigma_0 \otimes \dots \otimes \sigma_{n-1})$ in symplectic notation \cite{gottesman1997phd} ignoring the phase factor $\eta$, that is as a binary (row) vector $v_p\in\mathbb{F}_2^{2n}$ such that for $0 \leq i \leq n-1$, $(v_p)_i=1$ if and only if $\sigma_i\in\{X,Y\}$, and for $n \leq i \leq 2n-1$, $(v_p)_{i}=1$ if and only if $\sigma_{i-n}\in\{Y,Z\}$. As binary vectors, Paulis $p$ and $q$ commute if and only if $v_p \Lambda v_q^\top=0$ where $\Lambda=\left(\begin{smallmatrix}0&I_n\\I_n&0\end{smallmatrix}\right)\in\mathbb{F}_2^{2n\times 2n}$ written using $n\times n$ blocks \cite{gottesman1997phd}, $I_n$ is the $n \times n$ identity matrix, and all operations are performed over $\mathbb{F}_2$.

A list $\mathcal{H} = \{p_0, \dots, p_{\ell-1}\} \subseteq \mathcal{P}_n$ is called \textit{anticommuting} (resp. \textit{commuting}) if and only if $\{p_i,p_j\} = 0$, for all $i \neq j$ (resp. $[p_i,p_j] = 0$ for all $i,j$). Recall from Definition~\ref{def:CAL} that a CAL on $n$ qubits is a list $\mathcal{C} = \{p_0, \dots, p_{\ell-1}\} \subseteq \mathcal{P}_n$, with the property that for distinct $i$ and $j$, $\{p_i, p_j\} = 0$ if and only if $j = (i \pm 1) \mod \ell$. Singleton lists and the empty list are considered to be anticommuting lists and CALs by convention. Similarly we also say that a list $\mathcal{H} \subseteq \pauligrp{n}$ is anticommuting (resp. commuting) if the resulting list of $\mathcal{P}_n$ obtained by replacing each element in $\mathcal{H}$ by exactly one of its representatives (which can be chosen arbitrarily) is anticommuting (resp. commuting). We define a CAL of $\pauligrp{n}$ analogously to be a list $\mathcal{C} \subseteq \pauligrp{n}$, such that the list obtained by replacing each element of $\mathcal{C}$ by a representative is a CAL in $\mathcal{P}_n$, and it is called extremal if it has length $\ell \geq 1$, and there is no CAL of same length on fewer qubits.

We will define the product of all elements in a list $\mathcal{H} \subseteq \mathcal{P}_n$ to be $\prod\mathcal{H} := \prod_{p \in \mathcal{H}} p$, (this product is ordered). Product of multisets of $\pauligrp{n}$ is defined in a similar fashion, except the ordering does not matter as all elements of $\pauligrp{n}$ commute under the group operation. By convention the product of the empty set is defined to be the identity element $I \in \mathcal{P}_n$ (resp. $[I] \in \pauligrp{n}$). Given a multiset $\mathcal{H} = \{p_0, \dots, p_{\ell-1}\} \subseteq \mathcal{P}_n$, the symplectic representation of Paulis allows us to define its \textit{dimension}. Writing each Pauli in the symplectic notation gives a matrix $H \in \mathbb{F}_2^{\ell \times 2n}$ whose each row is a Pauli in $\mathcal{H}$, and we define $\dim (\mathcal{H}) := \rank{H}$. In words, $\dim (\mathcal{H})$ is the size of the smallest subset of $\mathcal{H}$ that generates all its elements up to phase factors --- thus it satisfies $\dim (\mathcal{H}) = \min \{ |\mathcal{H}'| : \mathcal{H}' \subseteq \mathcal{H}, \langle \mathcal{H}', iI \rangle = \langle \mathcal{H}, iI \rangle\}$. We say $\mathcal{H}$ is \textit{independent} if and only if $\dim (\mathcal{H}) = |\mathcal{H}|$. Dimension and independence of multisets of $\pauligrp{n}$ are now analogously defined in terms of the dimension and independence of the multiset of $\mathcal{P}_n$, obtained by choosing representatives. Finally, if $\hat{p} \in \pauligrp{n}$, and $\hat{q} \in \pauligrp{m}$, the Kronecker product of $\hat{p}$ and $\hat{q}$, is defined as $\hat{p} \otimes \hat{q} = [p \otimes q] \in \pauligrp{n+m}$, where $p \in \hat{p}$, and $q \in \hat{q}$.

\subsection{Constructing a list of Paulis given desired commutation relations}
\label{subsec:pauli-list-arbitrary-pattern}

Suppose $C\in\mathbb{F}_2^{\ell \times \ell}$ is a matrix specifying the commutation relations for an unknown list of Paulis $\{p_0,\dots,p_{\ell-1}\} \subseteq \mathcal{P}_n$ (not necessarily all distinct), that is $C_{ij}=\com{p_i}{p_j}$. Hence, $C$ must be symmetric and must have zero diagonal (i.e.~$C_{ii}=0$ for all $i$). In this subsection, we answer the following questions: 
\begin{enumerate}[(i)]
    \item From $C$ alone, can we find Paulis $p_i$ that obey the specified commutation relations?
    \item What is the minimum number of qubits $n$ required to achieve $C$?
\end{enumerate}

The key idea to answering these questions is to simultaneously row and column reduce the matrix $C$. We start by analyzing a symmetric rank-revealing factorization algorithm, which is important in the sequel. We make use of the elementary row operations $E^{(h,k)}\in\mathbb{F}_2^{\ell \times \ell}$, defined by $E^{(h,k)}_{ij}=1$ if and only if $i=j$, or $(i,j) = (h,k)$. If $h \neq k$, when a matrix is multiplied by $E^{(h,k)}$ on the left, its $k^{\text{th}}$ row is added to its $h^{\text{th}}$ row, while when multiplied by $\left(E^{(h,k)}\right)^\top$ on the right, its $k^{\text{th}}$ column is added to its $h^{\text{th}}$ column. We will refer to the latter as elementary column operations. Note also that $E^{(h,k)}$ is self-inverse. If $h=k$, $E^{(h,h)}$ equals the identity matrix. Now consider Algorithm~\ref{alg:reduce}, where all operations occur over the field $\mathbb{F}_2$ (note that Algorithm~\ref{alg:reduce} may be of independent interest elsewhere).

\begin{algorithm}
\caption{}
\label{alg:reduce}
\begin{algorithmic}[1]
    \Procedure{LBLdecompose}{$C$} \algorithmiccomment{$C \in \mathbb{F}_2^{\ell \times \ell}$ is symmetric with zero diagonal}
        \State Set $B\leftarrow C$ and $L\leftarrow I$
        \State For $i=1$ to $\ell -1$:
        \State \quad Find the smallest index $j_i$ such that $B_{ij_i}=1$. If none exists, set $j_i\leftarrow i$.
        \State \quad While there exists $k>i$ such that $B_{kj_i}=1$:
        \State \quad\quad $B\leftarrow E^{(k,i)} B\left(E^{(k,i)}\right)^\top$ \algorithmiccomment{Set $B_{kj_i}=B_{j_i k}=0$}
        \State \quad\quad $L\leftarrow LE^{(k,i)}$
        \State Return $B$ and $L$
    \EndProcedure
\end{algorithmic}
\end{algorithm}

\begin{lemma}
\label{lem:LBL}
Let $C\in\mathbb{F}_2^{\ell \times \ell}$ with $\ell \geq 1$, such that $C=C^\top$ and $C_{ii}=0$ for all $i$, and consider Algorithm~\ref{alg:reduce}. Let $B^{(0)}=C$, and for $i \geq 1$, let $B^{(i)}$ denote the matrix $B$ after the end of iteration $i$ of the for-loop. For $i \geq 0$, let $B^{(i,k)}$ denote the upper-left $k \times k$ block of $B^{(i)}$, i.e. all rows and columns numbered $1$ to $k$. Then for all $1 \leq i \leq \ell - 1$
\begin{enumerate}[(a)]
    \item $B^{(i)}$ is symmetric, has zero diagonal, and $B^{(i,i+1)} = B^{(r,i+1)}$ for all $r \geq i$. Moreover the upper-left $i \times i$ block of $B^{(i-1)}$ is unchanged during for-loop iteration $i$.
    \item If $j_i \neq i$, then $B^{(r)}_{i s} = B^{(r)}_{s i} = 0$ and $B^{(r)}_{i j_i} = B^{(r)}_{j_i i} = 1$ for all $r \geq i-1$ and $s < j_i$.  If $j_i = i$, then  $B^{(r)}_{i s} = B^{(r)}_{s i} = 0$ for all $r \geq i-1$ and $1 \leq s \leq \ell$.
    \item If $j_i \neq i$, then $B^{(r)}_{s j_i} = B^{(r)}_{j_i s} = 0$, for all $r \geq i$ and $s > i$.
    \item If $B^{(i,i)}_{rs} = 1$, then it is the only non-zero element of the $r^{\text{th}}$ and $s^{\text{th}}$ rows and columns of $B^{(i)}$.
    \item Each row and column of $B^{(i,i+1)}$ contains at most a single $1$.
    \item $j_{j_i} = i$.
\end{enumerate}
The output of the algorithm $L, B \in\mathbb{F}_2^{\ell \times \ell}$ satisfy $C = LBL^\top$, where $L$ is an invertible lower-triangular matrix, $B$ is symmetric and has zero diagonal, and each row and column of $B$ contains at most a single $1$. The algorithm runs in $O(\ell^3)$ time.
\end{lemma}

\begin{proof}
Note that by construction, throughout the algorithm we have that $L$ is lower triangular and invertible (since the $E^{(k,i)}$ matrix inside the while loop is lower triangular and invertible), and also that $C= L B L^\top$. The claims about the output $B$ of the algorithm follow as $B = B^{(\ell - 1)}$ and from (a) and (e). Also it is clear that the worst case run time of the algorithm is $O(\ell^3)$, where we must use the sparsity of $E^{(k,i)}$ while performing the matrix multiplication inside the while loop (otherwise we will get a $O(\ell^5)$ algorithm). Note also that the symmetry of $B^{(i-1)}$ dictates that in the case $j_i=i$, the while loop will exit immediately. 

We now prove (a)-(f). 
\begin{enumerate}[(a)]
    \item Let $\tilde L$ denote the matrix $L$ after the end of for-loop iteration $i$ of Algorithm~\ref{alg:reduce}. Then $C = \tilde L B^{(i)} \tilde L^\top$ and so $B^{(i)}$ is symmetric, for all $1 \leq i \leq \ell - 1$. $B^{(i)}$ has zero diagonal because for all $1 \leq r \leq \ell$
    \begin{equation}
    B^{(i)}_{rr}=\sum_{xy} \tilde{L}^{-1}_{rx} C_{xy} \tilde{L}^{-1}_{ry} = \sum_{x} \tilde{L}^{-1}_{rx} C_{xx} + \sum_{x<y} \tilde{L}^{-1}_{rx} \tilde{L}^{-1}_{ry} (C_{xy} + C_{yx}) = 0.
    \end{equation}
    Now note that during any for-loop iteration $i$, in each iteration of the while loop, multiplying $E^{(k,i)}$ on the left does not change any row of $B$ with index $1$ to $i$, while multiplying by $\left(E^{(k,i)}\right)^\top$ on the right does not change any column of $E^{(k,i)} B$ with index $1$ to $i$. It immediately follows that the upper-left $i \times i$ block of $B^{(i-1)}$ is unchanged during for-loop iteration $i$. We also easily deduce from this that $B^{(i,i+1)} = B^{(r,i+1)}$ for all $r \geq i$.
    \item First notice that from the beginning of for-loop iteration $i$, all rows numbered $1$ to $i$ are unaffected by elementary row operations alone. If $j_i = i$, it means that at the start of for-loop iteration $i$, the $i^{\text{th}}$ row of $B^{(i-1)}$ is zero. Any elementary column operations will not change this row subsequently. If instead $j_i \neq i$, then at the start of for-loop iteration $i$ we have $B^{(i-1)}_{i s} = 0$ for all $s < j_i$, and $B^{(i-1)}_{i j_i} = 1$. In this case, any subsequent elementary column operations will not change columns $1$ to $j_i$ of the $i^{\text{th}}$ row. The conclusion now follows by symmetry using (a).
\end{enumerate}

For parts (c)-(f), we will first prove a claim. We claim that for all $1 \leq i \leq \ell - 1$, Algorithm~\ref{alg:reduce} maintains the following invariants after for-loop iteration $i$:
\begin{enumerate}[(i)]
    \item For all $1 \leq k \leq i$, $B^{(i)}_{s j_k} = B^{(i)}_{j_k s} = 0$, for all $s > k$.
    \item For every $1 \leq r,s \leq i$ such that $B^{(i)}_{rs} = 1$, the only non-zero element of the $r^{\text{th}}$ and $s^{\text{th}}$ rows and columns of $B^{(i)}$ are $B^{(i)}_{rs} = B^{(i)}_{sr} = 1$. If in addition $r = i$, then the $i^{\text{th}}$ row and column of $B^{(i-1)}$ do not change during the iteration, while if $1 \leq r,s < i$, then the $r^{\text{th}}$ row and the $s^{\text{th}}$ column of $B^{(i-1)}$ do not change during iteration $i$.
    \item If $j_i \leq i$, then $j_{j_i} = i$.
\end{enumerate}

\vspace{-0.2cm}
\textit{Proof of claim.} It is helpful to make an observation about what happens during for-loop iteration $i$, assuming $j_i \neq i$. One then easily checks that in an iteration of the while loop, the effect of the update $B \leftarrow E^{(k,i)} B \left(E^{(k,i)}\right)^\top$ on the $j_i^{\text{th}}$ column of $B$ is to set the $(k,j_i)$ entry to zero without affecting any other entries, for some $k > i$. By symmetry the only effect on the $j_i^{\text{th}}$ row of $B$ is to set the $(j_i, k)$ entry to zero. Thus at the end of for-loop iteration $i$, the $j_i^{\text{th}}$ column of $B^{(i)}$ satisfies $B^{(i)}_{s j_i} = 0$ for all $s > i$, and $B^{(i)}_{i j_i} = 1$.

Returning to the proof of the claim, notice that for $i=1$ (i)-(iii) are clearly true: if $j_1 \neq 1$ then (i) holds by the observation in the previous paragraph, and if $j_1 = 1$ then it is true by (b), (ii) is vacuous as $B^{(1,1)}_{11} = 0$ by (a), and for (iii) note that $j_1 \leq 1$ implies $j_1 = 1$, thus $j_{j_1} = j_1 = 1$. Now for the inductive hypothesis, assume that the invariants are true after for-loop iteration $i-1$, for some $1 < i \leq \ell - 1$, and we want to show that they are also true after iteration $i$. We do this separately in three steps.

\vspace{0.2cm}
\textit{Step 1:} Let us first show that (iii) holds. If $j_{i} = i$ then it clearly holds, so assume $j_{i} \leq i-1$. Now at the start of for-loop iteration $i$ we have by symmetry $B^{(i-1)}_{j_i i} = 1$. Assume for contradiction that $j_{j_i} \neq i$, and consider the situation during for-loop iteration $j_i$. If $j_{j_i} > i$, then by (b) we would get $B^{(i-1)}_{j_i i} = 0$, so it must be the case that $j_{j_i} < i$. But then we get $B^{(i-1)}_{j_i j_{j_i}} = 1$, again by (b). Thus the $j_i^{\text{th}}$ row of $B^{(i-1)}$ contains two non-zero elements, which contradicts (ii) of the inductive hypothesis.

\vspace{0.2cm}
\textit{Step 2:} We now show that (ii) holds after for-loop iteration $i$, first in the limited setting $1 \leq r,s < i$. By symmetry, it suffices to prove that whenever $B^{(i)}_{rs} = 1$, it is the only non-zero element of the $r^{\text{th}}$ row of $B^{(i-1)}$ which does not change during the iteration, since one can repeat the same argument for the element $B^{(i)}_{sr} = 1$ in the $s^{\text{th}}$ row. 
Assuming such an element $B^{(i)}_{rs} = 1$ exists, we have from (a) that $B^{(i-1)}_{rs} = 1$, and by (ii) of the inductive hypothesis it is the only non-zero element of the $r^{\text{th}}$ row of $B^{(i-1)}$. Now for every while loop iteration within for-loop iteration $i$, the $r^{\text{th}}$ row of $B^{(i-1)}$ is not modified by elementary row operations alone, by the same argument as in the proof of (b), or by elementary column operations as $B^{(i-1)}_{ri} = 0$ (this continues to be true by (a) during the iteration). Thus the $r^{\text{th}}$ row of $B^{(i-1)}$ is unchanged proving this case. 

It remains to prove that (ii) also holds for any non-zero element in the $i^{\text{th}}$ row or column of $B^{(i,i)}$. For this consider the $i^{\text{th}}$ row of $B^{(i)}$. If $j_i \geq i$, then we are again done using (b) as the $i^{\text{th}}$ row and column of $B^{(i,i)}$ is zero, so suppose $j_i < i$. We will show that in this case the only non-zero element of the $i^{\text{th}}$ row and $j_i^{\text{th}}$ column of $B^{(i)}$ is $B^{(i)}_{i j_i} = 1$, and the $i^{\text{th}}$ row of $B^{(i-1)}$ is unchanged during the for-loop, and then we are done by symmetry.
We now have a series of consequences. By invariant (iii), which we have already proved above to hold after for-loop iteration $i$, we have $j_{j_i} = i$. Applying (i) of the inductive hypothesis then gives $B^{(i-1)}_{s i} = 0$ for all $s > j_i$, and using symmetry of $B^{(i-1)}$ we conclude that the only non-zero element of the $i^{\text{th}}$ row of $B^{(i-1)}$ is $B^{(i-1)}_{i j_i} = 1$. Next using exactly the same argument as in the last paragraph (for the case $1 \leq r,s < i$), we conclude that the $i^{\text{th}}$ row of $B^{(i-1)}$ is unchanged during for-loop iteration $i$. Now consider the $j_i^{\text{th}}$ column of $B^{(i-1)}$. By (ii) of the inductive hypothesis and (a) we first have $B^{(i-1)}_{s j_i} = B^{(i)}_{s j_i} = 0$ for all $s < i$, and by the observation stated above we also have $B^{(i)}_{s j_i} = 0$ for all $s > i$. This proves that the only non-zero element in the $j_i^{\text{th}}$ column of $B^{(i)}$ is $B^{(i)}_{i j_i} = 1$.

\vspace{0.2cm}
\textit{Step 3:} Finally we show that (i) holds. If $k = i$, then the case $j_i \neq i$ is already true by the observation made above, while the case $j_i = i$ is true by (b). So fix some $1 \leq k < i$. If $j_k \leq i$ then by (b) we have $B^{(i)}_{k, j_k} = 1$, so we are done because invariant (ii) holds after iteration $i$. If $j_i = i$, then $B^{(i-1)} =  B^{(i)}$ and we are done using (i) of the inductive hypothesis. So the interesting case is when $j_i \neq i$ and $j_k > i$. Now there are two subcases $j_i > k$ and $j_i < k$ which need to be handled separately (note that the case $j_i = k$ is ruled out by (b)). First suppose $j_i > k$, and consider the $j_k^{\text{th}}$ column of $B^{(i-1)}$. Since by (i) of the inductive hypothesis we have $B^{(i-1)}_{i j_k} = B^{(i-1)}_{j_i j_k}= 0$, elementary row operations in the while loops during for-loop iteration $i$ do not change this column (this uses $B^{(i-1)}_{i j_k}= 0$), while elementary column operations never act on the column as $B^{(i-1)}_{j_k j_i}= 0$ (this uses symmetry and $B^{(i-1)}_{j_i j_k}= 0$). Thus the $j_k^{\text{th}}$ columns are the same for both $B^{(i-1)}$ and $B^{(i)}$, and this case is proved. 
Finally consider the case $j_i < k$. Notice that since $j_i < k < i$ and we have proved that (iii) holds after for-loop iteration $i$, we have $j_{j_i} = i$, and then by (i) of the induction hypothesis we get $B^{(i-1)}_{si} = 0$ for all $s > j_i$. Moreover since we also proved that (ii) holds, we know that the $i^{\text{th}}$ column of $B^{(i-1)}$ does not change during the for-loop. Now again consider the $j_k^{\text{th}}$ column of $B^{(i-1)}$, focusing only on the rows numbered $k$ to $\ell$. As before, this part of the column is unchanged by elementary row operations during for-loop iteration $i$, so we only need to argue that elementary column operations also do not affect it. But this immediately follows from the pattern of zeros in the $i^{\text{th}}$ column of $B^{(i-1)}$, which finishes the proof of the claim.

\vspace{0.2cm}
We now finish up the proof of the lemma.
\setmargins
\begin{enumerate}[(a)]
\setcounter{enumi}{2}
    \item This follows directly from (i) of the above claim.
    \item This is a sub-part of invariant (ii) of the claim.
    \item Suppose this is not true, and let $k=i+1$. Then by (ii) of the claim and symmetry of $B^{(i)}$, we must have at least two non-zero elements in the $k^{\text{th}}$ row of $B^{(i,k)}$, one of which must be $B^{(i,k)}_{k j_k}$, and let the other non-zero element be $B^{(i,k)}_{k h}$, and note that $j_k < h < k$ (by definition of $j_k$). But by invariant (iii) applied to the $k^{\text{th}}$ for-loop, we also have $j_{j_k} = k$. Since $j_k \leq i$, we can now apply (i) of the claim, but this gives a contradiction as $B^{(i,k)}_{hk} = 1$ by symmetry.
    \item If $j_i \leq i$, then this is just the statement of (iii) of the claim. So let $j_i > i$, and suppose $j_{j_i} \neq i$. Then after for-loop iteration $j_i$, we have $B^{(j_i)}_{ j_{j_i} j_i} = 1 = B^{(j_i)}_{i j_i}$, by (b) and (i) of the claim. But this contradicts invariant (ii).
\end{enumerate}
\end{proof}

From the properties of $B$ in the above lemma, we deduce the following additional simple lemma.

\begin{lemma}
\label{lem:Bstruct}
Let $B\in\mathbb{F}_2^{\ell \times \ell}$ be a symmetric matrix with zero diagonal and at most a single 1 in each row and column. Let $|B|=\sum_{ij}B_{ij}$ be the total number of ones. Then $|B|$ is even and $\rank{B}=|B|$. In particular $\rank{C} = |B|$ is even.
\end{lemma}

\begin{proof}
It is clear that $|B|$ is even due to symmetry of $B$ and it having zero diagonal. Since each row has at most a single $1$, $B$ has exactly $|B|$ non-zero rows, all of which are independent since each column of $B$ has at most a single $1$. This shows $\rank{B} = |B|$. Finally, since $C=LBL^\top$ and $L$ is invertible, $\rank{C}=\rank{B}$.
\end{proof}

Let us apply this linear algebra to our study of Paulis. Let $P\in\mathbb{F}_2^{\ell\times2n}$ be the matrix representing the list of Paulis $\{p_0,p_1,\dots,p_{\ell-1}\}$ (which we are looking for) in symplectic notation. Therefore, $\{p_0,p_1,\dots,p_{\ell-1}\}$ satisfies the commutation relations specified by $C$ if and only if $P\Lambda P^\top=C$. If $C = 0$, it suffices to just use $1$ qubit and we can take $p_0 = \dots = p_{\ell-1} = I_2$. So let us assume that $C \neq 0$. Let us suppose that we decompose $C=LBL^\top$ as in Lemma~\ref{lem:LBL}. As we now explain, we can easily create a list of Paulis $\{q_0,q_1,\dots,q_{\ell-1}\}$ with commutation relations given by $B$, or equivalently a matrix $Q \in \mathbb{F}_2^{\ell \times 2n}$ such that $Q\Lambda Q^\top=B$. For each pair $(i,j)$ such that $i<j$ and $B_{ij}=1$, we introduce a qubit (call it qubit $(i,j)$), and let $q_i=X_{(i,j)}$ and $q_j=Z_{(i,j)}$ act on that qubit. For all zero rows $i$ of $B$, we set $q_i=I_2$. We have thus used $n=\rank{B}/2=\rank{C}/2$ qubits to make the Paulis $\{q_0,q_1,\dots,q_{\ell-1}\}$, and one easily verifies that this list of Paulis satisfy the commutation relations specified by $B$. Now $Q\Lambda Q^\top=B$ implies $(LQ)\Lambda(LQ)^\top=LBL^\top=C$, so $P=LQ$ represents a list of Paulis satisfying the commutation relations specified by $C$. It is also easy to check that $\rank{P}=\rank{LQ}=\rank{Q} = 2n = \rank{C}$. This discussion leads to the following theorem.

\begin{theorem}
\label{thm:plist_construct}
Let $C\in\mathbb{F}_2^{\ell \times \ell}$ be symmetric with zero diagonal. A list of $n$-qubit Paulis $\mathcal{C}=\{p_0,p_1,\dots,p_{\ell-}\} \subseteq \mathcal{P}_n$ such that $C_{ij}=\com{p_i}{p_j}$ for all $0 \leq i,j \leq \ell-1$, exists if and only if $n \geq \max (1, \rank{C}/2)$. Moreover, when $\mathcal{C}$ exists, $\dim(\mathcal{C}) \geq \rank{C}$, with equality if $n = \rank{C}/2$.
\end{theorem}

\begin{proof}
If $C = 0$, then as argued above $n=1$ is the minimum number of needed qubits to achieve the commutation relations specified by $C$; thus the first statement is true. So assume $C \neq 0$, which implies $\rank{C} \geq 2$ by Lemma~\ref{lem:Bstruct}. The discussion in the previous paragraph shows the construction of $\mathcal{C}$ when $n = \rank{C}/2$, and this clearly suffices as if $n > \rank{C}/2$, then one can simply act by $I_2$ on the extra qubits. We argue the converse, i.e. for the minimality of $n$, by contradiction. Let $P\in\mathbb{F}_2^{\ell \times 2n}$ be the matrix such that $P\Lambda P^\top=C$ with $n<\rank{C}/2$. Let $C=LBL^\top$ from Lemma~\ref{lem:LBL}. Then $L^{-1}P$ represents a list of Paulis with commutation relations given by $B$. But by the structure of $B$, this implies $\rank{B}/2=\rank{C}/2$ anticommuting pairs of Paulis (i.e. the Paulis can be grouped in pairs that anticommute with each other but commute with all other Paulis in the list), which is known to be impossible on fewer than $\rank{C}/2$ qubits.\footnote{Notice that, ignoring phases, the size of the group generated by $n$ anticommuting pairs of Paulis is $4^n$, since every paired Pauli is necessarily independent. Therefore, the Paulis must act on at least $n$ qubits.}

Now suppose $\mathcal{C} \subseteq \mathcal{P}_{n}$ exists satisfying the commutation relations specified by $C$, where we already know $n \geq \max (1, \rank{C}/2)$. Representing the Paulis in $\mathcal{C}$ by a matrix $P \in \mathbb{F}_2^{\ell \times 2n}$ again gives $P\Lambda P^\top=C$. Thus $\rank{C} \leq \rank{P} = \dim(\mathcal{C})$. If in addition $2n = \rank{C}$, then $\rank{P} \leq \min (\ell, \rank{C}) \leq \rank{C}$, which implies $\dim{\mathcal{C}} = \rank{C}$.
\end{proof}

We demonstrate the use of Theorem~\ref{thm:plist_construct} with a couple of examples, the first of which is well known, and the second one concerns CALs, which will be discussed much more thoroughly in the next subsection.

\begin{example}
\label{ex:anticomm}
Suppose
\begin{equation}
C=\left(\begin{array}{ccccc}
0&1&1&&1\\
1&0&1&\dots&1\\
1&1&0&&1\\
&\vdots&&\ddots&\vdots\\
1&1&1&\dots&0
\end{array}\right)\in\mathbb{F}_2^{\ell \times \ell}, \;\; \ell \geq 2,
\end{equation}
so that we are looking for a list of anticommuting Paulis. It is well known that
\begin{equation}
\rank{C}/2=\bigg\{\begin{array}{ll}(\ell-1)/2,& \ell\text{ odd}\\ \ell/2,& \ell\text{ even}\end{array}.
\end{equation}
Then, the construction from Theorem~\ref{thm:plist_construct} with $n = \rank{C}/2$, returns the Jordan-Wigner \cite{wigner1928paulische} list --- i.e.~for $h\in\{0,1,\dots,2\lfloor \ell/2\rfloor-1\}$,
\begin{equation}
p_h=\bigg\{\begin{array}{ll}
X_{h/2}\prod_{i=0}^{h/2-1}Y_i,&h\text{ even}\\
Z_{(h-1)/2}\prod_{i=0}^{(h-3)/2}Y_i,&h\text{ odd}
\end{array},
\end{equation}
and, if $\ell$ is odd, the final Pauli $p_{\ell-1}=\prod_{i=0}^{n-1}Y_i$.
\end{example}

\begin{example}
\label{ex:cal}
Suppose
\begin{equation}
C=\left(\begin{array}{ccccc}
0&1&0&&1\\
1&0&1&\dots&0\\
0&1&0&&0\\
&\vdots&&\ddots&\vdots\\
1&0&0&\dots&0
\end{array}\right)\in\mathbb{F}_2^{\ell \times \ell}, \;\; \ell \geq 3,
\end{equation}
so that we are looking for a cyclically anticommuting list of Paulis. We note that the rank of $C$ satisfies
\begin{equation}
\rank{C}/2=\bigg\{\begin{array}{ll}(\ell-1)/2,& \ell\text{ odd}\\(\ell-2)/2,& \ell \text{ even}\end{array}.
\end{equation}
The construction from Theorem~\ref{thm:plist_construct}, again with $n = \rank{C}/2$ qubits, now gives the list
\begin{equation}
    \begin{split}
        \left\{X_0,\ Z_0,\ X_0  X_1,\ Z_1,\ \dots,\ X_{n-2}X_{n-1},\ Z_{n-1},\ Y_{n-1}\prod_{i=0}^{n-2}Z_i\right\},& \quad \ell \text{ odd},\\
        \left\{X_0,\ Z_0,\ X_0 X_1,\ Z_1,\ \dots,\ X_{n-2}X_{n-1},\ Z_{n-1},\ X_{n-1},\ \prod_{i=0}^{n-1} Z_i\right\},&\quad \ell \text{ even}.
    \end{split}
\end{equation}
\end{example}

By Theorem~\ref{thm:plist_construct}, $\dim(\mathcal{C})=2n$ in the preceding examples, so it ensures that $\langle\mathcal{C}, iI\rangle=\mathcal{P}_n$. However, we can even generalize the construction to include $\dim(\mathcal{C})$ as an input parameter.

\begin{theorem}
Suppose $C\in\mathbb{F}_2^{\ell \times \ell}$ is symmetric and has zero diagonal, and let $k\in\mathbb{Z}$ such that $\rank{C}\le k\le \ell$. A list of Paulis $\mathcal{C}=\{p_0,p_1,\dots,p_{\ell-1}\} \subseteq \mathcal{P}_n$ satisfying (a) $C_{ij}=\com{p_i}{p_j}$ for all $0 \leq i,j \leq \ell - 1$, and (b) $\dim(\mathcal{C})=k$ exists if and only if $n \geq \max (1, k-\rank{C}/2)$.
\end{theorem}

\begin{proof}
First consider the case $C = 0$, then as argued in Theorem~\ref{thm:plist_construct} we need $n \geq 1$. But it is well known that on $n$ qubits, the maximum size of an independent commuting list of Paulis is $n$, so $k \leq n$, and this also proves the converse for this case. For the rest of the proof assume $C \neq 0$, which implies $\rank{C} \geq 2$ by Lemma~\ref{lem:Bstruct}. Let $[\ell] = \{0,1,\dots, \ell-1\}$. We begin the same way as Theorem~\ref{thm:plist_construct}, decomposing $C=LBL^\top$ and creating a set of Paulis $q_i=X_{(i,j)}$, $q_j=Z_{(i,j)}$ for all pairs $(i,j)$ such that $i<j$ and $B_{ij}=1$. Let $\mathcal{I}\subseteq[\ell]$ be the indices of the rows in $B$ that are all zero, and partition $\mathcal{I}$ into two sets $\mathcal{I}_1$ and $\mathcal{I}_2$ such that $|\mathcal{I}_1|=k-\rank{C}$. For each element $i\in\mathcal{I}_1$ we introduce a qubit (labeled by $i$) and set $q_i=Z_i$. We now have $n=k-\rank{C}/2$ qubits. We form $Q\in\mathbb{F}_2^{\ell \times 2n}$ representing the Pauli list $\{q_0,q_1,\dots,q_{\ell-1}\}$ as symplectic vectors and calculate $P=LQ$ representing the desired list $\mathcal{C}=\{p_0,p_1,\dots,p_{\ell-1}\}$. Note $Q\Lambda Q^\top=B$ by construction and so $P\Lambda P^\top=LBL^\top=C$. Also note that since $\dim(\{q_0,q_1,\dots,q_{\ell-1}\}) = \rank{Q} = k$ by construction and $L$ is invertible, $\dim(\mathcal{C})=\rank{P} = k$. This shows the existence of $\mathcal{C}$ when $n = k - \rank{C}/2$, and also covers the case $n \geq k - \rank{C}/2$, since one can act on the extra qubits by identity.

It remains to prove the converse. This argument by contradiction is analogous to that in Theorem~\ref{thm:plist_construct}. If $P\in\mathbb{F}_2^{\ell \times 2n'}$ represents the list of Paulis on $n'<n$ qubits satisfying $P\Lambda P^\top=C$ and $\rank{P}=k \geq \rank{C}$, then $Q=L^{-1}P$ represents the list of $n'$-qubit Paulis such that $Q\Lambda Q^\top=B$ and $\rank{Q}=k$. However, the structure of $B$ implies that there is an independent and fully commuting subset of $\{q_0,q_1,\dots,q_{\ell-1}\}$ with size $k-\rank{C}/2$ which is again known to be impossible on $n'< k-\rank{C}/2$ qubits \cite{sarkar-vandenberg2019}.
\end{proof}

\subsection{Properties of cyclically anticommuting lists of Paulis}
\label{subsec:cal}

We will work exclusively in this subsection with the group $\pauligrp{n}$, to not have to keep track of the phase factors of the Paulis. CALs of $\mathcal{P}_n$ were introduced in Definition~\ref{def:CAL}, and CALs of $\pauligrp{n}$ were defined in Section~\ref{subsec:pauli-definitions}. Here we present some more properties of CALs of $\pauligrp{n}$, that we know of currently. Translating these results to CALs of $\mathcal{P}_n$ is obvious. We recall the commutation relations that a CAL satisfies: if $\mathcal{C} = \{\hat{p}_0,\hat{p}_1,\dots,\hat{p}_{k-1}\} \subseteq \pauligrp{n}$ is a CAL, then for distinct $i$ and $j$
\begin{equation}
\label{eq:cal-comm}
\begin{cases}
\{\hat{p}_i, \hat{p}_j\} = 0, & \text{if}\;\; i=j \pm 1 \mod k \\
[\hat{p}_i, \hat{p}_j] = 0, & \text{otherwise}.
\end{cases}
\end{equation}

This subsection is organized as follows: Lemma~\ref{lem:cal-repeating-elements} - Lemma~\ref{lem:cal-maxsize} points out the relationship between CALs and anticommuting lists, and shows that all elements of a CAL must be unique if the CAL has length greater than four, Theorem~\ref{thm:cal-dim-odd} - Corollary~\ref{cor:extremal-cal-mult-identity} addresses the question of dimension and independence of a CAL, and finally Lemma~\ref{lem:cal-expand-contract} - Theorem~\ref{thm:cal-noY-even} looks at whether it is possible to create a CAL, where the representative of each element of the CAL is a Kronecker product of only $X$ and $Z$ Pauli matrices, and possibly the $2 \times 2$ identity matrix $I_2$.

Recall from the proof of Theorem~\ref{thm:CAL_construct}, that on $n=1$ qubits, $\mathcal{C} = \{[X], [Z], [X], [Z]\}$ is a CAL. Clearly this has repeating elements. A natural question is whether it is possible to have CALs of longer length with repeating elements. The next lemma completely addresses this question.
\begin{lemma}
\label{lem:cal-repeating-elements}
Let $\mathcal{C} \subseteq \pauligrp{n}$ be a CAL, for $n \geq 1$. If $\mathcal{C}$ has repeating elements, then $\mathcal{C}$ has one of the following forms
\begin{equation}
\label{eq:cal-repeating-elements}
    \mathcal{C} =
    \begin{cases}
        \{\hat{p},\hat{q},\hat{p},\hat{q}\} & \;\;\;\; \text{ with } \{\hat{p},\hat{q}\} = 0, \\
        \{\hat{p},\hat{q},\hat{p},\hat{s}\} & \;\;\;\; \text{ with } \{\hat{p},\hat{q}\} = \{\hat{p},\hat{s}\} = [\hat{q},\hat{s}] = 0, \\
        \{\hat{p},\hat{q},\hat{s},\hat{q}\} & \;\;\;\; \text{ with } \{\hat{p},\hat{q}\} = \{\hat{q},\hat{s}\} = [\hat{p},\hat{s}] = 0,
    \end{cases}
\end{equation}
where $\hat{p},\hat{q},\hat{s} \in \pauligrp{n}$, and distinct. The last two cases can only happen if $n > 1$.
\end{lemma}

\begin{proof}
First notice that if $\mathcal{C}$ has any of the forms in Eq.~\eqref{eq:cal-repeating-elements}, then the commutation relations imply that $\mathcal{C}$ is indeed a CAL. Now we claim that if $\mathcal{C}$ has at least one repeating element, then $|\mathcal{C}| = 4$. We first prove this claim. If $|\mathcal{C}| \leq 1$, then there can obviously be no repeating element. If $2 \leq |\mathcal{C}| \leq 3$, then $\mathcal{C}$ is an anticommuting list, and so again there can be no repeating element. Now suppose that $|\mathcal{C}| \geq 5$, and $\mathcal{C} = \{\hat{p}_0, \dots, \hat{p}_{k-1}\}$ with repeating elements $\hat{p}_{\ell} = \hat{p}_{\ell'}$, for $0 \leq \ell < \ell' \leq k$. Notice first that $\ell' \notin \{\ell + 1, \; (\ell - 1) \mod k\}$, because that would imply that $\hat{p}_{\ell}$ anticommutes with itself. Then there must exist an element $\hat{p}_m \in \mathcal{C}$, $m \in \{\ell + 1, \; (\ell - 1) \mod k\}$, such that $\{\hat{p}_{\ell}, \hat{p}_m\} = 0$, and $[\hat{p}_{\ell'}, \hat{p}_m] = 0$. The existence of $\hat{p}_m$ follows from the fact that $\mathcal{C}$ is a CAL, and $|\mathcal{C}| \geq 5$. But this is a contradiction because it implies that $\hat{p}_m$ commutes and anticommutes with the same element, and therefore $|\mathcal{C}| < 5$. This proves the claim.

Using the claim, we can suppose that $\mathcal{C} = \{\hat{p},\hat{q},\hat{r},\hat{s}\}$, for $\hat{p},\hat{q},\hat{r},\hat{s} \in \pauligrp{n}$, where all of them may not be distinct, and satisfies the CAL commutation relations in Eq.~\eqref{eq:cal-comm}. We already know that $\hat{p} \neq \hat{q}$, $\hat{q} \neq \hat{r}$, $\hat{r} \neq \hat{s}$, and $\hat{s} \neq \hat{p}$, because otherwise Eq.~\eqref{eq:cal-comm} is violated. Now there are two cases: (i) $\hat{p} = \hat{r}$, and (ii) $\hat{p} \neq \hat{r}$. Assuming the first case $\hat{p} = \hat{r}$, we have two subcases: either $\hat{q} = \hat{s}$ which gives that $\mathcal{C} = \{\hat{p},\hat{q},\hat{p},\hat{q}\}$, or $\hat{q} \neq \hat{s}$ which gives $\mathcal{C} = \{\hat{p},\hat{q},\hat{p},\hat{s}\}$, and both are thus of the forms in Eq.~\eqref{eq:cal-repeating-elements}. Next assume the second case $\hat{p} \neq \hat{r}$, and thus all $\hat{p},\hat{q}$, and $\hat{r}$ are distinct. Since there is at least one repeating element in $\mathcal{C}$, this means that $\hat{s}$ equals either $\hat{p},\hat{q},$ or $\hat{r}$. But we already argued above that $\hat{p} \neq \hat{s}$ and $\hat{r} \neq \hat{s}$, and so in fact $\hat{q} = \hat{s}$. This gives that $\mathcal{C} = \{\hat{p},\hat{q},\hat{r},\hat{q}\}$, which is the third form in Eq.~\eqref{eq:cal-repeating-elements}. The proof is finished by noting that when $n = 1$, only the first case in Eq.~\eqref{eq:cal-repeating-elements} can happen.
\end{proof}

The next lemma concerns the maximum possible length of a CAL on $n$ qubits. The arguments in the proof of parts (b) and (e) of the lemma already appear in the proof of Theorem~\ref{thm:CAL_construct}, when the length of the CAL is at least three. We record the lemma here for completeness, as it covers the case when the length is less than three. Also part (c) shows how to get a CAL from an anticommuting list of $\pauligrp{n}$.

\begin{lemma}
\label{lem:cal-maxsize}
Let $\mathcal{C} = \{\hat{p}_0, \dots, \hat{p}_{\ell-1}\} \subseteq \pauligrp{n}$ be a CAL, and $|\mathcal{C}| \geq 1$.

\setmargins
\begin{enumerate}[(a)]
\item The list $\mathcal{C}' = \{\hat{p}'_0, \dots, \hat{p}'_{\ell-1}\}$ defined as $\hat{p}'_i = \hat{p}_{j}$ with $j = (i+r) \mod \ell$, is also a CAL for all translations $r \in \mathbb{Z}$.
\item If $|\mathcal{C}| = 1$, let $\mathcal{A} = \emptyset$, otherwise let $\mathcal{A} = \{\hat{q}_0, \dots, \hat{q}_{\ell-2}\}$ defined as $\hat{q}_i = \prod_{j=0}^{i} \hat{p}_j$, for all $0 \leq i \leq \ell-2$. Then $\mathcal{A}$ is an anticommuting list. Thus an anticommuting list exists of size $|\mathcal{C}| - 1$.
\item Let $\mathcal{A} = \{\hat{q}_0, \dots, \hat{q}_{k-1}\} \subseteq \pauligrp{n}$ be an anticommuting list, and $|\mathcal{A}| \geq 2$. If we define $\mathcal{D} = \{\hat{d}_0, \dots, \hat{d}_{k}\}$ by $\hat{d}_0 = \hat{q}_0$, $\hat{d}_{k} = \hat{q}_{k-1}$, and $\hat{d}_{i+1} = \hat{q}_i \hat{q}_{i+1}$ for all $0 \leq i \leq k-2$, then $\mathcal{D}$ is a CAL.
\item $|\mathcal{C}| \leq 2n+2$. Thus there exists a CAL on $n$ qubits of every length $0 \leq \ell' \leq 2n+2$.
\item The minimum number $n$ of qubits needed to create a CAL of length $\ell'$ is 
\begin{equation}
\label{eq:cal-min-qubits}
    n = 
    \begin{cases}
    1 & \; \text{ if } \ell' \leq 2, \\
    (\ell'-1)/2 & \; \text{ if } \ell' > 2 \text{ and odd}, \\
    (\ell'-2)/2 & \; \text{ if } \ell' > 2 \text{ and even}. 
    \end{cases}
\end{equation}
\item $\mathcal{C} = \{[I]\}$ if and only if $[I] \in \mathcal{C}$.
\end{enumerate}
\end{lemma}

\vspace{0.5cm}
\begin{proof}
\setmargins
\begin{enumerate}[(a)]
\item This is easily checked and follows from the commutation relations of the elements of $\mathcal{C}$.
\item If $|\mathcal{C}| = 1$, $\mathcal{A} = \emptyset$ is anticommuting by convention. Otherwise take any two distinct elements $\hat{q}_i, \hat{q}_j \in \mathcal{A}$, and without loss of generality assume that $i < j$, and thus $\hat{q}_j = \hat{q}_i \prod_{k = i+1}^{j} \hat{p}_{k}$. From the commutation relations of CAL, we have that $\{\hat{p}_i, \hat{p}_{i+1}\} = 0$, $[\hat{p}_i, \hat{p}_{k'}] = 0$ for all $i+1 < k' \leq j$, and $[\hat{p}_k, \hat{p}_{k'}] = 0$ for all $0 \leq k < i$, and $i+1 \leq k' \leq j$. Then $\hat{q}_i \hat{q}_j = \hat{q}_i (\hat{q}_i \prod_{k = i+1}^{j} \hat{p}_{k}) = -(\hat{q}_i \prod_{k = i+1}^{j} \hat{p}_{k}) \hat{q}_i = -\hat{q}_j \hat{q}_i$. Thus $\mathcal{A}$ is an anticommuting list of size $|\mathcal{C}| - 1$.
\item As in (b), one easily checks that the elements of $\mathcal{D}$ satisfy the CAL commutation relations of Eq.~\ref{eq:cal-comm}.
\item For the sake of contradiction, assume that $|\mathcal{C}| \geq 2n+3$. Then by (b) there exists an anticommuting list $\mathcal{A}$ of size at least $2n+2$. But it is a well-known fact that the largest size of an anticommuting list of $\pauligrp{n}$ is $2n+1$ (see for example \cite{sarkar-vandenberg2019} and Proposition 9 in \cite{hrubes2016families}), which is a contradiction. For the last part, if $\ell' \leq 2$, then choose any anticommuting list of size $\ell'$. If $\ell' > 2$, choose an anticommuting list of size $\ell' - 1$, and then apply (c).
\item The case $\ell' \geq 3$ follows from (d) (and has been proved previously). For the case $\ell' \leq 2$, note that the set $\{[X]\}$ is a CAL of length 1, and the set $\{[X], [Y]\}$ is a CAL of length 2.
\item Suppose $[I] \in \mathcal{C}$ and $\mathcal{C} \neq \{[I]\}$. Then $\mathcal{C}$ must contain an element $\hat{p}$ that anticommutes with $[I]$ which is impossible.
\end{enumerate}
\end{proof}

It is clear from the above lemma that extremal CALs of length $\ell \geq 3$ correspond precisely to the last two cases of Eq.~\eqref{eq:cal-min-qubits}, while extremal CALs of length $1 \leq \ell \leq 2$ correspond to the first case. We now turn to the question of the dimension of a CAL. For the next few proofs, it will be useful to have the following result, that is proved in Section~4 of \cite{sarkar-vandenberg2019}. It uses the notion of commutativity maps in parts (e)-(g), that we defined in Section~\ref{subsec:pauli-definitions}.

\begin{lemma}
\label{lem:anticommuting-sets}
Let $\mathcal{A} = \{\hat{p}_0,\ldots,\hat{p}_{\ell-1}\} \subseteq \pauligrp{n}$ be an anticommuting list, and let $\text{Com}_{\mathcal{A}}$ denote the commutativity map with respect to $\mathcal{A}$. Then
\setmargins
\begin{enumerate}[(a)]
\item $\prod \mathcal{A} \neq [I]$ implies $\mathcal{A}$ is an independent generator of the subgroup $\langle \mathcal{A} \rangle$ of order $2^\ell$.
\item For any non-empty list $\mathcal{H} \subset \mathcal{A}$, we have $\prod\mathcal{H} \neq [I]$, so $\mathcal{H}$ is independent.
\item  $\prod\mathcal{A} = [I]$ implies $\ell$ is odd.
\item Suppose $\mathcal{A}$ is non-empty, independent and $\ell$ is even. Then for every $v \in \mathbb{F}_2^\ell$, there exists a unique element $\hat{p} \in \langle \mathcal{A} \rangle$ such that $\commap{\mathcal{A}}{\hat{p}} = v$.
\item Suppose $\mathcal{A}$ is independent, $\mathcal{A} \neq \{[I]\}$, and $\ell$ is odd. Then for every $v \in \mathbb{F}_2^\ell$ containing an odd number of zeros, there exists exactly two elements $\hat{p},\hat{q} \in \langle \mathcal{A} \rangle$ such that $\commap{\mathcal{A}}{\hat{p}} = \commap{\mathcal{A}}{\hat{q}} = v$.
\item Suppose $\mathcal{A}$ is independent, $\mathcal{A} \neq \{[I]\}$, and $\ell$ is odd. Then each coset $\mathcal{H}$ of $\langle \mathcal{A} \rangle$ in $\pauligrp{n}$ has the property that either for all $\hat{q} \in \mathcal{H}$, $\commap{\mathcal{A}}{\hat{q}}$ has an odd number of zeros, or for all $\hat{q} \in \mathcal{H}$, $\commap{\mathcal{A}}{\hat{q}}$ has an even number of zeros. Moreover the number of cosets of each type are equal.
\end{enumerate}
\end{lemma}

We now give the answer to the question of the dimension of a CAL. The next two theorems cover the cases when the CAL has lengths odd and even respectively, and the corollary after that specializes these results to the case of an extremal CAL. It is important to remember for the next proofs that by Lemma~\ref{lem:cal-maxsize}(f), whenever $[I] \not \in \mathcal{C} \subseteq \pauligrp{n}$, the dimension of a CAL $\mathcal{C}$ is the minimum number of elements of $\mathcal{C}$ that generates $\langle \mathcal{C} \rangle$.

\begin{theorem}
\label{thm:cal-dim-odd}
Let $\mathcal{C} \subseteq \pauligrp{n}$ be a CAL, and suppose $|\mathcal{C}|$ is odd. 

\setmargins
\begin{enumerate}[(a)]
\item The dimension of $\mathcal{C}$ satisfies
\begin{equation}
    \label{eq:cal-dim-odd}
    \dim(\mathcal{C}) = 
    \begin{cases}
        |\mathcal{C}| \;\; & \text{ if } \prod \mathcal{C} \neq [I], \\
        |\mathcal{C}| - 1 \;\; & \text{ if } \prod \mathcal{C} = [I].
    \end{cases}
\end{equation}
\item If $|\mathcal{C}| \geq 3$, it holds that $\mathcal{C} \setminus \{\hat{p}\}$ is independent for every $\hat{p} \in \mathcal{C}$.
\item For all non-empty multisets $\mathcal{H} \subset \mathcal{C}$, we have $\prod \mathcal{H} \neq [I]$.
\end{enumerate}
\end{theorem}

\begin{proof} In this proof let $\mathcal{C} = \{\hat{p}_0,\dots,\hat{p}_{2k}\}$ where $k \geq 0$, so $|\mathcal{C}|=2k+1$. The case $\mathcal{C} = \{[I]\}$ is true for part (a) as $\dim (\{[I]\}) = 0$, and also obvious for parts (b)-(c). So for the rest of the proof assume that $\mathcal{C} \neq \{[I]\}$.

\setmargins
\begin{enumerate}[(a)]
\item Since $\dim(\mathcal{C}) \leq |\mathcal{C}|$, it is easily seen that proving Eq.~\eqref{eq:cal-dim-odd}, is equivalent to proving the following two conditions: 

\begin{enumerate}[(i)]
\item $\dim(\mathcal{C}) \geq |\mathcal{C}| - 1$,
\item $\dim(\mathcal{C}) = |\mathcal{C}| - 1$ if and only if $\prod \mathcal{C} = [I]$.
\end{enumerate}

If $|\mathcal{C}| = 1$, then $\dim(\mathcal{C}) = 1$, and $\prod \mathcal{C}\neq [I]$, and so (i) and (ii) are satisfied. Now assume $k \geq 1$. Construct the anticommuting list $\mathcal{A} = \{\hat{q}_0, \dots, \hat{q}_{2k-1}\}$ similarly as in Lemma~\ref{lem:cal-maxsize}(b). Notice that because $|\mathcal{A}|$ is even, it follows from Lemma~\ref{lem:anticommuting-sets}(c) that $\prod \mathcal{A} \neq [I]$, and so $\mathcal{A}$ is an independent generator of $\langle \mathcal{A} \rangle$ by Lemma~\ref{lem:anticommuting-sets}(a). We also have $\langle \mathcal{C} \setminus \{\hat{p}_{2k}\} \rangle  = \langle \mathcal{A} \rangle$, which then implies that $\mathcal{C} \setminus \{\hat{p}_{2k}\}$ is independent, because $|\mathcal{C} \setminus \{\hat{p}_{2k}\}| = 2k$ and $|\langle \mathcal{C} \setminus \{\hat{p}_{2k}\} \rangle| = |\langle \mathcal{A} \rangle| = 2^{2k}$. It follows that $\dim(\mathcal{C}) \geq |\mathcal{C}| - 1$, finishing the proof of (i).

For the proof of (ii), first assume that $\prod \mathcal{C} = [I]$. But this implies that $\dim(\mathcal{C}) < |\mathcal{C}|$, and hence combining with (i) this gives us that $\dim(\mathcal{C}) = |\mathcal{C}| - 1$. To prove the converse, for contradiction assume that $\prod \mathcal{C} \neq [I]$, which is equivalent to the assumption that $\hat{p}_{2k} \neq \prod (\mathcal{C} \setminus \{\hat{p}_{2k}\}) = \hat{q}_{2k-1}$. We have already proved in the previous paragraph that $\mathcal{A}$ is an independent list, and it can also be checked that $\{\hat{p}_{2k}, \hat{q}_i\} = 0$, for all $0 \leq i \leq 2k-2$, and $[\hat{p}_{2k},\hat{q}_{2k-1}] = 0$. Since $|\mathcal{A}|$ is even, it follows from Lemma~\ref{lem:anticommuting-sets}(d) that $\hat{q}_{2k-1}$ is the unique element in $\langle \mathcal{A} \rangle$ generating the commutativity pattern $\commap{\mathcal{A}}{\hat{p}_{2k}}$ with respect to $\mathcal{A}$. But $\hat{p}_{2k} \neq \hat{q}_{2k-1}$ by assumption, implying that $\hat{p}_{2k} \notin \langle \mathcal{A} \rangle = \langle \mathcal{C} \setminus \{\hat{p}_{2k}\} \rangle$. So $\mathcal{C}$ is an independent list, because we also proved in the previous paragraph that $\mathcal{C} \setminus \{\hat{p}_{2k}\}$ is independent, and it follows that $\dim(\mathcal{C}) = |\mathcal{C}|$, which is a contradiction.
\item We have already proved in (a) that $\mathcal{C} \setminus \{\hat{p}_{2k}\}$ is independent. The fact that $\mathcal{C} \setminus \{\hat{p}_i\}$ is also independent for all $\hat{p}_i \in \mathcal{C}$ then follows by translating the elements of $\mathcal{C}$ to get a relabeling of the CAL $\mathcal{C} = \{\hat{p}'_0,\dots,\hat{p}'_{2k}\}$ by Lemma~\ref{lem:cal-maxsize}(a), where $\hat{p}'_j = \hat{p}_{i+j-2k}$ and indices are evaluated modulo $2k+1$.
\item If $|\mathcal{C}|=1$, there is nothing to prove. So let $\mathcal{H} \subset \mathcal{C}$ be a non-empty multiset, and $|\mathcal{C}| \geq 3$. Then there exists $\mathcal{C} \ni \hat{p} \not\in \mathcal{H}$. So by (b), $\mathcal{C} \setminus \{\hat{p}\}$ is an independent multiset. Since $\mathcal{H} \subseteq \mathcal{C} \setminus \{\hat{p}\}$, $\prod \mathcal{H} \neq [I]$.
\end{enumerate}
\end{proof}

\begin{theorem}
\label{thm:cal-dim-even}
Let $\mathcal{C} = \{\hat{p}_\ell: 0 \leq \ell \leq 2k-1\} \subseteq \pauligrp{n}$ be a non-empty CAL, with $|\mathcal{C}|$ even. 
\setmargins
\begin{enumerate}[(a)]
    \item Define the lists $\mathcal{C}_{\text{odd}} = \{\hat{p}_{2\ell+1}: 0 \leq \ell \leq k-1 \}$, and $\mathcal{C}_{\text{even}} = \{\hat{p}_{2\ell}: 0 \leq \ell \leq k-1 \}$. The dimension of $\mathcal{C}$ satisfies
    \begin{equation}
    \label{eq:cal-dim-even}
        \dim(\mathcal{C}) = 
        \begin{cases}
            |\mathcal{C}| \;\; & \text{ if } \prod \mathcal{C}_{\text{odd}} \neq [I] \neq \prod \mathcal{C}_{\text{even}}, \text{ and } \prod \mathcal{C}_{\text{odd}} \neq \prod \mathcal{C}_{\text{even}}, \\
            |\mathcal{C}| - 1 \;\; & \text{ if } \prod \mathcal{C}_{\text{odd}} \neq [I] \neq \prod \mathcal{C}_{\text{even}}, \text{ and } \prod \mathcal{C}_{\text{odd}} = \prod \mathcal{C}_{\text{even}}, \\
            |\mathcal{C}| - 1 \;\; & \text{ if either } \prod \mathcal{C}_{\text{odd}} = [I] \neq \prod \mathcal{C}_{\text{even}}, \text{ or } \prod \mathcal{C}_{\text{odd}} \neq [I] = \prod \mathcal{C}_{\text{even}}, \\
            |\mathcal{C}| - 2 \;\; & \text{ if } \prod \mathcal{C}_{\text{odd}} = \prod \mathcal{C}_{\text{even}} = [I].
        \end{cases}
    \end{equation}
    In particular, if $|\mathcal{C}| = 2$, then $\dim(\mathcal{C}) = 2$, while if $|\mathcal{C}| = 2n \geq 4$, then $\dim(\mathcal{C}) < 2n$.
    \item For all non-empty multisets $\mathcal{H} \subset \mathcal{C}$ such that $\mathcal{H} \neq \mathcal{C}_{\text{odd}}$ and $\mathcal{H} \neq \mathcal{C}_{\text{even}}$, we have $\prod \mathcal{H} \neq [I]$.
    \item If $|\mathcal{C}| \geq 4$, it holds that $\mathcal{C} \setminus \{\hat{p},\hat{q}\}$ is an independent multiset for all $\hat{p} \in \mathcal{C}_{\text{odd}}$, and $\hat{q} \in \mathcal{C}_{\text{even}}$.
\end{enumerate}
\end{theorem}

\vspace{0.5cm}
\begin{proof}
\setmargins
\begin{enumerate}[(a)]
\item If $|\mathcal{C}| = 2$, then $\mathcal{C}$ is an anticommuting list of two elements, and hence by Lemma~\ref{lem:anticommuting-sets}(a) and (d), we have that $\mathcal{C}$ is independent. Thus $\dim(\mathcal{C}) = 2$, $\prod \mathcal{C}_{\text{odd}} \neq [I] \neq \prod \mathcal{C}_{\text{even}}$, and also $\prod \mathcal{C}_{\text{odd}} \neq \prod \mathcal{C}_{\text{even}}$, which proves the statement for $|\mathcal{C}| = 2$. For the rest of the proof assume $|\mathcal{C}| \geq 4$. Construct the anticommuting list $\mathcal{A} = \{\hat{q}_0, \dots, \hat{q}_{2k-2}\}$ similarly as in Lemma~\ref{lem:cal-maxsize}(b). Then $\mathcal{A}  \setminus \{\hat{q}_{2k-2}\}$ is independent by Lemma~\ref{lem:anticommuting-sets}(c), and since it is easily checked that $\langle \mathcal{A}  \setminus \{\hat{q}_{2k-2}\} \rangle = \langle \mathcal{C} \setminus \{\hat{p}_{2k-2}, \hat{p}_{2k-1}\} \rangle$, we have that $\dim(\mathcal{C}) \geq |\mathcal{C}| - 2$. 

First suppose that $\prod \mathcal{C}_{\text{odd}} = \prod \mathcal{C}_{\text{even}} = [I]$. This implies that $\dim(\mathcal{C}_{\text{odd}}) \leq |\mathcal{C}_{\text{odd}}| - 1$, and $\dim(\mathcal{C}_{\text{even}}) \leq |\mathcal{C}_{\text{even}}| - 1$, and so $\dim(\mathcal{C}) \leq |\mathcal{C}| - 2$. This proves that $\dim(\mathcal{C}) = |\mathcal{C}| - 2$ in this case. 

Next suppose that either $\prod \mathcal{C}_{\text{odd}} = [I] \neq \prod \mathcal{C}_{\text{even}}, \text{ or } \prod \mathcal{C}_{\text{odd}} \neq [I] = \prod \mathcal{C}_{\text{even}}$. In the former case we have that $\dim(\mathcal{C}_{\text{odd}}) \leq |\mathcal{C}_{\text{odd}}| - 1$, while in the latter case we have that $\dim(\mathcal{C}_{\text{even}}) \leq |\mathcal{C}_{\text{even}}| - 1$, and so in both cases $\dim(\mathcal{C}) \leq |\mathcal{C}| - 1$. Let us now assume without loss of generality that $\prod \mathcal{C}_{\text{odd}} = [I] \neq \prod \mathcal{C}_{\text{even}}$, because the other case reduces to this one by translating $\mathcal{C}$ by one element, and relabeling the elements as in the proof of Theorem~\ref{thm:cal-dim-odd}(b). Then a simple computation shows that $\prod \mathcal{A} = \prod \mathcal{C}_{\text{even}}$, and so $\prod \mathcal{A} \neq [I]$ by assumption, which implies that $\mathcal{A}$ is an independent list by Lemma~\ref{lem:anticommuting-sets}(a). But as $\langle \mathcal{A} \rangle = \langle \mathcal{C} \setminus \{\hat{p}_{2k-1}\} \rangle$, this implies that $\dim(\mathcal{C}) \geq |\mathcal{C}| - 1$, and so in fact $\dim(\mathcal{C}) = |\mathcal{C}| - 1$.

Now assume that $\prod \mathcal{C}_{\text{odd}} \neq [I] \neq \prod \mathcal{C}_{\text{even}}$, and $\prod \mathcal{C}_{\text{odd}} = \prod \mathcal{C}_{\text{even}}$. Because $\prod \mathcal{C}_{\text{even}} \neq [I]$, the reasoning in the previous paragraph gives us that $\dim(\mathcal{C}) \geq |\mathcal{C}| - 1$ because $\mathcal{A}$ is independent, while we also have that $\prod \mathcal{C} = (\prod \mathcal{C}_{\text{odd}}) (\prod \mathcal{C}_{\text{even}}) = [I]$ and thus $\dim(\mathcal{C}) \leq |\mathcal{C}| - 1$, which means that $\dim(\mathcal{C}) = |\mathcal{C}| - 1$.

Finally assume that $\prod \mathcal{C}_{\text{odd}} \neq [I] \neq \prod \mathcal{C}_{\text{even}}$, and $\prod \mathcal{C}_{\text{odd}} \neq \prod \mathcal{C}_{\text{even}}$. As in the last paragraph, we still have that $\mathcal{A}$ is an independent anticommuting list. One can now verify that $\hat{p}_{2k-1}$, $\hat{q}_{2k-2}$ and $\prod (\mathcal{A} \setminus \hat{q}_{2k-2})$ generate the same commutativity pattern with respect to $\mathcal{A}$, i.e. if $\hat{p}$ is either $\hat{p}_{2k-1}$, $\hat{q}_{2k-2}$, or $\prod (\mathcal{A} \setminus \hat{q}_{2k-2})$, then $[\hat{p}, \hat{q}_{2k-2}] = 0$, and $\{\hat{p}, \hat{q}_i\} = 0$, for all $0 \leq i \leq 2k-3$. Now notice that $\hat{p}_{2k-1} \hat{q}_{2k-2} = \prod \mathcal{C} \neq [I]$ as $\prod \mathcal{C}_{\text{odd}} \neq \prod \mathcal{C}_{\text{even}}$, and also $\hat{p}_{2k-1} \prod (\mathcal{A} \setminus \hat{q}_{2k-2}) = \prod \mathcal{C}_{\text{odd}} \neq [I]$ by assumption, and so an application of Lemma~\ref{lem:anticommuting-sets}(e) implies that $\hat{p}_{2k-1} \notin \langle \mathcal{A} \rangle$. This implies that in this case $\dim(\mathcal{C}) = |\mathcal{C}|$.

It remains to show that if $|\mathcal{C}| = 2n \geq 4$, then $\dim(\mathcal{C}) < 2n$. From what has already been proved above, it suffices to show that the conditions $\prod \mathcal{C}_{\text{odd}} \neq [I] \neq \prod \mathcal{C}_{\text{even}}$, and $\prod \mathcal{C}_{\text{odd}} \neq \prod \mathcal{C}_{\text{even}}$ cannot both be true. However, for contradiction assume that they in fact both hold. Then the discussion in the last paragraph applies, and so $\hat{p}_{2n-1} \notin \langle \mathcal{A} \rangle$, and $\dim(\mathcal{A}) = 2n-1$. This also implies that there are exactly two cosets of $\langle \mathcal{A} \rangle$ in $\pauligrp{n}$, and so $\hat{p}_{2n-1}$ must be in the coset different from $\langle \mathcal{A} \rangle$. But this is a contradiction by an application of Lemma~\ref{lem:anticommuting-sets}(f), which rules out any element in the other coset, that generate the same commutativity pattern as $\hat{p}_{2k-1}$, with respect to the elements of $\mathcal{A}$.
\item The statement is vacuous for $|\mathcal{C}| =2$ as there are no non-empty proper multisets satisfying the conditions; so assume $|\mathcal{C}| \geq 4$. We will prove this using contradiction. So suppose $\mathcal{H} \subset \mathcal{C}$ is a non-empty multiset, $\mathcal{C}_{\text{even}} \neq \mathcal{H} \neq \mathcal{C}_{\text{odd}}$, and $\prod \mathcal{H} = [I]$. This means $\prod \mathcal{H}$ commutes with all elements of $\mathcal{C}$. Now there are two cases: (i) $\hat{p}_0 \in \mathcal{H}$, and (ii) $\hat{p}_0 \not \in \mathcal{H}$. Since $\prod \mathcal{H}$ commutes with $\hat{p}_1$, it must be that in the first case $\hat{p}_2 \in \mathcal{H}$, and in the second case  $\hat{p}_2 \not\in \mathcal{H}$. Continuing this argument iteratively for $\hat{p}_3, \dots, \hat{p}_{2k-3}$ then shows that (i) and (ii) are equivalent to the conditions $\mathcal{C}_{\text{even}} \subseteq \mathcal{H}$ and $\mathcal{C}_{\text{even}} \cap \mathcal{H} = \emptyset$ respectively. Similarly, since $\prod \mathcal{H}$ commutes with all elements of $\mathcal{C}_{\text{even}}$, one again concludes that exactly one of the two cases is true: either $\mathcal{C}_{\text{odd}} \subseteq \mathcal{H}$ or $\mathcal{C}_{\text{odd}} \cap \mathcal{H} = \emptyset$. Combining we have four cases: (i) $\mathcal{C}_{\text{odd}} \cup \mathcal{C}_{\text{even}} \subseteq \mathcal{H}$, (ii) $(\mathcal{C}_{\text{odd}} \cup \mathcal{C}_{\text{even}}) \cap \mathcal{H} = \emptyset$, (iii) $\mathcal{H} = \mathcal{C}_{\text{odd}}$, and (iv) $\mathcal{H} = \mathcal{C}_{\text{even}}$, all of which are impossible by assumption.
\item For contradiction assume that $\mathcal{C} \setminus \{\hat{p},\hat{q}\}$ is not independent. Then there exists a multiset $\mathcal{H} \subseteq \mathcal{C} \setminus \{\hat{p},\hat{q}\}$ such that $\prod \mathcal{H} = [I]$. But this is a contradiction by (b).
\end{enumerate}
\end{proof}

\begin{corollary}
\label{cor:cal-dim-extremal}
Let $\mathcal{C} \subseteq \pauligrp{n}$ be an extremal CAL of length $|\mathcal{C}| \geq 3$. Then the dimension of $\mathcal{C}$ satisifies
\begin{equation}
    \label{eq:cal-dim-extremal}
    \dim(\mathcal{C}) = 
    \begin{cases}
        |\mathcal{C}| - 1 \;\; & \text{ if } |\mathcal{C}| \text{ odd}, \\
        |\mathcal{C}| - 2 \;\; & \text{ if } |\mathcal{C}| \text{ even}.
    \end{cases}
\end{equation}

Thus $\mathcal{C}$ generates the group $\pauligrp{n}$.
\end{corollary}

\begin{proof}
First assume that $|\mathcal{C}|$ is odd. Then from Theorem~\ref{thm:cal-dim-odd} we have $\dim(\mathcal{C}) \geq |\mathcal{C}| - 1$. But because $\mathcal{C}$ is an extremal CAL we also have that $\dim(\mathcal{C}) \leq |\mathcal{C}| - 1$, because the number of qubits used is $(|\mathcal{C}| - 1) / 2$. Now assume that $|\mathcal{C}|$ is even. In this case Theorem~\ref{thm:cal-dim-even} gives that $\dim(\mathcal{C}) \geq |\mathcal{C}| - 2$, while again as $\mathcal{C}$ is an extremal CAL, it must also be true that $\dim(\mathcal{C}) \leq |\mathcal{C}| - 2$. In both cases $\dim(\mathcal{C}) = n$, and so $\langle \mathcal{C} \rangle = \pauligrp{n}$. This proves the lemma.
\end{proof}

As a consequence of the previous lemmas, we get the next corollary that states which multisets of an extremal CAL of length at least three, can possibly multiply to the identity.

\begin{corollary}
\label{cor:extremal-cal-mult-identity}
Let $\mathcal{C} \subseteq \pauligrp{n}$ be an extremal CAL of length $|\mathcal{C}| \geq 3$. Then the following statements are true.
\setmargins
\begin{enumerate}[(a)]
    \item Let $|\mathcal{C}|$ be odd. If $\mathcal{H} \subseteq \mathcal{C}$ is a non-empty multiset, then $\prod \mathcal{H} = [I]$ if and only if $\mathcal{H} = \mathcal{C}$.
    \item Let $|\mathcal{C}|$ be even. If $\mathcal{H} \subseteq \mathcal{C}$ is a non-empty multiset, then $\prod \mathcal{H} = [I]$ if and only if $\mathcal{H} = \mathcal{C}_{\text{odd}}$ or $\mathcal{H} = \mathcal{C}_{\text{even}}$, where $\mathcal{C}_{\text{odd}}$ and $\mathcal{C}_{\text{even}}$ are defined as in Theorem~\ref{thm:cal-dim-even}.
\end{enumerate}
\end{corollary}

\vspace{0.5cm}
\begin{proof}
\setmargins
\begin{enumerate}[(a)]
    \item As $\mathcal{C}$ is an extremal CAL, it is not independent by Corollary~\ref{cor:cal-dim-extremal}. The result then follows directly from Theorem~\ref{thm:cal-dim-odd}(a),(c).
    \item As $\mathcal{C}$ is an extremal CAL, $\dim(\mathcal{C}) = |\mathcal{C}| - 2$ by Corollary~\ref{cor:cal-dim-extremal}. The result follows by Theorem~\ref{thm:cal-dim-even}(a),(b). 
\end{enumerate}
\end{proof}

In the final part of this subsection, we turn to the question of whether it is possible to create CALs of $\pauligrp{n}$, such that the representatives of each element in the CAL are Kronecker products of $I_2$ and the Pauli matrices $X$ and $Z$ only, and not of $Y$. To motivate this section, note that on $n=1$ qubit, one has the CALs of lengths $\ell = 1$, $2$, and $4$ given by $\{[X]\}$, $\{[X], [Z]\}$, and $\{[X], [Z], [X], [Z]\}$ respectively, satisfying this property. But there is no CAL of length $\ell = 3$ with this property. Also notice that if the CAL has length $\ell \geq 2$, then it is not possible that all its element representatives are Kronecker products of $I_2$ and just one other Pauli matrix, for example $X$. This is because all elements of the CAL constructed this way will commute with one another. So we definitely need to use at least two of the Pauli matrices, for example $X$ and $Z$. 

For what follows, we need a couple of simple lemmas. The next result is well-known (see for instance Chapter 5 in \cite{warren2013hacker}).
\begin{lemma}
\label{lem:vector-sum-even-ham-weight}
For any $v \in \mathbb{F}_2^p$, let $\ham{v} = \sum_{j}v_j$ denote its Hamming weight. Then if $u,v \in \mathbb{F}_2^p$ have even Hamming weight, so does $w = u + v$ (where the addition is over $\mathbb{F}_2$). Thus any binary sum of even Hamming weight vectors also has even Hamming weight.
\end{lemma}

\begin{proof}
Let $J_u \subseteq \{0,\dots,p-1\}$ be the set of indices such that $u_j = 1$ if and only if $j \in J_u$, and similarly let $J_v \subseteq \{0,\dots,p-1\}$ be the set of indices such that $v_j = 1$ if and only if $j \in J_v$. Since addition is over $\mathbb{F}_2$, this implies that $w_j = 1$ if and only if $j \in J_w = (J_u \cup J_v) \setminus (J_u \cap J_v)$. But $|J_w| = |J_u| + |J_v| - 2|J_u \cap J_v|$, and we also have that $|J_u| = \ham{u}$, and $|J_v| = \ham{v}$, which are even by assumption. So $|J_w|$ is even, and since $|J_w| = \ham{w}$, this proves the first part of the lemma. The last part now follows by induction.
\end{proof}

We also give a construction that allows us to take a CAL and create a CAL of smaller length on the same number of qubits, or a longer CAL on more qubits, as given in the lemma below. The first part of this construction will be used later.
\begin{lemma}
\label{lem:cal-expand-contract}
Let $\mathcal{C} = \{\hat{p}_0,\dots,\hat{p}_{k-1}\} \subseteq \pauligrp{n}$ be a CAL of length $k \geq 2$, and define the lists $\mathcal{C}_{\text{odd}} = \{\hat{p}_{2\ell+1}: 0 \leq \ell \leq \floor{k/2}-1\}$ and $\mathcal{C}_{\text{even}}=\{\hat{p}_{2\ell}: 0 \leq \ell \leq \ceil{k/2}-1\}$. Consider the following constructions:
\setmargins
\begin{enumerate}[(i)]
    \item CAL expansion --- If $k$ is odd, let $\hat{r} = [Z]$, $\hat{s} = [X]$, and if $k$ is even, let $\hat{r} = [X]$, $\hat{s} = [Z]$. Define the list $\mathcal{C}' = \{\hat{p}'_0,\dots,\hat{p}'_{k+1}\} \subseteq \pauligrp{n+1}$ by $\hat{p}'_\ell = \hat{p}_\ell \otimes [I_2]$ for all $0 \leq \ell \leq k-3$, $\hat{p}'_{k-2} = \hat{p}_{k-2} \otimes \hat{r}$, $\hat{p}'_{k-1} = [I] \otimes \hat{s}$, $\hat{p}'_{k} = [I] \otimes \hat{r}$, and $\hat{p}'_{k+1} = \hat{p}_{k-1} \otimes \hat{s}$, where $[I] \in \pauligrp{n}$ is the identity element of $\pauligrp{n}$. Also define the lists $\mathcal{C}'_{\text{odd}} = \{\hat{p}'_{2\ell+1}: 0 \leq \ell \leq \floor{k/2}\}$ and $\mathcal{C}'_{\text{even}}=\{\hat{p}'_{2\ell}: 0 \leq \ell \leq \ceil{k/2}\}$.
    \item CAL contraction --- Let $\alpha \geq 0$ be any integer such that $\beta := k - 2\alpha \geq 2$. Define the list $\tilde{\mathcal{C}} = \{\hat{q}'_0,\dots,\hat{q}'_{\beta-1}\} \subseteq \pauligrp{n}$ by $\hat{q}'_\ell = \hat{p}_\ell$ for all $0 \leq \ell \leq \beta-3$, $\hat{q}'_{\beta - 2} = \prod_{\ell = 0}^{\alpha} \hat{p}_{\beta - 2 + 2 \ell}$, and $\hat{q}'_{\beta - 1} = \prod_{\ell = 0}^{\alpha} \hat{p}_{\beta - 1 + 2 \ell}$, and also define the lists $\tilde{\mathcal{C}}_{\text{odd}} = \{\hat{q}_{2\ell+1}: 0 \leq \ell \leq \floor{\beta/2}-1\}$ and $\tilde{\mathcal{C}}_{\text{even}}=\{\hat{q}_{2\ell}: 0 \leq \ell \leq \ceil{\beta/2}-1\}$.
\end{enumerate}

Then we have the following:
\setmargins
\begin{enumerate}[(a)]
    \item $\prod \mathcal{C}'_{\text{odd}} = \prod \mathcal{C}_{\text{odd}} \otimes [I_2]$, $\prod \mathcal{C}'_{\text{even}} = \prod \mathcal{C}_{\text{even}} \otimes [I_2]$, and $\prod \mathcal{C}' = \prod \mathcal{C} \otimes [I_2]$. If all elements of $\mathcal{C}_{\text{odd}}$ (resp. $\mathcal{C}_{\text{even}}$) are of $Z$-type (resp. $X$-type), then so are all elements of $\mathcal{C}'_{\text{odd}}$ (resp. $\mathcal{C}'_{\text{even}}$).
    \item If $k \geq 3$, then $\mathcal{C}'$ is a CAL of length $k+2$, and $\dim (\mathcal{C}') = \dim (\mathcal{C})$.
    \item $\prod \tilde{\mathcal{C}}_{\text{odd}} = \prod \mathcal{C}_{\text{odd}}$, $\prod \tilde{\mathcal{C}}_{\text{even}} = \prod \mathcal{C}_{\text{even}}$, and $\prod \tilde{\mathcal{C}} = \prod \mathcal{C}$. If all elements of $\mathcal{C}_{\text{odd}}$ (resp. $\mathcal{C}_{\text{even}}$) are of $Z$-type (resp. $X$-type), then so are all elements of $\tilde{\mathcal{C}}_{\text{odd}}$ (resp. $\tilde{\mathcal{C}}_{\text{even}}$).
    \item If $\beta \geq 3$, then $\tilde{\mathcal{C}}$ is a CAL of length $\beta$, and $\dim (\tilde{\mathcal{C}}) = \dim (\mathcal{C})$.
\end{enumerate}
\end{lemma}

\vspace{0.5cm}
\begin{proof}
\setmargins
\begin{enumerate}[(a)]
    \item It is easy to verify that by construction, if all elements of $\mathcal{C}_{\text{odd}}$ are $Z$-type, then the same is true for $\mathcal{C}'_{\text{odd}}$, and if all elements of $\mathcal{C}_{\text{even}}$ are $X$-type, then the same holds for all elements of $\mathcal{C}'_{\text{even}}$. Now $\prod \mathcal{C}'_{\text{odd}} = \prod_{\ell=0}^{\floor{k/2}} \hat{p}'_{2 \ell + 1} = \left( \prod_{\ell=0}^{\floor{k/2}-2} \hat{p}_{2 \ell + 1} \otimes [I_2] \right) \left(\hat{p}_{2 \floor{k/2} - 3} \otimes [Z] \right) \left([I] \otimes [Z]\right) = \left( \prod_{\ell=0}^{\floor{k/2}-1} \hat{p}_{2 \ell + 1} \right) \otimes [I_2]$, which shows $\prod \mathcal{C}'_{\text{odd}} = \prod \mathcal{C}_{\text{odd}} \otimes [I_2]$. A similar calculation also shows that $\prod \mathcal{C}'_{\text{even}} = \prod \mathcal{C}_{\text{even}} \otimes [I_2]$, and from both these we also conclude that $\prod \mathcal{C}' = \prod \mathcal{C} \otimes [I_2]$. 
    \item It is straightforward to verify that when $k \geq 3$, and since $\mathcal{C}$ is a CAL, the elements of $\mathcal{C}'$ satisfy the CAL commutation relations (Eq.~\eqref{eq:cal-comm}). Now by (a), we have that $\prod \mathcal{C}' = [I] \in \pauligrp{n+1}$ if and only if $\prod \mathcal{C} = [I] \in \pauligrp{n}$. So when $k$ is odd, we have from Theorem~\ref{thm:cal-dim-odd}(a) that $\dim (\mathcal{C}') = \dim (\mathcal{C})$. When $k$ is even, it follows similarly using (a) the following: (i) $\prod \mathcal{C}_{\text{odd}} = [I]$ if and only if $\prod \mathcal{C}'_{\text{odd}} = [I]$, (ii) $\prod \mathcal{C}_{\text{even}} = [I]$ if and only if $\prod \mathcal{C}'_{\text{even}} = [I]$, and (iii) $\prod \mathcal{C}_{\text{odd}} = \prod \mathcal{C}_{\text{even}}$ if and only if $\prod \mathcal{C}'_{\text{odd}} = \prod \mathcal{C}'_{\text{even}}$. Theorem~\ref{thm:cal-dim-even} now implies that $\dim (\mathcal{C}') = \dim (\mathcal{C})$.
    \item This also follows by construction: we have $\prod \tilde{\mathcal{C}}_{\text{odd}} = \left( \prod_{\ell=0}^{\floor{\beta/2}-2} \hat{p}_{2\ell + 1} \right) \left( \prod_{\ell = 0}^{\alpha} \hat{p}_{2 \floor{\beta/2} - 1 + 2\ell} \right) = \prod \mathcal{C}_{\text{odd}}$, and $\prod \tilde{\mathcal{C}}_{\text{even}} = \left( \prod_{\ell=0}^{\ceil{\beta/2}-2} \hat{p}_{2\ell} \right) \left( \prod_{\ell = 0}^{\alpha} \hat{p}_{2 \ceil{\beta/2} - 2 + 2\ell} \right) = \prod \mathcal{C}_{\text{even}}$, and multiplying them together gives $\prod \tilde{\mathcal{C}} = \prod \mathcal{C}$. The last part is a simple observation. 
    \item If $\beta \geq 3$, then one easily verifies that $\tilde{\mathcal{C}}$ is a CAL. The proof of the last part is similar to (b), after noting that (c) implies (i) $\prod \mathcal{C}_{\text{odd}} = \prod \tilde{\mathcal{C}}_{\text{odd}}$, and (ii) $\prod \mathcal{C}_{\text{even}} = \prod \tilde{\mathcal{C}}_{\text{even}}$.
\end{enumerate}
\end{proof}

We are now in a position to discuss the case of odd-length CALs, which is more interesting than the even-length case. It turns out that one must always use all the three Pauli matrices to construct an odd-length CAL of length at least three, if the CAL is not an independent set; otherwise it is possible to not use the Pauli matrix $Y$. This is proved in the next theorem.

\begin{theorem}
\label{thm:cal-noY-odd}
For all $n \geq 1$, and $1 \leq k \leq 2n+1$ such that $k$ is odd, we have the following:
\setmargins
\begin{enumerate}[(a)]
    \item If $3 \leq k \leq 2n+1$, then it is not possible to construct a CAL $\mathcal{C} \subseteq \pauligrp{n}$, with $|\mathcal{C}| = k$, and $\prod \mathcal{C} = [I]$, such that the representatives of each element of $\mathcal{C}$ are Kronecker products not involving the Pauli matrix $Y$. If $k=1$, then $\mathcal{C} = \{[I]\}$ is a CAL whose element representatives are Kronecker products not involving any Pauli matrix.
    \item If $1 \leq k \leq 2n-1$, then it is always possible to construct a CAL $\mathcal{C} \subseteq \pauligrp{n}$, with $|\mathcal{C}| = k$, and $\prod \mathcal{C} \neq [I]$, such that the representatives of each element of $\mathcal{C}$ are Kronecker products not involving the Pauli matrix $Y$.
\end{enumerate}
\end{theorem}

\vspace{0.5cm}
\begin{proof}
\setmargins
\begin{enumerate}[(a)]
    \item The case $k=1$ is obvious, so fix $n \geq 1$, $3 \leq k \leq 2n+1$, such that $k$ is odd. The proof is by contradiction. Suppose there exists a CAL $\mathcal{C} = \{\hat{p}_0,\dots,\hat{p}_{k-1}\} \subseteq \pauligrp{n}$ satisfying $|\mathcal{C}| = k$, and $\prod \mathcal{C} = [I]$, such that the representatives of each element of $\mathcal{C}$ are Kronecker products not involving $Y$. Let $\hat{p}_\ell^{(j)}$ denote the restriction of $\hat{p}_\ell$ to qubit $j$, and let us also define the lists $\mathcal{C}^{(1)}, \dots \mathcal{C}^{(n)} \subseteq \pauligrp{1}$ by $\mathcal{C}^{(j)} = \{\hat{p}_\ell^{(j)} : 0 \leq \ell \leq k-1\}$, for all $0 \leq j \leq n-1$. We observe that $\prod \mathcal{C} = [I]$ implies $\prod \mathcal{C}^{(j)} = [I]$ for all $j$ (here the latter $[I]$ denotes the identity element of $\pauligrp{1}$). Now by assumption we have that $[Y] \notin \mathcal{C}^{(j)}$, and so $[X]$ and $[Z]$ must occur an even number of times in each $\mathcal{C}^{(j)}$. Define the matrices $M, M^{(0)}, \dots, M^{(n-1)} \in \mathbb{F}_2^{k \times k}$ as 
    
    \begin{equation}
    \label{eq:cal-noY-odd-com-matrix}
        \begin{split}
            M (r,s) &= 
            \begin{cases}
                0 & \;\;\;\; \text{if } [\hat{p}_{r},\hat{p}_{s}] = 0 \\
                1 & \;\;\;\; \text{if } \{\hat{p}_{r},\hat{p}_{s}\} = 0
            \end{cases} \\
            M^{(j)} (r,s) &= 
            \begin{cases}
                0 & \;\;\;\; \text{if } [\hat{p}_{r}^{(j)},\hat{p}_{s}^{(j)}] = 0 \\
                1 & \;\;\;\; \text{if } \{\hat{p}_{r}^{(j)},\hat{p}_{s}^{(j)}\} = 0
            \end{cases}
        \end{split}
    \end{equation}
    for all $0 \leq r,s \leq k-1$, and $0 \leq j \leq n-1$. These matrices have the following properties:
    
    \begin{enumerate}[(i)]
        \item $M, M^{(0)}, \dots, M^{(n-1)}$ are symmetric, and have zeros on the diagonal.
        \item The matrices satisfy the identity $M = \sum_{j=0}^{n-1} M^{(j)}$, where addition is performed over $\mathbb{F}_2$. This is because $\hat{p}_{r}$ and $\hat{p}_{s}$ commute if and only if their restrictions $\hat{p}_{r}^{(j)}$ and $\hat{p}_{s}^{(j)}$ for qubits $0 \leq j \leq n-1$, anticommute an even number of times.
        \item The upper triangular part of $M$ contains exactly $k$ non-zeros. This is simply a consequence of the commutation relations Eq.~\eqref{eq:cal-comm}.
        %
        %
        \item For all $0 \leq j \leq n-1$, since the multiset $\mathcal{C}^{(j)}$ contains an even number of $[X]$ and $[Z]$, it follows from Eq.~\eqref{eq:cal-noY-odd-com-matrix} that the number of non-zeros in $M^{(j)}$ is a multiple of $8$. As $M^{(j)}$ is symmetric with zeros on the diagonal by property (i), we have that the number of non-zeros in the upper triangular part of $M^{(j)}$ is a multiple of $4$.
    \end{enumerate}
    
    By an application of Lemma~\ref{lem:vector-sum-even-ham-weight}, property (iv) above implies that the number of non-zeros in the upper triangular part of $\sum_{j =0}^{n-1} M^{(j)}$ is even, but this is a contradiction to property (iii) which states that the number of non-zeros in the upper triangular part of $M$ is $k$, which is odd by assumption. This gives the desired contradiction.
    \item First notice that if there exists a CAL $\mathcal{C} = \{\hat{p}_\ell : 0 \leq \ell \leq k-1\} \subseteq \pauligrp{n}$ of odd length $k$ and $\prod \mathcal{C} \neq I$, then it must necessarily satisfy $1 \leq k \leq 2n-1$, by Lemma~\ref{lem:cal-maxsize}(d) and Corollary~\ref{cor:extremal-cal-mult-identity}(a), and this also implies the existence of a CAL of length $k$ with the required properties on $n' > n$ qubits, because one can simply form the new CAL $\mathcal{C}' = \{\hat{p}_\ell \otimes [I] : 0 \leq \ell \leq k-1\}$, where $[I] \in \pauligrp{n'-n}$ is the identity element on $n'-n$ qubits. Thus for each odd $k \geq 1$, it suffices to find a CAL $\mathcal{C}_k$ of length $k$, with the required properties on $(k+1)/2$ qubits. If $k=1$, then $\mathcal{C}_1 = \{[X]\}$ is such a CAL. If $k=3$, then $\mathcal{C}_3 = \{[X] \otimes [I_2], [Z] \otimes [Z], [Z] \otimes [X]\}$ is a CAL with the desired properties. One can now iteratively use the CAL expansion construction in Lemma~\ref{lem:cal-expand-contract}, to obtain $\mathcal{C}_{k+2}$ from $\mathcal{C}_{k}$ for all odd $k \geq 3$. The construction ensures that if the restriction of each element of $\mathcal{C}_k$ to qubit $j$ is not $[Y]$, for all $0 \leq j \leq (k-1)/2$, then the same is true for each element of $\mathcal{C}_{k+2}$ for all qubit indices $0 \leq j \leq (k+1)/2$.
\end{enumerate}
\end{proof}

The final theorem of this section shows that for CALs of even-length, it always suffices to not use the Pauli matrix $Y$.

\begin{theorem}
\label{thm:cal-noY-even}
For all $n \geq 1$, and for all integers $k$ satisfying $2 \leq 2k \leq 2n+2$, there exists a CAL $\mathcal{C} \subseteq \pauligrp{n}$, with $|\mathcal{C}| = 2k$, such that the representatives of each element of $\mathcal{C}$ are Kronecker products not involving the Pauli matrix $Y$. Moreover defining $\mathcal{C}_{\text{odd}}$ and $\mathcal{C}_{\text{even}}$ according to Theorem~\ref{thm:cal-dim-even}, one can choose $\mathcal{C}$ to satisfy any of the cases
\setmargins
\begin{enumerate}[(i)]
    \item $\prod \mathcal{C}_{\text{odd}} \neq [I] \neq \prod \mathcal{C}_{\text{even}}$, and  $\prod \mathcal{C}_{\text{odd}} \neq \prod \mathcal{C}_{\text{even}}$
    \item $\prod \mathcal{C}_{\text{odd}} \neq [I] \neq \prod \mathcal{C}_{\text{even}}$ and $\prod \mathcal{C}_{\text{odd}} = \prod \mathcal{C}_{\text{even}}$
    \item $\prod \mathcal{C}_{\text{odd}} = [I] \neq \prod \mathcal{C}_{\text{even}}$
    \item $\prod \mathcal{C}_{\text{odd}} \neq [I] = \prod \mathcal{C}_{\text{even}}$
    \item $\prod \mathcal{C}_{\text{odd}} = \prod \mathcal{C}_{\text{even}} = [I]$
\end{enumerate}
if the case is allowed by Theorem~\ref{thm:cal-dim-even}(a). Except for case (ii), one can also choose all elements of $\mathcal{C}_{\text{odd}}$ to be of $Z$-type, and all elements of $\mathcal{C}_{\text{even}}$ to be of $X$-type.
\end{theorem}

\begin{proof}
As in the proof of Theorem~\ref{thm:cal-noY-odd}(b), notice that if we prove this theorem for any fixed values of $k$ and $n$, then we have also proved the theorem for the same fixed $k$ and for all $n' > n$. First consider CALs of length $2k = 2$. Then by Theorem~\ref{thm:cal-dim-even}(a) and its proof, we know that only case (i) is possible. In this case the CAL $\mathcal{C} = \{[X], [Z]\}$ has the required properties for $n=1$, so this proves the case $2k=2$. For the rest of the proof, assume $2k \geq 4$. Also notice that the last statement is true because if it were possible to choose all elements of $\mathcal{C}_{\text{odd}}$ and $\mathcal{C}_{\text{even}}$ to be of $Z$-type and $X$-type respectively, then the only way to have $\prod \mathcal{C}_{\text{odd}} = \prod \mathcal{C}_{\text{even}}$ is to have $\prod \mathcal{C}_{\text{odd}} = \prod \mathcal{C}_{\text{even}} = [I]$, which is a contradiction.

We now adopt the notation $\mathcal{C}_{2k,n} \subseteq \pauligrp{n}$ to denote a CAL of length $2k$. Note that the minimum value of $n$ for which $\mathcal{C}_{2k,n}$ can exist for cases (i)-(v) satisfying the properties of the theorem is given by $n \geq n_{2k}$, where $n_{2k} = k+1$ for case (i), $n_{2k} = k$ for cases (ii)-(iv), and $n_{2k} = k-1$ for case (v), all following from Lemma~\ref{lem:cal-maxsize}(d) and Theorem~\ref{thm:cal-dim-even}(a). We first exhibit examples of CALs $\mathcal{C}_{4,n_2}$ existing for each case --- these are (i) $\mathcal{C}_{4,3} = \{[X_1 X_2], [Z_2], [X_2 X_3], [Z_1 Z_3]\}$, (ii) $\mathcal{C}_{4,2} = \{[X_1], [Z_1], [X_1 X_2], [Z_1 X_2]\}$, (iii) $\mathcal{C}_{4,2} = \{[X_1], [Z_1], [X_1 X_2], [Z_1]\}$, (iv) $\mathcal{C}_{4,2} = \{[X_1], [Z_1], [X_1], [Z_1 Z_2]\}$, and (v) $\mathcal{C}_{4,1} = \{[X], [Z], [X], [Z]\}$, for the cases (i)-(v) respectively. Now for each case we use the CAL expansion construction of Lemma~\ref{lem:cal-expand-contract}(i) to iteratively construct the CALs $\mathcal{C}_{2k,n_k}$ for all $k \geq 2$, starting from $\mathcal{C}_{4,n_2}$. These CALs have the required properties in each case by Lemma~\ref{lem:cal-expand-contract}(a), since $\mathcal{C}_{4,n_2}$ has them as well. The theorem is then proved by the first sentence of the previous paragraph.
\end{proof}

%% file: appendix/majorana_qubit_code_proofs.tex
\edit{\section{Majorana surface code lemmas and proofs}
\label{app:majorana-proofs}

This section contains proofs of Lemmas~\ref{lem:majorana-group-props} and \ref{lem:stabilizer_dependence}, which establish properties of Majorana subgroups and Majorana surface codes. We also restate the Lemmas here.

\begin{replemma}{lem:majorana-group-props}
Let $\mathcal{S}$ be a subgroup of $\mathcal{J}_m$, and let $\mathcal{I} \subseteq \mathcal{J}_m$ be non-empty, such that elements of $\mathcal{I}$ commute and are Hermitian. Then $\langle \mathcal{I} \rangle$ is Abelian and Hermitian, and moreover the following holds:

\setmargins
\begin{enumerate}[(a)]
    \item There exists a set $\mathcal{I}'$ formed by multiplying each element of $\mathcal{I}$ by either $1$ or $-1$, such that $-I \not \in \langle \mathcal{I}' \rangle$.
    \item If $-I \not \in \mathcal{S}$, then $\mathcal{S}$ is Abelian and Hermitian. Conversely, if $\mathcal{S}$ is Abelian and Hermitian, then either $-I \not \in \mathcal{S}$, or $\mathcal{S} = \langle \mathcal{S}', -I \rangle$ for some subgroup $\mathcal{S}'$ of $\mathcal{S}$ with $-I \not\in \mathcal{S}'$.
\end{enumerate}
\end{replemma}

\begin{proof}
In this proof $\eta, \eta', \eta'' \in \{\pm 1, \pm i\}$, and $a,a',b \in \mathbb{F}_2^m$. We also note some useful facts for this proof which follow easily from Eqs.~(\ref{eq:Maj_rules},\ref{eq:majorana-commutation}): if $x \in \mathcal{J}_m$ then (i) $x^{-1} = x^{\dag}$, (ii) $x$ is either Hermitian or skew-Hermitian, and (iii) $x^2 \in \{I,-I\}$, with $x^2 = I$ if and only if $x = x^\dag$. Also note that for any non-empty $\mathcal{J} \subseteq \mathcal{J}_m$ whose elements commute and are Hermitian, if $x \in \mathcal{J}$, then by (i) and (iii) $x = x^\dag = x^{-1}$, so all elements in $\langle \mathcal{J} \rangle$ are products of elements of $\mathcal{J}$. That $\langle \mathcal{I} \rangle$ is Abelian is clear since any element of $\langle \mathcal{I} \rangle$ is a product of the elements in $\mathcal{I}$. To see that $\langle \mathcal{I} \rangle$ is Hermitian, note that if $\langle \mathcal{I} \rangle \ni x = \prod_{y \in \mathcal{I}'} y$ for some $\mathcal{I}' \subseteq \mathcal{I}$, then $x^2 = \prod_{y \in \mathcal{I}'} y^2 = I$, and we conclude using (iii).
\setmargins
\begin{enumerate}[(a)]
    \item Let $|\mathcal{I}| = k \geq 1$. We prove the result by induction on $|\mathcal{I}|$. If $\mathcal{I} = \{x\}$ we choose $\mathcal{I}' = \{x\}$ if $x \neq -I$, else $\mathcal{I}' = \{I\}$, proving the base case. Now suppose the result is true for all $\mathcal{I}$ such that $1 \leq |\mathcal{I}| \leq r < k$. If $|\mathcal{I}| = r + 1$, pick any $x \in \mathcal{I}$ and let $\mathcal{I}_1 = \mathcal{I} \setminus \{x\}$. By the inductive hypothesis, there exists $\mathcal{I}'_1$ formed by multiplying every element of $\mathcal{I}_1$ by either $1$ or $-1$, such that $-I \not\in \langle \mathcal{I}'_1 \rangle$. If $x \in \langle \mathcal{I}'_1 \rangle$ (resp. $-x \in \langle \mathcal{I}'_1 \rangle$), choose $\mathcal{I}' = \mathcal{I}'_1 \cup \{x\}$ (resp. $\mathcal{I}' = \mathcal{I}'_1 \cup \{-x\}$) and we are done. If $x,-x \not\in \langle \mathcal{I}'_1 \rangle$, choose $\mathcal{I}' = \mathcal{I}'_1 \cup \{x\}$. Now for contradiction, assume $-I \in \langle \mathcal{I}' \rangle$, so we have $-I = zx$ for some $z \in \langle \mathcal{I}'_1 \rangle$, using properties of $\mathcal{I}$ (commuting and Hermitian) and (iii). Letting $z = \eta \gamma_{a}$, and $x = \eta' \gamma_{a'}$ then gives $-I = \eta \eta' (-1)^{\xi(a,a')} \gamma_{a+a'}$ using Eq.~\eqref{eq:majorana-commutation}, and so $a = a'$. Combining with $z^2 = x^2 = I$ now gives $\eta = - \eta'$, or equivalently $z = -x$ which is a contradiction.
    \item First suppose $-I \not \in \mathcal{S}$. Note that if $x \in \mathcal{S}$ is not Hermitian, then $\mathcal{S} \ni x^2 = -I$ by (iii). Now suppose $\mathcal{S}$ is non-Abelian. Then there exists $\eta \gamma_{a},\; \eta' \gamma_{a'} \in \mathcal{S}$, such that $\gamma_{a} \gamma_{a'} = - \gamma_{a'}\gamma_{a}$. But this implies that $\pm \eta'' \gamma_b \in \mathcal{S}$, where $\eta'' \gamma_b = \eta \eta' \gamma_{a} \gamma_{a'}$, and so both $\pm (\eta'')^2 \gamma_b^2 \in \mathcal{S}$. But one of these must be $-I$, as $\gamma_b^2 =  \pm I$, giving a contradiction. For the converse assume that $-I \in \mathcal{S}$ for Hermitian and Abelian $\mathcal{S}$. Consider the cosets of $\langle -I \rangle$ in $\mathcal{S}$, and let $\mathcal{I}$ be the set formed by picking exactly one element from each coset. By (a), then there exists $\mathcal{I}'$ such that $-I \not \in \langle \mathcal{I}' \rangle$, and $\mathcal{S} = \langle \langle \mathcal{I}' \rangle, -I \rangle$.
\end{enumerate}
\end{proof}

For the next Lemma, refer to the definition of Majorana surface codes, Definition~\ref{def:majorana_surface_code} for the notation.

\begin{replemma}{lem:stabilizer_dependence}
The stabilizers of the Majorana surface code for the embedded graph $G(V,E,F)$ satisfy
\setmargins
\begin{enumerate}[(a)]
\item There is no non-empty subset $V'\subseteq V$ such that $\prod_{v\in V'}S_v= \pm I$.
\item There is no non-empty, proper subset $V'\subset V$ and no subset $F'\subseteq F$ such that $\prod_{v\in V'}S_v= \pm \prod_{f\in F'}S_f$. If there is any odd-degree vertex, then the statement holds for all non-empty subsets $V'\subseteq V$.
\item If $\prod_{v\in V}S_v=\prod_{f\in F'}S_f$ for some subset $F' \subseteq F$, then $G$ is checkerboardable. If $G$ is checkerboardable then $\prod_{v\in V}S_v=\prod_{f\in F'}S_f = \prod_{f\in F \setminus F'}S_f$ for a non-empty proper subset $F' \subset F$, which is determined uniquely up to taking complements.
\item There is no non-empty, proper subset $F'\subset F$ such that $\prod_{f\in F'}S_f= \pm I$.
\item $\prod_{f\in F}S_f=I$.
\item $|F|=1$ if and only if there is $f\in F$ such that $S_f=I$.
\item $-I \not \in \mathcal{S}(G)$.
\end{enumerate}
\end{replemma}

\begin{proof}
We start with an observation that is needed for proofs of (c) and (g). For any subset $F' \subseteq F$ and any Majorana associated to a half-edge $\gamma_{[h]_\tau}$, the product $\prod_{f\in F'}S_f$ either contains both $\gamma_{[h]_\tau}$ and $\gamma_{[\lambda h]_\tau}$, or neither of them. Thus the Majoranas appearing in the product $\prod_{f\in F'}S_f$ can be grouped in pairs labeled by edges, which allows us to identify the product with a vector $y \in \mathbb{F}_2^{|E|}$, with $y_e = 1$ for $e = \{h, \lambda h, \tau h, \lambda \tau h\} \in E$, if and only if both $\gamma_{[h]_\tau}$ and $\gamma_{[\lambda h]_\tau}$ appear in the product. It follows that $y = x \Phi$, where we recall that $\Phi\in\mathbb{F}_2^{|F|\times|E|}$ is the face-edge adjacency matrix of the graph embedding, and $x \in \mathbb{F}_2^{|F|}$ satisfying $x_f = 1$ if and only if $f \in F'$. We now return to the proof.
\setmargins
\begin{enumerate}[(a)]
    \item This follows because all the vertex stabilizers are independent from each other, as they have disjoint supports. 
    \item First assume that $V' \subset V$ is a non-empty proper subset. We find an edge $e \in E$ with one endpoint in $V'$ and one in $V \setminus V'$ (we use connectivity of $G$ here), and notice that since face stabilizers either include or exclude both Majoranas on an edge together, there is no way to take a product of face stabilizers to include only one Majorana out of the two on edge $e$. This proves the first part. Now assume that there is at least one odd-degree vertex $v$, and suppose $\prod_{v\in V'}S_v=\prod_{f\in F'}S_f$ for some non-empty subset $V' \subseteq V$ and some subset $F' \subseteq F$. Then by the first part, $V' = V$, so the Majorana $\bar{\gamma}_v$ is present in the product $\prod_{v\in V'}S_v$, which is a contradiction since no face stabilizer contains $\bar{\gamma}_v$.
    \item Suppose $\prod_{v\in V}S_v = \prod_{f\in F'}S_f$ for some subset $F' \subseteq F$. This cannot happen if there is any odd-degree vertex, as the product $\prod_{f\in F'}S_f$ cannot contain any Majorana associated to an odd-degree vertex. On the other hand, $\prod_{v\in V}S_v$ contains every Majorana associated to a half-edge, and so we can identify $\prod_{f\in F'}S_f$ with $\vec 1 \in \mathbb{F}_2^{|E|}$. Thus we have $x \Phi = \vec 1$ where $x \in \mathbb{F}_2^{|F|}$ satisfying $x_f = 1$ if and only if $f \in F'$, so by Lemma~\ref{lem:checkerboardability}(a) we conclude that $G$ is checkerboardable.
    Now suppose that $G$ is checkerboardable. So $G$ has no odd-degree vertices, and by Lemma~\ref{lem:checkerboardability}(b) there exist exactly two distinct vectors $x, x' \in \mathbb{F}_2^{|F|}$ (different from $\vec 0$ and $\vec 1$) such that $x \Phi = x' \Phi = \vec 1$, and $x + x' = \vec 1$. Taking $F' = \{f \in F : x_f = 1\}$ we obtain a non-empty proper subset for which these properties imply $\prod_{v\in V}S_v = \eta \prod_{f\in F'}S_f = \eta' \prod_{f\in F \setminus F'}S_f$ where $\eta, \eta' \in \{\pm 1, \pm i\}$. Now $\eta, \eta' \neq \pm i$ as $\mathcal{S}(G)$ is Hermitian by Lemma~\ref{lem:majorana-group-props}, so it remains to show that $\eta, \eta' \neq -1$. We defer this remaining part to the proof of (g).
    \item This is similar to (b), but for the dual graph embedding $\overline{G}$. Find a face $f'\in F'$ that shares an adjacent edge $e \in E$ with a face $f\not\in F'$ (this uses connectivity of the dual graph), and notice that the Majoranas on edge $e$ appear just once in the product, which can therefore not be identity.
    \item Note that each Majorana $\gamma_{[h]_{\tau}}$ is included in either two face stabilizers (if $h$ and $\tau h$ are in different faces) or none (if $h$ and $\tau h$ are in the same face). Moreover, Majoranas associated to odd-degree vertices are not included in any face stabilizer. Also note that for any flag $h \in H$, $i \gamma_{[h]_{\tau}} \gamma_{[\lambda h]_{\tau}}$ commutes with all Majoranas different from $\gamma_{[h]_{\tau}}$ and $\gamma_{[\lambda h]_{\tau}}$, and squares to $I$. Using these facts, a straightforward calculation shows that $\prod_{f\in F}S_f=I$.
    \item This follows directly from (d) and (e).
    \item Suppose this is false. Then $-I = \prod_{v\in V}S_v \prod_{f\in F'}S_f$ for some non-empty proper subset $F' \subset F$, as all other cases are ruled out by (a), (b), (d), (e), or equivalently $\prod_{v\in V}S_v = - \prod_{f\in F'}S_f$. By the same identification argument used in (c), we conclude that $G$ is checkerboardable. We will now show that in fact $\prod_{v\in V}S_v = \prod_{f\in F'}S_f$, which will complete the proof. Note that this also finishes the remaining part of the proof of (c), because if either $\eta=-1$ or $\eta' = -1$, then $-I \in \mathcal{S}(G)$. 
    
    Since $G$ is checkerboardable, it has no odd-degree vertices, and each edge in $E$ is adjacent to exactly one face in $F'$. Thus let us identify the Majoranas by their labels: so define $\gamma_j := \gamma_{[h]_{\tau}}$ if $\gamma_{[h]_{\tau}}$ has label $j$. Let $\tilde{v} \in V$ be the vertex that contains the Majoranas $\gamma_0$ and $\gamma_{2|E|-1}$ in its vertex stabilizer. Then for each $v \in V \setminus \tilde{v}$, one has $S_v = \prod_{j \in I_v} (i \gamma_{2j+1} \gamma_{2j+2})$, and $S_{\tilde{v}} = (i \gamma_0) \left( \prod_{j \in I_{\tilde{v}}} i \gamma_{2j+1} \gamma_{2j+2} \right)\gamma_{2|E|-1}$, where for every $v \in V$, $I_v \subseteq \{0 ,\dots, |E| - 2\}$ contains those integers $j$ such that $\gamma_{2j+1}$ is associated to a half-edge that is a subset of $v$. Noticing that $(i \gamma_{2j+1} \gamma_{2j+2})$ commutes with all Majoranas different from $\gamma_{2j+1}$ and $\gamma_{2j+2}$, we immediately get $\prod_{v\in V}S_v = i^{|E|} \gamma_0 \gamma_1 \dots \gamma_{2|E|-1}$. Similarly, each face stabilizer can be written as $S_f = \prod_{j \in I_f} i \gamma_{2j} \gamma_{2j+1}$ where $I_f \subseteq \{0 ,\dots, |E| - 1\}$ contains those integers $j$ for which $\gamma_{2j}$ belongs to an edge adjacent to $f$. It follows that $\prod_{f\in F'}S_f = i^{|E|} \gamma_0 \gamma_1 \dots \gamma_{2|E|-1}$, completing the proof.
\end{enumerate}
\end{proof}

}

%% file: appendix/majorana_qubit_code_equivalence_v1.tex
\section{Equivalence of Majorana and qubit surface codes}
\label{app:majorana-qubit-code-equivalence}

The goal of this section is to prove the equivalence between the Majorana surface code, defined on $m=M+ 2|E|$ Majoranas (where $m$ is even), and the qubit surface code, defined on $N=M/2 + |E| - |V|$ qubits, as explained in Section~\ref{subsec:qub_surface_codes} (the quantities $M, |E|, |V|$ are defined in Section~\ref{subsec:maj_codes_on_graphs}). The proof relies on some general facts that we also present here. We will use some of the notations and definitions introduced in Section~\ref{subsec:pauli-definitions}. Recall from Section~\ref{subsec:maj_ops_intro} that the Jordan-Wigner transformation identifies the Pauli group $\mathcal{P}_{m/2}$ with the Majorana group $\mathcal{J}_{m}$, in a way that preserves commutation relations. To remain consistent with the notation used for Paulis, define $\text{Com}: \mathcal{J}_{m} \times \mathcal{J}_{m} \rightarrow \mathbb{F}_2$ as $\com{\gamma}{\gamma'} = \com{\text{JW}(\gamma)}{\text{JW}(\gamma')}$, where $\gamma, \gamma' \in \mathcal{J}_{m}$, and $\text{JW}(\cdot)$ is the Jordan-Wigner map (the phase factors of the elements of $\mathcal{J}_{m}$ have been absorbed into $\gamma, \gamma'$ and we continue to do so for the rest of this section\footnote{This is in contrast to the notation in Section~\ref{subsec:maj_ops_intro} where the phase factors were made explicit.}). Due to the Jordan-Wigner identification, Lemma~\ref{lem:majorana-group-props} is also valid for the Pauli group. To simplify notation, let us define $r := m/2$.

Recall that a stabilizer group $\mathcal{S}$ on $r$-qubits is a subgroup of $\mathcal{P}_r$ satisfying $-I \not \in \mathcal{S}$, which implies that $p^2 = I$ for all $p \in \mathcal{S}$, and $\mathcal{S}$ is a commuting subgroup. Let $\mathcal{C}(\mathcal{S})$ be the centralizer of $\mathcal{S}$. If $\dim(\mathcal{S})=k$, for some $0 \leq k \leq r$, then $\dim(\mathcal{C}(\mathcal{S}))=2r-k$. Let $\mathcal{S}'$ be another stabilizer group of $\mathcal{P}_r$ such that $\mathcal{S} \subseteq \mathcal{S}'$. Then by the definition of the centralizer we have $\mathcal{S}' \subseteq \mathcal{C}(\mathcal{S})$. Now suppose that $\kappa : \mathcal{C}(\mathcal{S}) \rightarrow \mathcal{P}_N$ is a homomorphism. Notice that $\pm iI \not \in \kappa(\mathcal{S}')$, because if otherwise $\kappa(p) = \pm iI$ for some $p \in \mathcal{S}'$, then it implies $\kappa(I) = \kappa(p^2) = (\kappa(p))^2 = -I$. We say that $\kappa$ is \textit{commutation preserving} if $\com{p}{q} = \com{\kappa(p)}{\kappa(q)}$ for all $p,q \in \mathcal{C}(\mathcal{S})$. Let $\pi : \mathcal{P}_N \rightarrow \pauligrp{N}$ be the projection map taking each element of $\mathcal{P}_N$ to its equivalence class in $\pauligrp{N}$. It is easily checked that $\pi$ is a surjective homomorphism, so $\kappa$ descends to a homomorphism $\hat{\kappa} := \pi \circ \kappa : \mathcal{C}(\mathcal{S}) \rightarrow \pauligrp{N}$. The homomorphism $\pi$ is commutation preserving in the sense that for all $p,q \in \mathcal{P}_N$, $\com{p}{q} = \com{[p]}{[q]}$. Thus if $\kappa$ is commutation preserving also, then so is $\hat{\kappa}$, as for all $p,q \in \mathcal{C}(\mathcal{S})$, $\com{p}{q} = \com{\hat{\kappa}(p)}{\hat{\kappa}(q)}$. Morover, if $\kappa$ is surjective, it follows by surjectivity of $\pi$ that $\hat{\kappa}$ is surjective. Below in Lemma~\ref{lem:commutation-preserving-props}, we prove a result addressing the converse direction.

\begin{lemma}
\label{lem:commutation-preserving-props}
Let $\mathcal{S}$ be a stabilizer subgroup of $\mathcal{P}_r$, and let $\mu : \mathcal{C}(\mathcal{S}) \rightarrow \pauligrp{N}$ be a surjective, commutation preserving homomorphism. Then there exists a surjective homomorphism $\kappa : \mathcal{C}(\mathcal{S}) \rightarrow \mathcal{P}_N$ such that $\hat{\kappa} = \mu$, and $\kappa(\mathcal{S})$ is a stabilizer subgroup of $\mathcal{P}_N$. Moreover, for any such $\kappa$, if $\mathcal{S}'$ is a stabilizer subgroup of $\mathcal{P}_r$ satisfying $\mathcal{S} \subseteq \mathcal{S}'$, the following holds:
\setmargins
\begin{enumerate}[(a)]
    \item $\kappa$ is commutation preserving.
    \item $\kappa(p) \not \in \langle iI \rangle$ for all $p \in \mathcal{C}(\mathcal{S}) \setminus \langle iI,\mathcal{S} \rangle$.
    \item The kernel of $\kappa$ is $\mathcal{S}$, and $\kappa(iI) = \pm iI$.
    \item $\kappa(\mathcal{S}')$ is a stabilizer subgroup of $\mathcal{P}_N$.
    \item $\kappa(\mathcal{C}(\mathcal{S}')) = \mathcal{C}(\kappa(\mathcal{S}'))$.
    \item $\mathcal{T} / \mathcal{S} \cong \kappa(\mathcal{T})$, and $\dim(\kappa(\mathcal{T})) = \dim(\mathcal{T}) - \dim(\mathcal{S})$, where $\mathcal{T}$ is either $\mathcal{S}'$ or $\mathcal{C}(\mathcal{S}')$. The stabilizer codes given by the stabilizer groups $\mathcal{S}'$ in $\mathcal{P}_r$, and $\kappa(\mathcal{S}')$ in $\mathcal{P}_N$ encode the same number of logical qubits.
\end{enumerate}
\end{lemma}

\begin{proof}
For the proof, we first note some useful facts. Since $\mu$ is surjective and commutation preserving, we have $\mu(iI)=[I]$. This follows because if $\mu(iI) \neq [I]$, then by surjectivity there exists $p \in \mathcal{C}(\mathcal{S})$ such that $\com{\mu(p)}{\mu(iI)}=1$, while $\com{p}{iI}=0$ giving a contradiction. Also recall that for any $p \in \mathcal{C}(\mathcal{S}) \setminus \langle iI, \mathcal{S}\rangle$, there exists $q \in \mathcal{C}(\mathcal{S}) \setminus \langle iI, \mathcal{S}\rangle$ such that $\com{p}{q}=1$. Now let $\dim(\mathcal{S}) = k \geq 0$. Then there exists an independent list $\mathcal{G} = \{p_{1},\dots,p_{k}, p^{'}_{1},\dots,p^{'}_{2r-2k}\}$  of commuting, Hermitian elements of $\mathcal{P}_r$, 
such that $\mathcal{S} = \langle p_{1},\dots,p_{k} \rangle$, and $\mathcal{C}(\mathcal{S}) = \langle iI, \mathcal{G} \rangle$. Thus every element $p \in \mathcal{C}(\mathcal{S})$ can be uniquely represented in the form $p = (iI)^{\alpha_p} \prod_{y \in \mathcal{H}_p} y$, where $\alpha_p \in \{0,1,2,3\}$, and $\mathcal{H}_p \subseteq \mathcal{G}$ are uniquely determined by $p$. The elements in this product representation appear in the same order as in the list $\mathcal{G}$, with the convention that the product of the empty subset is $I$. 

We first show the existence of $\kappa$. Let us define a map $\psi : \mathcal{G} \rightarrow \mathcal{P}_N$ by letting $\psi(x)$ be any arbitrarily chosen Hermitian element in the equivalence class $\mu(x)$ for each $x \in \mathcal{G}$, and consider the list $\mathcal{T} = \{ \psi(p_1),\dots,\psi(p_k)\}$. Since $\mathcal{S}$ is a commuting subgroup and $\mu$ is a homomorphism, it follows that $\langle \mathcal{T} \rangle$ is a commuting subgroup. If $k \geq 1$, then $\mathcal{T}$ is non-empty, and so by Lemma~\ref{lem:majorana-group-props} we conclude $\langle \mathcal{T} \rangle$ is Hermitian, and the elements of $\mathcal{T}$ can be multiplied by $\pm 1$ to obtain a new list $\tilde{\mathcal{T}} = \{ \tilde{\psi}(p_1),\dots,\tilde{\psi}(p_k)\}$ such that $-I \not \in \langle \tilde{\mathcal{T}} \rangle$. Next define the map $\kappa: \mathcal{G} \cup \{iI\} \rightarrow \pauligrp{N}$ as follows
\begin{equation}
\label{eq:kappa-def}
\kappa(x) =
\begin{cases}
\tilde{\psi} (x) &\;\;\; \text{if} \;\; x = p_{1},\dots,p_{k} \\
\psi(x) &\;\;\; \text{if} \;\; x = p^{'}_{1},\dots,p^{'}_{2r-2k} \\
iI &\;\;\; \text{if} \;\; x = iI,
\end{cases}
\end{equation}
and we claim that $\kappa$ extends uniquely to a homomorphism $\kappa: \mathcal{C}(\mathcal{S}) \rightarrow \mathcal{P}_N$. For any such extension, if $x = (iI)^{\alpha_x} \prod_{y \in \mathcal{H}_x} y \in \mathcal{C}(\mathcal{S})$, then $\kappa(x) = (iI)^{\alpha_x} \prod_{y \in \mathcal{H}_x} \kappa(y)$, showing uniqueness. To prove existence, for any $x = (iI)^{\alpha_x} \prod_{y \in \mathcal{H}_x} y \in \mathcal{C}(\mathcal{S})$, define $\kappa(x) := (\kappa(iI))^{\alpha_x} \prod_{y \in \mathcal{H}_x} \kappa(y)$. This definition satisfies Eq.~\eqref{eq:kappa-def}, so we need to verify that $\kappa$ is a homomorphism. Let $x, x' \in \mathcal{C}(\mathcal{S})$ be given by $x = (iI)^{\alpha_x} \prod_{y \in \mathcal{H}_x} y$, and $x' = (iI)^{\alpha_{x'}} \prod_{y \in \mathcal{H}_{x'}} y$. Then $x'' := x x' = \left( (iI)^{\alpha_x} \prod_{y \in \mathcal{H}_x} y \right) \; \left( (iI)^{\alpha_{x'}} \prod_{y \in \mathcal{H}_{x'}} y \right) = (iI)^{\alpha_{x''}} \prod_{y \in \mathcal{H}_{x''}} y$, where $\mathcal{H}_{x''} =  (\mathcal{H}_{x} \cup \mathcal{H}_{x'}) \setminus (\mathcal{H}_{x} \cap \mathcal{H}_{x'})$, and $\alpha_{x''} = (\alpha_x + \alpha_{x'} + 2t) \mod 4$, with $t$ determined by the commutativity of the elements of $\mathcal{H}_{x}$ and $\mathcal{H}_{x'}$. Now from the definition of $\kappa$, we get $\kappa(x x') = (iI)^{\alpha_{x''}} \prod_{y \in \mathcal{H}_{x''}} \kappa(y)$. We also observe that since $\mu$ is commutation preserving, Eq.~\eqref{eq:kappa-def} implies $\com{x_1}{x_2} = \com{\kappa(x_1)}{\kappa(x_2)}$, for all $x_1, x_2 \in \mathcal{G} \cup \{iI\}$. This gives $\kappa(x) \kappa(x') = \left( (iI)^{\alpha_x} \prod_{y \in \mathcal{H}_x} \kappa(y) \right) \; \left( (iI)^{\alpha_{x'}} \prod_{y \in \mathcal{H}_{x'}} \kappa(y) \right) = (iI)^{\alpha_{x''}} \prod_{y \in \mathcal{H}_{x''}} \kappa(y)$, where in the last equality we have used that all elements of $\kappa(\mathcal{G})$ square to $I$ (as they are Hermitian by choice). This shows $\kappa(x x')=\kappa(x)\kappa(x')$ proving the claim. Next observe that by construction $\kappa(\mathcal{S}) = \langle \tilde{\mathcal{T}} \rangle$, and so $\kappa(\mathcal{S})$ is a stabilizer subgroup of $\mathcal{P}_N$. We moreover have $\hat{\kappa} = \mu$, because for any $x = (iI)^{\alpha_x} \prod_{y \in \mathcal{H}_x} y \in \mathcal{C}(\mathcal{S})$, we have $\hat{\kappa}(x) = \prod_{y \in \mathcal{H}_x} \hat{\kappa}(y) = \prod_{y \in \mathcal{H}_x} \mu (y) = \mu (\prod_{y \in \mathcal{H}_x} y) = \mu(x)$, where the last equality uses $\mu(iI)=[I]$. Finally the condition $\kappa(iI)=iI$ in Eq.~\eqref{eq:kappa-def}, together with the fact that $\mu$ is surjective, and $\hat{\kappa} = \mu$, imply that $\kappa$ is surjective.

We now prove parts (a)-(f) of the lemma, assuming the existence of $\kappa$.
\setmargins
\begin{enumerate}[(a)]
    \item For any two elements $p,q \in \mathcal{C}(\mathcal{S})$, since $\hat{\kappa}=\mu$ and both $\mu$ and $\pi$ are commutation preserving homomorphisms, we have $\com{p}{q}=\com{\hat{\kappa}(p)}{\hat{\kappa}(q)} = \com{\kappa(p)}{\kappa(q)}$.
    \item If $p \in \mathcal{C}(\mathcal{S}) \setminus \langle iI,\mathcal{S} \rangle$, then there exists $q \in \mathcal{C}(\mathcal{S}) \setminus \langle iI,\mathcal{S} \rangle$ such that $\com{p}{q}=1$. As $\kappa$ is commutation preserving by (a), we also have $\com{\kappa(p)}{\kappa(q)}=1$, which is impossible if $\kappa(p) \in \langle iI \rangle$.
    \item The center of $\mathcal{C}(\mathcal{S})$ is $\langle iI, \mathcal{S} \rangle$, and the center of $\mathcal{P}_N$ is $\langle iI \rangle$. Since $\kappa$ is a surjective homomorphism, this implies $\kappa(\langle iI, \mathcal{S} \rangle) \subseteq \langle iI \rangle$. Now fix any $p \in \mathcal{S}$. We have already argued in the paragraph preceding the lemma that $\kappa(p) \neq \pm iI$. Since $\kappa(\mathcal{S})$ is a stabilizer group, we also have $\kappa(p) \neq -I$. This shows $\kappa(\mathcal{S})=\langle I \rangle$ (equivalently $\mathcal{S}$ is in the kernel of $\kappa$). Now by the surjectivity of $\kappa$ and (b), it follows that $\kappa(\langle iI, \mathcal{S} \rangle) =\langle iI \rangle$, and since $\mu(iI) = [I]$, we also have $\kappa(iI) \in \langle iI \rangle$. If $\kappa(iI) = \pm I$, then $\kappa(\langle iI, \mathcal{S} \rangle) = \langle -I \rangle$, which is a contradiction. This proves $\kappa(iI) = \pm iI$. In each of these two possible cases, since $\kappa(\mathcal{S}) = \{I\}$, we have for any $p \in \langle iI,\mathcal{S} \rangle \setminus \mathcal{S}$, $\kappa(p) \neq I$, and this proves that the kernel of $\kappa$ is exactly $\mathcal{S}$.
    \item It suffices to show that $-I \not \in \kappa(\mathcal{S}')$, as then Lemma~\ref{lem:majorana-group-props}(b) implies $\kappa(\mathcal{S}')$ is also commuting and Hermitian. If $\mathcal{S}' = \mathcal{S}$ then $\kappa(\mathcal{S}')=I$ by (c). Now assume $\mathcal{S}' \neq \mathcal{S}$, and for contradiction, assume there exists $p \in \mathcal{S}'$ such that $\kappa(p) = -I$. Then by (c), $p \in \mathcal{S}' \setminus \mathcal{S}$, and since $\mathcal{S}'$ is a stabilizer group it must also hold that $p \not \in \langle iI,\mathcal{S}\rangle$. But this contradicts (b).
    \item As $\kappa$ is a homomorphism, we have $\kappa(\mathcal{C}(\mathcal{S}')) \subseteq \mathcal{C}(\kappa(\mathcal{S}'))$. We want to show $\mathcal{C}(\kappa(\mathcal{S}')) \subseteq \kappa(\mathcal{C}(\mathcal{S}'))$. Pick any $p \in \mathcal{C}(\kappa(\mathcal{S}'))$. By surjectivity of $\kappa$, there exists $q \in \mathcal{C}(\mathcal{S})$ such that $\kappa(q) = p$. Then for all $x \in \mathcal{S}'$, since $\kappa$ is commutation preserving, we have $\com{q}{x} = \com{p}{\kappa(x)} = 0$, which implies $q \in \mathcal{C}(\mathcal{S}')$.
    \item Since $\mathcal{S} \subseteq \mathcal{S}'$, we have $\mathcal{S} \subseteq \mathcal{S}' \subseteq \mathcal{C}(\mathcal{S}') \subseteq \mathcal{C}(\mathcal{S})$. Let $\mathcal{T}$ be either $\mathcal{S}'$ or $\mathcal{C}(\mathcal{S}')$, and let $\kappa_{\mathcal{T}} : \mathcal{T} \rightarrow \kappa(\mathcal{T})$ be the restriction of $\kappa$ to $\mathcal{T}$. Then $\kappa_{\mathcal{T}}$ is surjective, and the kernel of $\kappa_{\mathcal{T}}$ is $\mathcal{S}$. Thus we have the isomorphism $\mathcal{T} / \mathcal{S} \cong \kappa(\mathcal{T})$, which implies $\dim(\kappa(\mathcal{T})) = \dim(\mathcal{T} / \mathcal{S}) = \dim(\mathcal{T}) - \dim(\mathcal{S})$. Using this, we easily obtain the equality $\dim(\kappa(\mathcal{C}(\mathcal{S}'))) - \dim(\kappa(\mathcal{S}')) = \dim(\mathcal{C}(\mathcal{S}')) - \dim(\mathcal{S}')$, where the left (resp.~right) hand side of the equation denotes twice the number of encoded qubits by the stabilizer code specified by $\mathcal{\kappa(\mathcal{S}')}$ (resp.~$\mathcal{S}'$).
\end{enumerate}
\end{proof}

Similarly, a homomorphism $\mu: \mathcal{G} \rightarrow \pauligrp{N}$, where $\mathcal{G}$ is a subgroup of $\mathcal{J}_{m}$, is called commutation preserving if for all $\gamma, \gamma' \in \mathcal{G}$, $\com{\mu(\gamma)}{\mu(\gamma')} = \com{\gamma}{\gamma'}$. A simple fact is the following: if $\mathcal{T} \subseteq \mathcal{G}$ such that $\langle \mathcal{T} \rangle = \mathcal{G}$, and $\mu : \mathcal{G} \rightarrow \pauligrp{N}$ is a homomorphism such that $\com{\gamma}{\gamma'}=\com{\mu(\gamma)}{\mu(\gamma')}$ for all $\gamma,\gamma' \in \mathcal{T}$, then $\mu$ is commutation preserving on $\mathcal{G}$. To see this, if $\gamma,\gamma' \in \mathcal{G}$, there exists $\gamma_1,\dots,\gamma_{k_\gamma},\gamma'_1,\dots,\gamma'_{k_{\gamma'}} \in \mathcal{T}$ such that $\gamma = \prod_{s=1}^{k_\gamma} \gamma_s$, and $\gamma' = \prod_{t=1}^{k_{\gamma'}} \gamma'_t$. Then $\com{\mu(\gamma)}{\mu(\gamma')} = \com{\prod_{s=1}^{k_\gamma} \mu(\gamma_s)}{\prod_{t=1}^{k_{\gamma'}} \mu(\gamma'_t)} = \sum_{s=1}^{k_\gamma} \sum_{t=1}^{k_{\gamma'}} \com{\mu(\gamma_s)}{\mu(\gamma'_t)} = \sum_{s=1}^{k_\gamma} \sum_{t=1}^{k_{\gamma'}} \com{\gamma_s}{\gamma'_t} = \com{\gamma}{\gamma'}$, where the first equality is true since $\mu$ is a homomorphism, and the sums are evaluated over $\mathbb{F}_2$.

Let us discuss a consequence of Majorana and qubit partitioning on $\mathcal{J}_{m}$ and $\pauligrp{N}$, respectively. For $m,N \geq \ell \geq 1$, let $J_1,\dots,J_{\ell}$ (resp.~$J'_1,\dots,J'_{\ell}$) be a partition of $m$ Majoranas (resp.~$N$ qubits), and denote by $\mathcal{J}_{m,J_k}$ (resp.~$\pauligrp{N,J'_k}$) the subgroup of $\mathcal{J}_{m}$ (resp.~$\pauligrp{N}$) formed by those elements of $\mathcal{J}_{m}$ (resp.~$\pauligrp{N}$) whose support is contained in $J_k$ (resp.~$J'_k$). We then have $\mathcal{J}_m = \langle \mathcal{J}_{m,J_1}, \dots, \mathcal{J}_{m,J_\ell} \rangle$ and the direct sum $\pauligrp{N} = \bigoplus_{k=1}^{\ell} \pauligrp{N,J'_k}$. Recall that a stabilizer group $\mathcal{S}$ of $\mathcal{J}_m$ satisfies $-I \not \in \mathcal{S}$ such that all elements of $\mathcal{S}$ have even weight. Now suppose we are given stabilizer groups $\mathcal{S}_1,\dots,\mathcal{S}_{\ell}$  of $\mathcal{J}_m$, such that for each $k$, $\mathcal{S}_k \subseteq \mathcal{J}_{m,J_k}$, and moreover $\mathcal{S} = \langle \mathcal{S}_1,\dots,\mathcal{S}_{\ell} \rangle$ is a stabilizer group. Then the centralizer of $\mathcal{S}$ satisfies $\mathcal{C}(\mathcal{S}) = \langle \mathcal{C}(\mathcal{S}_1) \cap \mathcal{J}_{m,J_1}, \dots, \mathcal{C}(\mathcal{S}_\ell) \cap \mathcal{J}_{m,J_\ell} \rangle$. To prove this, note that for any fixed $k$, if $\gamma \in \mathcal{C}(\mathcal{S}_k) \cap \mathcal{J}_{m,J_k}$, then from the support of $\gamma$ we deduce using Eq.~\eqref{eq:majorana-commutation} that $\gamma \in \mathcal{C}(\mathcal{S})$, showing $\langle \mathcal{C}(\mathcal{S}_1) \cap \mathcal{J}_{m,J_1}, \dots, \mathcal{C}(\mathcal{S}_\ell) \cap \mathcal{J}_{m,J_\ell} \rangle \subseteq \mathcal{C}(\mathcal{S})$. Next let $\gamma \in \mathcal{C}(\mathcal{S})$. As $\gamma$ is an element of $\mathcal{J}_m$, we have a decomposition $\gamma = \prod_{k=1}^{\ell} \gamma_k$, where for each $k$ we have $\gamma_k \in \mathcal{J}_{m,J_k}$. Fixing a $k$, we note that for all $\gamma' \in \mathcal{S}_k$, $\com{\gamma}{\gamma'}=0$ implies $\com{\gamma_k}{\gamma'}=0$ (this uses $\com{\gamma_{k'}}{\gamma'}=0$ for all $k' \neq k$ from support conditions), and this shows $\mathcal{C}(\mathcal{S}) \subseteq \langle \mathcal{C}(\mathcal{S}_1) \cap \mathcal{J}_{m,J_1}, \dots, \mathcal{C}(\mathcal{S}_\ell) \cap \mathcal{J}_{m,J_\ell} \rangle$. This discussion is useful in proving the following lemma:

\begin{lemma}
\label{lem:disjoint-supp-props}
Suppose we are given
\setmargins
\begin{enumerate}[(i)]
    \item $m, N \geq \ell \geq 1$, $J_1,\dots,J_{\ell}$ is a partition of $m$ Majoranas, and $J'_1,\dots,J'_{\ell}$ is a partition of $N$ qubits.
    \item $\mathcal{S}_1,\dots,\mathcal{S}_{\ell}$ are stabilizer groups of $\mathcal{J}_m$, satisfying $\mathcal{S}_k \subseteq \mathcal{J}_{m,J_k}$ for all $k$, and  $\mathcal{S} = \langle \mathcal{S}_1,\dots,\mathcal{S}_{\ell} \rangle$ is a stabilizer group.
    \item For all $k$, and for all $\gamma \in \mathcal{C}(\mathcal{S}_k) \cap \mathcal{J}_{m,J_k}$, $\gamma$ has even weight.
    \item $\mu: \mathcal{C}(\mathcal{S}) \rightarrow \pauligrp{N}$ is a map with the properties: (a) $\mu(\mathcal{C}(\mathcal{S}_k) \cap \mathcal{J}_{m,J_k}) \subseteq \hat{\mathcal{P}}_{N,J'_k}$ for each $k$, and (b) if $\gamma = \prod_{k=1}^{\ell} \gamma_{k}$ with each $\gamma_{k} \in \mathcal{C}(\mathcal{S}_{k}) \cap \mathcal{J}_{m,J_{k}}$, then $\mu(\gamma)=\prod_{k=1}^{\ell} \mu(\gamma_{k})$.
\end{enumerate}
\vspace{0.1cm}
Denote by $\mu_k$ the restriction of $\mu$ to $\mathcal{C}(\mathcal{S}_k) \cap \mathcal{J}_{m,J_k}$. Then the following hold:
\setmargins
\begin{enumerate}[(a)]
    \item Two distinct $\mu$ cannot give rise to the same family $\{\mu_k\}_{k=1}^{\ell}$.
    \item If $\ell > 1$, then for each $k$, we have $\mu_k(iI)=[I]$, and $\mu_k(\gamma) = \mu_k(\gamma')$ whenever $\gamma, \gamma' \in \mathcal{C}(\mathcal{S}_k) \cap \mathcal{J}_{m,J_k}$, and $\gamma, \gamma'$ are equivalent up to phase factors. Moreover, if we are given a family of maps $\mu_k : \mathcal{C}(\mathcal{S}_k) \cap \mathcal{J}_{m,J_k} \rightarrow \hat{\mathcal{P}}_{N,J'_k}$ for $1 \leq k \leq \ell$, satisfying these two conditions, then there exists a corresponding $\mu$ satisfying (iv), and it is unique.
    \item $\mu$ is a surjective, commutation preserving homomorphism if and only if $\mu_k: \mathcal{C}(\mathcal{S}_k) \cap \mathcal{J}_{m,J_k} \rightarrow \hat{\mathcal{P}}_{N,J'_k}$ is a surjective, commutation preserving homomorphism, for each $k$.
\end{enumerate}
\end{lemma}

\vspace{0.5cm}
\begin{proof}
\setmargins
\begin{enumerate}[(a)]
    \item Since $\mathcal{C}(\mathcal{S}) = \langle \mathcal{C}(\mathcal{S}_1) \cap \mathcal{J}_{m,J_1}, \dots, \mathcal{C}(\mathcal{S}_\ell) \cap \mathcal{J}_{m,J_\ell} \rangle$, it implies that if $\gamma \in \mathcal{C}(\mathcal{S})$, then $\gamma = \prod_{k=1}^{\ell} \gamma_{k}$ with each $\gamma_{k} \in \mathcal{C}(\mathcal{S}_{k}) \cap \mathcal{J}_{m,J_{k}}$. Then by property (b) of (iv), we get $\mu(\gamma) = \prod_{k=1}^{\ell} \mu_k(\gamma_k)$ showing that distinct $\mu$ cannot give rise to the same family $\{\mu_k\}_{k=1}^{\ell}$.
    \item If there exists $k$ such that $\mu_k(iI) \neq [I]$, then property (a) of (iv) does not hold. Now suppose there exists $k$ and $\gamma_k, \gamma'_k \in \mathcal{C}(\mathcal{S}_k) \cap \mathcal{J}_{m,J_k}$, such that $\gamma_k, \gamma'_k$ are distinct and equivalent up to phase, and $\mu_k(\gamma_k) \neq \mu_k(\gamma'_k)$. Consider any $\gamma \in \mathcal{C}(\mathcal{S})$ given by $\gamma = \prod_{s=1}^{\ell} \gamma_{s}$, where for all $s \neq k$, $\gamma_{s} \in \mathcal{C}(\mathcal{S}_{s}) \cap \mathcal{J}_{m,J_{s}}$ are chosen arbitrarily. By adjusting the phases, one also has $\gamma = \prod_{s=1}^{\ell} \gamma'_s$ with $\gamma_s, \gamma'_s$ equivalent up to phase for all $s$. Then by property (b) of (iv) we get $\prod_{s=1}^{\ell} \mu_s(\gamma_s) = \prod_{s=1}^{\ell} \mu_s(\gamma'_s)$, and the support conditions imply $\mu_k(\gamma_k) = \mu_k(\gamma'_k)$ giving a contradiction. To prove the second part, assume we are given a family $\{\mu_k\}_{k=1}^{\ell}$ satisfying the two conditions of the first part. Take any $\gamma \in \mathcal{C}(\mathcal{S})$, and suppose $\gamma = \prod_{k=1}^{\ell} \gamma_k$ and $\gamma = \prod_{k=1}^{\ell} \gamma'_k$ are two different representations of $\gamma$, where $\gamma_k, \gamma'_k \in \mathcal{C}(\mathcal{S}_k) \cap \mathcal{J}_{m,J_k}$ for all $k$. From support conditions, we conclude that $\gamma_k, \gamma'_k$ are equivalent up to phase for each $k$. Thus  $\prod_{k=1}^{\ell} \mu_k(\gamma_k) = \prod_{k=1}^{\ell} \mu_k(\gamma'_k)$, and so we can define $\mu(\gamma) := \prod_{k=1}^{\ell} \mu_k(\gamma_k)$ unambiguously. This definition of $\mu$ satisfies (iv), and it is unique by (a).
    \item First assume that $\mu$ is a surjective, commutation preserving homomorphism, and fix any $k$. Then if $\gamma,\gamma' \in \mathcal{C}(\mathcal{S}_k) \cap \mathcal{J}_{m,J_k}$, we have $\com{\mu_k(\gamma)}{\mu_k(\gamma')}=\com{\mu(\gamma)}{\mu(\gamma')}=\com{\gamma}{\gamma'}$, and also $\mu_k(\gamma \gamma') = \mu(\gamma \gamma') = \mu(\gamma) \mu(\gamma') = \mu_k(\gamma) \mu_k(\gamma')$, showing that $\mu_k$ is a commutation preserving homomorphism. Next let $[q] \in \hat{\mathcal{P}}_{N,J'_k}$. By surjectivity of $\mu$, we have $\gamma \in \mathcal{C}(\mathcal{S})$ such that $\mu(\gamma) = [q]$. Since we can write $\gamma = \prod_{s=1}^{\ell} \gamma_{s}$ with $\gamma_{s} \in \mathcal{C}(\mathcal{S}_{s}) \cap \mathcal{J}_{m,J_{s}}$ for each $s$, it implies $\mu(\gamma) = \prod_{s=1}^{\ell} \mu(\gamma_{s}) = [q]$. But from property (a) of (iv), it must be that  $\mu(\gamma_{s}) = [I]$ for all $s \neq k$, and so $\mu(\gamma_k)=[q]$ showing that $\mu_k$ is surjective.
    For the converse direction, suppose $\mu_k: \mathcal{C}(\mathcal{S}_k) \cap \mathcal{J}_{m,J_k} \rightarrow \hat{\mathcal{P}}_{N,J'_k}$ is a surjective, commutation preserving homomorphism for all $k$. Take $\gamma,\gamma' \in \mathcal{C}(\mathcal{S})$, and we can write $\gamma = \prod_{k=1}^{\ell} \gamma_{k}$, and $\gamma' = \prod_{k=1}^{\ell} \gamma'_{k}$, with $\gamma_{k}, \gamma'_{k} \in \mathcal{C}(\mathcal{S}_{k}) \cap \mathcal{J}_{m,J_{k}}$ for each $k$, and thus $\gamma \gamma' = \pm \prod_{k=1}^{\ell} \gamma_{k} \gamma'_{k}$. It follows that $\mu(\gamma \gamma') = \mu(\pm \gamma_1 \gamma'_1) \prod_{k=2}^{\ell} \mu(\gamma_k \gamma'_k) = \prod_{k=1}^{\ell} \mu_k (\gamma_k \gamma'_k) = \prod_{k=1}^{\ell} \mu_k(\gamma_k) \mu_k (\gamma'_k)$, where the first equality is by property (b) of (iv), and the second equality uses that $\mu_1(iI)=[I]$ as $\mu_1$ is a surjective, commutation preserving homomorphism (see the first paragraph of the proof of Lemma~\ref{lem:commutation-preserving-props} for a similar reasoning). Similarly, $\mu(\gamma) = \prod_{k=1}^{\ell} \mu_k(\gamma_k)$ and $\mu(\gamma') = \prod_{k=1}^{\ell} \mu_k(\gamma'_k)$, from which we get $\mu(\gamma) \mu(\gamma') = \mu(\gamma \gamma')$. We also have $\com{\mu(\gamma)}{\mu(\gamma')} = \sum_{k=1}^{\ell} \com{\mu_k(\gamma_k)}{\mu_k(\gamma'_k)} = \sum_{k=1}^{\ell} \com{\gamma_k}{\gamma'_k}=\com{\gamma}{\gamma'}$, where the last equality is by (iii). This shows that $\mu$ is a commutation preserving homomorphism. To prove $\mu$ is surjective, let $[q] \in \pauligrp{N}$. Then by the direct sum representation of $\pauligrp{N}$, one can write $[q] = \prod_{k=1}^{\ell} [q_k]$ with $[q_k] \in \hat{\mathcal{P}}_{N,J'_k}$ for all $k$. By surjectivity of $\mu_k$, we have that for all $k$, there exists $\gamma_k \in \mathcal{C}(\mathcal{S}_{k}) \cap \mathcal{J}_{m,J_{k}}$ satisfying $\mu_k(\gamma_k) = [q_k]$. Defining $\gamma = \prod_{k=1}^{\ell} \gamma_k$, gives $\mu(\gamma) = \prod_{k=1}^{\ell} \mu_k(\gamma_k) = \prod_{k=1}^{\ell} [q_k] = [q]$.
\end{enumerate}
\end{proof}

Using the lemmas and facts presented above in this section, it is now easy to show that the Majorana and the qubit surface codes of Section~\ref{subsec:maj_codes_on_graphs} and Section~\ref{subsec:qub_surface_codes} respectively, encode the same number of qubits. Recall that the Majorana code is defined by the stabilizer group $\mathcal{S}(G) \subseteq \mathcal{J}_m$, generated by the vertex stabilizers $S_v$ and the face stabilizers $S_f$, for all $v \in V$ and $f \in F$, as defined by Eqs.~{\eqref{eq:s_v}, \eqref{eq:s_f}}. We partition the Majoranas into $|V|$ subsets $\{J_v\}_{v \in V}$, with each subset indexed by a vertex $v \in V$. If $v$ is an even degree vertex, then $J_v= \{\gamma_{[h]_\tau}: [h]_\tau \subseteq v\}$, and if $v$ is an odd degree vertex, then $J_v= \{\gamma_{[h]_\tau}: [h]_\tau \subseteq v\} \cup \{\bar{\gamma}_v\}$. Similarly, we also partition the $N$ qubits into $|V|$ subsets $\{J'_v\}_{v \in V}$, where $J'_v$ consists of the $N_v$ qubits placed at $v$ for the qubit surface code (see Definition~\ref{def:qubit_surface_code}). Using this, we complete the proof of Corollary~\ref{cor:number_encoded_qubits} below.

\begin{proof}[Proof of Corollary~\ref{cor:number_encoded_qubits}]
Let $\mathcal{S}_v$ be the stabilizer group generated by $S_v$, and note that $\mathcal{S}_v \subseteq \mathcal{J}_{m,J_v}$ for all $v$. Let $\mathcal{S}_{\text{maj}} = \langle \{\mathcal{S}_v : v \in V \} \rangle$, and it is a stabilizer group as $\mathcal{S}_{\text{maj}}  \subseteq \mathcal{S}(G)$. Define $m_v = |J_v|$ and note that $m_v$ is even. If $\gamma \in \mathcal{C}(\mathcal{S}_v) \cap \mathcal{J}_{m,J_v}$ then $|\supp{\gamma}|$ is even, as $\gamma$ must commute with $S_v$. Thus conditions (i)-(iii) of Lemma~\ref{lem:disjoint-supp-props} are satisfied (with $\ell = |V|$). Our goal is to define a family of surjective, commutation preserving homomorphisms $\mu_v: \mathcal{C}(\mathcal{S}_v) \cap \mathcal{J}_{m,J_v} \rightarrow \hat{\mathcal{P}}_{N,J'_v}$, for each vertex $v$. Fixing a $v \in V$, note that any such $\mu_v$ must satisfy $\mu_v(iI) = [I]$ (one argues similarly as for $\mu_1$ in the proof of the converse part of Lemma~\ref{lem:disjoint-supp-props}(c)), and this ensures $\mu_v(\gamma) = \mu_v(\gamma')$ whenever $\gamma, \gamma'$ are equivalent up to phase. Thus the two conditions in Lemma~\ref{lem:disjoint-supp-props}(b) are satisfied by $\mu_v$, so given such a family $\{\mu_v\}_{v \in V}$ one concludes that there exists a unique map $\mu: \mathcal{C}(\mathcal{S}_{\text{maj}}) \rightarrow \pauligrp{N}$ that gives rise to $\{\mu_v\}_{v \in V}$. This ensures $\mu$ satisfies condition (iv) of Lemma~\ref{lem:disjoint-supp-props}, and then Lemma~\ref{lem:disjoint-supp-props}(c) allows us to conclude that $\mu$ is a surjective, commutation preserving homomorphism.

Starting with any half-edge $h \in v$, assign the Majoranas $\gamma_{[h]_\tau}, \gamma_{[\tau \rho h]_\tau}, \dots, \gamma_{[(\tau \rho)^{\text{deg}(v)-1} h]_\tau}$ in $J_v$ the labels $1,\dots,\text{deg}(v)$ respectively, and if $v$ is an odd-degree vertex, assign the Majorana $\bar{\gamma}_v$ the label $\text{deg}(v)+1$ (note that this is independent of the labeling scheme discussed in Section~\ref{subsec:maj_codes_on_graphs}). Also consider the list of sector operators $\mathcal{T}_v = \{q_{[h]_\rho},q_{[\tau\rho h]_\rho},\dots,q_{[(\tau\rho)^{\text{deg}(v)-1}h]_\rho}\} \subseteq \mathcal{C}(\mathcal{S}_v) \cap \mathcal{J}_{m,J_v}$, and let $\chi_v : \mathcal{T}_v \rightarrow \pauligrp{N,J'_v}$ be the map that assigns each sector operator to an extremal CAL on $N_v$ qubits (as discussed in Section~\ref{subsec:qub_surface_codes}). Let $\Xi_v \in \mathbb{F}_2^{(\text{deg}(v)+1) \times m_v}$ be the matrix whose last row is all ones, and the top $\text{deg}(v)$ rows are defined by $(\Xi_v)_{ij} = 1$ if and only if the sector operator $q_{[(\tau\rho)^{i-1} h]_\rho}$ contains the Majorana with label $j$. Let $\mathcal{V}_v$ denote the row space of $\Xi_v$, and since all rows of $\Xi_v$ are binary vectors of even Hamming weight, Lemma~\ref{lem:vector-sum-even-ham-weight} implies that $\mathcal{V}_v \subseteq \{x \in \mathbb{F}_2^{m_v} : \ham{x} \text{ is even}\}$. Also define $\mathcal{N}_v = \{y \in \mathbb{F}_2^{\text{deg}(v)+1} : y \Xi_v = \vec{0} \in \mathbb{F}_2^{m_v} \}$.

We want to construct $\mu_v$ as an extension of $\chi_v$. Since elements of $\mathcal{C}(\mathcal{S}_v) \cap \mathcal{J}_{m,J_v}$ which are equivalent up to phase, must get mapped to the same element of $\hat{\mathcal{P}}_{N,J'_v}$, we may represent any $\gamma \in \mathcal{C}(\mathcal{S}_v) \cap \mathcal{J}_{m,J_v}$ by a row vector $x \in \mathbb{F}_2^{m_v}$ of even Hamming weight, where $x_j = 1$ if and only if $\gamma$ contains the  Majorana with label $j$. Given any such $x$, let $y \in \mathbb{F}_2^{\text{deg}(v)+1}$ be a solution of $x = y \Xi_v$. Then for any $\gamma$ represented by $x$ we define
\begin{equation}
\label{eq:mu_v}
    \mu_v(\gamma) = \prod_{j: y_j =1} \chi_v (q_{[(\tau\rho)^{j-1} h]_\rho}).
\end{equation}
To show that $\mu_v$ is a well defined map, we need to check (i) for all $x \in \mathbb{F}_2^{m_v}$ of even Hamming weight, the equation $x = y \Xi_v$ has a solution $y$, and (ii) if multiple solutions exist for the same $x$, then Eq.~\eqref{eq:mu_v} does not depend on the solution $y$ picked. Observe that $\rank{\Xi_v} = \text{deg}(v)-1$ if $\text{deg}(v)$ is even, and $\rank{\Xi_v} = \text{deg}(v)$ if $\text{deg}(v)$ is odd. Thus $\rank{\Xi_v} = m_v - 1$ in both cases, and we conclude that $|\mathcal{V}_v| = 2^{m_v - 1} = |\{x \in \mathbb{F}_2^{m_v} : \ham{x} \text{ is even}\}|$. This proves (i). To prove (ii), let $\vec{1} \in \mathbb{F}_2^{\text{deg}(v)}$ be the all-ones row vector, and we consider the $\text{deg}(v)$ odd or even cases separately. First suppose $\text{deg}(v)$ is odd. Then $\mathcal{N}_v = \text{span} \{\begin{bmatrix}\vec{1} & 0 \end{bmatrix}\}$, and thus if $y \neq y'$ satisy $x = y \Xi_v = y' \Xi_v$, then $y - y' = \begin{bmatrix}\vec{1} & 0 \end{bmatrix}$. Since $\prod_{j=1}^{\text{deg}(v)} \chi_v (q_{[(\tau\rho)^{j-1} h]_\rho}) = [I]$ by Corollary~\ref{cor:extremal-cal-mult-identity}(a), it follows that Eq.~\eqref{eq:mu_v} gives the same $\mu_v(\gamma)$ whether we choose $y$ or $y'$. Next suppose $\text{deg}(v)$ is even. Then $\mathcal{N}_v = \text{span} \{\begin{bmatrix}\vec{1} & 0 \end{bmatrix}, \; \begin{bmatrix}z & 1 \end{bmatrix}\}$, where $z \in \mathbb{F}_2^{\text{deg}(v)}$ satisfies $z_j = 1$ if and only if $j$ is even. In this case, both $\prod_{j=1}^{\text{deg}(v)} \chi_v (q_{[(\tau\rho)^{j-1} h]_\rho}) = [I]$ and $\prod_{j=1}^{\text{deg}(v)/2} \chi_v (q_{[(\tau\rho)^{2j-1} h]_\rho}) = [I]$ by Corollary~\ref{cor:extremal-cal-mult-identity}(b), and so again Eq.~\eqref{eq:mu_v} is well defined. We also see that $\mu_v$ is an extension of $\chi_v$. For any fixed $j$, consider the sector operator $q_{[(\tau\rho)^{j-1} h]_\rho}$ and let $x$ be its representation. Then a solution to $x = y \Xi_v$ is given by $y$ satisfying $y_{j'} = 1$ if and only if $j'=j$, and thus by Eq.~\eqref{eq:mu_v} we get $\mu_v(q_{[(\tau\rho)^{j-1} h]_\rho}) = \chi_v (q_{[(\tau\rho)^{j-1} h]_\rho})$.

We now show that the map $\mu_v$ is a surjective, commutation preserving homomorphism. Let $\gamma,\gamma' \in \mathcal{C}(\mathcal{S}_v) \cap \mathcal{J}_{m,J_v}$ be represented by $x,x' \in \mathbb{F}_2^{m_v}$ respectively, and suppose $y,y'$ satisfy $x = y \Xi_v$, $x' = y' \Xi_v$. Then we have $(x+x') = (y+y') \Xi_v$, and $\gamma \gamma'$ is represented by $x + x'$. From Eq.~\eqref{eq:mu_v} we conclude $\mu_v(\gamma \gamma') = \mu_v(\gamma) \mu_v(\gamma')$ proving $\mu_v$ is a homomorphism. Since the CAL $\{\chi_v(q_{[h]_\tau}),\dots,\chi_v (q_{[(\tau\rho)^{\text{deg}(v)-1} h]_\rho})\}$ generates $\pauligrp{N,J'_v}$ by Corollary~\ref{cor:cal-dim-extremal}, $\mu_v$ being a homomorphism also implies it is surjective. Next note that the fact $x = y \Xi_v$ always has a solution for any $x \in \mathbb{F}_2^{m_v}$ of even Hamming weight, implies $\mathcal{C}(\mathcal{S}_v) \cap \mathcal{J}_{m,J_v} = \langle iI, S_v, q_{[h]_\rho},\dots,q_{[(\tau\rho)^{\text{deg}(v)-1}h]_\rho} \rangle$, and from Eq.~\eqref{eq:mu_v} we have $\mu_v(iI) = \mu_v(S_v)=[I]$. Thus we see that $\mu_v$ is commutation preserving on a generating set of $\mathcal{C}(\mathcal{S}_v) \cap \mathcal{J}_{m,J_v}$, and by the same argument as in the paragraph below the proof of Lemma~\ref{lem:commutation-preserving-props} we conclude that $\mu_v$ is commutation preserving. 

As mentioned in the first paragraph of the proof, this finishes the construction of $\mu$ whose restriction to $\mathcal{C}(\mathcal{S}_v) \cap \mathcal{J}_{m,J_v}$ is $\mu_v$ for each $v$. Now after identifying $\mathcal{J}_m$ with $\mathcal{P}_{m/2}$ using the Jordan-Wigner map (which is a commutation preserving isomorphism), we can apply Lemma~\ref{lem:commutation-preserving-props} and conclude the existence of a surjective homomorphsim $\kappa: \mathcal{C}(\mathcal{S}_{\text{maj}}) \rightarrow \mathcal{P}_N$ such that $\hat{\kappa} = \mu$, and by Lemma~\ref{lem:commutation-preserving-props}(f) we conclude that the stabilizer codes $\mathcal{S}(G)$ and $\kappa(\mathcal{S}(G))$ encode the same number of logical qubits. It remains to show that $\mu(\mathcal{S}(G))$ equals the stabilizer group of the qubit surface code. By Definition~\ref{def:qubit_surface_code} and Eq.~\ref{eq:face_stabilizer_as_sectors}, for each $f \in F$ there is a stabilizer of the qubit surface code given by $\prod_{[h]_\rho\subseteq f} \chi_{v([h]_\rho)}(q_{[h]_\rho})$, where for each sector $[h]_\rho$, $v([h]_\rho)$ is the unique vertex containing the sector. But by Eq.~\eqref{eq:mu_v} this is equal to $\mu(S_f)$, and the conclusion follows as $\mu(S_v)=[I]$ for all $v \in V$.
\end{proof}

%% file: appendix/cyclic_codes_params_v2.tex
\section{A more general family of cyclic qubit stabilizer codes}
\label{sec:cyclic-code-families}

In this section we show that the cyclic toric code introduced in Section~\ref{subsec:qub_surface_codes} are part of a much larger 2-parameter code family of cyclic qubit stabilizer codes, which in turn are special cases of a 4-parameter cyclic code family. We will prove a few facts about these code families. The results here also provide a completely different proof for the number of encoded qubits for the cyclic toric code (see Theorem~\ref{thm:cyclic_toric_codes}).

\subsection{A four parameter cyclic code family}
\label{subsec:4param-cyclic-code}
We begin by choosing integers $N,p,q,r$ satisfying the conditions (i) $p,q,r \geq 1$, (ii) $r,q,p+r$ all distinct, and (iii) $\max\{q,p+r\} + 1 \leq N$, which ensures $N \geq 4$. We will always assume in this section that $N,p,q,r$ satisfies these conditions. Now consider a Pauli string $a_0 a_1 \dots a_{N-1}$ of length $N$ such that 
\begin{equation}
\label{eq:4param-cyclic-code-def}
    a_k = 
    \begin{cases}
        Z &\;\; \text{if } k \in \{0,q\},\\
        X &\;\; \text{if } k \in \{r,p+r\},\\
        I_2 &\;\; \text{otherwise}.
    \end{cases}
\end{equation}
Let $S_i$ denote the Pauli string in Eq.~\eqref{eq:4param-cyclic-code-def} cyclically shifted by $i$ indices to the right. Thus $S_0$ is the unshifted Pauli string. Let $\mathcal{S} = \{S_i : 0 \leq i \leq N-1\}$ denote the set of all cyclically shifted Pauli strings, and $\langle \mathcal{S} \rangle$ be the subgroup generated by $\mathcal{S}$. It will be convenient to associate subsets of $\mathcal{N} := \{0,\dots,N-1\}$ with subsets of $\mathcal{S}$, so we introduce the map $\Gamma: 2^{\mathcal{N}} \rightarrow 2^{\mathcal{S}}$ defined as $\Gamma(\mathcal{X}) = \{S_x : x \in \mathcal{X}\}$, for all $\mathcal{X} \in 2^\mathcal{N}$. 

If all the cyclic shifts of this Pauli string commute with one another, they can be used to define a stabilizer group with each shifted Pauli string defining a stabilizer up to a phase factor of $1$ or $-1$ (see Lemma~\ref{lem:majorana-group-props}(a) for the equivalent statement about the Majorana group, which also applies for the Pauli group due to identification of the two groups via the Jordan-Wigner transform). 
The quantum stabilizer code that results will be called a 4-parameter $(N,p,q,r)$ cyclic code. Depending on the values of $r,q$, and $p+r$, there are three distinct subfamilies of these codes. The explicit form of the Pauli string $S_0$ corresponding to these three subfamilies are given below:
\begin{equation}
\label{eq:3-types-4param-cyclic-code}
    S_0 = 
    \begin{cases}
        Z I^{\otimes r-1} X I^{\otimes q-r-1} Z I^{\otimes p+r-q-1} X I^{\otimes N - p - r - 1} \;\;\;\; & \text{if } r < q< p+r, \;\;\;\;\;\;\;\; \text{(subfamily A)} \\
        Z I^{\otimes r-1} X I^{\otimes p-1} X I^{\otimes q-p-r-1} Z I^{\otimes N - q - 1} \;\;\;\; & \text{if } r < p+r < q, \;\;\;\;\;\;\;\; \text{(subfamily B)}\\
        Z I^{\otimes q-1} Z I^{\otimes r - q - 1} X I^{\otimes p - 1} X I^{\otimes N - p - r - 1} \;\;\;\; & \text{if } q < r < p+r. \;\;\;\;\;\;\;\;\; \text{(subfamily C)}
    \end{cases}
\end{equation}

We would like to understand the conditions on the values of $N,p,q,r$ under which all cyclic shifts of $S_0$ for each subfamily in Eq.~\eqref{eq:3-types-4param-cyclic-code} commute with one another, which is necessary to define a valid stabilizer group. The following lemma is important towards achieving this goal.

\begin{lemma}
\label{lem:commuting-check-6-cases}
For any Pauli string defined in Eq.~\eqref{eq:4param-cyclic-code-def}, $\langle \mathcal{S} \rangle$ is a commuting subgroup if and only if the set $\{S_{0}, S_{q}, S_{r}, S_{p+r}\}$ is pairwise commuting.
\end{lemma}

\begin{proof}
In this proof all qubit indices will be evaluated modulo $N$. If $\langle \mathcal{S} \rangle$ is a commuting subgroup, the set $\{S_{0}, S_{q}, S_{r}, S_{p+r}\}$ is pairwise commuting because it is a subset of $\langle \mathcal{S} \rangle$. For the other direction, assume that $\{S_{0}, S_{q}, S_{r}, S_{p+r}\}$ is a pairwise commuting set. To show that $\langle \mathcal{S} \rangle$ is a commuting subgroup, it suffices to show that $\mathcal{S}$ is pairwise commuting, which is in turn equivalent to showing that $S_0$ commutes with each element in $\mathcal{S}$, because of cyclicity. For contradiction, assume that this is not true, so there exists $S_\ell \in \mathcal{S}$ that anticommutes with $S_0$. This is only possible if there exists a qubit index $0 \leq n \leq N-1$, such that both $S_0$ and $S_\ell$ contain a Pauli different from $I_2$ (not necessarily the same) at that qubit, and so this implies $\ell \in \{\pm p, \pm q, \pm r, \pm (p+r), \pm (q - r), \pm (p + r - q)\}$. Now for any index $\ell'$, it follows using cyclicity that $S_0$ commutes with $S_{\ell'}$ if and only if $S_0$ commutes with $S_{-\ell'}$, because the relative shifts in both cases are the same. Also by assumption $S_0$ commutes with each of $S_q$, $S_r$, and $S_{p+r}$. Thus we deduce that $\ell \in \{p, q - r, p + r - q\}$. Using cyclicity again, it also follows by the same reasoning that (i) $S_0$ and $S_p$ are commuting if and only if $S_r$ and $S_{p+r}$ are commuting, (ii) $S_0$ and $S_{q-r}$ are commuting if and only if $S_r$ and $S_{q}$ are commuting, and (iii) $S_0$ and $S_{p+r-q}$ are commuting if and only if $S_q$ and $S_{p+r}$ are commuting. But this is a contradiction as $S_q$, $S_r$ and $S_{p+r}$ are pairwise commuting by assumption.
\end{proof}

Because of Lemma~\ref{lem:commuting-check-6-cases}, given any values of $N,p,q,r$, one only needs to check whether all six pairs $\{S_{0}, S_{q}\}$, $\{S_{0}, S_{r}\}$, $\{S_{0}, S_{p+r}\}$, $\{S_{q}, S_{r}\}$, $\{S_{q}, S_{p+r}\}$ and $\{S_{r}, S_{p+r}\}$ are commuting or not. Each of these checks has complexity $O(1)$, independent of $N$, because each Pauli string in $\{S_{0}, S_{q}, S_{r}, S_{p+r}\}$ has bounded support, and so determining whether $\langle \mathcal{S} \rangle$ is a commuting subgroup or not also has complexity $O(1)$ for any input values of $N,p,q,r$. We will say that a quadruple $(N,p,q,r)$ is \textit{consistent} if and only if $\langle \mathcal{S} \rangle$ is a commuting subgroup. 
For example, the following special case is of particular interest.

\begin{lemma}
\label{lem:special-condition-subfamilyB}
A quadruple $(N,p,q,r)$ satisfying $r < p + r < q$ is consistent if $q = p + 2r$, and corresponds to a cyclic code of subfamily B.
\end{lemma}

\begin{proof}
In this proof all qubit indices will be implicitly assumed to be modulo $N$. If we prove that the quadruple is consistent then we know from Eq.~\eqref{eq:3-types-4param-cyclic-code} that the cyclic code belongs to subfamily B. First note that the Pauli string $S_0$ has the form $Z I^{\otimes r-1} X I^{\otimes p-1} X I^{\otimes r-1} Z I^{\otimes N - p - 2r - 1}$. Let $\begin{bmatrix}a & b\end{bmatrix}$ be the symplectic representation of $S_0$ and $\begin{bmatrix}a' & b'\end{bmatrix}$ be the symplectic representation of $S_m$, for some $0 \leq m \leq N-1$, where $a,a',b,b' \in \mathbb{F}_2^{N}$. Then $S_0$ and $S_m$ commute if and only if $(a \cdot b' + a' \cdot b) = 0$, where the left hand side is evaluated over $\mathbb{F}_2$. Now note that $a_i = 1$ if and only if at $i \in \{r,p+r\}$, $b_i = 1$ if and only if $i \in \{0,p+2r\}$, $a'_i = 1$ if and only if at $i \in \{m+r,m+p+r\}$, and $b'_i = 1$ if and only if $i \in \{m,m+p+2r\}$. Then
\begin{equation}
\begin{split}
    a \cdot b' + a' \cdot b =& \left( \delta_{r,m} + \delta_{r,m+p+2r} + \delta_{p+r,m} + \delta_{p+r,m+p+2r} \right) \\
    &+ \left( \delta_{m+r,0} + \delta_{m+r,p+2r} + \delta_{m+p+r,0} + \delta_{m+p+r,p+2r} \right) \\
    =& \left( \delta_{r,m} + \delta_{0,m+p+r} + \delta_{p+r,m} + \delta_{0,m+r} \right) + \left( \delta_{m+r,0} + \delta_{m,p+r} + \delta_{m+p+r,0} + \delta_{m,r} \right) \\
    =& \left ( \delta_{r,m} + \delta_{m,r} \right ) + \left ( \delta_{0,m+p+r} + \delta_{m+p+r,0} \right ) + \left ( \delta_{p+r,m} + \delta_{m,p+r} \right ) + \left ( \delta_{0,m+r} + \delta_{m+r,0} \right ) \\
    =& \; 0.
\end{split}
\end{equation}
Since $m$ is arbitrary this proves that $S_0$ and $S_m$ commute for all $0 \leq m \leq N-1$, and using cyclicity we then conclude that $(N,p,q,r)$ is consistent.
\end{proof}

Given a consistent quadruple $(N,p,q,r)$, we would like to know the number of qubits encoded by the corresponding cyclic code. The answer to this question is completely provided by the next theorem, which we state below.

\begin{theorem}
\label{thm:num-logical-ops-4param-cyclic-code}
For any 4-parameter $(N,p,q,r)$ cyclic code, the number of encoded qubits satisfies
\begin{equation}
\label{eq:num-logical-ops-4param-cyclic-code}
    K = \frac{\gcd (p,N) \; \gcd (q,N)}{\gcd (\lcm (p,q),N)}
\end{equation}
and is independent of $r$.
\end{theorem}

The proof of this theorem is delayed to first present a few general facts, on which the proof depends on. Below we start by proving Lemma~\ref{lem:gcd-equality-fact}, which is then used to prove Lemma~\ref{lem:subgroup-size-addition-modN}, the latter being the most important result that we will need. We then prove Lemma~\ref{lem:independent-4param-stabilizers}, which is a structure lemma that tells us exactly which subsets of the stabilizers of the 4-parameter $(N,p,q,r)$ cyclic code are independent. Combined with Lemma~\ref{lem:subgroup-size-addition-modN}, this immediately proves Theorem~\ref{thm:num-logical-ops-4param-cyclic-code}. In the proofs of Lemma~\ref{lem:independent-4param-stabilizers} and Theorem~\ref{thm:num-logical-ops-4param-cyclic-code}, we use the notation introduced in the paragraph below Eq.~\eqref{eq:4param-cyclic-code-def}, that associates subsets of $\{0,\dots,N-1\}$ to subsets of stabilizers for the 4-parameter $(N,p,q,r)$ cyclic code.

\begin{lemma}
\label{lem:gcd-equality-fact}
Let $a,b,m \geq 1$ be positive integers, and $\alpha, \beta, \gamma$ be real numbers. Then
\setmargins
\begin{enumerate}[(a)]
    \item
    \begin{equation}
    \label{eq:min-fact}
        \min(\gamma - \min(\alpha,\gamma), \gamma - \min(\beta,\gamma)) = \gamma - \min(\gamma, \alpha + \beta - \min(\alpha, \beta)),
    \end{equation}
    \item
    \begin{equation}
    \label{eq:gcd-equality-fact}
        \gcd \left( \frac{m}{\gcd(a,m)}, \frac{m}{\gcd(b,m)} \right) = \frac{m}{\gcd(\lcm(a,b),m)}.
    \end{equation}
\end{enumerate}
\end{lemma}

\vspace{0.5cm}
\begin{proof}
\setmargins
\begin{enumerate}[(a)]
    \item Without loss of generality assume that $\alpha \geq \beta$. Then $\alpha + \beta - \min(\alpha,\beta) = \alpha$, so the expression on the right-hand side of Eq.~\eqref{eq:min-fact} simplifies to $\gamma - \min(\gamma, \alpha)$. First consider the case $\gamma \leq \alpha$. We have $\gamma - \min(\gamma, \alpha) = 0$, while $\min(\gamma - \min(\alpha,\gamma), \gamma - \min(\beta,\gamma)) = \min(0, \gamma - \min(\beta,\gamma)) = 0$, since $\gamma - \min(\beta,\gamma) \geq 0$. Next consider the case $\gamma > \alpha \geq \beta$. Then $\gamma - \min(\gamma, \alpha) = \gamma - \alpha$, and $\min(\gamma - \min(\alpha,\gamma), \gamma - \min(\beta,\gamma)) = \min(\gamma - \alpha, \gamma - \beta) = \gamma - \alpha$. This proves Eq.~\eqref{eq:min-fact} in both cases.
    \item Let $\{p_i\}_{i \geq 1}$ be the enumeration of all the primes in ascending order. Then there exists an $n$ such that $a,b,m$ have the following unique representations
    \begin{equation}
    \label{eq:gcd-equality-fact-proof-1}
        a = \prod_{i=1}^{n} p_i^{\alpha_i}, \; b = \prod_{i=1}^{n} p_i^{\beta_i}, \; m = \prod_{i=1}^{n} p_i^{\gamma_i},
    \end{equation}
    where $\alpha_i, \beta_i,\gamma_i \in \{0,1,2,\dots\}$ for all $1 \leq i \leq n$. Using Eq.~\eqref{eq:gcd-equality-fact-proof-1} it then easily follows that
    \begin{equation}
    \label{eq:gcd-equality-fact-proof-2}
    \begin{split}
        \gcd \left( \frac{m}{\gcd(a,m)}, \frac{m}{\gcd(b,m)} \right) & = \gcd \left( \prod_{i=1}^{n} p_i^{\gamma_i - \min(\alpha_i, \gamma_i)}, \prod_{i=1}^{n} p_i^{\gamma_i - \min(\beta_i, \gamma_i)} \right) \\
        & = \prod_{i=1}^{n} p_i^{\min(\gamma_i - \min(\alpha_i, \gamma_i), \gamma_i - \min(\beta_i, \gamma_i))},
    \end{split}
    \end{equation}
    and
    \vspace*{-0.5cm}
    \begin{equation}
    \label{eq:gcd-equality-fact-proof-3}
        \frac{m}{\gcd(\lcm(a,b),m)} = \frac{\prod\limits_{i=1}^{n} p_i^{\gamma_i}}{\gcd \left( \prod\limits_{i=1}^{n} p_i^{\alpha_i + \beta_i - \min(\alpha_i, \beta_i)}, \prod\limits_{i=1}^{n} p_i^{\gamma_i} \right)} = \prod_{i=1}^{n} p_i^{\gamma_i - \min(\gamma_i, \alpha_i + \beta_i - \min(\alpha_i, \beta_i))},
    \end{equation}
    and thus by (a) these two expressions are equal.
\end{enumerate}
\end{proof}

\begin{lemma}
\label{lem:subgroup-size-addition-modN}
Let $a,b,m \geq 1$ be positive integers, and let ``$\pmb{+}$'' denote addition modulo $m$. Then the set of integers $\mathcal{G} = \{k_1 a \pmb{+} k_2 b : k_1, k_2 \in \mathbb{Z}\}$ is a cyclic subgroup of $(\mathbb{Z} / m \mathbb{Z}, \pmb{+})$, under $\pmb{+}$. The order of $\mathcal{G}$ satisfies
\begin{equation}
\label{eq:order-G}
|\mathcal{G}| = \frac{m \gcd (\lcm (a,b),m)}{\gcd (a,m) \; \gcd (b,m)}.
\end{equation}
The smallest positive integer that generates $\mathcal{G}$ is $\gcd(a,m) \gcd(b,m) / \gcd(\lcm(a,b),m)$.
\end{lemma}

\begin{proof}
$\mathcal{G}$ is closed under $\pmb{+}$, and for every $x \in \mathcal{G}$, $-x \bmod m \in \mathcal{G}$, and so each element has an additive inverse. Thus $\mathcal{G}$ is a subgroup of $(\mathbb{Z} / m \mathbb{Z}, \pmb{+})$ since $\mathcal{G} \subseteq \mathbb{Z} / m \mathbb{Z}$, and because $(\mathbb{Z} / m \mathbb{Z}, \pmb{+})$ is a cyclic group, it follows that $\mathcal{G}$ is cyclic.

We now derive the expression for $|\mathcal{G}|$ in Eq.~\eqref{eq:order-G}, which we break up into a few steps.
\begin{enumerate}[(i)]
    \item \textbf{Step 1:} We first construct two subsets of $\mathcal{G}$: $\mathcal{G}_a = \{k_1 a \bmod m : k_1 \in \mathbb{Z}\}$, and $\mathcal{G}_b = \{k_2 b \bmod m : k_2 \in \mathbb{Z}\}$. Using the same argument as for $\mathcal{G}$, we have that $\mathcal{G}_a$ and $\mathcal{G}_b$ are cyclic subgroups of $\mathcal{G}$ under $\pmb{+}$, hence also of $(\mathbb{Z} / m \mathbb{Z}, \pmb{+})$. Since $\lcm (a,m)$ is the smallest positive multiple of $a$ divisible by $m$, and noting that $ \lcm(a,m) /a = m / \gcd(a,m)$, it follows that $\mathcal{G}_a = \{k_1 a \bmod m : 0 \leq k_1 \leq m / \gcd(a,m) - 1\}$. Moreover for any two distinct integers $0 \leq k_1, k'_1 \leq  m / \gcd(a,m) - 1$, we have $k_1 a \not\equiv k'_1 a \pmod{m}$ because otherwise $|k_1 - k'_1| a$ is a positive multiple of $a$ strictly smaller than $\lcm(a,m)$ that is also divisible by $m$. Thus $|\mathcal{G}_a| = m / \gcd(a,m)$, and similarly $|\mathcal{G}_b| = m / \gcd(b,m)$. 
    
    \item \textbf{Step 2:} Our next goal is to show that $\mathcal{G}_a = \{k_1 \gcd(a,m) : 0 \leq k_1 \leq m/\gcd(a,m) - 1\}$. If $\gcd(a,m) = m$ then the result is clearly true, so assume that $\gcd(a,m) < m$. By B\'{e}zout's identity $\gcd(a,m) = xa + ym$ for some integers $x,y$, i.e. $\gcd(a,m) \equiv xa \pmod{m}$, and so $\gcd(a,m) \in \mathcal{G}_a$. But this implies that $\{k_1 \gcd(a,m) : 0 \leq k_1 \leq m/\gcd(a,m) - 1\} \subseteq \mathcal{G}_a$, since $\mathcal{G}_a$ forms a group under $\pmb{+}$. But all the elements in this subset are distinct and the number of elements in it equal $m/\gcd(a,m)$, which is the order of $\mathcal{G}_a$, and so we have in fact proved that $\mathcal{G}_a = \{k_1 \gcd(a,m) : 0 \leq k_1 \leq m/\gcd(a,m) - 1\}$, and analogously $\mathcal{G}_b = \{k_2 \gcd(b,m) : 0 \leq k_2 \leq m/\gcd(b,m) - 1\}$. This also shows that $\gcd(a,m)$ and $\gcd(b,m)$ are the smallest positive integers that cyclically generate $\mathcal{G}_a$ and $\mathcal{G}_b$ respectively.
    
    \item \textbf{Step 3:} We now want to locate the smallest non-zero integer $p$ in $\mathcal{G}_a$ that is also in $\mathcal{G}_b$. For this we consider two subsets: $\mathcal{G}_a \cap \mathcal{G}_b$ and $\mathcal{G}':= \{k_3 \lcm(a,b) \bmod m : k_3 \in \mathbb{Z}\}$. First note that $\mathcal{G}_a \cap \mathcal{G}_b$ forms a subgroup of both $\mathcal{G}_a$ and $\mathcal{G}_b$ under $\pmb{+}$, because it is the intersection of two subgroups; hence Lagrange's theorem implies $|\mathcal{G}_a \cap \mathcal{G}_b| \leq \gcd(|\mathcal{G}_a|, |\mathcal{G}_b|)$. Also since each element in $\mathcal{G}_a \cap \mathcal{G}_b$ is a multiple of both $a$ and $b$ modulo $m$, it follows that $\mathcal{G}' \subseteq \mathcal{G}_a \cap \mathcal{G}_b$, and so $|\mathcal{G}_a \cap \mathcal{G}_b| \geq |\mathcal{G}'|$. Next, reasoning similarly as in Steps 1 and 2 for $\mathcal{G}_a$ and $\mathcal{G}_b$, we conclude that $\mathcal{G}'$ is a cyclic group under $\pmb{+}$, $\mathcal{G}' = \{k_3 \gcd(\lcm(a,b),m) : 0 \leq k_3 \leq m / \gcd(\lcm(a,b),m) - 1\}$, and $|\mathcal{G}'| = m / \gcd(\lcm(a,b),m)$. But applying Lemma~\ref{lem:gcd-equality-fact}(b) gives that $m / \gcd(\lcm(a,b),m) = \gcd(m/\gcd(a,m), m/\gcd(b,m)) = \gcd(|\mathcal{G}_a|, |\mathcal{G}_b|)$. We have thus proved that $\mathcal{G}' = \mathcal{G}_a \cap \mathcal{G}_b$, and it follows that $p = \gcd(\lcm(a,b),m)$.
    
    \item \textbf{Step 4:} We now give an equivalent characterization of $\mathcal{G}$ that will help us count the number of elements in it. Let $q = p / \gcd(a,m)$, where $p$ is from Step 3, and notice that $q$ is an integer because
    \begin{equation}
    \begin{split}
        \gcd(\lcm(a,b),m) &= \gcd \left( \gcd(a,m) \frac{a}{\gcd(a,m)} \frac{b}{\gcd(a,b)}, \gcd(a,m) \frac{m}{\gcd(a,m)} \right) \\
        &= \gcd(a,m) \gcd \left( \frac{a}{\gcd(a,m)} \frac{b}{\gcd(a,b)}, \frac{m}{\gcd(a,m)} \right),
    \end{split}
    \end{equation}
    and moreover $m$ is a multiple of $q$, because $p$ divides $m$. Now construct the set
    \begin{equation}
    \label{eq:G-alternate-construction}
        \mathcal{G}'' = \bigcup_{k_1 = 0}^{q-1} \left( k_1 \gcd(a,m)  \pmb{+} \mathcal{G}_b \right).
    \end{equation}
    We want to prove that $\mathcal{G}'' = \mathcal{G}$. Clearly we have that $\mathcal{G}'' \subseteq \mathcal{G}$, so we only need to show that $\mathcal{G} \subseteq \mathcal{G}''$. If $r \in \mathcal{G}$ then there exists $k_1, k_2 \in \mathbb{Z}$, such that $r = k_1 a \pmb{+} k_2 b$. But $k_1 a \pmb{+} k_2 b = (k_1 a \bmod m) \pmb{+} (k_2 b \bmod m)$, and so $r = x \pmb{+} y$, for some $x \in \mathcal{G}_a$ and $y \in \mathcal{G}_b$. Note that any element $x \in \mathcal{G}_a$ can be expressed as $x = x' + z$ for some $z \in \mathcal{G}_a \cap \mathcal{G}_b$ and $x' \in \{k_1 \gcd(a,m) : 0 \leq k_1 \leq q-1\}$, and thus we have $r = x' \pmb{+} y \pmb{+} z = x' \pmb{+} y'$ for some $y' \in \mathcal{G}_b$, since both $y,z \in \mathcal{G}_b$. This proves $\mathcal{G} \subseteq \mathcal{G}''$.
    
    \item \textbf{Step 5:} We now want to show that the union appearing in Eq.~\eqref{eq:G-alternate-construction} is a disjoint union. Note that this would imply $|\mathcal{G}| = \sum_{k_1 = 0}^{q-1} |k_1 \gcd(a,m)  \pmb{+} \mathcal{G}_b| = \sum_{k_1 = 0}^{q-1} |\mathcal{G}_b| = q |\mathcal{G}_b|$ by Step 4, thus proving Eq.~\eqref{eq:order-G}. So for the sake of contradiction, assume that there exists distinct integers $0 \leq k_1 < k'_1 \leq q-1$ such that $\left( k_1 \gcd(a,m)  \pmb{+} \mathcal{G}_b \right) \cap \left( k'_1 \gcd(a,m)  \pmb{+} \mathcal{G}_b \right) \neq \emptyset$. Then there exists $b,b' \in \mathcal{G}_b$ such that $k_1 \gcd(a,m) \pmb{+} b = k'_1 \gcd(a,m) \pmb{+} b'$, and so $(k'_1 - k_1) \gcd(a,m) \equiv (b - b') \pmod{m}$ which implies $(k'_1 - k_1) \gcd(a,m) \in \mathcal{G}_b$. But $k'_1 - k_1 < q$ and so $(k'_1 - k_1) \gcd(a,m) < p$, and since $p$ is the smallest non-zero element in $\mathcal{G}_a \cap \mathcal{G}_b$ by Step 3, this implies $k'_1 = k_1$ which is a contradiction.
\end{enumerate}
For the last part, let $s$ be the smallest non-zero integer in $\mathcal{G}$; so $\{ks : 0 \leq k \leq \ceil{m/s} - 1\} \subseteq \mathcal{G}$ because $\mathcal{G}$ is a group. Then $s$ divides $m$ because otherwise $t = s \ceil{m/s} \bmod m \in \mathcal{G}$, $0 \neq t < s$. Suppose $u \in\mathcal{G}$ and $u \notin \{ks : 0 \leq k \leq m/s - 1\}$. Then $0 \neq u - s \floor{u/s} \in \mathcal{G}$ and $u - s \floor{u/s} < s$ which is a contradiction. This proves that $\mathcal{G} = \{ks : 0 \leq k \leq m/s - 1\}$, and so $m/s = |\mathcal{G}|$, or equivalently $s = \gcd(a,m) \gcd(b,m) / \gcd(\lcm(a,b),m)$. Thus $s$ is the smallest positive integer that cyclically generates $\mathcal{G}$.
\end{proof}

\begin{lemma}
\label{lem:independent-4param-stabilizers}
Let ``$\pmb{+}$'' denote addition modulo $N$, $\mathcal{G} = \{k_1 p \pmb{+} k_2 q : k_1, k_2 \in \mathbb{Z}\}$, $\mathcal{G}_p = \{k p \pmb{+} 0 : k \in \mathbb{Z}\}$, and $\mathcal{G}_q = \{k q \pmb{+} 0 : k \in \mathbb{Z}\}$. If $\mathcal{A} \subseteq \mathcal{N}$ and $s \in \mathbb{Z}$, denote $\mathcal{A} \pmb{+} s := \{x \pmb{+} s : x \in \mathcal{A}\}$. Then for any quadruple $(N,p,q,r)$ (not necessarily consistent), and for any non-empty $\mathcal{H} \subseteq \mathcal{N}$, the following holds:
\setmargins
\begin{enumerate}[(a)]
    \item $\Gamma(\mathcal{H})$ multiplies to a $X$-type Pauli string if and only if $\mathcal{H}$ is a union of the cosets of $\mathcal{G}_q$ in $(\mathbb{Z} / N \mathbb{Z}, \pmb{+})$.
    \item $\Gamma(\mathcal{H})$ multiplies to a $Z$-type Pauli string if and only if $\mathcal{H}$ is a union of the cosets of $\mathcal{G}_p$ in $(\mathbb{Z} / N \mathbb{Z}, \pmb{+})$.
    \item $\Gamma(\mathcal{H})$ multiplies to $I$ (up to phase factors) if and only if $\mathcal{H}$ is a union of the cosets of $\mathcal{G}$ in $(\mathbb{Z} / N \mathbb{Z}, \pmb{+})$.
\end{enumerate}
\end{lemma}

\begin{proof}
In this proof all qubit indices and cyclic shifts will be implicitly assumed to be modulo $N$, and products of Pauli strings in $\mathcal{S}$ will be considered modulo phase factors. Note that we have already proved $\mathcal{G}$, $\mathcal{G}_p$, and $\mathcal{G}_q$ are subgroups of $(\mathbb{Z} / N \mathbb{Z}, \pmb{+})$ under $\pmb{+}$ by Lemma~\ref{lem:subgroup-size-addition-modN}, so cosets of $\mathcal{G}$, $\mathcal{G}_p$, and $\mathcal{G}_q$ are well-defined. It will also be useful to remember the following fact: if $\mathcal{T} \subseteq \mathcal{S}$, and $\mathcal{T}'$ is obtained by cyclically shifting each Pauli string in $\mathcal{T}$ by some integer $k$, then $\prod \mathcal{T}'$ is obtained by cyclically shifting $\prod \mathcal{T}$ by $k$.

We next make the following claims, which we prove at the end of the proof: (i) $\prod \Gamma(\mathcal{G}_q)$ is a $X$-type Pauli string, and $\prod \Gamma(\mathcal{G}_p)$ is a $Z$-type Pauli string, (ii) if $\prod \Gamma(\mathcal{H}')$ is a $X$-type Pauli string for any $\mathcal{H}' \subseteq \mathcal{N}$ such that $x \in \mathcal{H}'$, then $\mathcal{G}_q \pmb{+} x \subseteq \mathcal{H}'$, (iii) if $\prod \Gamma(\mathcal{H}')$ is a $Z$-type Pauli string for any $\mathcal{H}' \subseteq \mathcal{N}$ such that $x \in \mathcal{H}'$, then $\mathcal{G}_p \pmb{+} x \subseteq \mathcal{H}'$, and (iv) if $\prod \Gamma(\mathcal{H}')$ equals $I$ (up to phase factors) for any $\mathcal{H}' \subseteq \mathcal{N}$ such that $x \in \mathcal{H}'$, then $\mathcal{G} \pmb{+} x \subseteq \mathcal{H}' $.

\setmargins
\begin{enumerate}[(a)]
    \item Suppose $\mathcal{H}$ is a union of the cosets of $\mathcal{G}_q$ in $(\mathbb{Z} / N \mathbb{Z}, \pmb{+})$. By claim (i) and the fact mentioned at the beginning of the proof, if $\mathcal{H}'$ is a coset of $\mathcal{G}_q$ in $(\mathbb{Z} / N \mathbb{Z}, \pmb{+})$, then $\prod \Gamma(\mathcal{H}')$ is a X-type Pauli string, since $\mathcal{H}' = \mathcal{G}_q \pmb{+} s$ for some $s \in \mathcal{N}$. Thus $\prod\Gamma(\mathcal{H})$ is also a X-type Pauli string. For the converse, suppose for contradiction that $\mathcal{H}$ is non-empty, $\prod\Gamma(\mathcal{H})$ is a X-type Pauli string, and $\mathcal{H}$ is not a union of the cosets of $\mathcal{G}_q$. Then there exists an integer $x$ belonging to some coset $\mathcal{H}'$ of $\mathcal{G}_q$ such that $x \in \mathcal{H}$ and $\mathcal{H}' \not\subseteq \mathcal{H}$. So $\mathcal{H}' = \mathcal{G}_q \pmb{+} x$, which gives a contradiction by claim (ii).
    \item This follows from (a) by performing Clifford transformations on every qubit interchanging $X$ and $Z$, and using cyclicity.
    \item Assume $\mathcal{H}$ is a union of the cosets of $\mathcal{G}$ in $(\mathbb{Z} / N \mathbb{Z}, \pmb{+})$. Since $\mathcal{G}_q$ (resp. $\mathcal{G}_p$) is a subgroup of $\mathcal{G}$, it implies $\mathcal{G}$ is a union of the cosets of $\mathcal{G}_q$ (resp. $\mathcal{G}_p$) in  $(\mathbb{Z} / N \mathbb{Z}, \pmb{+})$. Thus by (a) and (b), $\prod \Gamma(\mathcal{H})$ is both a $X$-type and $Z$-type Pauli string, and so it must equal $I$ up to phase factors. To prove the converse, suppose for contradiction that $\mathcal{H}$ is non-empty, $\prod \Gamma(\mathcal{H})$ equals $I$ (up to phase factors), and $\mathcal{H}$ is not a union of the cosets of $\mathcal{G}$. Then similar to the proof of the converse in (a), there exists an integer $x$ belonging to some coset $\mathcal{H}'$ of $\mathcal{G}$ such that $x \in \mathcal{H}$ and $\mathcal{H}' \not \subseteq \mathcal{H}$. This implies $\mathcal{H}' = \mathcal{G} \pmb{+} x$, but this contradicts claim (iv).
\end{enumerate}

We now prove the claims made above.
\setmargins
\begin{enumerate}[(i)]
    \item We observe that for any qubit index $n \in \mathcal{N}$, there are exactly two Pauli strings $S_{n}$ and $S_{n-q}$ containing $Z$ at index $n$, exactly two other Pauli strings $S_{n-r}$ and $S_{n-p-r}$ containing $X$ at index $n$, while all others have $I_2$ at index $n$ (note that $S_{n}$, $S_{n-q}$, $S_{n-r}$, and $S_{n-p-r}$ are distinct by assumptions on $p,q,r$). Since the relative shift of $S_{n}$ and $S_{n-q}$ is $q$, and the relative shift of $S_{n-r}$ and $S_{n-p-r}$ is $p$, it follows from the group structure of $\mathcal{G}_q$ and $\mathcal{G}_p$ that $S_{n} \in \Gamma(\mathcal{G}_q)$ if and only if $S_{n-q} \in \Gamma(\mathcal{G}_q)$, and $S_{n-r} \in \Gamma(\mathcal{G}_p)$ if and only if $S_{n-p-r} \in \Gamma(\mathcal{G}_p)$. Thus $\prod \Gamma(\mathcal{G}_q)$ has either $X$ or $I_2$ at index $n$, and $\prod \Gamma(\mathcal{G}_p)$ has either $Z$ or $I_2$ at index $n$.
    \item By cyclicity, it suffices to prove this for the special case $x=0$. So suppose $0 \in \mathcal{H}' \subseteq \mathcal{N}$, $\prod \Gamma(\mathcal{H}')$ is a $X$-type Pauli string, and for contradiction also assume $\mathcal{G}_q \not \subseteq \mathcal{H}'$. Noting that $\mathcal{G}_q = \{k q \bmod N : 0 \leq k \leq N / \gcd(q,N) - 1\}$, then there exists $0 < k' \leq N / \gcd(q,N) - 1$, such that $k'q \bmod N \notin \mathcal{H}'$ and $k'' q \bmod N \in \mathcal{H}'$ for all $0 \leq k'' < k'$. Thus the Pauli string $S_{(k' - 1) q} \in \Gamma(\mathcal{H}')$, which contains $Z$ at qubit index $k' q$. The only other Pauli string that has $Z$ at the same index is $S_{k'q}$, but $S_{k'q} \notin \Gamma(\mathcal{H}')$. Thus $\prod \Gamma(\mathcal{H}')$ is not $X$-type, giving a contradiction. 
    \item This follows from claim (ii) by cyclicity, i.e. cyclically shifting every element of $\mathcal{S}$ by $r$ units to the left, and then performing Clifford transformations on every qubit swapping $X$ and $Z$.
    \item Let $\mathcal{H}'$ be as in the claim. Since $\prod \Gamma(\mathcal{H}')$ is a $X$-type Pauli string, we conclude by claim (ii) that $\mathcal{G}_q \subseteq \mathcal{H}'$. Next fix any $q' \in \mathcal{G}_q$, and consider the set $\mathcal{G}_{p} \pmb{+} q'  = \{ k p \pmb{+} q' : 0 \leq k \leq N / \gcd(p,N) - 1 \}$. Since $\prod \Gamma(\mathcal{H}')$ is a $Z$-type Pauli string, by claim (iii) we have $\mathcal{G}_{p} \pmb{+} q' \subseteq \mathcal{H}'$. As $q'$ is arbitrary, this proves $\mathcal{G} = \bigcup_{q' \in \mathcal{G}_q} (\mathcal{G}_{p} \pmb{+} q') \subseteq \mathcal{H}'$.
\end{enumerate}
\end{proof}

\begin{proof}[Proof of Theorem~\ref{thm:num-logical-ops-4param-cyclic-code}]
Let $\mathcal{G}$ be defined as in Lemma~\ref{lem:independent-4param-stabilizers}. For each coset $\mathcal{H}$ of $\mathcal{G}$ in $(\mathbb{Z} / N \mathbb{Z}$, \pmb{+}), we can then remove exactly one stabilizer from $\Gamma(\mathcal{H})$ and take the union of the resulting sets, to get a set of stabilizers $\mathcal{H}'$ that still generates $\langle \mathcal{S} \rangle$ by Lemma~\ref{lem:independent-4param-stabilizers}(c). Since the number of cosets of $\mathcal{G}$ is $N/|\mathcal{G}|$, it follows by Lemma~\ref{lem:subgroup-size-addition-modN} that $|\mathcal{H}'| = N - \gcd(p,N) \gcd(q,N) / \gcd(\lcm(p,q),N)$. Moreover $\Gamma(\mathcal{H}')$ is independent as it has no subset that multiplies to $I$, again by Lemma~\ref{lem:independent-4param-stabilizers}(c), and so the number of independent generators of $\langle \mathcal{S} \rangle$ equals $|\mathcal{H}'|$. The proof is complete noting that $K = N - |\mathcal{H}'|$.
\end{proof}

Combining the results of Lemma~\ref{lem:subgroup-size-addition-modN}, Lemma~\ref{lem:independent-4param-stabilizers}(c), and Theorem~\ref{thm:num-logical-ops-4param-cyclic-code} we also have the following corollary, which allows us to extract an independent subset of stabilizers that generate $\langle \mathcal{S} \rangle$ from the knowledge of $N$ and $K$ alone.

\begin{corollary}
\label{cor:subsets-prod-I-from-K}
If a 4-parameter $(N,p,q,r)$ cyclic code encodes $K$ qubits, then $K \geq 1$ and divides $N$. Moreover a non-empty subset of $\mathcal{S}$ multiplies to $I$ (up to phase factors) if and only if it is a non-empty union of the subsets of $\mathcal{S}$ corresponding to the cosets of $\mathcal{G}$ in $(\mathbb{Z} / N \mathbb{Z}, \pmb{+})$, where $\mathcal{G} = \{k K  : 0 \leq k \leq N/K - 1\}$.
\end{corollary}

\begin{proof}
The second part is a mere restatement of Lemma~\ref{lem:independent-4param-stabilizers}(c), combining the result proved in the last paragraph of the proof of Lemma~\ref{lem:subgroup-size-addition-modN}, and the expression of $K$ in Theorem~\ref{thm:num-logical-ops-4param-cyclic-code}. That $K \geq 1$ and $K$ divides $N$ follows because $\gcd(p,N) \gcd(q,N) / \gcd(\lcm(p,q),N)$ divides $N$, the proof of which is already contained in the proof of Lemma~\ref{lem:subgroup-size-addition-modN}.
\end{proof}

As the final result of this subsection, we prove in Lemma~\ref{lem:allX-allZ} an interesting property of the subgroup $\langle \mathcal{S} \rangle$ for any quadruple $(N,p,q,r)$, which finds use in the next subsection. The lemma studies the membership of the Pauli strings $X^{\otimes N}$ and $Z^{\otimes N}$ in $\langle \mathcal{S} \rangle$, where membership is decided up to phase factors. Thus we will say for instance that $X^{\otimes N}$ belongs to the subgroup $\langle \mathcal{S} \rangle$ if and only if $\{\pm X^{\otimes N}, \pm i X^{\otimes N}\} \cap \langle \mathcal{S} \rangle \neq \emptyset$. For the proof, we need the following fact (Lemma~\ref{lem:combinatorial-fact}) that we first prove. Lemma~\ref{lem:allX-allZ} is stated and proved after that.

\begin{lemma}
\label{lem:combinatorial-fact}
Let $1 \leq a,m \in \mathbb{Z}$, $\mathcal{M} = \{x \in \mathbb{Z} : 0 \leq x \leq m-1\}$, and ``$\pmb{+}$'' denote addition modulo $m$. For any $\mathcal{A} \subseteq \mathcal{M}$, denote $\mathcal{A} \pmb{+} a = \{x \pmb{+}a : x \in \mathcal{A}\}$. A partition $\mathcal{M} = \mathcal{M}_1 \sqcup \mathcal{M}_2$ satisfying $\mathcal{M}_2 = \mathcal{M}_1 \pmb{+} a$ exists, if and only if $m / \gcd(a, m) = \lcm(a,m)/a \equiv 0 \pmod{2}$.
\end{lemma}

\begin{proof}
Let $\mathcal{G}_a=\{ka\pmb{+}0:k\in\mathbb{Z}\}$ denote an orbit of the function that adds $a$ modulo $m$, i.e.~$f(x)=x\pmb{+}a$. Each orbit $\mathcal{G}_a\pmb{+}x$ has size $\lcm(a,m)/a$, and for $x=0,1,\dots,\gcd(a,m)-1$ they are all disjoint. We may write $\mathcal{M}=\bigsqcup_{x=0}^{\gcd(a,m)-1}\left(\mathcal{G}_a\pmb{+}x\right)$.

Suppose $\lcm(a,m)/a\equiv0 \pmod{2}$ so the orbits have even size. Create the sets $\mathcal{M}_1$ and $\mathcal{M}_2$ as follows. For each $\mathcal{G}_a\pmb{+}x$, place $x\pmb{+}ka$ in $\mathcal{M}_1$ if $k$ is even and $x\pmb{+}ka$ in $\mathcal{M}_2$ if $k$ is odd. Because the orbits have even size, $\mathcal{M}_1$ and $\mathcal{M}_2$ thus defined are disjoint. Moreover, $\mathcal{M}_2=\mathcal{M}_1\pmb{+}a$.

Now suppose $\mathcal{M}=\mathcal{M}_1\sqcup\mathcal{M}_2$ where $\mathcal{M}_2=\mathcal{M}_1\pmb{+}a$. If $x\in\mathcal{M}_1$, then $x\pmb{+}ka$ is in $\mathcal{M}_1$ if $k$ is even and in $\mathcal{M}_2$ if $k$ is odd. Notice $x\equiv x\pmb{+}(\lcm(a,m)/a)a \pmod{m}$. If $\lcm(a,m)/a$ were odd then $x$ would be in both $\mathcal{M}_1$ and $\mathcal{M}_2$, a contradiction with those sets being disjoint. So $\lcm(a,m)/a$ is even.
\end{proof}

\begin{lemma}
\label{lem:allX-allZ}
For any quadruple $(N,p,q,r)$, the following holds:
\setmargins
\begin{enumerate}[(a)]
    \item $Z^{\otimes N}$ belongs to $\langle \mathcal{S} \rangle$ if and only if $\gcd(p,N) / \gcd(q, \gcd(p,N)) \equiv 0 \pmod{2}$.
    \item $X^{\otimes N}$ belongs to $\langle \mathcal{S} \rangle$ if and only if $\gcd(q,N) / \gcd(p, \gcd(q,N)) \equiv 0 \pmod{2}$.
\end{enumerate}
\end{lemma}

\begin{proof}
We only prove (a), as (b) follows from (a) by cyclicity, i.e. cyclically shifting every element of $\mathcal{S}$ by $r$ units to the left, and then performing Clifford transformations on every qubit swapping $X$ and $Z$. For the proof, we introduce some notation. Let $m = \gcd(p,N)$, $\mathcal{M} = \{x \in \mathbb{Z} : 0 \leq x \leq m-1\}$, ``$\pmb{+}$ denote addition modulo $N$, and ``$+_m$'' denote addition modulo $m$. For any $\mathcal{A} \subseteq \mathbb{Z}$ and $a \in \mathbb{Z}$, denote $\mathcal{A} +_m a = \{x +_m a : x \in \mathcal{A}\}$ and $\mathcal{A} \pmb{+} a = \{x \pmb{+} a : x \in \mathcal{A}\}$. To prove (a), it suffices to show $Z^{\otimes N}$ belongs to $\langle \mathcal{S} \rangle$ if and only if there exists a partition $\mathcal{M} = \mathcal{M}_1 \sqcup \mathcal{M}_2$ with $\mathcal{M}_2 = \mathcal{M}_1 +_m q$, because by Lemma~\ref{lem:combinatorial-fact} the latter is true if and only if $m / \gcd(q, m) \equiv 0 \pmod{2}$. In the proof, we consider two Pauli strings equal if and only if they are the same up to phase factors. Let $\mathcal{G} = \{k_1 p \pmb{+} k_2 q: k_1,k_2 \in \mathbb{Z}\}$, $\mathcal{G}_p = \{kp \pmb{+} 0 : k \in \mathbb{Z}\}$, and recall that $\mathcal{G}_p$ has $m$ distinct cosets $\{\mathcal{G}_p \pmb{+} x\}_{x=0}^{m-1}$. Notice when $q \equiv 0 \pmod{m}$, we have $q \in \mathcal{G}_p$ (since $m \in \mathcal{G}_p$ by Step 2 of the proof of Lemma~\ref{lem:subgroup-size-addition-modN}) implying $\mathcal{G} = \mathcal{G}_p$.

Let $x \in \mathcal{M}$. Observe that $x = x +_m q$ if and only if $q \equiv 0 \pmod{m}$, so the cosets $\mathcal{G}_p \pmb{+} x$ and $\mathcal{G}_p \pmb{+} (x +_m q)$ are different if and only if $q \not \equiv 0 \pmod{m}$. We claim $\prod \Gamma(\mathcal{G}_p \pmb{+} x)$ is a $Z$-type Pauli string with $Z$ only at qubit indices $(\mathcal{G}_p \pmb{+} x) \sqcup (\mathcal{G}_p \pmb{+} (x +_m q))$ if $q \not \equiv 0 \pmod{m}$, and $\prod \Gamma(\mathcal{G}_p \pmb{+} x)=I$ if $q \equiv 0 \pmod{m}$. The case $q \equiv 0 \pmod{m}$ follows by Lemma~\ref{lem:independent-4param-stabilizers}(c) as $\mathcal{G} = \mathcal{G}_p$. So now assume $q \not \equiv 0 \pmod{m}$. Lemma~\ref{lem:independent-4param-stabilizers}(b) implies $\prod \Gamma(\mathcal{G}_p \pmb{+} x)$ is a $Z$-type Pauli string. For each $y \in \mathcal{G}_p \pmb{+} x$, the qubit indices where the Pauli string $S_y$ contributes $Z$ are $y$ and $y \pmb{+} q$, where $y \pmb{+} q \in \mathcal{G}_p \pmb{+} (x +_m q)$. Thus we conclude using the disjointness of $\mathcal{G}_p \pmb{+} x$ and $\mathcal{G}_p \pmb{+} (x +_m q)$, that $\prod \Gamma(\mathcal{G}_p \pmb{+} x)$ has $Z$ at qubit indices $\{y : y \in \mathcal{G}_p \pmb{+} x \} \sqcup \{ y \pmb{+} q : y \in \mathcal{G}_p \pmb{+} x \} = (\mathcal{G}_p \pmb{+} x) \sqcup (\mathcal{G}_p \pmb{+} (x +_m q))$.

Assume that there exists a partition $\mathcal{M} = \mathcal{M}_1 \sqcup \mathcal{M}_2$ such that $\mathcal{M}_2 = \mathcal{M}_1 +_m q$. Lemma~\ref{lem:combinatorial-fact} implies that $m/\gcd(q,m) \equiv 0 \pmod{2}$, and so $q \not \equiv 0 \pmod{m}$. Let $\mathcal{H} = \bigsqcup_{x \in \mathcal{M}_1} (\mathcal{G}_p \pmb{+} x)$. By the claim, if $x \in \mathcal{M}_1$, then $\prod \Gamma(\mathcal{G}_p \pmb{+} x)$ has $Z$ at qubit indices $(\mathcal{G}_p \pmb{+} x) \sqcup (\mathcal{G}_p \pmb{+} (x +_m q))$. Moreover for distinct $x,x' \in \mathcal{M}_1$, the conditions on the partitions $\mathcal{M}_1, \mathcal{M}_2$ imply the cosets $\mathcal{G}_p \pmb{+} x$, $\mathcal{G}_p \pmb{+} x'$, $\mathcal{G}_p \pmb{+} (x +_m q)$ and $\mathcal{G}_p \pmb{+} (x' +_m q)$ are distinct. This means that there is no cancellation in the Pauli $Z$s of $\prod\Gamma(\mathcal{G}_p\pmb{+}x)$ and $\prod\Gamma(\mathcal{G}_p\pmb{+}x')$. Since $\mathcal{M} = \mathcal{M}_1 \sqcup \mathcal{M}_2$, we conclude $\prod \Gamma (\mathcal{H})$ has $Z$ at qubit indices $\bigsqcup_{x \in \mathcal{M}} (\mathcal{G}_p \pmb{+} x) = \mathcal{N}$.

Now assume $Z^{\otimes N}$ belongs to $\langle \mathcal{S} \rangle$. By Lemma~\ref{lem:independent-4param-stabilizers}(b), $Z^{\otimes N} = \prod \Gamma(\mathcal{H})$, where $\mathcal{H} = \bigsqcup_{x \in \mathcal{M}_1} (\mathcal{G}_p \pmb{+} x)$ with $\emptyset \neq \mathcal{M}_1 \subseteq \mathcal{M}$. Then $q \not \equiv 0 \pmod{m}$ by claim above, because otherwise $\prod \Gamma(\mathcal{G}_p \pmb{+} x) = I$ for each $x \in \mathcal{M}_1$, contradicting $\prod \Gamma(\mathcal{H})=Z^{\otimes N}$. Define $\mathcal{M}_2 := \mathcal{M}_1 +_m q \subseteq \mathcal{M}$. The claim also implies that the qubit indices where $\prod \Gamma(\mathcal{H})$ contains $Z$ is a subset of $\bigsqcup_{x \in \mathcal{M}_1 \cup \mathcal{M}_2} (\mathcal{G}_p \pmb{+} x)$. Thus we already have $\mathcal{M} = \mathcal{M}_1 \cup \mathcal{M}_2$ using $\prod \Gamma(\mathcal{H})=Z^{\otimes N}$. We want to show $\mathcal{M}_1$ and $\mathcal{M}_2$ are disjoint. For contradiction, suppose there exists $x \in \mathcal{M}_1 \cap \mathcal{M}_2$. Then there exists $x' \in \mathcal{M}_1$ satisfying $x' +_m q = x$. The Pauli string $\prod \Gamma(\mathcal{G}_p \pmb{+} x')$ has $Z$ at qubit indices $(\mathcal{G}_p \pmb{+} x') \sqcup (\mathcal{G}_p \pmb{+} (x' +_m q))$, while $\prod \Gamma(\mathcal{G}_p \pmb{+} x)$ has $Z$ at qubit indices $(\mathcal{G}_p \pmb{+} x) \sqcup (\mathcal{G}_p \pmb{+} (x +_m q))$. Thus there is cancellation of Pauli Zs at the qubit indices in the coset $\mathcal{G}_p \pmb{+} x$ in the product $\prod \Gamma(\mathcal{H})$. This is a contradiction to $Z^{\otimes N} = \prod \Gamma(\mathcal{H})$, so we have proved $\mathcal{M} = \mathcal{M}_1 \sqcup \mathcal{M}_2$.
\end{proof}

\subsection{A two parameter cyclic code family}
\label{subsec:2-param-cyclic-code}
Choose two non-negative integers $s,t \geq 0$. Setting $r = 1, \; p = s+1, \; q = s+3$, and $N = s+t+4$ in the 4-parameter $(N,p,q,r)$ cyclic code family of the last subsection gives rise to a 2-parameter $(s,t)$ cyclic code family whose stabilizers are cyclic shifts of $ZXI^{\otimes s}XZI^{\otimes t}$, up to phase factors. Notice that in this case $q = p + 2r$, and $r < p+r < q$, and so by Lemma~\ref{lem:special-condition-subfamilyB} the quadruple $(N,p,q,r)$ is consistent. We have computed the $\llbracket N,K,D\rrbracket$ values for this 2-parameter family for all $0 \leq s,t \leq 9$, which we provide in Table~\ref{tab:cyclic_code_params}.

\begin{table}[ht]
    \centering
    \begin{tabular}{|c||c|c|c|c|c|c|c|c|c|c|}
         \hline
         $s$\textbackslash$t$ & 0 & 1 & 2 & 3 & 4 & 5 & 6 & 7 & 8 & 9 \\
         \hline\hline
        0 & 4,1,2 & 5,1,3 & 6,1,2 & 7,1,3 & 8,1,3 & 9,1,3 & 10,1,3 & 11,1,3 & 12,1,3 & 13,1,3 \\\hline
        1 & --- & 6,2,2 & 7,1,3 & 8,2,2 & 9,1,3 & 10,2,3 & 11,1,3 & 12,2,3 & 13,1,3 & 14,2,3 \\\hline
        2 & --- & --- & 8,1,2 & 9,1,3 & 10,1,2 & 11,1,3 & 12,1,3 & 13,1,3 & 14,1,4 & 15,1,3 \\\hline
        3 & --- & --- & --- & 10,2,2 & 11,1,3 & 12,2,2 & 13,1,5 & 14,2,3 & 15,1,3 & 16,2,3 \\\hline
        4 & --- & --- & --- & --- & 12,1,2 & 13,1,3 & 14,1,2 & 15,1,3 & 16,1,4 & 17,1,3 \\\hline
        5 & --- & --- & --- & --- & --- & 14,2,2 & 15,1,3 & 16,2,2 & 17,1,5 & 18,2,3 \\\hline
        6 & --- & --- & --- & --- & --- & --- & 16,1,2 & 17,1,3 & 18,1,2 & 19,1,5 \\\hline
        7 & --- & --- & --- & --- & --- & --- & --- & 18,2,2 & 19,1,3 & 20,2,2 \\\hline
        8 & --- & --- & --- & --- & --- & --- & --- & --- & 20,1,2 & 21,1,3 \\\hline
        9 & --- & --- & --- & --- & --- & --- & --- & --- & --- & 22,2,2 \\\hline
    \end{tabular}
    \caption{The code parameters $\llbracket N,K,D\rrbracket$ of cyclic quantum codes with stabilizer group generated by $ZXI^{\otimes s}XZI^{\otimes t}$ and its cyclic shifts. Interchanging $s$ and $t$ does not change the code parameters (see Lemma~\ref{lem:interchange-s-t}), so we show just the cases $t\ge s$.}
    \label{tab:cyclic_code_params}
\end{table}

One should note that Table~\ref{tab:cyclic_code_params} is symmetric about the diagonal -- so interchanging $s$ and $t$ does not change the code parameters $\llbracket N,K,D\rrbracket$. This is proved in the next lemma.

\begin{lemma}
\label{lem:interchange-s-t}
For any 2-parameter $(s,t)$ cyclic code, the $\llbracket N,K,D\rrbracket$ code parameters are invariant upon interchange of $s$ and $t$.
\end{lemma}

\begin{proof}
The cyclic code generated by cyclic shifts of $ZXI^{\otimes s}XZI^{\otimes t}$ is the same as that generated by cyclic shifts of $XZI^{\otimes t}ZXI^{\otimes s}$. Now we perform a Clifford transformation on every qubit that interchanges $X$ and $Z$ (note that Clifford transformations do not change commutativity of Pauli strings). This results in a new cyclic code generated by cyclic shifts of $ZXI^{\otimes t}XZI^{\otimes s}$. Clifford transformations are group isomorphisms of the Pauli group, thus $N$ and $K$ are unchanged. Moreover single qubit Clifford transformations don't change the distance of the code, and hence the same is true for their compositions. Thus $D$ is also unchanged under interchange of $s$ and $t$.
\end{proof}

The number of encoded qubits for the $(s,t)$ cyclic code family is explained by the following theorem.
\begin{theorem}
\label{thm:num-logical-ops-2param-cyclic-code}
For any 2-parameter $(s,t)$ cyclic code, the number of encoded qubits satisfies
\begin{equation}
\label{eq:num-logical-ops-2param-cyclic-code}
    K = 
    \begin{cases}
        2 & \;\; \text{if} \;\; s\text{ and }t \text{ are odd}, \\
        1 & \;\; \text{otherwise}.
    \end{cases}
\end{equation}
If $K=1$, the only subset of $\mathcal{S}$ that multiplies to $I$ up to phase factors is $\mathcal{S}$. If $K = 2$, there are two disjoint subsets $\mathcal{S}_{\text{even}}$ and  $\mathcal{S} \setminus \mathcal{S}_{\text{even}}$ of $\mathcal{S}$ that multiplies to $I$ up to phase factors, where $\mathcal{S}_{\text{even}} = \{S_{2n}: 0 \leq n \leq N/2 - 1\}$.
\end{theorem}

\begin{proof}
In this proof we will use the fact that for any integers $a,b,c \geq 1$, if $\gcd(a,b) = 1$ then $\gcd(ab,c) = \gcd(a,c) \gcd(b,c)$, which easily follows by considering the prime factorizations of $a,b,c$. Also the variables $x$ and $y$ appearing in this proof are always non-negative integers. We start by noting that when $s+1$ is odd, $\gcd (s+1, s+3) = 1$ which implies that $\lcm (s+1,s+3) = (s+1) (s+3)$, while if $s+1$ is even, then $\gcd (s+1, s+3) = 2$ and so $\lcm (s+1,s+3) = (s+1) (s+3) / 2$. Applying Theorem~\ref{thm:num-logical-ops-4param-cyclic-code} then gives
\begin{equation}
\label{eq:num-logical-ops-2param-cyclic-code-proof-1}
    K = 
    \begin{cases}
        \gcd (s+1,N) \; \gcd (s+3,N) / \gcd ((s+1)(s+3), N) & \;\; \text{if} \;\; s+1 \text{ is odd}, \\
        \gcd (s+1,N) \; \gcd (s+3,N) / \gcd \left( \frac{(s+1)(s+3)}{2}, N \right) & \;\; \text{otherwise}.
    \end{cases}
\end{equation}
In the first case ($s+1$ odd), by the fact mentioned above, we have $\gcd ((s+1)(s+3), N) = \gcd(s+1,N) \gcd(s+3,N)$, and so $K = 1$. Now assume that $s+1 = 2x$, i.e. is even, and so $(s+1)(s+3) / 2 = 2x(x+1)$. Then there are two subcases which we consider separately below.
\begin{enumerate}[(i)]
    \item $t$ is odd: In this case $N = s+t+4$ is even, and so suppose that $N = 2y$. Then we have $\gcd(s+1,N) = \gcd(2x,2y) = 2 \gcd(x,y)$, $\gcd(s+3,N) = \gcd(2x+2,2y) = 2 \gcd(x+1,y)$, and $\gcd((s+1)(s+3)/2, N) = \gcd(2x(x+1),2y) = 2 \gcd(x(x+1),y)$. But $\gcd(x,x+1) = 1$ and so $\gcd(x(x+1),y) = \gcd(x,y) \gcd((x+1),y)$. Plugging these into Eq.~\eqref{eq:num-logical-ops-2param-cyclic-code-proof-1} gives $K = 2$.
    
    \item $t$ is even: In this case $N$ is odd, and so $\gcd(2x,N) = \gcd(x,N)$, $\gcd(2(x+1),N) = \gcd(x+1,N)$, and $\gcd(2x(x+1),N) = \gcd(x(x+1),N) = \gcd(x,N) \gcd(x+1,N)$. Eq.~\eqref{eq:num-logical-ops-2param-cyclic-code-proof-1} then gives $K=1$.
\end{enumerate}
We have thus proved Eq.~\eqref{eq:num-logical-ops-2param-cyclic-code}. The last part of the theorem now follows directly from Corollary~\ref{cor:subsets-prod-I-from-K}.
\end{proof}

Theorem~\ref{thm:num-logical-ops-2param-cyclic-code} suggests that the case when $s$ and $t$ are both odd is special. We note another curious property that is also true for this case in the next lemma, which will find an use in the proof of Theorem~\ref{thm:allXYZ-2 param}.

\begin{lemma}
\label{lem:st-odd-coset-property}
For any 2-parameter $(s,t)$ cyclic code with both $s$ and $t$ odd, let $\mathcal{S}_{\text{even}} = \{S_{2n}: 0 \leq n \leq N/2 - 1\}$ and $\mathcal{S}_{\text{odd}} = \mathcal{S} \setminus \mathcal{S}_{\text{even}}$. Fix any qubit index $0 \leq n \leq N-1$. Then
\setmargins
\begin{enumerate}[(a)]
    \item If $n$ is even, $\mathcal{S}_{\text{even}}$ contains both stabilizers that contain $Z$ at index $n$, and $\mathcal{S}_{\text{odd}}$ contains both stabilizers that contain $X$ at index $n$.
    
    \item If $n$ is odd, $\mathcal{S}_{\text{even}}$ contains both stabilizers that contain $X$ at index $n$, and $\mathcal{S}_{\text{odd}}$ contains both stabilizers that contain $Z$ at index $n$.
    
    \item For all subsets $\mathcal{S}' \subset \mathcal{S}_{\text{even}}$ and $\mathcal{S}'' \subset \mathcal{S}_{\text{odd}}$, the products $\prod \mathcal{S}'$ and $\prod \mathcal{S}''$ do not contain $Y$ at any qubit index.
\end{enumerate}
\end{lemma}

\begin{proof}
All qubit indexes will be assumed to be modulo $N$ in this proof. Since $s,t$ are odd, first note that $N = s+t+4$ is even, $s+3$ is even, and $s+2$ is odd. Now for any qubit index $0 \leq n \leq N-1$, the two stabilizers containing $Z$ at index $n$ are $S_n$ and $S_{n-s-3}$, and the two stabilizers containing $X$ at index $n$ are $S_{n-1}$ and $S_{n-s-2}$ (note that these four stabilizers are distinct). If $n$ is even, we have $(n-s-3) \bmod N$ is even, $(n-1) \bmod N$ is odd, and $(n-s-2) \bmod N$ is odd; while if $n$ is odd, we have $(n-s-3) \bmod N$ is odd, $(n-1) \bmod N$ is even, and $(n-s-2) \bmod N$ is even. We have thus proved (a) and (b). For (c), notice that in order to obtain $Y$ at any fixed qubit index $n$, one needs to at least multiply a stabilizer containing $X$ at index $n$, and another stabilizer containing $Z$ at index $n$. But by (a) and (b) this is impossible for subsets of $\mathcal{S}_{\text{even}}$ or $\mathcal{S}_{\text{odd}}$.
\end{proof}

The next theorem states the conditions that determine the membership of the Pauli strings $X^{\otimes N}$, $Y^{\otimes N}$ and $Z^{\otimes N}$ in the stabilizer group $\langle \mathcal{S} \rangle$ for the $(s,t)$ cyclic code family. As with Lemma~\ref{lem:allX-allZ}, we will only determine membership upto phase factors.

\begin{theorem}
\label{thm:allXYZ-2 param}
For any 2-parameter $(s,t)$ cyclic code, we have the following:
\setmargins
\begin{enumerate}[(a)]
    \item $Z^{\otimes N}$ belongs to $\langle \mathcal{S} \rangle$ if and only if $s \equiv 3 \pmod{4}$ and $t \equiv 1 \pmod{4}$.
    
    \item $X^{\otimes N}$ belongs to $\langle \mathcal{S} \rangle$ if and only if $s \equiv 1 \pmod{4}$ and $t \equiv 3 \pmod{4}$.
    
    \item At most one of $X^{\otimes N}$, $Y^{\otimes N}$, or $Z^{\otimes N}$ can belong to $\langle \mathcal{S} \rangle$.
    
    \item $Y^{\otimes N}$ belongs to $\langle \mathcal{S} \rangle$ if and only if $s \equiv 0 \pmod{2}$ and $t \equiv 0 \pmod{2}$.
\end{enumerate}
\end{theorem}

\vspace{0.5cm}
\begin{proof}
\setmargins
\begin{enumerate}[(a)]
    \item Define $\alpha = (s+3) \bmod \gcd(s+1,N)$. Note that $\gcd(s+3, \gcd(s+1,N)) = \gcd(\alpha, \gcd(s+1,N))$, and its useful to also observe that $\gcd(s+1,N) / \gcd(s+3, \gcd(s+1,N)) \equiv 0 \pmod{2}$ implies $\alpha \neq 0$. By Lemma~\ref{lem:allX-allZ}(a), then $Z^{\otimes N}$ belongs to $\langle \mathcal{S} \rangle$ if and only if $\gcd(s+1,N) / \gcd(\alpha, \gcd(s+1,N)) \equiv 0 \pmod{2}$. We proceed to check the various cases. The variables $x$ and $y$ appearing in this proof are always non-negative integers.
    
    \begin{enumerate}[(i)]
        \item At least one of $s$ or $t$ is even: First assume that both $s,t$ are even. Then $N$ is even, while $s+1$ is odd, and so $\gcd(s+1,N)$ is odd. Now assume exactly one of $s,t$ is even. Then $N$ is odd, so again $\gcd(s+1,N)$ is odd. Now if $\gcd(s+1,N) = 1$ we have $\alpha = 0$, while if $\gcd(s+1,N) \geq 3$ then $\alpha = 2$. In the latter case this implies that $\gcd(\alpha, \gcd(s+1,N)) = \gcd(2, \gcd(s+1,N)) = 1$, and so $\gcd(s+1,N) / \gcd(\alpha, \gcd(s+1,N)) = \gcd(s+1,N) \equiv 1 \pmod{2}$. Thus $Z^{\otimes N}$ does not belong to $\langle \mathcal{S} \rangle$ in these cases.
        
        \item $s,t \equiv 1 \pmod{4}$; $s,t \equiv 3 \pmod{4}$; $s \equiv 1 \pmod{4} \text{ and } t \equiv 3 \pmod{4}$: We handle these three cases together. So assume that $s = 4x+1, t = 4y+1$ in the first case, $s = 4x+3, t = 4y+3$ in the second, and $s = 4x+1, t = 4y+3$ in the third. The values of $N$ are $4x+4y+6$, $4x+4y+10$, and $4x+4y+8$ for the three cases respectively. It then follows that
        \begin{equation}
            \gcd(s+1,N) = 
            \begin{cases}
                2 \gcd(2x+1,2y+2) &\;\; \text{if} \;\; s = 4x+1, t = 4y+1, \\
                2 \gcd(2x+2, 2y+3) &\;\; \text{if} \;\; s = 4x+3, t = 4y+3, \\
                2 \gcd(2x+1,2y+3) &\;\; \text{if} \;\; s = 4x+1, t = 4y+3.
            \end{cases}
        \end{equation}
        In each of the above cases $\gcd(s+1,N)$ is twice an odd number, so either $\gcd(s+1,N) = 2$ or $\gcd(s+1,N) \geq 6$. If $\gcd(s+1,N) = 2$ we have $\alpha = (s+3) \bmod 2 = 0$, since $s+3$ is even in all three cases; so $Z^{\otimes N}$ does not belong to $\langle \mathcal{S} \rangle$. If $\gcd(s+1,N) \geq 6$, then first we have $\alpha = (s+3) \bmod \gcd(s+1,N) = 2$, and so $\gcd(\alpha, \gcd(s+1,N)) = \gcd(2, \gcd(s+1,N)) = 2$ in all three cases. It follows that $\gcd(s+1,N) / \gcd(\alpha, \gcd(s+1,N)) = \gcd(s+1,N) / 2 \equiv 1 \pmod{2}$ in each case, and $Z^{\otimes N}$ does not belong to $\langle \mathcal{S} \rangle$.
        
        \item $s \equiv 3 \pmod{4} \text{ and } t \equiv 1 \pmod{4}$: Finally suppose that $s = 4x+3, t = 4y+1$; so $N = 4(x+y+2)$ and $\gcd(s+1,N) = \gcd(4x+4,4x+4y+8) = \gcd(4x+4,4y+4) = 4\gcd(x+1,y+1)$ is a multiple of $4$. This first implies that $\alpha = 2$, and so $\gcd(\alpha, \gcd(s+1,N)) = \gcd(2,4\gcd(x+1,y+1)) = 2$. This in turn implies that $\gcd(s+1,N) / \gcd(\alpha, \gcd(s+1,N)) = 4\gcd(x+1,y+1) / 2 = 2\gcd(x+1,y+1) \equiv 0 \pmod{2}$. Thus $Z^{\otimes N}$ belongs to $\langle \mathcal{S} \rangle$.
    \end{enumerate}
    
    \item Let $\mathcal{S}'$ be the set of stabilizers obtained by performing the Clifford transformation that interchanges $X$ and $Z$ at every qubit, for each stabilizer in $\mathcal{S}$. Thus it suffices to show that $Z^{\otimes N}$ belongs to $\langle \mathcal{S}' \rangle$ if and only if $s \equiv 1 \pmod{4}$ and $t \equiv 3 \pmod{4}$. Now the stabilizers in $\mathcal{S}'$ are generated by cyclic shifts of the Pauli string $ZXI^{\otimes t}XZI^{\otimes s}$. The result then follows from (a).
    
    \item For the sake of contradiction assume that any two of $X^{\otimes N}$, $Y^{\otimes N}$, or $Z^{\otimes N}$ belong to $\langle \mathcal{S} \rangle$. Then their product also belongs to $\langle \mathcal{S} \rangle$, and so in fact all three $X^{\otimes N}$, $Y^{\otimes N}$, and $Z^{\otimes N}$ belong to $\langle \mathcal{S} \rangle$. But by (a) and (b) $X^{\otimes N}$ and $Z^{\otimes N}$ cannot belong to $\langle \mathcal{S} \rangle$ simultaneously, which is a contradiction.
    
    \item We analyze the different cases separately.
    
    \begin{enumerate}[(i)]
        \item $s,t$ are even: In this case $N = s+t+4$ is even. Then for the subset $\mathcal{S}_{\text{even}} = \{S_{2n} : 0 \leq n \leq N/2 - 1 \}$ it is easily calculated that $\prod \mathcal{S}_{\text{even}} = \prod (\mathcal{S} \setminus \mathcal{S}_{\text{even}}) = Y^{\otimes N}$ up to phase factors.
        
        \item Exactly one of $s,t$ is even: In this case $N = s+t+4$ is odd, and so the Pauli strings $X^{\otimes N}$, $Y^{\otimes N}$, and $Z^{\otimes N}$ are pairwise anticommuting. Noticing that each stabilizer in $\mathcal{S}$ also commutes with each of $X^{\otimes N}$, $Y^{\otimes N}$, and $Z^{\otimes N}$, we then conclude that none of the Pauli strings $X^{\otimes N}$, $Y^{\otimes N}$, or $Z^{\otimes N}$ belong to $\langle \mathcal{S} \rangle$.
        
        \item $s,t$ are odd: If $s \equiv 3 \pmod{4} \text{ and } t \equiv 1 \pmod{4}$, then $Z^{\otimes N}$ belongs to $\langle \mathcal{S} \rangle$ by (a); while if $s \equiv 1 \pmod{4}$ and $t \equiv 3 \pmod{4}$, then $X^{\otimes N}$ belongs to $\langle \mathcal{S} \rangle$ by (b). In both cases $Y^{\otimes N}$ does not belong to $\langle \mathcal{S} \rangle$ by (c).  We are now left with two cases: either $s,t \equiv 1 \pmod{4}$, or $s,t \equiv 3 \pmod{4}$, both satisfying $N = s+t+4 \equiv 2 \pmod{4}$. We handle these cases together and show below that $Y^{\otimes N}$ does not belong to $\langle \mathcal{S} \rangle$.
        
        The proof is by contradiction, so let us suppose that $Y^{\otimes N}$ belongs to $\langle \mathcal{S} \rangle$, and $\prod \mathcal{S}' = Y^{\otimes N}$ up to phase factors, for some subset $\mathcal{S}' \subseteq \mathcal{S}$. Since both $s$ and $t$ are odd, consider the subsets $\mathcal{S}_{\text{even}}$ and $\mathcal{S}_{\text{odd}}$ defined in Lemma~\ref{lem:st-odd-coset-property}, and let $\mathcal{S}_{\text{even}}^{'} = \mathcal{S}' \cap \mathcal{S}_{\text{even}}$, and $\mathcal{S}_{\text{odd}}^{'} = \mathcal{S}' \cap \mathcal{S}_{\text{odd}}$. It must then be true that the products $\prod \mathcal{S}_{\text{even}}^{'}$ and $\prod \mathcal{S}_{\text{odd}}^{'}$ have the property that for any fixed qubit index $n$, exactly one of them contains $X$ and the other one contains $Z$ at index $n$ (this follows from Lemma~\ref{lem:st-odd-coset-property}(c) because otherwise one cannot get $Y$ at index $n$). Thus both $\prod \mathcal{S}_{\text{even}}^{'}$ and $\prod \mathcal{S}_{\text{odd}}^{'}$ don't contain $I_2$ at any qubit index. Moreover any two distinct Pauli strings $S_i, S_j \in \mathcal{S}_{\text{even}}^{'}$ must have disjoint supports, because otherwise by Lemma~\ref{lem:st-odd-coset-property}(a),(b) the product $\prod \mathcal{S}_{\text{even}}^{'}$ will contain $I_2$ at every qubit index in the intersection of the supports of $S_i$ and $S_j$. Since each stabilizer is supported at exactly four qubit indices this implies that $N \equiv 0 \pmod{4}$, which is a contradiction.
    \end{enumerate}
\end{enumerate}
\end{proof}

Finally, we can find the distance of the two-parameter cyclic toric codes when $s$ is close to $t$, the first three diagonals of Table~\ref{tab:cyclic_code_params}.

\begin{theorem}
Consider a two-parameter $(s,t)$ cyclic toric code. If $t=s$ or $t=s+2$, its code distance is two. If $t=s+1$, then its code distance is three.
\end{theorem}
\begin{proof}
Letting $D$ indicate the code distance, in all cases $D>1$ because no Pauli with weight one commutes with all stabilizers. 

In the $t=s$ case, we note that both $ZI^{\otimes s+1}XI^{\otimes s+1}$ and $Y^{\otimes s+2}I^{\otimes s+2}$ commute with all stabilizers. They also anticommute with each other. Thus, they are nontrivial logical operators and $D=2$.

In the $t=s+1$ case, we note $YI^{\otimes s+1}XXI^{\otimes s+1}$ and $X^{\otimes N}$ are nontrivial logical operators and therefore $D\le3$. We must also show there is no weight-two Pauli commuting with all stabilizers. A weight-two Pauli has the form $T=PI^{\otimes a}QI^{\otimes N-a-2}$ for single-qubit Paulis $P,Q\in\{X,Y,Z\}$. If $a<s+1$ or $N-a-2=2s-a+3<s+1$, then $T$ clearly must anticommute with some stabilizer. This leaves two cases $a=s+1$ and $a=s+2$. In the former case, commuting with all stabilizers requires that $Q=X$ and also that $P\otimes Q$ commutes with $Z\otimes Z$ and with $Z\otimes X$, which is clearly impossible. In the latter case, the analogous argument shows $P=X$ and $P\otimes Q$ commutes with $Z\otimes Z$ and with $X\otimes Z$, again a contradiction.

In the $t=s+2$ case, we note that $XI^{\otimes s+2}XI^{\otimes s+2}$ and $Y^{\otimes s+2}ZXI^{\otimes s+2}$ are a pair of anticommuting logical operators, implying $D=2$.
\end{proof}

%% file: appendix/medial-graphs-and-trisection.tex
\section{Embedding of the medial graph}\label{app:medial-graphs-and-trisection}
The goal here is to prove some properties of medial graphs claimed in Section~\ref{subsec:relate_to_homological}. See Definition~\ref{def:medial_graph} for the definition of a medial rotation system and, by extension, a medial graph.

\begin{lemma}
The rotation system $R=(H,\lambda,\rho,\tau)$ and its medial rotation system $\widetilde R=(\widetilde H,\widetilde\lambda,\widetilde\rho,\widetilde\tau)$ embed in the same manifold.
\end{lemma}
\begin{proof}
We show that both rotation systems have the same orientability and the same genus. 

For orientability, we show $R$ is orientable if and only if $\widetilde R$ is. Start by assuming $R$ is orientable. Therefore, we can partition $H$ into two sets $H_{\pm1}$ such that $\lambda$, $\rho$, $\tau$ applied to any element of $H_{k}$ takes it to an element of $H_{-k}$. Define two sets $\widetilde H_{\pm1}\subseteq\widetilde H=H\times\{-1,1\}$ such that $(h,j)\in\widetilde H_{k}$ if and only if $h\in H_{jk}$. These two sets clearly partition $\widetilde H$. By definition, if $(h,j)\in\widetilde H_k$ then $\widetilde\tau(h,j)=(h,-j)\in\widetilde H_{-k}$. Also, by using the orientability of $R$, $\widetilde\lambda(h,j)=(\rho(h),j)\in\widetilde H_{-k}$ and $\widetilde\rho(h,j)$ is either $(\lambda(h),j)$ or $(\tau(h),j)$ both of which are in $\widetilde H_{-k}$. We conclude that $\widetilde R$ is orientable. The other direction -- if $\widetilde R$ is orientable, then so is $R$ -- follows the same argument in reverse.

The genus question boils down to counting the numbers of vertices, edges, and faces. Let $V,E,F$ and $\widetilde V,\widetilde E,\widetilde F$ be the sets of vertices, edges, faces for $R$ and $\widetilde R$. Then, we note
\begin{align}
|\widetilde V|&=|E|,\\
|\widetilde E|&=\sum_{v\in V}\text{deg}(v)=2|E|,\\
|\widetilde F|&=|V|+|F|.
\end{align}
Therefore, $\widetilde\chi=|\widetilde V|-|\widetilde E|+|\widetilde F|=\chi$ and thus the genuses are the same.
\end{proof}

\begin{lemma}
Let $R'=(H',\lambda',\rho',\tau')$ be a rotation system. Then, $R'=\widetilde R$ is the medial rotation system of some rotation system $R=(H,\lambda,\rho,\tau)$ if and only if $R'$ is checkerboardable and 4-valent.
\end{lemma}
\begin{proof}
We begin by showing a medial rotation system is 4-valent and checkerboardable. The degree of a vertex is the size of an orbit of $\widetilde\rho\widetilde\tau$. Clearly, any orbit must have even size $L$ because $\widetilde\rho\widetilde\tau$ applied to $(h,j)\in\widetilde H= H\times\{-1,1\}$ flips the sign of $j$. Applying $(\widetilde\rho\widetilde\tau)^L$ to $(h,j)$ gives $((\tau\lambda)^{L/2}h,j)$ if $j=1$ and $((\lambda\tau)^{L/2}h,j)$ if $j=-1$. By definition, $(\lambda\tau)$ has order two, so $L=4$ is the size of any orbit of $\widetilde\rho\widetilde\tau$ and the degree of any vertex in the medial graph.

Checkerboardability demands that there is a partition of $\widetilde H$ into two sets $\widetilde H_{\pm1}$ such that $\tau$ applied to any element $(h,j)\in\widetilde H_k$ gives an element of $\widetilde H_{-k}$ and that applying $\rho$ or $\lambda$ to $(h,j)$ gives an element of $\widetilde H_k$. This clearly achieved by putting $(h,j)$ for any $h$ into the set $\widetilde H_j$.

Now we must show the other direction, that a 4-valent, checkerboardable rotation system $R'$ is the medial graph of another rotation system. Because $R'$ is checkerboardable there is a partition $H'=H'_{-1}\sqcup H'_1$ where $\tau'$ moves elements from $H'_{k}$ to $H'_{-k}$ while $\lambda'$ and $\rho'$ preserve the sets. We let $H=\left\{(h,\tau'(h)):h\in H'_1\right\}$ and define permutations of $H$ called $\rho$, $\lambda$, $\tau$ so that
\begin{align}
\lambda(h,\tau'(h))&=(\tau'\rho'\tau'(h),\rho'\tau'(h)),\\
\rho(h,\tau'(h))&=(\lambda'(h),\tau'\lambda'(h)),\\
\tau(h,\tau'(h))&=(\rho'(h),\tau'\rho'(h)).
\end{align}
Each of these are clearly involutions. To be a legitimate rotation system, $\lambda$ and $\tau$ must also commute. They do so because $R'$ is 4-valent, i.e.~$(\rho'\tau')^4=1$.

We claim the medial rotation system of $R=(H,\lambda,\rho,\tau)$ is $R'$. We must therefore establish a one-to-one map $\phi:\widetilde H=H\times\{-1,1\}\rightarrow H'$. We do this by defining, for all $h\in H'$ and $j\in\{-1,1\}$,
\begin{align}
\phi((h,\tau'(h)),j)&=\bigg\{\begin{array}{ll}h,&j=1\\\tau'(h),&j=-1\end{array},\\
\phi^{-1}(h)&=\bigg\{\begin{array}{ll}((h,\tau'(h)),1),&h\in H'_1\\((\tau'(h),h),-1),&h\in H'_{-1}\end{array}.
\end{align}
Then one can easily show $\phi(\widetilde\lambda(\phi^{-1}(h)))=\lambda'(h)$ for all $h$ and likewise for $\lambda$ replaced with $\rho$ or $\tau$. This completes the proof.
\end{proof}


%% file: appendix/GAP_example.tex
\section{Using GAP to find normal subgroups}\label{app:GAP_example}
To elucidate the process of finding hyperbolic codes, we give an example of finding low-index normal subgroups in GAP. In particular, we use the LINS package \cite{GAP_LINS}, version $0.5$. For the sake of illustration, we focus on the index 160 normal subgroup $S$ that defines the $\llbracket20,5,4\rrbracket$ code in Fig.~\ref{fig:hyperbolic_codes}(a) and the first entry of Table~\ref{tab:noncheckerboardable_hyperbolic_codes}.

\begin{align*}
\begin{array}{lr}
\mathrm{LoadPackage}(``\mathrm{LINS}");;&\\
F:=\mathrm{FreeGroup}(``l",``r",``t");;&\text{//}l=\lambda, r=\rho, t=\tau\\
\mathrm{AssignGeneratorVariables}(F);;&\text{//makes $l,r,t$ into variable names}\\
G:=F/[l^2, r^2, t^2, (l*r)^5, (l*t)^2, (r*t)^4];;&\text{//quotient group, face-degree 5, vertex-degree 4}\\
H:=\mathrm{LowIndexNormalSubgroupsSearchForAll}(G,200);&\text{//LINS graph of normal subgroups, index $\le200$}\\
L:=\mathrm{List}(H);&\text{//convert LINS graph into a list of LINS nodes}\\
S:=\mathrm{Grp}(L[9]);;&\text{//$S$ happens to be the 9th entry in the list}\\
\mathrm{IsNormal}(G,S);&\text{//check it is a normal subgroup; returns true}\\
\mathrm{GeneratorsOfGroup}(S);&\text{//display (overcomplete) relations of the group}
\end{array}
\end{align*}


Reducing overcomplete generating sets to independent ones was done by hand and verified by drawing the codes as those in Fig.~\ref{fig:hyperbolic_codes} are drawn.